%% file: main.tex
\documentclass[11pt]{article}
\usepackage{fullpage}
\usepackage{url}
\usepackage[colorlinks=true,linkcolor=Emerald,citecolor=RoyalBlue]{hyperref}
\usepackage{amsmath,nccmath, amsfonts,amsthm,multirow,mdframed,amssymb}
\usepackage{thmtools,thm-restate}

\usepackage{times}
\usepackage{paralist}
\usepackage{graphicx}
\usepackage{floatpag}
\usepackage{float}
\DeclareMathOperator{\supp}{supp}
\usepackage{bbold}
\usepackage{enumitem}
\usepackage{subcaption} 
\usepackage{soul}
\usepackage{comment}
\usepackage{calc}
\usepackage{upgreek}
\usepackage{mathtools}
\usepackage{color}
\usepackage[dvipsnames]{xcolor}
\usepackage{xspace}
\usepackage{tcolorbox}
\usepackage{cleveref}
\usepackage{stackengine}
\usepackage{tocloft}
\usepackage{physics}
\usepackage{quantikz}

\stackMath
\allowdisplaybreaks

\DeclareMathOperator*{\E}{\mathbb{E}}

\def \R {\mathbb{R}}

\def \N {\mathbb{N}}

\def \cB {\mathcal{B}}

\def \cE {\mathcal{E}}

\def \cK {\mathcal{K}}

\def \cP {\mathcal{P}}
\def \cQ {\mathcal{Q}}
\def \cR {\mathcal{R}}
\def \cS {\mathcal{S}}

\def \cW {\mathcal{W}}
\def \Id {\mathrm{Id}}

\def \sF {\mathsf{F}}

\def \poly {\mathrm{poly}}
\def \polylog {\mathrm{polylog}}
\def \eps {{\varepsilon}}

\def \maj {\operatorname{Maj}}

\def\ggg{\gtrsim}
\def\lll{\lesssim}

\def \LC {\text{LC}}

\def \deg {\text{deg}}

\def \pl {\operatorname{polylog}}
\def \pll {\operatorname{polyloglog}}
\def \plll {\operatorname{polylogloglog}}

\def \Enc {\mathrm{Enc}}
\def \REnc {\mathrm{REnc}}
\def \Rep {\mathrm{Rep}}
\def \NRep {\mathrm{NRep}}
\def \BL {\mathrm{BL}}
\def \FL {\mathrm{FL}}

\def \syn {\mathrm{Syndrome}}

\renewcommand{\leq}{\leqslant}
\renewcommand{\le}{\leqslant}
\renewcommand{\geq}{\geqslant}
\renewcommand{\ge}{\geqslant}

\newcommand{\encq}{\mathrm{Enc}_\mathcal{Q}}
\newcommand{\encc}{\mathrm{Enc}_\mathrm{rep}}

\usepackage{tikz}

\newtheorem{theorem}{Theorem}[section]

\newtheorem{fact}[theorem]{Fact}

\newtheorem{proposition}[theorem]{Proposition}
\newtheorem{prop}[theorem]{Proposition}

\newtheorem{lemma}[theorem]{Lemma}
\newtheorem*{lemma*}{Lemma}
\newtheorem*{theorem*}{Theorem}
\newtheorem{claim}[theorem]{Claim}
\newtheorem*{claim*}{Claim}

\newtheorem{assumption}[theorem]{Assumption}

\newtheorem{remark}[theorem]{Remark}
\newtheorem*{remark*}{Remark}
\newtheorem{definition}[theorem]{Definition}
\theoremstyle{definition}

\newtheorem{agr-test}{Agreement-Test}
\newtheorem{list-agr-test}{List-Agreement-Test}
\newtheorem{protocol}{Protocol}

\usetikzlibrary{decorations.pathmorphing}

\tikzset{snake it/.style={decorate, decoration=snake}}

\makeatletter
\let\c@fconjecture\c@conjecture
\makeatother

\makeatletter
\let\c@fconj\c@conj
\makeatother

\title{\fontsize{15.5}{19}\selectfont Gap Amplification for Local Hamiltonians with Combinatorial Soundness}

\author{
Mitali Bafna\textsuperscript{1}\thanks{mitalib@cs.washington.edu}
\qquad
Quynh T. Nguyen\textsuperscript{2}\thanks{thequynhqb@pm.me}
\qquad
Tina Zhang\textsuperscript{3}\thanks{tinaz@mit.edu}
\\[0.5em]
{\normalsize
\textsuperscript{1}University of Washington
\qquad
\textsuperscript{2}Harvard \& UC Berkeley
\qquad
\textsuperscript{3}MIT
}
}

\date{\vspace{-5ex}}
\begin{document}
\maketitle
\thispagestyle{empty}

\begin{abstract}
    The quantum PCP conjecture is one of the major open problems in quantum complexity theory. It has resisted attack in part because many primitives used in the proof of the classical PCP theorem, such as \emph{locality-preserving gap amplification} and \emph{alphabet reduction}, have no obvious quantum analogues due to quantum no-cloning.

    Locality-preserving gap amplification is a procedure that takes as input a local Hamiltonian problem instance and produces a new instance with a larger promise gap, without increasing the locality of the Hamiltonian, and instead moderately increasing its local qudit dimension. Obtaining this kind of control over the locality during gap amplification is critical to the success of many known strategies for proving the classical PCP theorem.
    
    In this work, we put forth the first known viable template for quantum locality-preserving gap amplification, and we prove that our procedure amplifies the \emph{combinatorial gap} of local Hamiltonians. Our work introduces a new framework for reasoning about quantum gap amplification in terms of fault-tolerant computation, and illuminates a route toward importing one of the central ingredients in classical PCPs into the quantum setting. In particular, we build upon ideas from the recent classical PCP of~\cite{BafnaMV25} based on high-dimensional expanders, and the circuit-to-Hamiltonian construction of~\cite{Anshu_2024}, in addition to  several recent advances in quantum coding theory.

\end{abstract}

\newpage

{
\pagenumbering{gobble}
\setcounter{tocdepth}{2}
\tableofcontents
\clearpage
}

\pagenumbering{arabic}
\setcounter{page}{1}
\pagestyle{plain}

\input{intro}

\input{prelims}

\input{mpc}

\input{routing}

\input{ckt-to-protocol-v2}

\input{nc-protocol}
\input{protocol-to-ham}

\input{reduction}

\paragraph{Acknowledgements:} We
thank Anurag Anshu and Umesh Vazirani for helpful discussions. Part of this work was done while the three authors were visiting the Simons Institute at UC Berkeley.

\paragraph{AI Disclosure:} None of the ideas in this paper were generated by AI. We used ChatGPT 5.4 and Claude Opus 4.6 to assist with elementary calculations and choosing parameters in proofs of lemmas in Sections 4, 8 and Appendix A.

\bibliographystyle{alpha}
\bibliography{references}

\appendix
\input{dl-new}

\end{document}

%% file: intro.tex
\section{Introduction}
The classical PCP theorem~\cite{AS, ALMSS, Dinur07} is one of the most celebrated achievements of classical complexity theory, and the quantum PCP conjecture is one of the major open problems of quantum complexity theory. The vast literature on classical PCPs contains a variety of ways to prove the classical PCP theorem. Many of these proof strategies rely on the following two steps:
\begin{enumerate}
\item \emph{Locality-preserving gap amplification}, in which the promise gap of some family of classical constraint satisfaction problems (CSPs) is amplified at the cost of blowing up the alphabet size moderately, but without otherwise increasing the locality of the constraints in the CSPs,
\item \emph{Alphabet reduction}, in which the alphabet size of the constraint satisfaction problems is reduced (to offset the alphabet blowup from the gap amplification step) and the promise gap is shown not to shrink too much in the process. %
\end{enumerate}

It is difficult to see how to approach the quantum PCP conjecture in part because neither of these two steps translates easily to the quantum setting. Quantum gap amplification has been studied, but so far every known quantum gap amplification procedure is \emph{locality-increasing}~\cite{AharonovALV09, bergamaschi2025derandomised}. This makes existing quantum gap amplification procedures difficult to use in constructions which aim to approach the quantum PCP conjecture using the two-step template above, because it is difficult (even classically) to reduce the locality of a CSP family in a black-box way while preserving its promise gap, while it is considerably easier to reduce its alphabet size.\footnote{In fact, achieving locality-increasing gap amplification is almost trivial in the classical setting. Quantumly, even achieving locality-increasing gap amplification requires some finesse, and the proof techniques which led to the successful analysis of locality-increasing gap amplification~\cite{AharonovALV09} have since seen widespread application.} The problem of quantizing alphabet reduction appears to have received less attention, but it too seems likely to present serious technical difficulties.

Locality-\emph{preserving} gap amplification has eluded quantization due in part to the well-known fact that quantum information cannot be cloned. (For a thorough survey of the obstacles that no-cloning creates for approaches to quantum PCP, the reader can refer to \cite{aharonov2013guest}.) Classical locality-preserving gap amplification usually relies on asking the prover to encode the original (input) assignment redundantly into a larger proof, and then locally enforcing consistency between different appearances of the same input variable. Redundant encoding produces amplification because (for an honest prover) each large variable in the output assignment contains enough information to check several input constraints rather than only one, so for the same number of variables queried the verifier is more likely to notice violations. Classical \emph{parallel repetition} is a simple example of this idea: each large output variable represents an assignment to a \emph{group} of $k$ input variables, and every possible group of size $k$ appears once. The verifier makes two types of queries: either it chooses a large variable (representing a group of $k$ input variables) uniformly at random and checks all the input constraints that involve those input variables, or it checks that different `large variables' in the new proof are consistent with some global underlying assignment by choosing two at random and comparing them on their intersection.

Na\"ively speaking, this strategy is completely unworkable in the quantum setting, because the original ``assignment'' (the ground state of a local Hamiltonian) could be a highly entangled state, in which case the reduced state on any small subsystem may be close to maximally mixed, so that local views contain very little information about the global state and local consistency checks cannot certify that different pieces come from the same underlying assignment. Moreover, there are concrete no-go results against locality-preserving quantum gap amplification in certain parameter regimes. For example, the 2013 results of Brandao and Harrow~\cite{brandao2013product} show that parallel repetition is impossible for $\mathsf{QMA}$-hard quantum Hamiltonian families unless either (1) the quantum PCP conjecture is false, or (2) $\mathsf{QMA} = \mathsf{NP}$. This surprising result led to a marked increase in skepticism about the quantum PCP conjecture, because parallel repetition is in some sense the most basic classical locality-preserving gap amplification for CSPs and a natural starting point for the kinds of locality-preserving gap amplification procedures that are used for the PCP theorem in the classical setting. In fact, the gap amplification procedures used in proofs of the classical PCP theorem tend to be `derandomized' versions of parallel repetition, in which expander walks or other derandomization techniques are used in order to avoid the unrecoverable blowup in the randomness complexity of the verifier which would occur as a result of applying na\"ive parallel repetition as a gap amplification procedure.

For these reasons, despite the success of \cite{AharonovALV09, bergamaschi2025derandomised} in quantizing locality-increasing gap amplification, quantizing locality-preserving gap amplification remains a major open problem. In this work, we present substantial progress towards this objective: we set forth arguably the first meaningful template for a quantum locality-preserving gap amplification procedure in the literature, and we prove that it has \emph{combinatorial soundness}, defined below,

\begin{definition}[Combinatorial energy]
    Let $H$ be a normalized sum of weighted projective terms: $H = \sum_{i=1}^m p_i \Pi_i$. We say that a Hamiltonian $H$ has combinatorial energy $\epsilon$ if for any state $\ket{\psi}$, if we sample a random constraint $\Pi_i$ from the probability distribution with $\Pr[\Pi_i] = p_i$, then the probability that $\ket{\psi}$ violates (has nonzero trace against) $\Pi_i$ is at least $\epsilon$.
\end{definition}

\begin{theorem}[Informal main theorem]\label{thm:main-intro}
Let $H_\mathrm{in}$ be a normalized sum of projective terms acting on $n$ qubits. Further suppose that $H_\mathrm{in}$ is $O(1)$-local with degree $O(1)$, i.e., each qubit is part of $O(1)$ terms. There is a polynomial-time reduction that maps any such Hamiltonian
$H_{\mathrm{in}}$ to a $2$-local Hamiltonian $H$ on $O(n)$ qudits with the following properties:
\begin{itemize}
\item  \textbf{Completeness}: if $H_{\mathrm{in}}$ has ground energy 0,
then the ground energy of $H$ is at most $2^{-\operatorname{polylog}(n)}$.
\item \textbf{Combinatorial soundness}: If $H_{\mathrm{in}}$ has normalized ground energy
$\eps_{\mathrm{in}}$, then $H$ has combinatorial energy
\[
\eps =\min\{\eps_{\mathrm{in}} \cdot \poly\log(n),\,\eps_{\max}(n)\},
\]
for $\eps_{\max}(n) =  \frac{1}{\mathrm{polylogloglog}\ n}$.
\item The degree and alphabet blowups are as follows:
\begin{itemize}
    \item \textbf{Low-gap regime.}
    If the amplified promise gap is $O(1/\log n)$, then
    $H$ has degree and log-alphabet size (log qudit dimension) $\operatorname{polylog}(n)$.
    \item \textbf{High-gap regime.}
    If the amplified promise gap is $\Omega(1/\log n)$, then $H$ has degree and log-alphabet size (log qudit dimension) $2^{\operatorname{polyloglog}(n)}$.
\end{itemize}
\end{itemize}
Furthermore, assuming asymptotically good quantum locally testable codes, one can take $\eps_{\max}(n)=\Theta(1)$.\footnote{It will be clear from our proof overview why we use good qLTCs in the highest soundness regime: see~\Cref{sec:solving-post-prob-1}. We believe however that this dependence is an artifact of the current proof strategy and can be removed.}
\end{theorem}

Combinatorial soundness is a restricted form of soundness in which, instead of showing that (in the no case) there does not exist any state that has energy less than $\epsilon$ with respect to our Hamiltonian, we show that \emph{every state violates at least $\epsilon$ fraction of the Hamiltonian terms}. This notion is technically more tractable because it reduces the analysis to states that satisfy at least $(1-\epsilon)$ fraction of the terms \emph{exactly}. This weakening has often served as a useful first step in past work: for instance, the proof of combinatorial NLTS~\cite{anshu2022construction} was in some sense a precursor to a recent circuit-to-Hamiltonian construction \cite{Anshu_2024}. As such, combinatorial soundness is a natural milestone for any candidate gap-amplification framework, and a construction that achieves it provides meaningful evidence for the viability of the underlying approach. We discuss the prospects for generalizing our analysis beyond combinatorial soundness in \Cref{sec:discussion}.

Our new gap amplification procedure combines two main ingredients from prior work: the classical PCP of~\cite{BafnaMV25} based on \emph{high-dimensional expanders} (HDX)\footnote{Although the main  advance that \cite{BafnaMV25} made was to construct PCPs in the small-soundness regime, their work also achieves en route gap amplification in the inverse-polylog-to-constant-gap regime. The gap amplification results of~\cite{BafnaMV25} were recently generalized by~\cite{BT} to regimes other than the inverse-polylog-to-constant regime.}, and the circuit-to-Hamiltonian procedure of~\cite{Anshu_2024}. One of our key conceptual contributions, which helps us overcome the no-cloning barrier, is to observe that locality-preserving gap amplification needs not redundantly repeat \emph{the witness itself} across the large-alphabet variables. Inspired by~\cite{BafnaMV25}’s view of a PCP as a \emph{tableau of a protocol} that fault-tolerantly computes whether the original CSP constraints are violated, we observe that it suffices to repeat and distribute only this \emph{classical} ``violation information'' across the large-alphabet variables. Of course, there needs to be some way to certify that this ``violation information'' was computed from the original witness. Thus, using some of the tools from~\cite{BafnaMV25}, we firstly build a protocol on an HDX communication graph that fault-tolerantly computes whether the original Hamiltonian's constraints are violated or not. Then, like in~\cite{BafnaMV25}, we get a PCP by converting this protocol into  a local Hamiltonian. This basic template makes it possible to use \emph{secure quantum multiparty computation}, based on the \emph{secret sharing} rather than the repetition of the witness, to compute and distribute the violation information fault-tolerantly.

This idea gives us arguably the first viable template for locality-preserving gap amplification in the quantum setting, but proving soundness is nontrivial. The main difficulty is that, while converting a protocol into a Hamiltonian can be achieved using a minor modification of Cook-Levin in the classical setting~\cite{BafnaMV25}, analyzing the quantum equivalent of this step is difficult. We use the circuit-to-Hamiltonian procedure of~\cite{Anshu_2024} because the standard Feynman-Kitaev construction is not robust enough for our purposes. The~\cite{Anshu_2024} construction guarantees that combinatorial low-energy states of its output Hamiltonian encode (adversarially) faulty executions of our protocol, but the~\cite{Anshu_2024} fault model allows \emph{coherent} combinations of faults that can effectively implement non-unitary operations such as postselection. Thus our protocol must remain sound against adversarial players capable of \emph{postselection attacks}. Dealing with this type of attack is technically challenging, and our solution requires a number of advances which we believe could be more generally useful: (1) we develop a novel method to perform fault-tolerant computation (in particular fault-tolerant classical readout) in the presence of~\cite{Anshu_2024}-style faults by injecting \emph{random noise} and \emph{friendly postselection}, which may be useful more broadly in applications of~\cite{Anshu_2024}; (2) we prove a version of the detectability lemma from~\cite{AharonovALV09} that is robust even in the presence of adversarial postselection; (3) we develop a template for quantum multiparty computation with \emph{zero error} (perfect security), using as ingredients several recent coding-theoretic advances~\cite{dinur2024expansion, nguyen2025good, golowich2025asymptotically}, which solves an open problem of independent interest~\cite[Section 6]{crepeau2002secure}. We elaborate on how these ingredients fit together under our proof strategy in~\Cref{sec:overview}.

Finally, we remark that the work of~\cite{BT} shows a classical analogue of~\Cref{thm:main-intro}, and also constructs classical PCPs (albeit with  slightly super-polynomial $n^{\exp(\log^*(n))}$ size) by combining this gap amplification procedure with recent PCPs of proximity (PCPP) constructed by~\cite{amireddyBSSW2025}.\footnote{We use sampling bounds on high-dimensional expanders from~\cite{BT}; see~\Cref{lem:concentration-improved}.} \cite{BT}'s approach is analogous to Dinur's strategy of interleaving gap amplification with composition with the Hadamard PCPP and iterating these two steps~\cite{Dinur07}. However, it is not known how to quantize either proof composition or PCPPs, and so a corresponding route to quantum PCPs remains out of reach (see~\Cref{sec:discussion}). Even so, this observation indicates that \Cref{thm:main-intro} is already of a strength comparable to what suffices for the classical PCP theorem, which has not been true for any previous quantum gap amplification procedure.

\section{Proof Overview}
\label{sec:overview}
\subsection{Obstacles to Quantizing Gap Amplification}

Classical locality-preserving gap amplification typically works by \emph{redundantly encoding} the witness: each output variable stores information about many input variables, and local checks enforce consistency among different output variables that involve the same input variables. This paradigm produces amplification because it increases the ``per-variable'' / ``per-query'' visibility of violated input constraints: each output variable contains enough information to check many input constraints, and so, assuming that the output-variable assignments are consistent with some global assignment to the input CSP (which condition must be enforced by local equality checks between output variables), the verifier can detect violations of the input constraints more often without increasing the number of queries it makes. Parallel repetition has this form, as does Dinur's gap amplification procedure~\cite{Dinur07}, where a large output variable represents an assignment to the $t$-\emph{neighborhood} of an input variable. In both cases, a single input variable appears many times in an honest output assignment.

The first obstacle to quantizing this paradigm is that there is no useful quantum analogue of repeating witness bits across many large variables. Repeating reduced states of an entangled witness is not useful for the reasons discussed in the introduction. One could instead try to apply repetition encoding \emph{in superposition}, for example by replacing a repeated bit by a GHZ state, but this is neither locally testable nor compatible with the need to measure the same qubit in different bases.
One might also try to replace repetition with a quantum error-correcting code. However, for a quantum encoding to play the role that repetition plays classically, each local view of the codeword would have to suffice to recover the encoded qubit, so the code would have to tolerate erasure of all but one local view. This is impossible because of local indistinguishability of quantum codes: any view of a quantum codeword smaller than the distance of the code looks maximally mixed (over the codespace), so it contains no  information about the logical state.

\subsection{A Template for Gap Amplification: Multiparty Computation on an HDX}
\label{sec:overview-template}

These obstacles suggest abandoning the redundant-encoding view of locality-preserving gap amplification. The recent work of~\cite{BafnaMV25} (BMV) offers an alternative paradigm. Their principal goal  is to achieve improved classical PCPs in the small soundness regime. One corollary of their work, however, is a locality-preserving gap amplification procedure that can amplify the promise gap of an input 2CSP family from $1/\polylog(n)$ to $\Omega(1)$ using high-dimensional expanders (HDX).

At a high level, they obtain this amplification by designing a \emph{distributed computing protocol} on an HDX graph $G$ (i.e.~parties/players in the protocol are associated with vertices of $G$, and all communication occurs on the edges of this graph). They then convert this distributed protocol into a CSP with amplified gap using a modification of Cook--Levin. We describe the overall structure of this distributed computing protocol here, and elaborate on the protocol-to-CSP conversion in `Converting the classical fault-tolerant protocol into a gap-amplified 2CSP' later in this overview.

The overall goal of the distributed protocol is to distribute a witness for the input 2CSP fault-tolerantly among the players; this is in some sense an interactive analogue of redundantly encoding the input witness across large variables. Fix any input 2CSP, whose gap we want to amplify, and fix any witness $x$ that violates $1/\pl n$ fraction of the constraints to the input 2CSP. Each player in the distributed protocol begins by receiving as input some $\pl n$--sized subset of the bits of $x$. The players then share their witness information with each other by executing a sequence of \emph{fault-tolerant routing protocols} so that, at the end of the overall protocol, each player holds the correct evaluations on $x$ of $\pl n$ constraints of the input 2CSP. This holds true even when communication across a constant fraction of `bad' edges in $G$ may be arbitrarily corrupted (hence `fault tolerant routing'). The properties of the HDX $G$ allow the constraints to be `spread out' among the players in a way such that an $\Omega(1)$ fraction of the players see at least one violation. These players will output `NO' at the end of the protocol. (We will achieve this guarantee in the quantum world: see~\Cref{lem:protocol-informal}.) This increased \emph{per-player visibility} of violations is the source of the $\pl n$--factor gap amplification: when the distributed protocol is converted into an output CSP (using a Cook--Levin-type reduction, which we discuss later), each player in the protocol will be associated with one variable in the output CSP, and so a constant fraction of players outputting NO will translate to a constant fraction of large variables in an assignment to the output CSP witnessing at least one input violation.

\paragraph{High-dimensional expanders} To achieve these guarantees, BMV exploit the sampling properties of high-dimensional expanders. For this overview, it suffices to think of a $d$-dimensional high-dimensional expander (HDX) as a sparse graph $G$ with many overlapping $d$-sized cliques. We will use $X(d)$ to denote the set of $d$-cliques in $G$. HDXs possess, among others, the following nice properties:
\begin{enumerate}
    \item \emph{Vertex-clique sampling.} As before, let $x$ be a witness for the input 2CSP. If the input constraints are distributed so that each clique in $X(d)$ is associated with one constraint (in the sense that every vertex/player in that clique ends the protocol knowing the value of the constraint evaluated on $x$), and $x$ violates a $1/\pl n$ fraction of the input constraints, then $\Omega(1)$ fraction of the vertices are a part of some clique responsible for a violated constraint.
    \item \emph{Clique-edge sampling.} Even if some constant (say $0.001$ fraction) of edges in $G$ are corrupted, only a $1/\pl n$ fraction of cliques will contain more than $0.1$ fraction of corrupted edges. This sampling property allows HDXs to support the kind of fault-tolerant routing protocols that BMV use.
\end{enumerate}
Sparse constructions of HDXs in which each vertex belongs to $2^{\pll n}$ cliques are known~\cite{LSV2, kaufman2018construction}, and any such construction suffices here.

\paragraph{Idea 1: Introducing multiparty computation}
We firstly build on the distributed-computing viewpoint from BMV in order to create a quantum-friendly distribution protocol for a witness to an input 2-local Hamiltonian, and then we discuss how we convert our protocol into an output local Hamiltonian with larger promise gap. 
The protocol in BMV still begins by repeating each bit of the input witness $x$ across some clique of players in the HDX, so it does not by itself resolve the quantum obstruction.

The key observation is that one need not repeat the witness $x$ itself. To increase the \emph{per-player visibility} of violations of $x$, we only need many players to be able to  obtain consistent copies of the input constraints evaluated on $x$. This suggests replacing repeated copies of $x$ with a more indirect mechanism: the relevant information about $x$ should be made available to the players only in an encoded or shared form, from which they jointly compute the violation information, and only that classical information is then copied. One must still ensure that this violation information was obtained from a correct evaluation of the input constraints on $x$, rather than being supplied arbitrarily by the prover; in the classical BMV protocol this is automatic because the players themselves locally perform the computation from the repeated inputs. The quantum analogue should therefore let the players jointly evaluate the constraint on shared quantum data in a way that guarantees that the output bit is exactly the measurement outcome of that constraint.

This suggests using multiparty computation (MPC). In an MPC protocol, mutually distrustful players jointly evaluate a circuit with one-bit output on a \emph{secret-shared} input, where each player holds one share. The adversarial fault-tolerance of MPC guarantees that, provided the number of dishonest players is small enough, the shared output bit has the correct distribution. Quantum secret sharing is possible even though quantum cloning is not: for example, the quantum Reed--Solomon code~\cite{aharonov1999fault, cleve1999share} yields a secret sharing scheme, just as classical Reed--Solomon does in the classical setting. Indeed, the fact that local views below the code distance are maximally mixed is exactly the secrecy property needed for secret sharing. As in the classical setting, \emph{verifiable} quantum secret sharing can be used to build secure quantum multiparty computation~\cite{crepeau2002secure, ben2006secure}. The idea that we can replace witness copying with secret sharing and MPC is the basis of our template for quantum locality-preserving gap amplification.

Fortunately, an HDX communication graph is very compatible with MPC and secret-sharing because HDXs contain many cliques. In~\Cref{sec:discussion} we discuss why extending our strategy to quantize Dinur's gap amplification procedure (that uses generic expander graphs) might be hard.

\paragraph{Our fault-tolerant protocol}
We now present a simplified version of our fault-tolerant protocol for distributing violation information among players that we will eventually compile into a local Hamiltonian.\footnote{In reality, \Cref{prot:protocol-intro} is not quite our final protocol because of difficulties that we will introduce later in this technical overview.} At a high level, this protocol simulates the \emph{sequential measurement circuit} $C$ of the input Hamiltonian $H_{\mathrm{in}}$ (the circuit that measures all the terms of $H_{\mathrm{in}}$ in turn) on an arbitrary witness state $\ket{\psi}$. We will assign the constraints of $C$ to cliques of $G$ at the start of the protocol. The qubits of the input witness that each constraint acts on will be secret shared among the players in the corresponding clique. At the end of the protocol, each player will learn the outcome of many constraint measurements, even in the presence of communication faults. Like in BMV, this is the source of gap amplification.

For this overview, we assume that our input Hamiltonian is 2-local and that its sequential measurement circuit has the following form.
\begin{definition}[Form of input Hamiltonian and its measurement circuit]
\label{def:hamiltonian-circuit-form}
We assume that $H_\mathrm{in} = \sum_{j=1}^{M} h_j$ acts on $n$ qubits, that each $h_j$ is 2-local, and that each qubit is involved in at most $O(1)$ many Hamiltonian terms.

Given \(H_{\mathrm{in}}\), we can construct its \emph{sequential measurement circuit} \(C\) as in \Cref{def:sequential-measurement-circuit-appendix}. \(C\) has a depth of $T=O(1)$ timesteps and $N=n/2$ gates per timestep. Every gate in $C$ represents a measurement circuit for one term $h_j$ of $H_{\mathrm{in}}$, has two input and two output wires, and writes a single classical bit on an additional wire as measurement outcome. The circuit when applied on any state $\ket{\psi}$ produces an output state of the form
\begin{equation}\label{eq:honest-exec-C}
C\ket{\psi}=\sum_{\mathbf y}\ket{\mathbf y}\otimes Q_{\mathbf y}\ket{\psi},
\end{equation}
where $\ket{\mathbf y}$ denotes the $NT$-length string of measurement outcomes, and $Q_y\ket{\psi}$ denotes the final output state conditioned on measuring $\ket{\mathbf y}$.

We will denote the $N$ gates in the $t^{th}$ timestep as $G_{1,t},\ldots,G_{N,t}$. Note that the wiring of $C$ between timesteps is a permutation: each of the $2N$ output wires at time $t$ will become one of the $2N$ input wires at time $t+1$.
\end{definition}

Given \(H_{\mathrm{in}}\) and $C$, we use an HDX \(G\) with \(N\)
\(d\)-cliques, i.e. equal to the number of gates in each timestep of $C$. Clique \(i\) will be responsible for executing
\(G_{i,t}\) at every time step \(t\), and in the protocol we allow for communication along the edges of $G$.

\begin{protocol}[Honest execution of protocol for simulating sequential measurement circuit]
\label{prot:protocol-intro}
\end{protocol}
\begin{enumerate}
\item Let $\ket{\psi}$ be a $2N$-qubit state on which we want to simulate $C$. In an honest execution of the protocol, every qubit of $\ket{\psi}$ is encoded into a $[d,1,\Delta]$ quantum error-correcting code $\mathcal{Q}$ that `supports' MPC (we will suppress the precise conditions on the code for now) with sufficiently large distance $\Delta$. This produces $d$ `shares' of every qubit of $\ket{\psi}$. For every $i \in [N]$, the $d$ shares of each wire that feeds into $G_{i,1}$ are distributed among the $d$ players in the $i$th $d$-clique.
\item For every time step \(t \in [T]\):
    \begin{enumerate}
        \item Each clique \(i\) executes a logical version of \(G_{i,t}\) using a quantum multiparty computation protocol (see~\Cref{thm:distributed-FTQC}). At the end of this step, the two output wires of $G_{i,t}$ remain separately encoded in $\mathcal{Q}$, whereas the measurement bit is encoded in a repetition code, i.e.~each player in clique $i$ holds one classical bit representing the measurement outcome of $G_{i,t}$.
        \item In this step, we route the output wires of $G_{i,t}$, so that  clique $i$ gets the encodings of the input wires of $G_{i,t+1}$. We do this by noting that the wiring of $C$ between timesteps $t$ and $t+1$ can be represented as a 2-regular bipartite graph with $N$ vertices on each side, representing the cliques at times $t$ and $t+1$. This set of edges can be decomposed into two perfect matchings (equivalently, permutations) $\pi_{1,t}$ and $\pi_{2,t}$, using Hall's theorem. We route these two matchings using two calls to the fault-tolerant routing procedure in~\Cref{prop:ft-routing}.
    \end{enumerate}
    \item Every player computes its final output by outputting YES if and only if every measurement outcome that it learned during step 2 is correct (i.e. indicates a satisfied term). %
\end{enumerate}

We use the fault-tolerance of the multiparty computation protocol in step 2(a) (\Cref{thm:distributed-FTQC}) and the routing protocol in 2(b) (\Cref{prop:ft-routing}; this routing protocol is a straightforward quantization of~BMV routing) to show that \Cref{prot:protocol-intro} satisfies the following guarantee.\footnote{We focus on the gap regime of $1/\pl n \to\Omega(1)$ for simplicity; our arguments easily generalize to any regime $\eps \to \eps\pl n$, using the appropriate properties about HDX.} 
\begin{lemma}[Informal]\label{lem:protocol-informal}
In the YES case, when there are no communication faults in the protocol, an input state equal to the prescribed encoding of a zero-energy state of $H_\mathrm{in}$ causes every player to output YES. In the NO case, if the input Hamiltonian $H_\mathrm{in}$ has ground energy at least $1/\pl n$, and a constant fraction (say 0.001 fraction) of communication edges in~\Cref{prot:protocol-intro} are faulty, then every input state to~\Cref{prot:protocol-intro} causes an $\Omega(1)$ fraction of vertices to output NO with high probability.
\end{lemma}

\begin{proof}[Proof sketch]
The YES case is straightforward. In the NO case, either the input state to the protocol is (close to being) the legal distributed encoding---as described in step 1---of some logical state $\ket{\psi}$, or it is far from a legal encoding. In the latter case, the first step in our MPC protocol in 2(a) will cause many players to reject and abort the protocol. Therefore, we assume for this proof sketch that the protocol's input is a legal encoding of some state $\ket{\psi}$. We will show that our protocol simulates the execution of $C$ on $\ket{\psi}$, up to some quantitative error bounds, and distributes the measurement outcomes resulting from the simulation of $C$ among the players so that a constant fraction of players output NO at the end.

Call any clique in $X(d)$ with less than $10\%$ corrupted internal edges `good', and other cliques `bad'. By the fault-tolerance guarantees of the quantum MPC protocol we use in step 2(a) (\Cref{thm:distributed-FTQC}), each good clique will execute every gate assigned to it correctly. A bad clique, on the other hand, may introduce an arbitrary logical error on the output wires of any gate assigned to it. Recall the `clique-edge sampling' property of HDXs that we introduced earlier. This property guarantees that, even when a constant fraction of communication edges are faulty, only a $1/\pl n$ fraction of cliques in $X(d)$ are bad. Meanwhile, using the fault-tolerance of the routing protocol (\Cref{prop:ft-routing}), the routing in step 2(b) also introduces logical faults on at most $1/\pl n$ fraction of the $2N$ wires. Thus our protocol amounts to an encoded simulation of $C$ that has up to $1/\pl n$ fraction of faulty logical gates in each layer/timestep. 

Denote by $\tilde C$ the version of $C$ with $1/\pl n$ fraction faulty logical gates per timestep that our protocol simulates. In analogy with~\eqref{eq:honest-exec-C}, let the output of $\tilde{C}$ on $\ket{\psi}$ be denoted
\begin{equation}\label{eq:final-state}
\widetilde{C}\ket{\psi}=\sum_{\mathbf y}\ket{\mathbf y}\otimes \widetilde{Q}_{\mathbf y}\ket{\psi},
\end{equation}
where $\ket{\mathbf y}$ denotes the $NT$-length string of measurement outcomes, and $\tilde Q_y\ket{\psi}$ denotes the final output state conditioned on measuring $\ket{\mathbf y}$. In~\Cref{sec:DL-new} we prove a strengthening of the detectability lemma of~\cite{AharonovALV09} which allows us to conclude that the distribution over $\mathbf{y}$ that $\tilde C$ outputs is not so different from that which the uncorrupted $C$ would output. More precisely, \Cref{lem:dl-appendix} allows us to conclude that, if $H_{\mathrm{in}}$ has ground energy at least $1/\pl n$, then measuring the $\ket{\mathbf{y}}$ register in the output of $\widetilde{C}$ results in a distribution over outcomes $\mathbf{y}$ with $1-o(1)$ weight on strings with at least $1/\pl n$ fraction of 1s (violations).

Recall that, at the end of~\Cref{prot:protocol-intro}, all the players in a clique responsible for gate $G_{i,t}$ hold one copy of the bit of the $\ket{\mathbf{y}}$ register associated with the measurement outcome of $G_{i,t}$. The `vertex-clique sampling' property of HDXs (discussed above) then implies that, if $\mathbf{y}$ is a string with $1/\pl n$ fraction of 1s, a constant fraction of players will see at least one 1 during the protocol.

\end{proof}

\paragraph{Converting the classical fault-tolerant protocol to a gap-amplified 2CSP} 
The classical analogue of~\Cref{lem:protocol-informal} immediately implies classical gap amplification, as observed in BMV. To see why, we must specify how a classical distributed protocol can be converted into an output CSP. This is done by associating each player in $G$ with one large variable in the output CSP. More specifically, an honest assignment to the output CSP in BMV is a \emph{Cook–Levin execution tableau} for the BMV distributed protocol, where each large variable records the complete work-tape history of one player in the protocol. The output CSP has two types of constraints: 2-local communication constraints and 1-local computation constraints. Each communication constraint uses equality checks to verify that the tableaux of two players agree on every transmission along one communication edge, and each computation constraint checks that all the local computation on a single player's work tape was propagated correctly and that the player accepted at the end. Since each large variable holds one player's work tape, the log of the alphabet size of the output CSP is the same as the per-player communication of the protocol, which in our case will be roughly the number of cliques that each player is part of, i.e. $2^{\pll n}$ bits. 

The argument that a classical analogue of~\Cref{lem:protocol-informal} translates into soundness for the output 2CSP is as follows. If an assignment to the output 2CSP (which is meant to be a tableau of a protocol execution) contains faulty transmissions for more than a constant fraction of `communication edges', then at least a constant fraction of the 2-local communication constraints fail. Otherwise, the fact that the protocol is fault-tolerant to a constant fraction of faulty edges, in the sense of~\Cref{lem:protocol-informal}, guarantees that a constant fraction of players output NO. This means that a constant fraction of the 1-local `computation constraints' will fail. In either case, a constant fraction of the terms in the output CSP are violated when the input CSP has unsat value at least $1/\pl n$.

While this is sufficient to show gap amplification in the classical setting, the quantum setting presents unique difficulties. The reason is that performing a Cook-Levin reduction that converts a circuit/protocol into a local Hamiltonian is significantly more complicated than it is in the classical setting, and this circuit-to-Hamiltonian step requires analyzing the soundness of~\Cref{prot:protocol-intro} in a very strong noise model. We discuss these challenges in the following section.

\subsection{Circuit-to-Hamiltonian and the Postselection Problem}
\label{sec:intro-circ-to-H}

We now explain how to convert the distributed protocol $\mathcal P$ of~\Cref{prot:protocol-intro}, which simulates the logical sequential measurement circuit $C$, into a Hamiltonian, and why postselection complicates its soundness analysis.

\paragraph{The ABN circuit-to-Hamiltonian procedure}

The usual quantum analogue of the Cook-Levin reduction, the Feynman--Kitaev clock Hamiltonian \cite{kitaev2002classical}, is not an adequate substitute here for the role that Cook--Levin played in BMV. The issue is robustness: when Cook--Levin is applied carefully to a \emph{fault-tolerant} circuit, an assignment to the resulting CSP with a small fraction of unsatisfied clauses still recognizably encodes a slightly faulty execution of the input circuit. By contrast, Feynman--Kitaev is extremely brittle: violating even a single propagation term allows the prover to cheat arbitrarily.

We therefore use the recent circuit-to-Hamiltonian construction of Anshu, Breuckmann, and Nguyen (ABN) \cite{Anshu_2024}, which is based on gate teleportation. We include an exposition of this construction based on a toy example in~\Cref{sec:ckt-to-H}.

\paragraph{Combinatorial soundness of the ABN construction}
At the level of completeness, the ABN construction can be adapted to a distributed computation just as BMV adapt Cook--Levin in the classical setting. The main difficulty is soundness. The only general soundness analysis currently known for the ABN construction would only allow us to obtain meaningful control over low-energy states of the output Hamiltonian when the protocol $\mathcal{P}$ has near constant physical depth.

However, ABN also provide a cleaner \emph{combinatorial} picture. Namely, they show a combinatorial soundness statement along the following lines (note that we have adapted the following lemma for our distributed-computation setting).
\begin{lemma}[Informal combinatorial soundness of~\cite{Anshu_2024} circuit-to-Hamiltonian]
\label{lem:abn-comb-soundness-intro}
Suppose that $\mathcal P$ is a distributed protocol whose circuit implementation has physical depth $D$, and $H$ is the result of applying the ABN circuit-to-Hamiltonian construction to this implementation\footnote{More precisely, we use a modification of~\cite{Anshu_2024} for distributed computation in~\Cref{sec:ckt-to-H}.}. Let $W_t$ denote the unitary of its $t$th layer. Then any witness which exactly satisfies all but $\epsilon_{\mathrm{ABN}}$ fraction of the local terms of $H$ has the form
\[
\sum_{\vec F : F_1, \dots, F_D \in S} F_D W_D \dots F_1 W_1 \ket{\psi},
\]
for some state $\ket{\psi}$ (ancilla initialisation can be checked by adding extra Hamiltonian terms) and some fixed set $S$ of size $O(\epsilon_{\mathrm{ABN}})$ of edges in the protocol $\mathcal{P}$'s communication graph, where the shorthand $F \in S$ is being used to denote ``$F$ is a Pauli error on a specified set $S$''. The set $S$ can be computed explicitly if it is known which Hamiltonian terms in $H$ are violated.
\end{lemma}

In particular, \Cref{lem:abn-comb-soundness-intro} shows that a violated Hamiltonian term in $H$ has only a bounded influence: unlike in Feynman--Kitaev, violating one local term does not let the prover arbitrarily corrupt the entire future computation. Instead, it lets the prover introduce errors only on a bounded set of `communication edges' from $\mathcal{P}$. If the protocol $\mathcal P$ is adversarially fault-tolerant against faults on a fixed bounded-size set of edges, the final outcome is essentially unchanged. Faults on a fixed bounded-size set of circuit wires/edges are exactly the type of fault pattern that our multiparty computation protocol is designed to tolerate.

\paragraph{Postselection and its effect on MPC}
There is, however, a subtlety here that does not arise in the classical setting. The sum in the expression
\[
\sum_{\vec F : F_1, \dots, F_D \in S} F_D W_D \dots F_1 W_1 \ket{\psi}
\]
is a \emph{coherent sum} over the Pauli errors represented by different $\vec F$s. This means that the errors that the fault-tolerance of $\mathcal P$ must handle include \emph{non-unitary} operations, such as postselection, which can be produced as linear combinations of unitaries. As a toy example, suppose that $\mathcal P$ consists of a one-qubit circuit. Then
\[
\Big(\frac{I+Z}{2}\Big) W_D \dots W_1 \ket{\psi}
\]
is a valid low-energy witness state with only one fault location, and its operational interpretation is that the adversarial prover executes the protocol $\mathcal P$ correctly and then \emph{postselects on $\ket{0}$}. Thus the prover should be viewed as having the ability to postselect as part of its cheating strategy.

This becomes a serious issue when the ABN Hamiltonian is used to encode an MPC protocol. Fix a gate $G_{i,t}$ and consider step 2(a) of \Cref{prot:protocol-intro}. The first step in any MPC protocol (including ours) is that the participating players run \emph{verification of secret sharing} (VSS), so as to jointly certify that the input state shared between them is a valid distributed encoding of some underlying state. The players then run the MPC protocol for $G_{i,t}$. Provided the protocol succeeds, every player obtains one copy of the classical bit that indicates whether the constraint $G_{i,t}$ measures was violated (1) or satisfied (0). The witness data remains encoded in a quantum code and shared between the players.

Under the usual assumptions that cheating consists only of trace-preserving operations and that not too many players cheat, this procedure is secure: the measurement outcome that the honest players obtain at the end will reflect an execution of $G_{i,t}$ on the input witness. Allowing postselection creates two new problems.

\begin{enumerate}
\item \emph{Postselection Problem 1: VSS must have zero error.}
Existing verifiable quantum secret sharing (VQSS) protocols \cite{crepeau2002secure, ben2006secure} have a small but nonzero probability of failure: even when the number of cheaters is below threshold, the protocol may very rarely make the wrong decision about whether the shared state is a legal codeword. This is harmless in cryptographic applications, where the failure probability is exponentially small in the number of players. Here it is fatal, because cheaters can simply postselect on that failure event.

\item \emph{Postselection Problem 2: Postselection can bias the final classical readout.} 
For data that remain encoded in a quantum codeword, postselection attacks on a bounded number of shares are no more dangerous than ordinary faults, because quantum error correction protects against arbitrary local quantum operations. The problem arises from the classical violation data that effectively becomes encoded in a repetition code that is distributed between the players at the end of the MPC protocol.

Recall (see~\Cref{def:hamiltonian-circuit-form}) that an honest execution of the sequential measurement circuit $C$ on $\ket{\psi}$ produces a state of the form $C\ket{\psi} = \sum_{\mathbf{y}} \ket{\mathbf{y}}\otimes Q_{\mathbf{y}}\ket{\psi}$. Now suppose for simplicity that this sum is supported over only two measurement strings $0^{NT}$ and $1^{NT}$, and furthermore $\|Q_{0^{NT}} \ket{\psi}\|^2 = \epsilon$ and $\|Q_{1^{NT}} \ket{\psi}\|^2 = 1-\epsilon$. Thus $\ket{\psi}$ violates every constraint in $H_\mathrm{in}$ with high probability and every player should reject except with small probability. If we ignore the quantum shares, then at the end of step 2 of~\Cref{prot:protocol-intro} the players jointly hold the repetition encoding of $\mathbf{y}$,
\begin{equation}
\label{eq:final-joint-state-classical}
\epsilon \: \ketbra{0^{NTd}} + (1-\epsilon)\ketbra{1^{NTd}}.
\end{equation}
Under ordinary circumstances, the honest players would report violation except with probability $\epsilon$. But if even a single cheater now postselects one of its own shares of~\Cref{eq:final-joint-state-classical} onto $\ket{0}$, the all-$1$ branch is removed and the honest players always accept. Since our protocol must tolerate far more than one cheater, applying the ABN circuit-to-Hamiltonian to \Cref{prot:protocol-intro} as written cannot yield a sound reduction.
\end{enumerate}

\subsection{Solving the Postselection Problem Part 1: Zero-Error Quantum MPC}
\label{sec:solving-post-prob-1}

\paragraph{Idea 2: Zero-error MPC using QLTCs and gate teleportation}
We first address Postselection Problem~1. The issue is not that existing VQSS protocols fail too often; it is that they fail at all. To obtain soundness against postselection, we need a \emph{zero-error} VQSS primitive. The authors who proposed the first nonzero-error VQSS protocol \cite{crepeau2002secure} also attempted to obtain a zero-error protocol, but did not succeed, leaving it as an open question.

In this work we resolve that question---at least against a $\frac{1}{\poly\log(n)}$ fraction of cheaters in $n$-party protocols, and in all parameter regimes if we assume quantum locally testable codes with constant soundness. Interestingly, our solution is genuinely quantum: there is no obvious classical analogue.

Classically, the standard zero-error VSS paradigm is the protocol of Ben-Or, Goldwasser, and Wigderson (BGW) \cite{BGW1988}. Their protocol is designed around a particular classical code which is a tensor product of two Reed--Solomon codes. However, there is no obvious quantum analogue of such an object: the tensor product of two quantum Reed--Solomon codes does not have good distance. To design zero-error VQSS, we therefore isolate the structural properties of the Reed--Solomon tensor product that the BGW classical protocol actually uses.

\begin{enumerate}
\item \emph{Good distance.} The code corrects errors on a constant $\delta$ fraction of coordinates.

\item \emph{Localized verification.} The checks of the code can be partitioned into sets $q_1,\dots,q_n$---one set for each player. If $p_i$ denotes the coordinates touched by the checks in $q_i$, then:
\begin{enumerate}
\item There is a constant $c\in(0,1)$ such that the union of the $p_i$'s belonging to any $c$ fraction of players covers at most a $\delta$ fraction of the coordinates. This prevents a constant fraction of cheating players from corrupting the codeword beyond repair.
\item There is a constant $\epsilon$ such that if only an $\epsilon$ fraction of the check-sets $q_i$ are misreported or violated, while all the others are correctly reported and pass, then the word differs from a codeword only on the union of the corresponding bad coordinate sets $p_i$.\footnote{In the usual MPC setting, repeated uses of VQSS are controlled by localizing all corruptions on the output state of VQSS to the shares of the state held by the cheating players, so that faults introduced by the protocol  always happen on the same set of coordinates and do not spread when the protocol is executed repeatedly. Generic QLTC local testability does not guarantee this localization property as a black box. For repeated uses of VQSS, one therefore also needs a sufficiently strong robust decoder~\cite{nguyen2025quantum}; see \Cref{algo:error-correction}.} This means that, even if a constant fraction of cheating players misreport their checks, the checks on the honest players alone can guarantee that the shared word is close to a codeword.
\end{enumerate}

\item \emph{Support for MPC.} One can implement the required logical operations fault-tolerantly on a distributed codeword using local operations and communication among the players.
\end{enumerate}

Our first observation is that Properties~1 and~2 are close to what a quantum locally testable code (QLTC) with good distance, bounded degree, and constant soundness already guarantees. The difficult property is Property~3. Classically, only very special codes such as Reed--Solomon are known to support general MPC efficiently.

In the quantum setting, however, we can obtain Property~3 by \emph{gate teleportation}~\cite{gottesman1999demonstrating}. Every CSS code already supports transversal Pauli and CNOT gates. To achieve a universal gate set, we need to add support for $T$ or CCZ gates. This can be done using a CCZ state distillation protocol followed by teleportation, similarly to how it is done in the fault-tolerance literature~\cite{nguyen2025quantum}. However, we do need to carefully design a distillation protocol to handle MPC adversarial corruptions, which we achieve using a linear-distance quantum code with transversal CCZ gates\footnote{We remark that good codes with transversal CCZ operations were developed only recently \cite{nguyen2025good, golowich2025asymptotically}, and so our strategy relies in more than one way on recent coding-theoretic advances.}.
This uniquely quantum strategy allows us to achieve universal MPC with \emph{any} base code admitting VQSS, which is unlikely to be possible classically. Letting the base code be any QLTC, we get a secure zero-error quantum MPC protocol (\Cref{sec:mpc}). We instantiate this template with a recent almost-good QLTC construction~\cite{dinur2024expansion} and obtain a $1/\pl n$-tolerant protocol, which in turn leads to the unconditional part of our main~\Cref{thm:main-intro}.

\paragraph{The commuting case}
For commuting input Hamiltonians, Postselection Problem~2 disappears. The reason is that, when the terms of the input Hamiltonian have a simultaneous eigenbasis, there is always an optimal witness $\ket{\psi}$ for the sequential measurement circuit $C$ which, in the formulation of~\eqref{eq:honest-exec-C}, causes $C$ to produce a single \emph{deterministic} measurement outcome string $\mathbf{y}$. Therefore, intuitively at least, the sum in~\Cref{eq:honest-exec-C} has only a single term, and the sort of postselection attack that works on states like~\Cref{eq:final-joint-state-classical} is useless.

Note that commuting Hamiltonians are not known to be in $\mathsf{NP}$; therefore, this is already an advance over the classical case. This analysis is carried out in~\Cref{sec:protocol-for-circuit} in order to show the soundness theorem~\Cref{thm:commuting-main}. (Although we also analyze our gap amplification procedure for non-commuting input Hamiltonians, we retain the commuting analysis because we hope it may illustrate some subset of the key ideas in our work without the additional complications---and modifications to the procedure---which the non-commuting case requires.)

This reasoning fails completely in the non-commuting case. A state having high energy with respect to a non-commuting Hamiltonian does not preclude that a sequential measurement of the terms can, with small probability, return a string of outcomes indicating no violation. For example, the Hamiltonian
\[
\frac12(\proj{0}+\proj{+})
\]
has no common eigenstate for its two terms, yet one can first measure in the $\{\ket0,\ket1\}$ basis, obtain outcome $\ket0$, and then measure in the $\{\ket+,\ket-\}$ basis and obtain outcome $\ket+$ with nonzero probability. Thus, for non-commuting Hamiltonians, further ideas are needed to control postselection.

\subsection{Solving the Postselection Problem Part 2}
\label{sec:solving-post-prob-2}

We saw when we introduced Postselection Problem 2 that, because a classical repetition code does not protect against quantum postselection errors, even a \emph{single} cheater among the players might be able to destroy soundness using a postselection attack on the classical violation data. The postselection attack that we presented worked when the final joint state shared by the players after~\Cref{prot:protocol-intro} had the form given in~\Cref{eq:final-joint-state-classical}. For this to happen, we required two assumptions:
\begin{enumerate}
    \item \emph{Assumption 1:} the \emph{logical distribution} of measurement outcomes (prior to encoding in the length-$d$ repetition code) was supported only on $0^{NT}$ and $1^{NT}$.
    \item \emph{Assumption 2:} at the end of step 2(a) from~\Cref{prot:protocol-intro}, every player obtained a perfect copy of the measurement outcome of $G_{i,t}$.
\end{enumerate}
As it turns out, our solution to Postselection Problem 2 will require breaking both of these assumptions.

For simplicity, we assume here that all cliques are `good', in the sense that they contain an honest majority. We also assume that cheating players behave honestly until the end of the protocol (at which point we allow them to postselect the measurement outcomes they obtained during step 2(a) of~\Cref{prot:protocol-intro}). We will be able to illustrate all of our remaining key ideas by showing how we deal with cheating in this toy setting; the analysis is done in full generality in~\Cref{sec:nc-protocol}.

\paragraph{Idea 3: Add random noise to the classical readout}
We begin by showing how breaking Assumption 2 can help with the problem. The issue with step 2(a) as it is written in~\Cref{prot:protocol-intro} is that every player, including cheaters, receives a perfect copy of the measurement outcome bit, which means that cheaters can manipulate the global state using postselection with perfect confidence. We can partially counter postselection attacks by making every player's information about the measurement bit \emph{slightly imperfect}. This can be achieved by altering 2(a) of~\Cref{prot:protocol-intro}. Recall that step 2(a) consists of an encoded simulation of a gate/constraint $G_{i,t}$, and as written it outputs a repetition encoding of the measurement outcome of $G_{i,t}$ which is distributed between the players in the $i$th clique. We modify this step so that, instead of outputting a perfect repetition code encoding, the simulation protocol outputs a \emph{non-orthogonal} or \emph{noisy} repetition encoding of the measurement outcome of $G_{i,t}$.

In particular, suppose we construct the multiparty measurement protocol used in step 2(a) of~\Cref{prot:protocol-intro} so that each measurement outcome bit which a player receives at the end is independently flipped from its true value with probability $0<\gamma<1/2$. (Specifically, this is achieved by using~\Cref{algo:noisy-rep-code-output} in the place of~\Cref{algo:rep-code-output}.) Then, instead of~\Cref{eq:final-joint-state-classical}, the final joint state held by the players (under the assumptions as before that the protocol was executed honestly until the end) would take the form
\begin{equation}
\label{eq:final-joint-state-classical-noisy}
\sum_{x \in \{0,1\}^{NTd}} \gamma^{|x|} (1-\gamma)^{NTd - |x|} X^{x} \Big( \epsilon \: \ketbra{0^{NTd}} + (1-\epsilon)\ketbra{1^{NTd}} \Big) X^x,
\end{equation}
where $X^x$ represents a bit-flip operator on the qubits where $x_i = 1$.

This is already an effective countermeasure against a \emph{single} cheater (who postselects only a single one of its measurement outcomes). The point is that this deliberately introduced noise partially decouples a player's local value from the global branch. A cheating player who sees a local $0$ can no longer infer that the global state is close to the all-$0$ branch: its own bit may have been flipped by the injected noise. Hence postselecting its local share onto $\ket0$ does not simply remove the violating branch; when the cheater's own share has been flipped, postselecting that share on 0 actually selects the 1 branch globally. In the example above, the honest players' probability of seeing $1$ after one cheater postselects a single bit on $0$ becomes approximately $\gamma + \gamma(1-\epsilon)/\epsilon$ to first order in $\gamma$. Meanwhile, if the noise is handled carefully, we can show that completeness is not affected too much by this modification.

However, this idea does not protect soundness against \emph{many cheaters which collude across cliques}. For example, suppose there is exactly one cheater in every $d$-clique, and that every such cheater postselects its own shares of the joint state \Cref{eq:final-joint-state-classical-noisy} onto $\ket{0}$. Writing $L=NT(d-1)$, the (unnormalized) residual state on the honest players' shares after this coordinated postselection is
\begin{equation}
     \sum_{x \in \{0,1\}^{L}} \gamma^{|x|}(1-\gamma)^{L-|x|} X^x \Big[ \epsilon(1-\gamma)^{NT} \proj{0^L} + (1-\epsilon)\gamma^{NT} \proj{1^L} \Big] X^x,
\end{equation}
where the noise $X^x$ is now applied only on the honest players' shares.

As such, unless $\epsilon$ is extremely small in this example, the all-one branch will still be exponentially suppressed. The possible size of $\epsilon$, and more generally the form of the logical distribution of measurement outcomes for the input Hamiltonians of interest, is determined by the behaviour of the sequential measurement circuit on $\ket{\psi}$, which (as we mentioned in the proof sketch of~\Cref{lem:protocol-informal}) we control using a robust modification of the detectability lemma of~\cite{AharonovALV09} that we prove in~\Cref{sec:DL-new}. We were not able to obtain enough control over the logical measurement-outcome distribution in order to counter this kind of across-clique adversarial postselection attack; as such, we need another idea in order to prove the combinatorial soundness of our construction.

\paragraph{Idea 4: friendly noisy postselection}
Our next idea is to allow the \emph{honest} players to use postselection too in order to counter adversarial postselection by cheating players. More specifically, we introduce a single-qubit non-unitary operator $A$ of the following form:
\begin{align}
\label{eq:A-def-intro}
    A_\eta = \sqrt{1-\eta} \ketbra{1} + \sqrt{\eta} \ketbra{0}
\end{align}
For small $\eta > 0$, $A_\eta$ implements a \emph{noisy} postselection onto $\ket{1}$. Interestingly, the combinatorial-soundness analysis for the ABN circuit-to-Hamiltonian formalism is able to handle input circuits that have non-unitary operations like $A_\eta$. As we explain at the beginning of~\Cref{sec:ckt-to-H}, the essential step in the combinatorial analysis of ABN is to find a change of basis under which the Hamiltonian witness becomes a tensor product of \emph{resource states} that are used to implement the computation by teleportation. The formalism does not distinguish between resource states that implement unitary gates and resource states that implement non-unitary gates. A cheating prover can supply a witness state that encodes postselection by putting in resource states for non-unitary computations (indeed this is how the fault model expressed in~\Cref{lem:abn-comb-soundness-intro} arises). We can also design our protocol to take advantage of this by putting resource states for non-unitary operations in the honest witness.

Suppose we have the same setup as in `Idea 3': step 2 of \Cref{prot:protocol-intro} is executed honestly with noisy readout, producing~\Cref{eq:final-joint-state-classical-noisy} as the final joint state, and as before we assume every $d$-clique has exactly one cheater who postselects its final shares onto $\ket{0}$. Now suppose that the $(d-1)$ honest players in each clique perform $A_\eta$ on every share that they hold. Let
\begin{align*}
    r_0 &\coloneqq \eta(1-\gamma)+(1-\eta)\gamma \\
r_1 &\coloneqq \eta\gamma+(1-\eta)(1-\gamma).
\end{align*}
It can be shown by calculation that, after both operations, the all-zero and all-one branches have unnormalized weights
\begin{equation}
\epsilon\bigl((1-\gamma)r_0^{d-1}\bigr)^{NT}
\qquad\text{and}\qquad
(1-\epsilon)\bigl(\gamma r_1^{d-1}\bigr)^{NT}
\end{equation}
respectively.
We can therefore choose $\eta$ so that
\begin{equation*}
(1-\gamma)r_0^{d-1}
=
\gamma r_1^{d-1}.
\end{equation*}
With this choice, the friendly postselection performed by the $(d-1)$ honest players exactly cancels the bias introduced by the cheater in each clique. Although the norm of the state is exponentially reduced, the common postselection factor is the same in both branches, so after normalization their relative weights are again $\epsilon$ and $(1-\epsilon)$.

Thus, in our final modification of~\Cref{prot:protocol-intro}, we implement the noisy readout from `Idea 3' and we also ask the honest players to implement $A_\eta$ on all their classical shares after step 2. In effect, after these modifications, cheating and honest players alike have \emph{noisy} postselection devices, and the cheaters try to bias the outcome towards $0$, while the honest players bias it towards $1$. If $\eta$ is chosen carefully---small enough relative to $\gamma$ that the honest majority in a good clique overpowers the colluding cheaters, but still large enough that the yes-case behavior is not too badly disrupted---then the honest players' noisy friendly postselection dominates the adversarial postselection in NO instances without destroying completeness in YES instances. This is analyzed in~\Cref{sec:nc-protocol}.

\subsection{Discussion}\label{sec:discussion}
In this section we discuss some interesting observations about our procedure and clarify the landscape for future work.

\paragraph{Expanders vs HDX}
The fact that an HDX contains many cliques seems important to our gap amplification procedure. To see why, we could imagine trying to quantize Dinur's classical locality-preserving gap amplification procedure~\cite{Dinur07} using our MPC-based paradigm.

Dinur assumes that the input CSP to her amplification procedure has an interaction graph described by an ordinary expander $G$: that is, each input variable is identified with a vertex in $G$, and the 2-local input constraints act only on vertices connected by edges in $G$. In her amplification scheme, a large variable in her output CSP associated with vertex $v$ contains the assignments to all the input variables in the \emph{$t$-neighbourhood} in $G$ of $v$ (vertices reachable from $v$ by a $t$-step walk). If we tried to quantize her procedure using our framework, we would want to distribute the input witness among the players/vertices using secret-sharing (rather than witness copying). Additionally, we would want all the players who need to know the violation value of a particular input constraint to run MPC together on a secret-sharing of the necessary witness information, in order to compute the violation value jointly. For Dinur's scheme, this means that we would want all the players in a $t$-neighbourhood to be able to run MPC together.

The issue is that computing on a secret-shared input using MPC requires a robust communication graph (ideally a clique) between the players. More brittle communication graphs (e.g.~trees) do not support adversarially fault-tolerant quantum MPC, because a ``bottleneck'' of a few edges can be corrupted to disconnect the graph. Unfortunately, local neighborhoods in ordinary expanders are tree-like. For our framework, we more or less require that local neighbourhoods in the communication graph be cliques (or something that can simulate a clique). Furthermore, this requirement must coexist with the requirement that the degree of the communication graph, which decides the alphabet size of our output CSP, is not too large. Fortunately, using an HDX allows us to satisfy exactly these constraints.

\paragraph{Overcoming the Brandao-Harrow barrier}
As we mentioned in the introduction section, the work of Brandao and Harrow \cite{brandao2013product} has long been seen as a significant barrier to a quantum locality-preserving gap amplification procedure. We can identify at least two sufficient reasons that we overcome this barrier:
\begin{enumerate}
\item Brandao-Harrow does not preclude all kinds of locality-preserving quantum gap amplification, but only those within certain parameter regimes (of which parallel repetition is the most obvious example). In particular, as long as the \emph{degree} of the output Hamiltonian remains significantly smaller than its \emph{alphabet size} after gap amplification, there is no formal contradiction with \cite[Corollary 5]{brandao2013product}. This is exactly what happens in our procedure---any fault-tolerant routing protocol on a graph must have work complexity at least the degree, which means that the alphabet size is at least $\exp(\deg(G))$, which is within the parameter regime allowed by~\cite{brandao2013product}. Another way to view this phenomenon is to say that we achieve some form of \emph{derandomized} parallel repetition without being able to achieve ``vanilla'' parallel repetition, which is surprising but does not contradict Brandao-Harrow.
\item We only get combinatorial soundness for our output Hamiltonian. Brandao-Harrow does not apply to amplifying the combinatorial promise gap of $\mathsf{QMA}$-hard Hamiltonians.
\end{enumerate}
Because the first item alone is a sufficient reason that Brandao-Harrow would not apply to our procedure, we hope that it may be possible to analyze some variant of our procedure eventually in the general-soundness setting.

\paragraph{Reliance on combinatorial soundness}
We specialize to combinatorial soundness principally due to the fact that a good general soundness analysis technique is not known for the ABN Hamiltonian. It is natural to ask: how strongly do we rely on the combinatorial soundness assumption? That is, supposing that a way to analyze the general soundness of the ABN circuit-to-Hamiltonian procedure were found, how hard would it be to get general soundness for our procedure?

There are two main reasons that the generalization would not be immediate even supposing that a general-soundness analysis for ABN were known.
\begin{enumerate}
\item Quantum MPC protocols (including ours) are tolerant to a \emph{bounded fraction} of adversarial players, which plays well with combinatorial soundness because combinatorial soundness allows us to assume that communication edges that are not corrupted are perfect. A general-soundness noise model might require analyzing the quantum MPC protocol in a setting where (say) \emph{every} communication may be corrupted, but only by a small amount.

\item Our use of ``friendly postselection'' seems potentially difficult to generalize to the general soundness setting, for the following reason. The combinatorial soundness statement we have for the ABN circuit-to-Hamiltonian construction (\Cref{lem:abn-comb-soundness-intro}) can be stated qualitatively along the following lines: given $H_C$ that comes from applying the ABN circuit-to-Hamiltonian construction to a circuit $C$, every state with combinatorial energy $\leq \epsilon$ with respect to $H_C$ encodes an $O(\epsilon)$-faulty execution of $C$. This statement holds even for circuits $C$ with postselection gates. This is precisely the feature that we exploit in order to introduce friendly noisy postselection in ``Idea 4''.

If a similar statement were to hold for general soundness to energy $\epsilon$, and in addition postselection circuits could be made \emph{adversarially} fault-tolerant enough to tolerate $\epsilon$ adversarial error density, then (we claim) there would be a $\mathsf{QMA}$ protocol for $\mathsf{PostBQP}$ problems, which is unlikely. The protocol would be: start with any circuit $C$ with postselection gates; construct a fault-tolerant version $\tilde C$; apply ABN to get a Hamiltonian $H_{\tilde C}$; ask the prover to provide a state and measure the energy of that state.

It is not known whether postselection circuits can be made fault-tolerant against adversarial errors. As such, this observation does not preclude a general-soundness analysis of the ABN circuit-to-Hamiltonian procedure that would allow friendly postselection, but it raises a potential difficulty.

\end{enumerate}
We think the second reason is likely to be the more substantial obstacle, and we leave the question of whether it is possible to recover our results without the use of friendly postselection as an interesting open problem.

\paragraph{Proof composition and alphabet reduction}
Given that we achieve locality-preserving gap amplification with large alphabet (for combinatorial soundness), a natural next step is to study alphabet reduction. In the classical world, alphabet reduction using proof composition and PCPPs is a standard technique, and nearly every classical PCP is also a PCPP \cite{Ben-SassonGHSV06, DinurReingold}. In particular, as we mentioned earlier, a classical gap amplification procedure with similar parameters to ours is already sufficient to construct classical PCPs with slightly super-polynomial $n^{\exp(\log^*(n))}$ size~\cite{BT}, and the construction uses proof composition with a subexponential-sized PCPP~\cite{amireddyBSSW2025} in order to achieve alphabet reduction.

Quantumly, however, alphabet reduction for \emph{generic} Hamiltonians seems difficult to achieve: consider, for example, the case of a Hamiltonian on `large variables' of size $k$ qubits each, whose unique ground state is a tensor product of many copies of the GHZ state on $k$ qubits. This observation does not rule out that alphabet reduction might be possible for \emph{specific} Hamiltonians and in particular for the Hamiltonians that we construct. 
We likewise leave this as a direction for future work.

\subsection{Organization}

In \Cref{sec:prelims} we define necessary preliminaries. In \Cref{sec:mpc} we describe and prove the properties we require of our zero-error MPC protocol, which we use for within-clique computations in \Cref{prot:protocol-intro}. In \Cref{sec:routing} we show how to implement the clique-to-clique routing required by \Cref{prot:protocol-intro} in the quantum setting. \Cref{sec:protocol-for-circuit} shows the necessary fault-tolerance properties of a more basic version of our protocol in order to analyze the combinatorial soundness of our output Hamiltonian when the input Hamiltonian is \emph{commuting}. \Cref{sec:nc-protocol} shows how to modify the protocol for the case where the input Hamiltonian is non-commuting, and does the necessary fault-tolerance analysis in that setting. \Cref{sec:ckt-to-H} shows how to convert a protocol into a Hamiltonian using a modification of the ABN circuit-to-Hamiltonian procedure. \Cref{sec:final-reduction} puts all the pieces together, and \Cref{sec:DL-new} proves a robust variant of the~\cite{AharonovALV09} detectability lemma which we need in our analysis.

\medskip\medskip

%% file: prelims.tex
\section{Preliminaries}
\label{sec:prelims}
\paragraph{Notation:} 
We use standard big-$O$ notation: we denote $A = O(B)$ 
or $A\lll B$ if $A\leq C\cdot B$ for some absolute constant
$C>0$. For a parameter $\epsilon$ that we think of as going to zero, we denote by $c(\epsilon)$ a function that is bounded by an absolute constant times $\epsilon$. Similarly, we denote $A = \Omega(B)$ or $A\ggg B$ 
if $A\geq c B$ for some absolute constant $c>0$. We also 
write $k\ll d$ to denote the fact that $d$ is taken to 
be sufficiently large compared to any function of $k$. 
If $A$ is a finite set and 
$i \leq |A|$, the notation
$B\subseteq_{i} A$ means that 
we sample a subset of size $i$
of $A$ uniformly.
Given a domain $X$ and a measure $\mu$ over 
it, we denote by $L_2(X;\mu)$ the space
of real-valued functions over $X$ endowed with
the expectation inner product.

$\mathsf{Pauli}$ is the set of 4 Pauli matrices $\{I,X,Y,Z\}$.

\subsection{High-Dimensional Expanders}\label{sec:hdx-prelims}
A $d$-dimensional simplicial complex $X = (X(0),\ldots,X(d))$ with vertex set $X(1) = [n]$ is a downwards closed collection of subsets of $[n]$. We follow the convention that $X(0) = \{\emptyset\}$, and for each $i>1$ the set of $i$-faces $X(i)$ is a collection of subsets of $X(1)$ of size $i$. The size of $X$ is the total number of faces in $X$. The degree of a vertex $v\in X(1)$ is the number of faces containing it, and the degree of $X$ is the maximum degree over all $v\in X(1)$.
\begin{definition}
For a $d$-dimensional simplicial complex $X = (X(0),X(1),\ldots,X(d))$, $0\leq i\leq d-2$
and $I\in X(i)$, the link of $I$ is the $(d-i)$-dimensional complex $X_I$ whose faces are given as
\[
X_I(j-i) = \{J\setminus I~|~J\in X(j), J\supseteq I\}.
\]
\end{definition}
For a $d$-dimensional complex $X=(X(0),X(1),\ldots,X(d))$ 
and $I\in X$ of size at most $d-2$, the 
graph underlying the link of $I$ is 
the graph whose vertices are $X_I(1)$ 
and whose edges are $X_I(2)$. 
We associate with $X$ a collection 
of distributions over faces. The distribution $\mu_d$ is the uniform 
distribution over $X(d)$, and 
for each $i<d$ the distribution 
$\mu_i$ is a distribution over 
$X(i)$ which results by picking 
$D\sim \mu_d$, and then taking 
$I\subseteq D$ of size $i$ 
uniformly. The distribution $\mu_{I,j-i}$ associated to the link of $I \in X(i)$ is the conditional distribution, $\mu_{j}~|~J \supseteq I$.

\paragraph{Density, average and average on links:}
for a set $S \subseteq X(i)$, let $\mu_i(S)$
denote its density with respect to $\mu_i$.
For $j>i$, $I\in X(i)$ and $S \subseteq X_I(j-i)$ let $\mu_{I,j-i}(S)$ denote the density of $S$ with respect to $\mu_{I,j-i}$. We often omit the subscript and simply write $\mu(S)$ and $\mu_{I}(S)$ in these cases when there is no risk of confusion. We extend the definition for functions, and for $F: X(j) \rightarrow \R$ we define $\mu(F) = \E_{J \sim \mu_j}[F(J)]$ as well as $\mu_{I}(F) = \E_{J \sim \mu_{I,j-i}}[F(J)]$. For a set $K \in X(k)$ for $k \geq j$, 
we define the restricted function $F|_{K}\colon \{J\in X(j)~|~J\subseteq K\}\to\mathbb{R}$ by
$F|_{K}(J) = F(J)$; we think of $\{J\in X(j)~|~J\subseteq K\}$ as being endowed with 
the natural condition measure of $\mu_j$ on this set, and hence define 
\[
\mu(F|_K)
=\E_{\substack{J\sim\mu_j\\ J \subseteq_j K}}[F(J)].
\]

We will use a special class of complexes known as `partite' complexes: 
\begin{definition}
A $d$-dimensional complex $X$ is said to be $d$-partite if $X(1)$ can be partitioned into $d$ disjoint parts called ``colors'', $X(1)=\cup_{i\in [d]} X_i(1)$, such that every $d$-face in $X(d)$ contains one vertex from each color. For a face $I \in X$ we 
let $\text{col}(I)$ denote the set of colors of the vertices it contains. For an $i$-sized set $S \subseteq [d]$, we will use $X_{S}(i) \subseteq X(i)$ to denote $i$-faces $I$ for which $\text{col}(I)=S$.
\end{definition}

Then a high-dimensional expander is defined as,
\begin{definition}
We say a $d$-dimensional simplicial complex $X$ is 
a $\lambda$ one-sided local spectral expander if for every $I\in X$ of size at most $d-2$, the second eigenvalue of the normalized adjacency matrix of the graph $(X_I(1), X_I(2))$ is at most $\lambda$, and it is a $d$-partite $\lambda$-HDX if it is also $d$-partite.
\end{definition}

Our arguments require that the HDX is $d$-partite with $\lambda \leq 1/\pl n$, as well as the distribution over the top faces and the vertices is regular. This can be satisfied by two well-known explicit constructions of HDX by~\cite{LSV1, kaufman2018construction}. 
\begin{theorem}[Sparse HDX construction]\label{thm:hdx-construction} 
For all universal constants $c$, large enough $n\in \N$ and some prime $q=\Theta(\log^c n)$, one can construct in time ${\sf poly}(n)$ a $d$-partite $\lambda$-HDX $X$, with $\lambda \leq 1/\sqrt{q}$, any $d$ smaller than  $\mathrm{loglog}^c n$ and $n \leq |X(1)|\leq 2n$. Furthermore, every vertex participates in at most $q^{O(d^2)}$ $d$-cliques, and all vertices have isomorphic links; in particular they have the same degree.\footnote{The interested reader can check out the original papers of~\cite{LSV1, kaufman2018construction} or the calculations in Lemma D.2 in~\cite{BafnaMV25}.
}
\end{theorem}

\subsubsection{Sampling Properties of HDX}
Recall that we associated with each $d$-dimensional simplicial 
complex $X$ a sequence of measures $\{\mu_k\}_{1\leq k\leq d}$, where 
$\mu_k$ is a probability measure over $X(k)$. Note that for all $0 \leq t \leq r \leq d$, a sample according to $\mu_t$ can be drawn by first sampling $R \sim \mu_r$, and then sampling $T\subseteq_{t} R$ uniformly. The converse is also true: a sample from $\mu_r$ can be drawn by first sampling $T \sim \mu_t$, and then sampling $R$ from $\mu_r$ conditioned on containing $T$. These observations 
give rise to the standard ``up'' and ``down'' operators, which we present next. %

\begin{definition}%
The operator $U_i^{i+1}$ is a map from $L_2(X(i); \mu_i)$ to $L_2(X(i+1); \mu_{i+1})$ defined as
\[
U_i^{i+1}f(u) 
= 
\E_{v \subseteq_i u}\big[f(v)\big]
\]
for all $u\in X(i+1)$. For $j\geq k+1$, we define $U_k^j$ via composition of up operators: $U_k^j = U_{j-1}^j \circ \ldots \circ U_k^{k+1}$.
\end{definition}

\begin{definition}%
The operator $D_i^{i+1}$ is a map from $L_2(X(i+1); \mu_{i+1})$ to $L_2(X(i); \mu_i)$ defined as
\[
D_i^{i+1}f(u) = \E_{v \supseteq_{i+1} u}\big[f(v)\big]
\]
for all $u \in X(i)$.
For $j\geq k+1$, we define $D_k^j$ via composition of down operators: $D_k^j = D_{k}^{k+1} \circ \ldots \circ D^j_{j-1}$.
\end{definition}

Abusing notations, we use the notations $U^j_k, D^j_k$ 
to denote the operators, as well as the real valued 
matrices associated with them. A key property of
the down and up operators is that they are adjoint:

\begin{claim}%
For all $k \leq j \leq d$, $U_k^{j}$ and $D^{j}_k$ are adjoint operators: for all functions $f\colon X(k)\to\mathbb{R}$ and $g\colon X(j)\to\mathbb{R}$ it holds that $\ip{U_k^{j}f,g} = \ip{f,D^{j}_kg}$. \end{claim}

We need the following lemma regarding the second eigenvalue of the down-up walks $U^j_{k}D^j_{k}$ on $X(j)$ ($j \geq k$), that can be found in~\cite{AlevL20}. Roughly speaking, the lemma asserts that for a
one-sided spectral expander $X$, the singular
values of the operators $U_{\alpha i}^{i}$ 
and $D_{\alpha i}^{i}$ for $\alpha\in (0,1)$
are upper bounded by the eigenvalues
of the corresponding operators in the 
complete complex, up to an additive factor if 
${\sf poly}(i)\gamma$:
\begin{lemma}\label{lem:spectral_gap_of_graphs_from_HDX}
Let $(X, \mu)$ be a $d$-dimensional $\lambda$ one-sided local spectral expander. For all $i \leq d$ and $\alpha \in (1/i, 1)$, the largest singular value of $U^i_{\alpha i}$ and $D^i_{\alpha i}$ is at most $\sqrt{\alpha}+\poly(i)\lambda$. Thus the down-up random walk $U^i_{\alpha i}D^i_{\alpha i}$ on $X(i)$ has second largest singular value at most $\alpha + \poly(i)\lambda$. 
\end{lemma}

\subsubsection{Concentration on HDX}\label{sec:concentration-prelims}
We will use the following Chernoff bound for HDXs due to~\cite{dikstein2024chernoff}. The version below is tailored for our purposes and it follows in a straightforward manner from~\cite[Theorem 5.1]{dikstein2024chernoff}.

\begin{lemma}[HDX Chernoff]\label{lem:chernoff-hdx}
There exist absolute constants $c,c_1,c_2>0$ such that if a $d$-dimensional complex $X$ is a $\lambda$-one-sided local spectral expander for $\lambda \leq \frac{c}{d}$, then for all $i < k \leq O(\sqrt{d})$ and all functions $F:X(i) \rightarrow [0,1]$,
\begin{align}
    \Pr_{K \sim X(k)}[|\mu(F|_K)-\mu(F)| \geq \eps] \leq \frac{c_1}{\eps}\exp\left(-c_2\eps^2 \frac{k}{i}\right).
    \label{eq:hdx-chernoff-1}
\end{align}
Furthermore, if $i \leq O(\sqrt{d})$, $k \geq \Omega(\sqrt{d})$ and $\frac{d}{i}$ is large enough as a polynomial function of $1/\eps$ then,
\begin{align}
\Pr_{K \sim X(k)}[|\mu(F|_K)-\mu(F)| \geq \eps] \leq \exp\left(-c'\varepsilon^2 \frac{\sqrt{d}}{i}\right),
\label{eq:hdx-chernoff-2}
\end{align}
for some absolute constant $c'>0$.
\end{lemma}

\begin{proof}
Let $k_0=c_3 \sqrt{d}$ for some fixed constant $c_3$. Theorem 5.1 in~\cite{dikstein2024chernoff} implies the first part of the lemma above. 
Let us therefore discuss the second part when $k \geq k_0$. Assume the result holds for $k_0$ and let $\cK_0$ be the set of $K \in X(k_0)$ where $|\mu(F|_K)-\mu(F)| \geq \eps/2$. By applying the concentration inequality in the first part we know that $$\mu(\cK_0) \leq \frac{c_1}{\varepsilon/2} \exp(-\frac{c_2 c_3 \varepsilon^2}{4} \frac{\sqrt{d}}{i}  ) \leq \exp(-c'' \varepsilon^2 \frac{\sqrt{d}}{i}) := p,$$
where the second inequality holds when $\frac{d}{i} \gg \frac{1}{\varepsilon^5}$. 
Furthermore, note that $p \ll \varepsilon^2$.

Define $\cK \subseteq X(k)$ by 
\[
\cK = \{K\in X(k)~|~\mu(\cK_0|_{K})\geq \sqrt{p}\},
\]
so that by Markov's inequality $\mu(\cK)\leq \sqrt{p}$.
On $K \notin \cK$, 
\[\mu(F|_K) = \E_{K' \subset_{k_0} K}[\mu(F|_{K'})] \leq (1-\sqrt{p})\cdot\left(\mu(F)+\frac{\eps}{2}\right)+\sqrt{p}\cdot 1 \leq \mu(F)+\eps.\]
Similarly for $K \notin \cK$, 
\[
\mu(F|_K)
=
\E_{K' \subset_{k_0} K}[\mu(F|_{K'})] 
\geq (1-\sqrt{p})\cdot\left(\mu(F)-\frac{\eps}{2}\right)
\geq \mu(F) - \eps.
\]
Combining the two inequalities we get that 
$|\mu(F|_{K})-\mu(F)|\leq \eps$ for all 
$K\in \cK$, and the second part of the lemma
follows with $c':= c''/2$.
\end{proof}

We will also use a sampling bound due to~\cite{BT} that works in the regime where $\mu(F)$ is small,

\begin{lemma}[HDX concentration in the low mean regime~\cite{BT}]\label{lem:concentration-improved}
    Let $X$ be a $d$-dimensional partite complex such that it is a $\lambda$-link expander, $\varepsilon > 0$, and let $f: X(i) \to \{0,1\}$ such that $\E_{v \sim \pi_k} f(v) = \mu$. Denoting $\overline{f} (\sigma) := \E_{\tau \sim \sigma} [f(\tau)]$ the simple averaging over $\tau \subset \sigma$, we have,

    $$\Pr_{\sigma \sim \pi_d} ( \overline{f}(\sigma) \geq \varepsilon ) \leq\frac{ 1 + O\left(\frac{i^2}{d} \right)}{\varepsilon^2 - O\left(\frac{i^2}{d} \right)} \cdot \mu \cdot (\mu  + i^2 \lambda).$$
\end{lemma}

\subsection{Quantum Error Correction}
We recall the stabilizer-code notions in quantum error correction. Throughout this subsection, equalities between Pauli operators are understood up to a global phase.
Let $\cQ$ be an $[[n,k,\Delta]]$ stabilizer code. Fix an independent generating set
$S_1,\dots,S_{n-k}$
for its stabilizer group, and fix logical $Z$ operators
$\overline Z_1,\dots,\overline Z_k$. For a Pauli operator $E$, define its syndrome $\syn(E)\in \{0,1\}^{n-k}$ by
\[
S_i E = (-1)^{\syn(E)_i} E S_i
\qquad\text{for } i=1,\dots,n-k.
\]

\begin{definition}[Clifford encoding unitary]\label{def:Clifford-encoding}
A \emph{Clifford encoding unitary} for $\cQ$ is an $n$-qubit Clifford circuit $V$, with no ancillas, such that
\[
V Z_{k+i} V^\dagger = S_i
\qquad\text{for } i=1,\dots,n-k,
\]
and
\[
V Z_i V^\dagger = \overline Z_i
\qquad\text{for } i=1,\dots,k.
\]
The corresponding \emph{encoding isometry} is
\[
\Enc_\cQ(\ket{\psi}) = V\bigl(\ket{\psi}\otimes \ket{0}^{\otimes (n-k)}\bigr),
\qquad
\ket{\psi}\in (\mathbb C^2)^{\otimes k}.
\]
Equivalently, one can check that for every $k$-qubit state $\ket{\phi}$ and every $s\in\{0,1\}^{n-k}$, the state $V(\ket{\phi}\otimes\ket{s})$ equals $E\Enc_{\cQ}(\ket{\phi})$ where $E$ is a Pauli error with syndrome $\ket{s}$. 
\end{definition}

Given a minimal stabilizer basis and logical operators, such a unitary exists and can be found efficiently~\cite{gottesman1997stabilizer}\footnote{Note that for quantum LDPC codes and quantum locally testable codes, the chosen generators $S_1,\dots,S_{n-k}$ need not coincide with the local check family, which may be redundant and much lower weight.}. 
For each syndrome $s\in\{0,1\}^{n-k}$, let
\[
\mathcal H_s := \Bigl\{\ket{\xi} : S_i\ket{\xi}=(-1)^{s_i}\ket{\xi}\text{ for all }i=1,\dots,n-k\Bigr\}.
\]
Then $\mathcal H_0$ is the codespace, the spaces $\mathcal H_s$ form an orthogonal decomposition of the full Hilbert space, and
\[
V^\dagger \mathcal H_s = (\mathbb C^2)^{\otimes k}\otimes \ket{s}.
\]
Indeed, if $\ket{\xi}\in \mathcal H_s$, then for every $i$,
\[
Z_{k+i}V^\dagger\ket{\xi}
=
V^\dagger S_i \ket{\xi}
=
(-1)^{s_i}V^\dagger\ket{\xi},
\]
so $V^\dagger\ket{\xi}$ lies in $(\mathbb C^2)^{\otimes k}\otimes \ket{s}$. Conversely, for every $k$-qubit state $\ket{\phi}$ and every $i$,
\[
S_iV(\ket{\phi}\otimes \ket{s})
=
V Z_{k+i}(\ket{\phi}\otimes \ket{s})
=
(-1)^{s_i}V(\ket{\phi}\otimes \ket{s}),
\]
so $V(\ket{\phi}\otimes \ket{s})\in \mathcal H_s$. Thus $V^\dagger$ exposes the stabilizer syndrome on the last $n-k$ qubits, although it might not decode to the ``closest codeword'' on the first $k$ qubits. A decoder $D$ can therefore be viewed as choosing one Pauli representative for each syndrome: it need not recover the actual physical error, only a Pauli with the same syndrome, and a good decoder is one which can identify the smallest Pauli correction given a syndrome.

\begin{definition}[Syndrome decoder]
A \emph{decoder} for $\cQ$ is a map
$D:\{0,1\}^{n-k}\to \mathsf{Pauli}^{\otimes n}$
such that
$\syn(D(s))=s$ for every $s\in\{0,1\}^{n-k}.$
A $t$-bounded-distance decoder is a decoder with the additional property that whenever a syndrome $s$ is realized by some Pauli error of weight at most $t$, the output $D(s)$ also has weight at most $t$. When $t<\Delta/2$, this implies
\[
D(\syn(E))\,E \in \operatorname{Stab}(\cQ)
\qquad\text{for every Pauli } E \text{ with } |E|\le t,
\]
because $D(\syn(E))E$ has trivial syndrome and weight at most $2t<\Delta$. 
\end{definition}
We remark that above we have stated a decoder as a map which may not have an efficient circuit that implements it. This is because we only use the \emph{existence} of such a bounded-distance decoder in the analysis of our protocols to move from a Pauli-deviated codewords to its nearest codeword and its decoding. Also note that such a decoder always exists for all $t$, by defining $D(s)$ to be the minimum weight Pauli error with syndrome $s$.

Given such a map $D$, we define a coherent version of ``read the syndrome, apply the chosen correction, and decode.'' On the syndrome-$s$ sector, the correction $D(s)$ returns the state to syndrome $0$, then $V^\dagger$ decodes the logical qubits, and finally the value $s$ is written into the last $n-k$ output qubits. Keeping this syndrome register is what makes the overall map reversible.

\begin{definition}[Reversible decoding unitary~\cite{he2025composable}]
\label{def:reversible-decoding}
For each $s\in\{0,1\}^{n-k}$, let
\[
P_s := V\bigl(I\otimes \ketbra{s}{s}\bigr)V^\dagger
\]
denote the projector onto the syndrome-$s$ subspace $\mathcal H_s$. The \emph{reversible decoding unitary} associated with $D$ is
\begin{equation}\label{eq:rev-dec}
U_\cQ
=
\sum_{s\in\{0,1\}^{n-k}}
\bigl(I\otimes \ketbra{s}{0^{n-k}}\bigr)\,V^\dagger D(s)\,P_s
=
\sum_{s\in\{0,1\}^{n-k}}
\bigl(I\otimes \ketbra{s}{0^{n-k}}\bigr)\,
V^\dagger D(s) V
\bigl(I\otimes \ketbra{s}{s}\bigr)
V^\dagger .
\end{equation}
\end{definition}

\begin{claim}[\cite{he2025composable}]
The operator $U_\cQ$ defined in~\eqref{eq:rev-dec} is unitary.
\end{claim}

\begin{proof}
    For an (arbitrarily) corrupted codeword $\ket{\xi} = E \ \mathrm{Enc}(\ket{\psi})$ for some Pauli $E$, we interpret $U\ket{\xi}$ as follows. When applying $V^\dagger$ on $\ket{\xi}$, we obtain the syndrome on the last $n-k$ qubits (while the first $k$ qubits could be in a completely wrong logical state, unless $E=I$). So $V (I \otimes \ketbra{s}{s}) V^\dagger$ is essentially a projector onto the subspace of syndrome $s$. Then we apply the correction $D(s)$, which could be wrong (not stabilizer-equivalent to $E$), depending on the performance guarantee of $D$ (e.g. if the error $E$ is too large). After $D(s)$ we are guaranteed the state has zero syndrome (it's in the codespace, but the encoded message could be different from $\ket{\psi}$ if $D(s)$ is not equivalent to $E$). Finally, $V^\dagger$ undoes the Clifford encoding, and the $(I \otimes \ketbra{s}{0})$ replaces the $\ket{0^{n-k}}$ syndrome string by $\ket{s}$.
    
    With the above interpretation, one can easily verify that $U$ is a unitary
    \begin{align}
        UU^\dagger &= \sum_{s,s' \in \{0,1\}^{n-k}} (I \otimes \ketbra{s}{0}) V^\dagger D(s) V (I \otimes \ketbra{s}{s}) V^\dagger  V (I \otimes \ketbra{s'}{s'}) V^\dagger D(s') V (I \otimes \ketbra{0}{s'})\\
        &= \sum_{s \in \{0,1\}^{n-k}} (I \otimes \ketbra{s}{0}) \underbrace{V^\dagger D(s) V (I \otimes \ketbra{s}{s}) V^\dagger D(s) V}_{I \otimes \ketbra{0}{0}} (I \otimes \ketbra{0}{s})\\
        &= \sum_{s \in \{0,1\}^{n-k}} I \otimes \ketbra{s}{s} = I. 
    \end{align}    
\end{proof}

We next record the effect of $U_\cQ$ on corrupted codewords. %

\begin{claim}[Reversible decoding on a stabilizer code]\label{claim:reversible-decoding-properties}
Let $\cQ$ be an $[[n,k,\Delta]]$ stabilizer code with encoding isometry $\Enc_\cQ$, and fix a minimal generating set for its stabilizer group. For a Pauli operator $E$, let $\syn(E)\in\{0,1\}^{n-k}$ denote its syndrome in this basis. Let
$D:\{0,1\}^{n-k}\to \mathsf{Pauli}^{\otimes n}$
be any decoder, and let $U_\cQ$ be the reversible decoding unitary associated with $D$ from \Cref{def:reversible-decoding}.

Then for every Pauli error $E$ there exist a phase $\omega_E\in U(1)$ and a $k$-qubit Pauli operator $P_E$ such that for every $k$-qubit state $\ket{\psi}$,
\[
U_\cQ E\Enc_\cQ(\ket{\psi}) = \omega_E\, P_E\ket{\psi}\ket{\syn(E)}.
\]

Moreover, if $D$ is a $t$-distance (possibly inefficient) decoder for $t<\Delta/2$ then for every error $E$ with weight at most $t$ and every $k$-qubit state $\ket{\psi}$,
\[
U_\cQ E\Enc_\cQ(\ket{\psi}) = \omega_E\, \ket{\psi}\ket{\syn(E)}
\]
for some phase $\omega_E\in U(1)$. 
\end{claim}

\begin{proof}
Let $s=\syn(E)$. Since $E\Enc_\cQ(\ket{\psi})$ lies entirely in the syndrome-$s$ subspace $\mathcal H_s$, only the $s$-summand in the definition of $U_\cQ$ contributes. Thus
\[
U_\cQ E\Enc_\cQ(\ket{\psi})
=
\bigl(I\otimes \ketbra{s}{0^{n-k}}\bigr)\,V^\dagger D(s)E\Enc_\cQ(\ket{\psi}).
\]
Now
\[
\syn(D(s)E)=\syn(D(s))\oplus \syn(E)=s\oplus s=0,
\]
so $D(s)E$ has trivial syndrome and therefore preserves the codespace. Since $D(s)E$ is a Pauli operator in the normalizer of the stabilizer, its action on the codespace is, up to a phase, an encoded logical Pauli. Hence there exist a phase $\omega_E\in U(1)$ and a $k$-qubit Pauli operator $P_E$ such that
\[
D(s)E\Enc_\cQ(\ket{\psi}) = \omega_E\, \Enc_\cQ(P_E\ket{\psi})
\]
for every $\ket{\psi}$. Applying $V^\dagger$ gives
\[
V^\dagger D(s)E\Enc_\cQ(\ket{\psi})
=
\omega_E\, P_E\ket{\psi}\ket{0}^{\otimes (n-k)},
\]
and then
\[
U_\cQ E\Enc_\cQ(\ket{\psi})
=
\omega_E\, P_E\ket{\psi}\ket{s},
\]
which proves the first statement.

If $|E|\le t$ and $D(\syn(E))E\in \operatorname{Stab}(\cQ)$, then the induced logical action of $D(s)E$ on the codespace is trivial, so the corresponding logical Pauli is $P_E=I$. Therefore
\[
U_\cQ E\Enc_\cQ(\ket{\psi}) = \omega_E\, \ket{\psi}\ket{\syn(E)}.
\]
Finally, if $D$ is a $t$-bounded-distance decoder with $t<\Delta/2$ and $|E|\le t$, then $D(\syn(E))E \in \operatorname{Stab}(\cQ)$ as required.
\end{proof}

We also need the corresponding statement for the classical repetition code interpreted as a quantum code. Let $\Rep$ denote the length-$d$ repetition code with encoding isometry
\[
\Enc_{\Rep}(\alpha\ket{0}+\beta\ket{1})=\alpha\ket{0^d}+\beta\ket{1^d}.
\]
For a string $x\in\{0,1\}^d$, let $\syn_{\Rep}(x)\in\{0,1\}^{d-1}$ denote its syndrome with respect to the standard parity checks $x_i\oplus x_{i+1}$, and let $b(x)\in\{0,1\}$ be the majority bit of $x$. If $d$ is even, fix a tie-breaking rule satisfying $b(\overline{x})=1\oplus b(x)$. Then the map
\[
x\longmapsto \bigl(b(x),\syn_{\Rep}(x)\bigr)
\]
is a bijection, and therefore there exists a unitary $U_{\Rep}$ such that
\[
U_{\Rep}\ket{x}=\ket{b(x)}\ket{\syn_{\Rep}(x)}
\qquad\text{for every } x\in\{0,1\}^d.
\]

\begin{claim}[Reversible decoding for the repetition code]\label{claim:dec-rep}
For each $I\subseteq [d]$, let
\[
x_I:=1_I0_{[d]\setminus I}
\qquad\text{and}\qquad
b(I):=b(x_I).
\]
Then for every Pauli operator $E=Z^J X^I$ on the repetition block and every one-qubit state $\ket{\psi}$,
\[
U_{\Rep}E\Enc_{\Rep}(\ket{\psi})
=
X^{b(I)}Z^{|J|\bmod 2}\ket{\psi}\ket{\syn_{\Rep}(x_I)},
\]
up to a global phase. Moreover, if $|I|<d/2$, then $b(I)=0$. In particular, when fewer than $d/2$ qubits suffer an $X$-error, the decoded logical error is either $I$ or $Z$; no logical $X$ component appears.
\end{claim}

\begin{proof}
Write $\ket{\psi}=\alpha\ket{0}+\beta\ket{1}$ and set $x:=x_I$. Then
\[
X^I\Enc_{\Rep}(\ket{\psi})
=
\alpha\ket{x}+\beta\ket{\overline{x}},
\]
where $\overline{x}=1^d\oplus x$. Applying $Z^J$ contributes a phase $(-1)^{|I\cap J|}$ to $\ket{x}$ and a phase $(-1)^{|J\setminus I|}$ to $\ket{\overline{x}}$. Factoring out the global phase $(-1)^{|I\cap J|}$, we obtain
\[
E\Enc_{\Rep}(\ket{\psi})
=
(-1)^{|I\cap J|}\Bigl(\alpha\ket{x}+(-1)^{|J|}\beta\ket{\overline{x}}\Bigr).
\]
By definition of $U_{\Rep}$,
\[
U_{\Rep}\ket{x}=\ket{b(I)}\ket{\syn_{\Rep}(x_I)}
\qquad\text{and}\qquad
U_{\Rep}\ket{\overline{x}}=\ket{1\oplus b(I)}\ket{\syn_{\Rep}(x_I)},
\]
because $\syn_{\Rep}(\overline{x})=\syn_{\Rep}(x)$ and $b(\overline{x})=1\oplus b(x)$ by construction. Therefore
\[
U_{\Rep}E\Enc_{\Rep}(\ket{\psi})
=
(-1)^{|I\cap J|}\Bigl(\alpha\ket{b(I)}+(-1)^{|J|}\beta\ket{1\oplus b(I)}\Bigr)\ket{\syn_{\Rep}(x_I)}.
\]
The qubit in parentheses is exactly $X^{b(I)}Z^{|J|\bmod 2}\ket{\psi}$, up to the same global phase. This proves the claimed formula.

Finally, if $|I|<d/2$, then the string $x_I$ has strict majority $0$, so $b(I)=0$.
\end{proof}

\subsection{Fault Tolerance Notions}\label{sec:prelims-ft}
Here, we will define essential notions and properties of our fault-tolerance protocols constructed for our relevant noise model. We borrow ideas from~\cite{nguyen2025quantum, he2025composable}, although we specialize to codes and protocols with bounded-error guarantees.

\begin{definition}[Deviations] Let $ 0\leq t \leq 1$.
    For two (not necessarily normalized) pure states \(\ket{\phi}\) and \(\ket{\psi}\), where $\ket{\phi}$ is on $n$ qubits,
  \(\ket{\psi}\) is said to be a \emph{$t$-deviation} of \(\ket{\phi}\) if there exists a decomposition
  \begin{align}
    \label{eq:deviation-decomp}
    \ket{\psi} = \sum_i \ket{\phi^{(i)}}\otimes \ket{\xi^{(i)}},
  \end{align}
  where $\ket{\phi^{(i)}}$ are states on the same qubits as $\ket{\phi}$, and $\ket{\xi^{(i)}}$ are states on an auxiliary system, such that for each $i$, there exist a set \(S^{(i)} \subseteq [n]\) with $|S^{(i)}| \leq t n$ and Pauli operator \(N_{S^{(i)}}\) supported on $S^{(i)}$ such that \(\ket{\phi^{(i)}} = \left(\mathrm{Id}_{[n]\setminus S^{(i)}}\otimes N_{S^{(i)}}\right)(\ket{\phi})\).
  When the sets $S^{(i)}=S$ for all $i$, we sometimes explicitly use the terminology \emph{$S$-deviation}. When the sum has only one term, we say that $\ket{\psi}$ is a \emph{Pauli $t$-deviation} of $\ket{\phi}$.
\end{definition}

\begin{definition}[Generalized deviations]
    Generalizing the above definition to the case when $\ket{\phi}$ is viewed as composed of multiple registers $[n] \cong R_1 \cup R_2 \cup ... \cup R_m$, we say that $\ket{\psi}$ is a \emph{$\vec{t}$-deviation} of $\ket{\phi}_{R_1R_2...R_m}$ for a vector $\vec{t} = [t_1,t_2,...,t_m]$ if in the above decomposition of~\Cref{eq:deviation-decomp}, each $S^{(i)}$ includes at most $t_j |R_j|$ qubits on register $j$. If $t_1 = t_2 = ...=t_m$, we simply say $t$-deviation and omit the register labels when clear from context. Similarly we can define a notion of fixed-set $(S_1 \cup S_2 \cup...\cup S_m)$-deviation for the generalized case. %
\end{definition}

Our most frequent use of the above notion of (generalized) deviation is when $\ket{\phi}$ is a codeword of multiple code blocks. This means that each codeblock is not too deviated from the desired codeword.

\begin{remark}
        The above notion of deviation is strictly stronger than other notions such as in~\cite{aharonov1999fault, crepeau2002secure,nguyen2025quantum}. In particular,~\cite{aharonov1999fault}'s deviation notion only requires that a state and be transformed to another via a local map.~\cite{crepeau2002secure}'s notion is similar to our $S$-deviation notion, but restricts the faults giving rise to it to be physical quantum channels. In this work, we deal with a non-physical noise model from the circuit-to-Hamiltonian reduction that allows a linear combination of adversarial Pauli errors, hence we need a strong notion of deviation that can handle such superpositions. 
\end{remark}

We will construct distributed fault-tolerant protocols on graphs that maintain deviations from the desired code states, formalized as below.
\begin{definition}[Multiparty computation protocol]
    \label{def:protocol}
    A \emph{multiparty computation protocol} $\cP$ on a multigraph $G=(V,E)$ is a protocol involving $|V|$ players, where each player holds a register of qubits and can perform local quantum computations on their register and communicate qubits to other players across edges in $E$. The protocol $\cP$ specifies the local computations and communications to be performed by the players. When viewed as a quantum circuit consisting of bounded-locality operations, the protocol has \emph{depth} $D$. The \emph{work complexity} $m$ is the maximum complexity of the computation and communication performed by any vertex throughout the protocol. In the $\mathsf{QMA}$ setting, the protocol also takes in a witness state $\ket{\xi}$, distributed across the players as specified by the protocol.
\end{definition}

\begin{definition}[Edge Pauli faults and supports]
    \label{def:edge-pauli-fault}
    Let $\cP$ be a distributed protocol on a multigraph $G=(V,E)$ with set of wires $A$ when viewed as a circuit of bounded-locality operations, and let $\mu$ be a probability measure over $E$. In a faulty execution of $\cP$, a list of Pauli operators $\mathsf{F} \in \mathsf{PAULI}^A$ are inserted at a subset of locations in $A$. We say that the execution is faulty on edge $(i,j)$ if there is a non-trivial Pauli fault in $\mathsf{F}$ touching a qubit either in player $i$ or player $j$'s registers, or occurring in a communication between them. Given a set of edges $F \subseteq E$, an \emph{$F$-edge faulty} execution of $\cP$ is one where the set of faulty edges is contained in $F$. We define the support of $\mathsf{F}$ as $\operatorname{supp}_{\cP}(\mathsf{F})=F$ and denote the faulty protocol as $\cP_{\mathsf{F}}$. If $\mu(F) \leq \delta$, we say the execution is a \emph{$\delta$-edge faulty} execution of $\cP$.
\end{definition}

\begin{definition}[Fault-tolerant protocols] \label{def:fault-tolerant-protocol}
Let $\mathcal{P}$ be a quantum MPC protocol on a multigraph $G=(V,E)$ with edge measure $\mu$, with $n$ designated input qubits and $n'$ designated output qubits.
Let $C$ be a quantum operation mapping $k$ to $k'$ qubits.
Let \(\mathcal{Q}, \mathcal{Q}'\) be quantum codes with parameters \([n,k]\), \([n',k']\) and associated physical qubit partitionings $[n]\cong R_1 \cup ... \cup R_{m}, [n']\cong R'_1 \cup ... \cup R'_{m'}$ with deviation thresholds $\vec{t} \in [0,1]^{m}$, $\vec{t'} \in [0,1]^{m'}$, respectively. For $F \subseteq E$, the protocol $\cP$ is said to be a \emph{$(\mathcal{Q},\vec{t},\mathcal{Q'},\vec{t'},F)$-fault-tolerant protocol} for $C$ if the following statement holds: 

For any $k$-qubit pure state $\ket{\psi}$, upon any input state $\ket{\Psi}$ that is Pauli \(\vec{t}\)-deviated from \(\mathrm{Enc}_\mathcal{Q}(\ket{\psi})\), an $F$-edge faulty execution of $\cP$ specified by Pauli faults $\mathsf{F}
$ outputs a state \(\cP_\mathsf{F}(\ket{\Psi})\) that is \(\vec{t'}\)-deviated from \(\mathrm{Enc}_{\mathcal{Q}'}(C\ket{\psi})\).
    
\noindent  Furthermore, if the above statement holds for all $F$ such that $\mu(F) \leq \epsilon$, then we say $\cP$ is a $(\mathcal{Q},\vec{t},\mathcal{Q'},\vec{t'},G_\epsilon)$-FT protocol for $C$. In the special case where $\mathcal{Q}=\mathcal{Q'}$ and $\vec{t} =\vec{t'}$, or when one of the codes is trivial (e.g., when $C$ is a state preparation operation) we simply use the terms $(\mathcal{Q},\vec{t},F)$-FT or $(\mathcal{Q},\vec{t}, G_\epsilon)$-FT protocol.
    
\end{definition}

We will mainly consider the case where $\mathcal{Q}$ (resp. $\mathcal{Q}'$) are composed of $k$ (resp. $k'$) code blocks of a few underlying codes (e.g. quantum locally testable code or repetition code), each encoding a single logical qubit. Then, the partitioning into registers in the above definition is the natural one according to the code blocks, so $m=k$ and $m'=k'$. If the codes are of the same type we will also often use uniform thresholds $\vec{t}=[t]^k$, $\vec{t'}=[t']^{k'}$ and accordingly abbreviate $(\mathcal{Q},t, \mathcal{Q}',t',F)$-FT.

\begin{definition}[Friendliness]
     In the above definition, the protocol is further said to be \emph{friendly} if for any input state $\ket{\psi}$ (arbitrarily) Pauli-deviated from a $\mathcal{Q}$-codeword, the output state $\cP_\mathsf{F}(\ket{\psi})$ is $\vec{t'}$-deviated from \(\mathrm{Enc}_{\mathcal{Q}'}(G\ket{\xi})\) for \emph{some} state $\ket{\xi}$.
\end{definition}

\begin{remark}\label{remark:pauli}
    Note that by definition, it holds that \(\sum_{\mathsf{F}: \supp(\mathsf{F}) \in F} \alpha_\mathsf{F} \cP_\mathsf{F}(\ket{\psi})\) is \(t'\)-deviated from \(\mathrm{Enc}_{\mathcal{Q}'}(C\ket{\phi})\) for all coefficients $\alpha_\mathsf{F}$, which is a strong guarantee that can handle superpositions of Pauli errors. Thus, by linearity, it suffices to show a gadget is fault-tolerant when the input state is a Pauli deviation of an encoded state and the circuit fault is also Pauli.
\end{remark}

We will also extend these definitions to the setting of \emph{partial} fault tolerance where we allow for some of the codeblocks to suffer from arbitrary deviation. These notions will be required in the analysis of our routing and distributed computing protocols later on,  so we introduce them in~\Cref{sec:routing}.

We end this section by observing that fault-tolerant gadgets can be composed with each other, which makes our fault-tolerance proofs modular. 

\begin{proposition}[Sequential composability of fault-tolerant gadgets]\label{def:seq-composability}
Let \(\mathcal{Q}_1,\mathcal{Q}_2,\mathcal{Q}_3\) be codes with parameters \([n_1,k_1]\), \([n_2,k_2]\), \([n_3,k_3]\) and associated deviation thresholds \(t_1,t_2,t_3\). Let \(G_1,G_2\) be operations mapping \(k_1\) to \(k_2\) qubits and \(k_2\) to \(k_3\) qubits. Let \(C_1,C_2\) be \((\mathcal{Q}_1,t_1,\mathcal{Q}_2,t_2,F_1)\)- and \((\mathcal{Q}_2,t_2,\mathcal{Q}_3,t_3,F_2)\)-fault-tolerant gadgets for \(G_1,G_2\), respectively. Then the circuit
$C=C_2\circ C_1$ is a \((\mathcal{Q}_1,t_1,\mathcal{Q}_3,t_3,F_1\cup F_2)\)-fault-tolerant gadget for
the operation $G_2\circ G_1.$
\end{proposition}

\begin{proof}
Fix any \(k_1\)-qubit pure state \(\ket{\psi}\), any Pauli \(t_1\)-deviation \(\ket{\Psi}\) of \(\mathrm{Enc}_{\mathcal{Q}_1}(\ket{\psi})\), and any \((F_1\cup F_2)\)-edge faulty execution of \(C_2\circ C_1\), specified by a Pauli fault pattern \(\mathsf{F}\). Let \(\mathsf{F}_1,\mathsf{F}_2\) denote the restrictions of \(\mathsf{F}\) to the locations belonging to \(C_1\) and \(C_2\), respectively. Then
\[
(C_2\circ C_1)_{\mathsf{F}}
=
(C_2)_{\mathsf{F}_2}\circ (C_1)_{\mathsf{F}_1},
\qquad
\operatorname{supp}_{C_1}(\mathsf{F}_1)\subseteq F_1,
\quad
\operatorname{supp}_{C_2}(\mathsf{F}_2)\subseteq F_2.
\]

Since \(C_1\) is \((\mathcal{Q}_1,t_1,\mathcal{Q}_2,t_2,F_1)\)-fault-tolerant for \(G_1\), the intermediate state
$\ket{\Omega}:=(C_1)_{\mathsf{F}_1}\ket{\Psi}$
is a \(t_2\)-deviation of \(\mathrm{Enc}_{\mathcal{Q}_2}(G_1\ket{\psi})\). Hence, by the definition of deviation, there exists a decomposition
\[
\ket{\Omega}
=
\sum_i P_i\,\mathrm{Enc}_{\mathcal{Q}_2}(G_1\ket{\psi})\otimes \ket{\xi_i},
\]
where each \(P_i\) is a Pauli operator supported on at most \(t_2 n_2\) qubits.
Now apply \((C_2)_{\mathsf{F}_2}\) termwise. For each \(i\), the state
$P_i\,\mathrm{Enc}_{\mathcal{Q}_2}(G_1\ket{\psi})$
is a Pauli \(t_2\)-deviation of \(\mathrm{Enc}_{\mathcal{Q}_2}(G_1\ket{\psi})\). Since \(C_2\) is \((\mathcal{Q}_2,t_2,\mathcal{Q}_3,t_3,F_2)\)-fault-tolerant for \(G_2\), we obtain
\[
(C_2)_{\mathsf{F}_2}\!\left(P_i\,\mathrm{Enc}_{\mathcal{Q}_2}(G_1\ket{\psi})\right)
=
\sum_j Q_{i,j}\,\mathrm{Enc}_{\mathcal{Q}_3}(G_2G_1\ket{\psi})\otimes \ket{\eta_{i,j}},
\]
where each \(Q_{i,j}\) is a Pauli operator supported on at most \(t_3 n_3\) qubits.
Therefore, by linearity,
\[
(C_2\circ C_1)_{\mathsf{F}}\ket{\Psi}
=
\sum_{i,j}
Q_{i,j}\,\mathrm{Enc}_{\mathcal{Q}_3}(G_2G_1\ket{\psi})
\otimes
(\ket{\eta_{i,j}}\otimes \ket{\xi_i}),
\]
which is a \(t_3\)-deviation of \(\mathrm{Enc}_{\mathcal{Q}_3}(G_2G_1\ket{\psi})\). This is exactly the desired fault-tolerance guarantee.
\end{proof}

\begin{lemma}[Parallel composability]
\label{lem:parallel-composability}
Let \(\mathcal{Q}_1,\mathcal{Q}_2,\mathcal{Q}_1',\mathcal{Q}_2'\) be codes with parameters \([n_1,k_1]\), \([n_2,k_2]\), \([n_1',k_1']\), \([n_2',k_2']\) and associated deviation thresholds \(t_1,t_2,t_1',t_2'\). Let \(G_1,G_2\) be operations mapping \(k_1\) to \(k_1'\) qubits and \(k_2\) to \(k_2'\) qubits, respectively. Let \(C_1,C_2\) be \((\mathcal{Q}_1,t_1,\mathcal{Q}_1',t_1',F_1)\)- and \((\mathcal{Q}_2,t_2,\mathcal{Q}_2',t_2',F_2)\)-fault-tolerant gadgets for \(G_1,G_2\), respectively\footnote{To be precise, we assume in addition that the fault-tolerance guarantee in \Cref{def:fault-tolerant-protocol} is preserved under adjoining an untouched reference register: namely, whenever \(C\) is \((\mathcal{Q},t,\mathcal{Q}',t',F)\)-fault-tolerant for \(G\), then for every reference register \(R\), the protocol $C$ leaves $R$ untouched. This is always satisfied in our setup.
}. 

Then
$C_1\otimes C_2$
is a
$(\mathcal{Q}_1\otimes \mathcal{Q}_2,\,(t_1,t_2),\,\mathcal{Q}_1'\otimes \mathcal{Q}_2',\,(t_1',t_2'),\,F_1\cup F_2)
$
-fault-tolerant gadget for
$G_1\otimes G_2$.
\end{lemma}

\begin{proof}
Fix any \((k_1+k_2)\)-qubit pure state \(\ket{\psi}\), and let
\[
\ket{\Gamma}
:=
(\mathrm{Enc}_{\mathcal{Q}_1}\otimes \mathrm{Enc}_{\mathcal{Q}_2})(\ket{\psi}).
\]
Let \(\ket{\Psi}\) be any Pauli \((t_1,t_2)\)-deviation of \(\ket{\Gamma}\). Then, with respect to the natural partition into the two code blocks, we may write
\[
\ket{\Psi}=(P\otimes Q)\ket{\Gamma},
\]
where \(P\) is supported on at most \(t_1 n_1\) qubits of the first block and \(Q\) is supported on at most \(t_2 n_2\) qubits of the second block.

Now fix any \((F_1\cup F_2)\)-edge faulty execution of \(C_1\otimes C_2\), specified by a Pauli fault pattern \(\mathsf{F}\), and let \(\mathsf{F}_1,\mathsf{F}_2\) be its restrictions to \(C_1\) and \(C_2\). Then
$(C_1\otimes C_2)_{\mathsf{F}}
=
(C_1)_{\mathsf{F}_1}\otimes (C_2)_{\mathsf{F}_2}.$

Since \(C_1\) is fault-tolerant, we can view the second code block as an untouched reference register. Applying this to the state \((P\otimes I)\ket{\Gamma}\), we obtain a decomposition
\[
\bigl((C_1)_{\mathsf{F}_1}\otimes I\bigr)(P\otimes I)\ket{\Gamma}
=
\sum_i (P_i\otimes I)\ket{\Theta}\otimes \ket{\xi_i},
\]
where each \(P_i\) is supported on at most \(t_1' n_1'\) qubits of the first output block, and $\ket{\Theta}=
(\mathrm{Enc}_{\mathcal{Q}_1'}G_1\otimes \mathrm{Enc}_{\mathcal{Q}_2})(\ket{\psi}).$ Since \(Q\) acts only on the second block, it commutes with \((C_1)_{\mathsf{F}_1}\otimes I\), and hence
\[
\bigl((C_1)_{\mathsf{F}_1}\otimes I\bigr)\ket{\Psi}
=
\sum_i (P_i\otimes Q)\ket{\Theta}\otimes \ket{\xi_i}.
\]

Next define
$\ket{\Theta'}
:=
(\mathrm{Enc}_{\mathcal{Q}_1'}G_1\otimes \mathrm{Enc}_{\mathcal{Q}_2'}G_2)(\ket{\psi}).$ Again for \(C_2\) with the first block treated as the untouched reference register, we get
\[
\bigl(I\otimes (C_2)_{\mathsf{F}_2}\bigr)(I\otimes Q)\ket{\Theta}
=
\sum_j (I\otimes R_j)\ket{\Theta'}\otimes \ket{\eta_j},
\]
where each \(R_j\) is supported on at most \(t_2' n_2'\) qubits of the second output block.
Therefore,
\begin{align*}
(C_1\otimes C_2)_{\mathsf{F}}\ket{\Psi}
&=
\bigl(I\otimes (C_2)_{\mathsf{F}_2}\bigr)
\bigl((C_1)_{\mathsf{F}_1}\otimes I\bigr)\ket{\Psi} \\
&=
\sum_i (P_i\otimes I)
\bigl(I\otimes (C_2)_{\mathsf{F}_2}\bigr)(I\otimes Q)\ket{\Theta}
\otimes \ket{\xi_i} \\
&=
\sum_{i,j}
(P_i\otimes R_j)\ket{\Theta'}
\otimes
(\ket{\eta_j}\otimes \ket{\xi_i}).
\end{align*}
This is a \((t_1',t_2')\)-deviation of \(\ket{\Theta'}\), i.e. of
$(\mathrm{Enc}_{\mathcal{Q}_1'}\otimes \mathrm{Enc}_{\mathcal{Q}_2'})((G_1\otimes G_2)\ket{\psi}),$ as required.
\end{proof}

%% file: mpc.tex
\section{Quantum Multiparty Computation with Zero Error}
\label{sec:mpc}

In this section, we construct a quantum multiparty computing (MPC) protocol with zero error (or perfect information-theoretic security) on the complete graph $K$ that can tolerate an inverse polylogarithmic fraction of corrupted players and communication links. The zero-error guarantee means that, if the number of corrupted players and communication links is below threshold, then the output state is guaranteed to be exactly correct, up to faults on some bounded number of qubits. This is in contrast to the standard notion of security in quantum MPC~\cite{crepeau2002secure}, which only guarantees that the output state is exactly correct up to faults on some bounded number of qubits when the protocol does not `fail', which is an exponentially unlikely event in the number of players. The zero-error guarantee is necessary because of the strong non-physical noise model that our circuit-to-Hamiltonian reduction produces. Our quantum MPC protocol relies on the existence of an explicit family of quantum locally testable codes (qLTCs)~\cite{dinur2024expansion} with efficient and robust decoding algorithms~\cite{nguyen2025quantum}. We also need asymptotically good codes with transversal CCZ~\cite{nguyen2025good, golowich2025asymptotically} in order to do gate teleportation for a universal gate set.

\Cref{thm:distributed-FTQC} is the main theorem of this section, and its proof constructs a zero-error quantum MPC protocol. Broadly speaking, this protocol is built along similar lines to previous classical and quantum MPC protocols. At the start, the players jointly hold an encoding of the input state $\ket{\psi}$ under a quantum error-correcting code $\mathcal{Q}$: every qubit of $\ket{\psi}$ is encoded individually into one codeblock of $\mathcal{Q}$, and one physical qubit from each encoded block is given to each player. This encoding functions as a \emph{quantum secret-sharing} of the input. Before the computation begins, the players run a verifiable quantum secret-sharing (VQSS) protocol to certify jointly that each purported sharing is close to a valid codeword. (VQSS must be sound/secure even in the presence of a bounded number of corrupted communication edges.) Once the inputs have been verified, the players compute jointly on their shared state by performing local operations and communicating while the data remain encoded, and hence without reconstructing the private logical state at any one player.

As we mention in the introduction, we are able to achieve zero-error VQSS by using a quantum locally testable code (qLTC). The local checks of a qLTC suggest a natural VQSS procedure (as described in the introduction). Each application of VQSS may introduce a fresh set of faults onto the shared state, and a qLTC with a robust decoder supplies an error-suppression procedure that keeps faults that arise from multiple uses of VQSS from spreading.

Some logical gates, including Pauli and CNOT gates, can be applied transversally for any CSS code. To implement the remaining gates, we use auxiliary-state gate teleportation~\cite{gottesman1999demonstrating}. In its simplest form, gate teleportation works as follows. Let $\ket{\Phi^+}=(\ket{00}+\ket{11})/\sqrt{2}$, and prepare the auxiliary state $(I\otimes G)\ket{\Phi^+}$. A Bell measurement between the input $\ket{\psi}$ and the first half of this auxiliary state produces two bits $(a,b)$, after which the remaining half is in the state
\begin{equation}
G X^a Z^b\ket{\psi}
=
\bigl(GX^aZ^bG^\dagger\bigr)G\ket{\psi}.
\end{equation}
The measurement outcomes $(a,b)$ therefore specify the correction $\bigl(GX^aZ^bG^\dagger\bigr)^\dagger$, and applying this correction leaves exactly $G\ket{\psi}$. When $G$ is Hadamard or CZ, the required corrections can be performed using Pauli gates, while for CCZ the corrections can be applied using only Pauli and CZ gates. Consequently, once these simpler correction gates are available, implementing $G$ reduces to preparing a reliable encoded copy of its associated auxiliary state $\ket{\phi_{G}} \coloneqq (I\otimes G)\ket{\Phi^+}$.

It remains to prepare the auxiliary states reliably. Since each such state is fixed and contains no private information, every player can locally prepare a copy, encode it using the code $\mathcal{Q}$, and distribute its shares among the other players. Honest preparers supply the intended state, while corrupted preparers may supply arbitrary states. The players run VQSS on each candidate state, and candidates that are too far from the code space of $\mathcal{Q}$ are rejected. An accepted candidate may nevertheless encode the wrong logical state, so the players then execute a \emph{state-distillation protocol} at the logical level of the code $\mathcal{Q}$. That is, each wire of the circuit $C_{\mathrm{dist}}$ for the state-distillation protocol is associated physically with a distributed block of the code $\mathcal{Q}$, and each gate in the circuit $C_{\mathrm{dist}}$ is implemented jointly by the players using the distributed logical implementation of that gate. We require that the state-distillation protocol for some auxiliary state $\ket{\phi_{G}}$ takes as input a set of candidate auxiliary states (some bounded fraction of which may be corrupted) and outputs exactly the intended auxiliary state $\ket{\phi_{G}}$ (provided that the number of corrupted candidates in the input does not exceed threshold). We also require that the circuit $C_\mathrm{dist}$ of the distillation protocol uses only gates whose logical/distributed implementations have already been constructed, so that the players can implement $C_{\mathrm{dist}}$ logically using existing operations. A distillation protocol for stabilizer states with the required properties (in particular, where $C_\mathrm{dist}$ uses only CNOT gates, measurements, and adaptive Pauli gates) is supplied by~\cite{nguyen2025quantum}. The recent construction of good quantum codes with transversal CCZ~\cite{nguyen2025good,golowich2025asymptotically} means that there is also a CCZ-state distillation protocol that uses only CNOTs, measurements, and adaptive Pauli and CZ gates. Since Hadamard and CCZ form a universal gate set, we can implement arbitrary quantum circuits in this way.

The components of the section therefore fit together as follows. The qLTC verification provides zero-error VQSS (\Cref{algo:vqss}, \Cref{prop:vqss}), and the robust decoder provides repeated error suppression (\Cref{algo:error-correction}, \Cref{prop:error-suppression}). Transversal Pauli and CNOT gates, together with their adaptive variants, supply the basic operations used by the distillation and teleportation circuits (\Cref{algo:adaptive-gates}, \Cref{prop:adaptive,remark:transversal}). Stabilizer-state distillation (\Cref{lem:stabilizer-distill}) is incorporated into the distributed state-preparation protocol (\Cref{algo:distributed-state-prep}, \Cref{prop:distributed-state-prep}) to produce encoded $\ket{0}$, $\ket{+}$, and $\ket{\mathrm{CZ}}=\mathrm{CZ}\ket{++}$ states. These states allow us to initialize ancillas and teleport Hadamard and CZ gates (\Cref{fact:gate-teleportation}, \Cref{algo:distributed-Clifford}, \Cref{prop:distributed-Clifford}). Once CZ is available, CCZ-state distillation (\Cref{lem:magic-distill}, \Cref{algo:distributed-magic-prep}) produces $\ket{\mathrm{CCZ}}=\mathrm{CCZ}\ket{+++}$, which is consumed to teleport CCZ (\Cref{fact:CCZ-teleportation}, \Cref{algo:teleported-CCZ}, \Cref{prop:teleported-CCZ}). The final construction first verifies each encoded logical input wire using VQSS with state replacement (\Cref{algo:vqss-state-replacement}), then evaluates the desired circuit gate by gate and applies error suppression between successive gates. Because every component restores the same error bound required by the next, sequential composability (\Cref{def:seq-composability}) allows these components to be composed for circuits of arbitrary length while preserving an exactly correct logical output, as formalized in \Cref{thm:distributed-FTQC}.

\subsection{Noise Model and Code Requirements}

We will consider MPC with an edge fault model, which is a stronger model than the vertex-corruption model considered in standard MPC literature~\cite{crepeau2002secure}.

Next we state the code needed for our protocol.

\begin{definition}[Single-shot decoder]
\label{def:single-shot-decoder}
     Let $\mathcal{Q}$ be a quantum CSS code specified by the parity check matrices $H_X \in \mathbb{F}_2^{r_X \times n}$ and $H_Z \in \mathbb{F}_2^{r_Z \times n}$. Suppose the data state carries the Pauli errors $e_X, e_Z \in \mathbb{F}_2^{n}$ and the syndrome measurements incur the errors $s_X \in \mathbb{F}_2^{r_Z}$, $s_Z \in \mathbb{F}_2^{r_X}$, such that the associated noisy syndromes are $\tilde{\sigma}_X = H_Z e_X + s_X$ and $\tilde{\sigma}_Z = H_X e_Z + s_Z$. Let $\mathcal{D}$ be a decoding algorithm that receives $\tilde{\sigma}_X$, $\tilde{\sigma}_Z$ and proposes the Pauli corrections $f_X$, $f_Z$. The decoder is said to be $(\alpha(n), \delta(n), \gamma(n),\eta)$-single-shot, where $\alpha, \beta, \gamma: \mathbb{N}^+ \mapsto \mathbb{R}^{+}$ and $0\leq \eta <1$, if the following holds. For each $P \in \{X,Z\}$, provided that
\begin{align*}
    |e_P|_R + \alpha(n) |s_P| \leq \delta(n) n,
\end{align*}
then the proposed Pauli correction $f_P$ satisfies
\begin{align*}
    |e_P + f_P|_R \leq \eta |e_P|_R + \gamma(n) |s_P|.
\end{align*}
Here $|\cdot|_R$ denotes the stabilizer-reduced weight. When $\eta=0$ we simply say $(\alpha,\delta,\gamma)$-single-shot.
\end{definition}

\begin{assumption}[Nice qLTC]\label{assumption:nice-qLTC} Suppose there exists an explicit family $[n,1,d]$ quantum CSS LTC $\mathcal{Q}$ with a set of $m=\Theta(n)$ bounded-degree checks (there exists $w=O(1)$ such that each qubit is involved in $\leq w$ checks, and each check acts on $\leq w$ qubits), and a function $\delta(n)$ such that the following hold. (a) The checks have soundness $\geq \delta(n)$. (b) The code admits an efficient and robust decoding algorithm with decoding radius $t \geq n \delta$, more precisely:
    \begin{enumerate}
    \item (\textbf{Efficient decoder}) A $\mathrm{poly}(n)$-time decoder that is $(\alpha, \delta(n), \alpha,0)$-single-shot (see \Cref{def:single-shot-decoder}), where $\alpha$ is a universal constant.
    \item (\textbf{Error suppressor})  A $O(1)$-depth circuit (including computing and applying a Pauli correction on the code qubits) that is $(\alpha, \delta(n), \alpha, \frac{1}{100})$-single-shot (see \Cref{def:single-shot-decoder}), where  $\alpha$ is a universal constant.
\end{enumerate} 
\end{assumption}

An example is the Dinur-Lin-Vidick almost-good qLTC with $\delta(n) = 1/\mathrm{polylog}(n)$~\cite{dinur2024expansion}, whose efficient and robust (i.e. single-shot) decoder is given in~\cite{nguyen2025quantum}. The code there encodes $\Theta(n)$ logical qubits, for our purposes we can just use one logical qubit and ignore the rest.

\begin{lemma}[Almost-good qLTC~\cite{dinur2024expansion,nguyen2025quantum}]\label{fact:almost-good-qLTC}
    There exists a family of quantum codes $\mathcal{Q}$ satisfying the above assumption with $\delta= 1/\mathrm{polylog}(n)$.
\end{lemma}

\subsection{Verification of Distributed Codewords}
Our starting point is a protocol to verify that the players are sharing a state close to a codeword of $\mathcal{Q}$. The task is similar to that which is achieved by the VQSS protocol in~\cite{crepeau2002secure} for any CSS code. Our protocol only works for qLTC codes of the form given in~\Cref{assumption:nice-qLTC}; however, it is simpler and most importantly it is \emph{zero-error} (i.e. if the number of violations detected by the protocol is small, then the state is guaranteed to be close to the code space).

\begin{mdframed}
\begin{protocol}[Zero-error VQSS with voting threshold $c_1$ (without state replacement)]\mbox{}\label{algo:vqss}\\
\textbf{Input:} An $n$-qubit pure state $\ket{\Psi}$, purported to be a codeword of $\mathcal{Q}$, jointly held by the $n$ players. \\
\textbf{Output:} The original qLTC block and a length-$n$ repetition code block $\ket{\cdot}_\mathrm{rep}$ that is in logical state $\ket{0}$ (accept) or $\ket{1}$ (reject), jointly held by the players.
\begin{enumerate}
\setcounter{enumi}{-1}
\item Partition the $m=\Theta(n)$ checks of $\mathcal{Q}$ into $\chi=O(1)$ groups $G_1, \ldots, G_{\chi}$ such that each group contains checks that act on disjoint sets of qubits. Pre-assign each player $i$ to $O(1)$ number of checks involving player $i$ (arbitrarily).
\item \textbf{(Measuring qLTC checks)} \\
For $q=1,...,\chi$: \\
 \hphantom{a} For each check $P$ in $G_q$:
\begin{itemize}
    \item The players send the qubits in $P$ to the assigned player $j$ via the following one-to-all-to-one protocol\footnote{This is needed only because we are in the edge fault model. Otherwise, the players could just communicate directly with each other.}:
    \begin{enumerate}[label=(\alph*)]
        \item For each $i$ in $\operatorname{supp}(P)$, player $i$ encodes its qubit into $\mathcal{Q}_\mathrm{inner}$ (inner encoding)\footnote{Any linear-distance quantum code with efficient decoder works (e.g.~\cite{aharonov1999fault, dinur2023good, nguyen2025quantum}).} and distributes to all players.
        \item The players send the shares of the above encoding to player $j$.
    \end{enumerate}
    \item Player $j$ error-corrects and unencodes the inner encodings of the qubits involved in $P$, (coherently) measures $P$ into a register $\ket{\cdot}_{P}$, then encodes this measurement qubit into the length-$n$ classical repetition code and distributes to all players. 
    Here, $P$ performs this distribution by---similarly to (a) and (b) in the previous step---firstly encoding every bit $i$ of the classical repetition code in another length-$n$ error-correcting code, and then distributing each of the resulting $n$ codeblocks among the $n$ players with the $j$th player out of $n$ receiving the $j$th bit of codeblock $i$ out of $n$, who are then expected to individually send the bits from codeblock $i$ to player $i$. Each player then decodes the codeblock they received from the other players.
     \item Player $j$ reencodes the inner encodings of the qubits involved in $P$ and sends the shares back to the respective players. The players send back to the original players, and the original players unencode the inner encodings.
\end{itemize}
\item \textbf{(Computing threshold)} Each player now holds a single bit for each check in the code $\mathcal{Q}$. Every player counts the number of 1s in the $m$ bits that they hold. If this number is greater than a threshold $t$ which we choose later, the player outputs 1 as their `opinion' bit; otherwise, the player outputs 0.
\item \textbf{(Byzantine agreement)} The players execute a Byzantine agreement protocol (with all communications done using an all-to-one-to-all gadget, as in previous steps) so that every player broadcasts their output bit to every other player. If there is a majority opinion that at most $\sqrt{|F|}/2$ players disagree with, all (honest) players replace their opinion bit with the majority opinion; otherwise, all (honest) players replace their opinion bit with 0.
\end{enumerate}
\end{protocol}
\end{mdframed}

This protocol will be used to verify the states provided by the QMA prover, and as a subroutine in the distributed state preparation procedures in the next subsection.

\begin{proposition}[Distributed code verification]\label{prop:vqss}
Assume \Cref{assumption:nice-qLTC}, and additionally assume that the number of checks satisfies \(m\ge n\).
Let \(w=O(1)\) be the common bound on check weight and qubit degree, and let \(\chi=O(1)\) be the number of check groups used in step 1 of~\Cref{algo:vqss}. Define
\(\kappa:=\frac{w+2w^{\chi+2}}{2\sqrt 2}\).
Fix any absolute constant \(c_0>0\).

Then there exist \(\eps^*_\mathrm{ver}=\poly(\delta)\) and a constant \(c'_\mathrm{ver}>0\) such that, for any \(\eps\le \eps^*_\mathrm{ver}\), the following holds. Consider any \(F\)-edge faulty execution of~\Cref{algo:vqss} with \(|F|\le \eps n^2/2\), and suppose that the threshold in step 2 is set to \((w c_0+2\kappa)\sqrt{\eps}\,m\). Then there exists a subset \(C'_\mathrm{ver}\subseteq [n]\) of size \(|C'_\mathrm{ver}| \le c'_\mathrm{ver}\sqrt{\eps}\,n\) depending only on $F$ such that the following guarantees hold. Suppose the input is a pure state of the form \(\ket{\Psi}=P\,\encq(\ket{\psi})\) for some Pauli operator \(P\), where \(P\) is the \emph{minimal-weight} Pauli operator which would put \(\ket{\Psi}\) in this form.

\begin{itemize}
\item \textbf{(Completeness)} If \(|\supp(P)| \le c_0\sqrt{\eps}\,n\), then the output of~\Cref{algo:vqss}, including all workspaces and ancilla qubits, is of the form
\[
E\!\left(\ket{\Psi}_\mathrm{in}\otimes\ket{0^n}_\mathrm{out}\otimes\ket{\mathrm{rest}}\right),
\]
where the support of \(E\) on the \(\ket{\cdot}_\mathrm{in},\ket{\cdot}_\mathrm{out}\) registers is restricted to players in \(C'_\mathrm{ver}\).

\item \textbf{(Soundness)} If \(|\supp(P)| \ge \frac{w c_0+4\kappa}{\delta}\sqrt{\eps}\,n\), then the output state is of the form
\[
E\!\left(\ket{\Psi}_\mathrm{in}\otimes\ket{1^n}_\mathrm{out}\otimes\ket{\mathrm{rest}}\right),
\]
where again the support of \(E\) on the \(\ket{\cdot}_\mathrm{in},\ket{\cdot}_\mathrm{out}\) registers is restricted to players in \(C'_\mathrm{ver}\).

\item If \(c_0\sqrt{\eps}\,n \le |\supp(P)| \le \frac{w c_0+4\kappa}{\delta}\sqrt{\eps}\,n\), then the output state is of one of the above two forms.
\end{itemize}

Furthermore, each player is involved in at most \(\poly(n)\) qubits and gates during the protocol.
\end{proposition}

 Notice that if \(\mathsf{Rep}\) denotes the classical repetition code of length \(n\), then for every fixed \(c_0>0\) and every \(\eps\le \eps^*_\mathrm{ver}\), the completeness guarantee implies that~\Cref{algo:vqss} is a \((\mathcal{Q},c_0\sqrt{\eps} , \mathcal{Q}\otimes \mathsf{Rep}, [(c_0+c'_\mathrm{ver})\sqrt{\eps}, c'_\mathrm{ver}\sqrt{\eps} ], K_{\eps})\)-fault-tolerant protocol for the quantum operation \(G(\ket{\psi}) = \ket{\psi} \ket{0}\) in the sense of~\Cref{def:fault-tolerant-protocol}.

\begin{proof}
Let \(\tau\) be the distance of the inner code used in the one-to-all-to-one transmission gadget in step 1 of~\Cref{algo:vqss}. Since that inner code has linear distance, there exists a constant \(\beta>0\) such that \(\tau\ge \beta n\). Choose \(\eps^*_\mathrm{ver}\) sufficiently small that \(\eps^*_\mathrm{ver}\le \beta^2/32\), and if needed shrink it further so that \(\eps^*_\mathrm{ver}=\poly(\delta)\).

Now fix any \(\eps\le \eps^*_\mathrm{ver}\). If \(|F|\le \eps n^2/2\), then
\[
8\sqrt{|F|}\le 8\sqrt{\eps/2}\,n \le 8\sqrt{\eps^*_\mathrm{ver}/2}\,n < \beta n \le \tau.
\]

\paragraph{Step 1: Reduction from edge faults to a small set of bad players.}
Define \(C:=\{\,i\in[n]: \deg_F(i)>4\sqrt{|F|}\,\}\), where \(\deg_F(i)\) is the number of faulty edges incident on player \(i\).

Consider one use of the one-to-all-to-one gadget to transmit a qubit from player \(i\) to player \(j\). If \(i,j\notin C\), then \(\deg_F(i)+\deg_F(j)\le 8\sqrt{|F|}<\tau\). Hence fewer than \(\tau\) positions of the inner-code word can be corrupted during that transmission, so the receiver decodes correctly. Therefore every transmission between players outside \(C\) is correct. It follows that the original \(F\)-edge faulty execution of \Cref{algo:vqss} produces no more faults than an auxiliary execution in which
\begin{itemize}
\item all one-to-all-to-one gadgets are replaced by ideal direct transmissions,
\item the players in \(C\) are treated as arbitrarily corrupted,
\item all players outside \(C\) behave honestly.
\end{itemize}
It therefore suffices to analyze this auxiliary execution.

The set \(C\) is small. Since every player in \(C\) has more than \(4\sqrt{|F|}\) faulty incident edges, we have \(|C|\cdot 4\sqrt{|F|}<\sum_{i\in C}\deg_F(i)\le 2|F|\), and hence
\[
|C| \le \frac{\sqrt{|F|}}{2}\le \frac{\sqrt{\eps}}{2\sqrt2}\,n.
\]
In particular, the number of bad players is at most \(\sqrt{|F|}/2\), exactly the tolerance built into the Byzantine-agreement step of~\Cref{algo:vqss}.

\paragraph{Step 2: Bounding the lightcone of the bad players.}
For \(t=0,1,\dots,\chi\), let \(L_t\) denote the set of qLTC codeword qubits lying in the future lightcone of the bad players' actions after round \(t\) of step 1, where initially \(L_0\) is the set of qLTC qubits held by the players in \(C\).

In a single round, the bad players can directly touch at most \((w^2+1)|C|\) fresh qubits: each bad player is assigned at most \(w\) checks, each such check involves at most \(w\) qubits, and we also include the bad player's own qLTC qubit. Moreover, each qubit already in the lightcone can spread to at most \(w\) qubits in the next round, because checks within one round are disjoint and each check has weight at most \(w\). Therefore \(|L_t|\le (w^2+1)|C|+w|L_{t-1}|\). Solving this recurrence gives \(|L_\chi|\le 2|C|w^{\chi+1}\).

Since player \(i\) holds qLTC qubit \(i\), we identify subsets of qLTC qubits with the corresponding subsets of players and define \(C'_\mathrm{ver}:=L_\chi\). Thus
\[
|C'_\mathrm{ver}|=|L_\chi| \le 2|C|w^{\chi+1}\le \frac{w^{\chi+1}}{\sqrt2}\sqrt{\eps}\,n.
\]
This proves the claim for any constant \(c'_\mathrm{ver}\ge \frac{w^{\chi+1}}{\sqrt2}\).

\paragraph{Step 3: Bounding how many check outcomes can be corrupted.}
Fix an honest player in the auxiliary execution, and compare the list of check outcomes that they hold after step 1 with the list they would hold in a completely honest execution of the same protocol on the same input state \(\ket{\Psi}=P\,\encq(\ket{\psi})\). A given check outcome can differ from the completely honest execution only if either
\begin{enumerate}
\item the check was assigned to a bad player, or
\item the check touches a qubit in the lightcone \(L_\chi\).
\end{enumerate}
The number of checks of type (1) is at most \(w|C|\), because each bad player is assigned at most \(w\) checks. The number of checks of type (2) is at most \(w|L_\chi|\le 2|C|w^{\chi+2}\), because each qubit participates in at most \(w\) checks. Hence the total number of check outcomes that may differ from the completely honest execution is at most \(B:=w|C|+2|C|w^{\chi+2}\). Using the bound on \(|C|\), we obtain
\[
B\le \frac{w+2w^{\chi+2}}{2\sqrt2}\sqrt{\eps}\,n
=\kappa\sqrt{\eps}\,n
\le \kappa\sqrt{\eps}\,m,
\]
where the last inequality uses \(m\ge n\).

\paragraph{Step 4: Completeness.}
Assume \(|\supp(P)|\le c_0\sqrt{\eps}\,n\). In a completely honest execution, every qubit in \(\supp(P)\) can cause at most \(w\) violated checks, so the total number of violated checks is at most
\[
w|\supp(P)|\le w c_0\sqrt{\eps}\,n\le w c_0\sqrt{\eps}\,m.
\]
By Step 3, the actual number of violated checks seen by any honest player can exceed this by at most \(\kappa\sqrt{\eps}\,m\). Hence every honest player sees at most \((w c_0+\kappa)\sqrt{\eps}\,m\) violated checks. Since \(w c_0+\kappa < w c_0+2\kappa\), every honest player outputs \(0\) in step 2.

Since the number of bad players is at most \(\sqrt{|F|}/2\), the Byzantine-agreement step makes all honest players agree on the majority honest opinion, namely \(0\). Thus every player outside \(C\) holds the ideal output bit \(0\).

Finally, outside the lightcone \(L_\chi=C'_\mathrm{ver}\), step 1 is exactly the honest coherent measurement-and-uncomputation of stabilizer checks, so the qLTC input register is unchanged there. Since \(C\subseteq C'_\mathrm{ver}\), it follows that outside \(C'_\mathrm{ver}\) both the qLTC input register and the output repetition register are ideal. Therefore the final state has the form
\[
E\!\left(\ket{\Psi}_\mathrm{in}\otimes\ket{0^n}_\mathrm{out}\otimes\ket{\mathrm{rest}}\right),
\]
where the support of \(E\) on the input/output registers is restricted to players in \(C'_\mathrm{ver}\).

\paragraph{Step 5: Soundness.}
Assume \(|\supp(P)|\ge \frac{w c_0+4\kappa}{\delta}\sqrt{\eps}\,n\). By qLTC soundness, in a completely honest execution the fraction of violated checks is at least \(\delta\cdot \frac{|\supp(P)|}{n}\), so the total number of violated checks is at least
\[
\delta\cdot \frac{|\supp(P)|}{n}\cdot m
\ge (w c_0+4\kappa)\sqrt{\eps}\,m.
\]
By Step 3, the actual number of violated checks seen by any honest player can be smaller by at most \(\kappa\sqrt{\eps}\,m\). Hence every honest player sees at least \((w c_0+3\kappa)\sqrt{\eps}\,m\) violated checks. Since \(w c_0+3\kappa > w c_0+2\kappa\), every honest player outputs \(1\) in step 2.

After Byzantine agreement, all honest players agree on \(1\), so every player outside \(C\) holds the ideal output bit \(1\). As in the completeness case, the only possible deviation on the input/output registers is inside \(C'_\mathrm{ver}\), and therefore the final state is of the form
\[
E\!\left(\ket{\Psi}_\mathrm{in}\otimes\ket{1^n}_\mathrm{out}\otimes\ket{\mathrm{rest}}\right),
\]
with support on the input/output registers restricted to players in \(C'_\mathrm{ver}\).

\paragraph{Step 6: Intermediate regime.}
If \(c_0\sqrt{\eps}\,n \le |\supp(P)| \le \frac{w c_0+4\kappa}{\delta}\sqrt{\eps}\,n\), then the uncertainty caused by the faulty communications may place the honest players' violated-check counts on either side of the threshold. Nevertheless, Byzantine agreement guarantees that all honest players agree on a common output bit, so the final state is of one of the two forms claimed above.

\paragraph{Complexity.}
The protocol performs only \(O(1)\) rounds of constant-weight check measurements, together with polynomial-size encoding, decoding, and Byzantine-agreement subroutines. Hence each player handles at most \(\poly(n)\) qubits and participates in at most \(\poly(n)\) gates.
\end{proof}

\subsection{Error Suppression}

In typical MPC protocols, faults introduced during repeated uses of VSS and during computation are controlled by showing that the faults remain localised on the shares of the joint state held by the cheating players or some fixed superset thereof. (Note that this is more specific than bounding the \emph{number} of fault locations by the number of cheaters: this means in addition that faults always occur in the same places, and so avoids a potential factor-two blowup in the fault weight every time VSS is repeated.) However, we are not always able to achieve this localisation property for our protocols. In addition, while typical MPC protocols are executed on a complete graph with a fixed set of cheaters, we will later have to transfer secret sharings between cliques in our HDX. The receiving clique could have a different fault pattern from the sending clique, which could lead to faults occurring both on shares held by the sending bad players and shares held by the receiving bad players. For these reasons, we will require the following \emph{error suppression} protocol, which is built from the robust decoder/error suppressor in~\Cref{assumption:nice-qLTC}.

\begin{mdframed}
\begin{protocol}[Repeated error suppression]\mbox{}\label{algo:error-correction}\\
\textbf{Input:} An \(n\)-qubit state \(\Psi\) jointly held by the \(n\) players, purported to be an erroneous codeword of \(\mathcal{Q}\).\\
\textbf{Ideal output:} The same code block, with the number of residual errors reduced to \(O(\sqrt{|F|})\).

Let
\[
R:=\left\lceil \log_{100}(2n) \right\rceil .
\]
For \(r=1,\dots,R\), the players perform the following two-step routine:
\begin{enumerate}
\item \textbf{(Measuring qLTC checks)} Perform the same syndrome measurements as in \Cref{algo:vqss} to obtain the (noisy) syndromes \(\tilde{\sigma}_X\) and \(\tilde{\sigma}_Z\), whose bits are held by the pre-assigned players.

\item \textbf{(Error suppression)} Implement the constant-depth error suppressor circuit from \Cref{assumption:nice-qLTC}. Specifically, for each layer \(\ell=1,\dots,D=O(1)\), and for each gate \(g\) in layer \(\ell\):
\begin{itemize}
    \item Any player holding a qubit involved in \(g\) sends that qubit to a pre-assigned player \(i_g\) for gate \(g\). Here, a player \(j\) sends a qubit to player \(i_g\) by encoding the qubit into a linear-distance quantum code of length \(n\) and distributing the resulting code block among the \(n\) players, who then individually send their qubits to player \(i_g\), who decodes the block. This is the same one-to-all-to-one gadget as in \Cref{algo:vqss}.
    \item The player \(i_g\) applies the gate \(g\), and then sends the qubits back to their original owners using the same one-to-all-to-one gadget.
\end{itemize}
\end{enumerate}
\end{protocol}
\end{mdframed}

\begin{proposition}[Repeated error suppression]\label{prop:error-suppression}
There exist \(\eps^*_{\mathrm{es}}=\poly(\delta)\), \(c_{\mathrm{es}}(\delta)=\poly(\delta)\), and a universal constant \(T_{\mathrm{es}}>0\) such that, for any \(c \leq c_{\mathrm{es}}(\delta)\), any \(\eps \leq \eps^*_{\mathrm{es}}\), and any set \(F\) of faulty edges with \(|F| \leq \eps n^2/2\), \Cref{algo:error-correction} is a \((\mathcal{Q},c,\mathcal{Q},T_{\mathrm{es}}\sqrt{\eps},F)\)-fault-tolerant gadget for the identity operation.
Furthermore, each player is involved in at most \(\poly(n)\) qubits and gates during the protocol.
\end{proposition}

\begin{proof}
Let the input be \(\ket{\Psi}=P\,\encq(\ket{\psi})\), where \(P\) is a Pauli operator on the code block with \(\operatorname{wt}(P)\le cn\) and \(c\le c_{\mathrm{es}}(\delta)\).

Let \(s:=\sqrt{2|F|}\). We first analyze one repetition of the two-step routine in \Cref{algo:error-correction}.

\medskip
\noindent\textbf{One round of suppression.}
Exactly as in the proof of \Cref{prop:vqss}, after the syndrome-measurement step, the state is of the form
\begin{equation}
\sum_{E_1} E_1 \Bigl(
P\,\encq(\ket{\psi})_{\mathrm{in}}
\otimes
\ket{\mathrm{syndrome}(P)}
\otimes
\ket{\mathrm{rest}'_{E_1}}
\Bigr),
\label{eq:es-step1}
\end{equation}
where the operators \(E_1\) are independent of \(\ket{\psi}\), and each \(E_1\) acts nontrivially on the code/syndrome blocks only on qubits held by at most \(c_1 s\) players, for some universal constant \(c_1\).

Now consider Step 2. As in the proof of \Cref{prop:vqss}, faults in the one-to-all-to-one gadgets can only cause faulty gate implementations at players that are bad for the corresponding inner code. The number of such bad players is \(O(s)\), and since the error-suppressor circuit has constant depth, the total number of extra code-qubit errors incurred during Step 2 is at most \(c_2 s\), for some universal constant \(c_2\).

If Step 2 were fault-free, then by the guarantee in \Cref{assumption:nice-qLTC} that the error suppressor circuit is \((\alpha,\delta,\alpha,\frac{1}{100})\) single-shot, we have the following: provided that \(\operatorname{wt}(P)+\alpha c_1 s \le \delta n\), the residual error on the code block after the suppressor would have weight at most \(\frac{1}{100}\bigl(\operatorname{wt}(P)+c_1 s\bigr)+\alpha c_1 s\). Including the \(c_2 s\) additional errors caused by faulty gate implementations in Step 2, the total residual error is at most \(\frac{1}{100}\operatorname{wt}(P)+\left(\frac{c_1}{100}+\alpha c_1+c_2\right)s\). Define \(A:=\frac{c_1}{100}+\alpha c_1+c_2\). Thus one round of the protocol maps any Pauli \(t\)-deviation with \(t+\alpha c_1 s \le \delta n\) to a \(\bigl(t/100+As\bigr)\)-deviation of the same codeword.

\medskip
\noindent\textbf{Choice of parameters.}
Choose \(c_{\mathrm{es}}(\delta)\) and \(\eps^*_{\mathrm{es}}\) small enough that whenever \(|F|\le \eps^*_{\mathrm{es}}n^2/2\), equivalently \(s\le \sqrt{\eps^*_{\mathrm{es}}}\,n\), the following two conditions hold:
\begin{align}
c_{\mathrm{es}}(\delta)n+\alpha c_1 s &\le \delta n,
\label{eq:es-cond1}
\\
As &\le \frac{1}{2}c_{\mathrm{es}}(\delta)n.
\label{eq:es-cond2}
\end{align}
This is possible with \(c_{\mathrm{es}}(\delta)=\Theta(\delta)\) and \(\eps^*_{\mathrm{es}}=\Theta(\delta^2)\), since \(A\), \(\alpha\), and \(c_1\) are universal constants.

\medskip
\noindent\textbf{Iterating the recurrence.}
Let \(t_r\) denote the code-block error weight after the \(r\)-th repetition of the two-step routine, with \(t_0:=\operatorname{wt}(P)\le cn\le c_{\mathrm{es}}(\delta)n\). Using \Cref{remark:pauli}, we can wlog assume that the code-block error is a Pauli deviation.

We claim that for all \(r\ge 0\), \(t_r\le c_{\mathrm{es}}(\delta)n\). This is true for \(r=0\). If it holds for some \(r\), then by \eqref{eq:es-cond1} the one-round analysis applies, so
\[
t_{r+1}\le \frac{t_r}{100}+As
\le \frac{c_{\mathrm{es}}(\delta)n}{100}+\frac{1}{2}c_{\mathrm{es}}(\delta)n
< c_{\mathrm{es}}(\delta)n,
\]
using \eqref{eq:es-cond2}. This proves the claim.

Hence the one-round recurrence applies at every repetition: \(t_{r+1}\le \frac{t_r}{100}+As\). Unrolling it for \(R=\left\lceil \log_{100}(2n)\right\rceil\) rounds gives
\[
t_R
\le
\frac{t_0}{100^R}
+
A s \sum_{j=0}^{R-1} 100^{-j}
\le
\frac{cn}{100^R}
+
\frac{100}{99}As.
\]
Since \(c\le 1\) and \(100^R\ge 2n\), we have \(\frac{cn}{100^R}\le \frac{1}{2}\). Therefore \(t_R\le \frac{1}{2}+\frac{100}{99}As\).

If \(F=\varnothing\), then \(s=0\), so \(t_R<1\). Since \(t_R\) is an integer, it follows that \(t_R=0\).

If \(F\neq \varnothing\), then \(|F|\ge 1\), so \(s=\sqrt{2|F|}\ge \sqrt{2}> \frac{1}{2}\). Hence \(t_R\le \left(1+\frac{100}{99}A\right)s\).

Setting \(T_{\mathrm{es}}:=1+\frac{100}{99}A\), we conclude that in all cases the final code block is a Pauli \(T_{\mathrm{es}}\sqrt{2|F|}\)-deviation of \(\encq(\ket{\psi})\).

This proves that \Cref{algo:error-correction} is \((\mathcal{Q},cn,\mathcal{Q},T_{\mathrm{es}}\sqrt{2|F|},F)\)-fault-tolerant for every \(c\le c_{\mathrm{es}}(\delta)\). Since \(\sqrt{2|F|}\le \sqrt{\eps}\,n\) by hypothesis, the protocol is \((\mathcal{Q},c,\mathcal{Q},T_{\mathrm{es}}\sqrt{\eps},F)\)-fault-tolerant.

Finally, each repetition uses at most \(\poly(n)\) qubits and gates per player, and the number of repetitions is \(R=\left\lceil \log_{100}(2n)\right\rceil = O(\log n)\). Hence the total per-player resource usage remains \(\poly(n)\).
\end{proof}

\subsection{Transversal Gates and Classical Adaptivity}
\paragraph{Transversal gates.}
The underlying qLTC is CSS so we automatically have transversal logical Pauli and CNOT gates. Of course we also have transversal logical SWAP gate.

The classically controlled version of these gates can be implemented transversally as well, provided that each player holds a copy of the classical control bit. Here, `classically controlled' means the control bit(s) is discarded right after the gate. This observation allows us to coherently implement adaptive gates that appear in elementary circuits such as gate teleportation or Steane error corretion.

\begin{mdframed}
\begin{protocol}[Adaptive transversal gates]\mbox{}\label{algo:adaptive-gates}\\
\textbf{Ideal logical operation:} Measure the first $\ell \geq 0$ qubits of a state $\ket{\psi}$, where each qubit can be measured in either the computational or Hadamard basis as specified by a vector $\mathcal{B} \in \{0,1\}^{\ell}$, obtaining a classical string $x \in \mathbb{F}_2^\ell$. Then, compute a function $f(x)$ and apply operation $G_{f(x)}$ on the remaining $m \geq 1$ qubit, promised that, for every $x$, $G_{f(x)}$ admits transversal implementation on the $m$ codeblocks $\mathcal{Q}$ (as tensor product of $m$-qubit gates). \\
\textbf{Input:} A pure state $\ket{\Psi}_\mathrm{AB}$, where registers AB are purportedly an encoding of some $(\ell+m)$-qubit logical state $\ket{\psi}_\mathrm{AB}$ into (multiple blocks of) the code $\mathcal{Q}$. Here, each $\mathcal{Q}$-codeblock is jointly held by the $n$ players. \\
\textbf{Output:} A state that is supposedly the result of applying the above ideal logical operation on $\ket{\psi}_\mathrm{AB}$, encoded in the remaining $m$ codeblocks.
\begin{enumerate}
\item For each codeblock that is supposed to be measured in the computational basis, each player CNOTs its share into an $n$-qubit string and distributes qubit $i$ of the string to player $i$. For Hadamard basis measurement, the players first apply Hadamard gates before performing the above step.
\item Using the received qubits from step 1 (and its own qubits), each player decodes the $\ell$ logical values using the efficient decoder from~\Cref{assumption:nice-qLTC}. Each player then computes $f$ locally from the decoded logical values. Let $s_i \in \Im(f)$ be the result of $f$ computed by player $i$.
\item On the remaining $m$ codeblocks, each  player $i$ applies the local gates that it is responsible for in the transversal implementation of $G_{s_i}$.
\item On each of the remaining \(m\) codeblocks, the players execute the repeated error-suppression protocol of~\Cref{algo:error-correction}.
\end{enumerate}
\end{protocol}
\end{mdframed}

\begin{proposition}[Adaptive transversal gates]\label{prop:adaptive}
    There exist $\eps^*_\mathrm{tr}= \poly(\delta)$, $c_\mathrm{tr}(\delta)=\poly(\delta)$, and a constant $c'_\mathrm{tr}>0$ such that for any faulty execution of~\Cref{algo:adaptive-gates} with faulty edge set
$ |F|\le \eps n^2/2 \le \eps^*_\mathrm{tr} n^2/2$,
    the following promise holds on the output. Upon input an $(\ell+m)$-qubit state $\ket{\Psi}_{AB}$ that is Pauli-deviated from a pure codeword $\Enc_{\mathcal Q}(\ket{\psi}_{AB})$ on a subset of players $C_\mathrm{tr}\subseteq [n]$ with $|C_\mathrm{tr}|\le c_\mathrm{tr}(\delta)\,n$,
    where
    \[
    \ket{\psi}=\sum_{x\in\{0,1\}^\ell}\alpha_x H^{\mathcal B}\ket{x}\ket{\psi_x},
    \]
    with each $\ket{\psi_x}$ normalized, the protocol outputs a state of the form
    \[
    \sum_{x,i} P_{x,i}\Bigl(\encq(G_{f(x)}\ket{\psi_x})_B \otimes \ket{\mathrm{rest}_{x,i}}\Bigr),
    \]
    where $\ket{\mathrm{rest}_{x,i}}$ are states on all other qubits (including workspace), and each $P_{x,i}$ is a Pauli operator on the output register $B$ whose restriction to each of the \(m\) output codeblocks has support size at most $c'_\mathrm{tr}\sqrt{\eps}\,n$.
    
    Equivalently, on the \(m\) output codeblocks the final state is a \([c'_\mathrm{tr}\sqrt{\eps}]^m\)-deviation of the ideal branchwise output. Furthermore, each player is involved in at most $\poly(n)$ qubits and gates during the protocol.
\end{proposition}

\begin{remark}\label{remark:transversal} 
    We will particularly apply this proposition to two cases:
    \begin{itemize} 
    \item \textbf{(Transversal gates)} When $\ell=0$, for any $\eps \leq \eps^*_\mathrm{tr}$, ~\Cref{algo:adaptive-gates} is a
    \[
    (\mathcal{Q}^{\otimes m}, c_\mathrm{tr}(\delta)/m,\mathcal{Q}^{\otimes m},c'_\mathrm{tr}\sqrt{\eps}, K_{\eps})
    \]
    -fault-tolerant protocol for an $m$-qubit gate $G$.

    \item \textbf{(Gate teleportation)} When $\ket{\psi}$ satisfies
    \[
    G_{f(x)} \ket{\psi_x} \propto G_{f(x')} \ket{\psi_{x'}}
    \qquad\text{for all }x,x',
    \]
    as in gate teleportation, for all $\eps \leq \eps^*_\mathrm{tr}$, ~\Cref{algo:adaptive-gates} is a
    \[
    (\mathcal{Q}^{\otimes (\ell+m)},[c_\mathrm{tr}(\delta)/m]^m,\mathcal{Q}^{\otimes m},c'_\mathrm{tr}\sqrt{\eps}, K_{\eps})
    \]
    -fault-tolerant protocol for the operation of measuring the first $\ell$ qubits and applying $G_{f(x)}$ on the remaining $m$ qubits adaptively based on the measurement outcomes.
    \end{itemize}
\end{remark}

\begin{proof}[Proof of \Cref{prop:adaptive}]
    For simplicity, we consider the case where one codeblock is measured ($\ell=1$) in the computational basis ($\mathcal{B}=\{0\}$), the feedback gate is on one code block ($m=1$), the function $f$ is identity $f(x)=x$, and $G_0$ is identity while $G_1$ is a single-qubit logical gate $G$ whose transversal physical implementation reads $\overline{G} = g_1 \otimes ... \otimes g_n$. The general case follows similarly (including the trivial $\ell=0$ case.) More concretely, let $\ket{\Psi} = P \encq^{\otimes 2}(\ket{\psi})$ denotes the input where $\ket{\psi}$ is a two-qubit state and $P$ is a Pauli operator with support on a subset $|C| \leq cn$ of vertices. 
    Expanding $\ket{\psi} = \sum_{x \in \{0,1\}} \alpha_x \ket{x} \ket{\psi_x}$, where $\ket{\psi_x}$ is normalized, then $\ket{\psi_0} \propto G \ket{\psi_1}$ as per the promise. 
    
    \paragraph{Faults:} When an $\eps n^2/2$ edges are faulty ($\eps \leq \eps^*_\mathrm{tr}$), there are at most $\sqrt{\eps} n$ corrupted vertices, and at most $(2 \eps/\delta) n$ (possibly non-corrupted) vertices that are incident to more than $ \delta n/2$ faulty edges. Let $C'$ be the set of such `bad' vertices (as opposed to `good' ones), we have
    \begin{align}
        |C'| \leq (\sqrt{\eps}  + 2 \eps/\delta) n.
    \end{align}
    On the bad vertices $C'$ we will be conservative and assume the relevant qubits and operations could be arbitrarily corrupted. In contrast, the perfomance of each good vertex is unaffected by the $\delta n/2$ faulty edges incident to it thanks to the decoding radius of the code $\mathcal{Q}$.

    We track the state after each step of the protocol, and, when helpful, compare the states between a faulty and faultless execution of the step. The initial state is
     \begin{align}
       \ket{\Psi}= \sum_{x \in \{0,1\}} \alpha_x 
        \begin{tikzpicture}[baseline=(current bounding box.center), scale=0.5]
    \draw[xstep=1.5cm, ystep=1cm, black] (0,0) grid (1.5,5);
    \fill[black] (0,1) rectangle (1.5,2);
    \node[fill=white, rotate = -90,
          minimum width=2cm, 
          minimum height=0.5cm, 
          inner sep=0pt] 
          at (0.75, 2.5) 
          { $\encq(|x\rangle)$};
\end{tikzpicture}  
 \otimes 
\begin{tikzpicture}[baseline=(current bounding box.center), scale=0.5]
        \draw[xstep=1.6cm, ystep=1cm, black] (0,0) grid (1.6,5);
        \fill[black] (0,1) rectangle (1.6,2);
        \node[fill=white, rotate=-90] at (0.8, 2.5) { $\encq(\ket{\psi_x})$ };
    \end{tikzpicture}.
    \end{align}
    Here, the blacked-out cells represent the input noise $P$ on top of the encoding on the corrupted set $C$. (In this illusration the input error $P$ is suppored on $C=\{4\}$ under the convention the top row being row 1.)
    
    \paragraph{Step 1:} After the first step, we imagine the qubits are arranged in a $n \times n$ grid where row $i$ corresponds to the qubits held by player $i$ and column $j$ corresponds to the qubits that were sent to player $j$.  In a faultless execution, the state is of the form
    \begin{align}
        \ket{\Psi_1^\mathrm{ideal}} =  \left(
    \begin{tikzpicture}[baseline=(current bounding box.center), scale=0.5]
    \draw[step=1cm, gray, thick] (0,0) grid (5,5);
    \fill[black] (3,0) rectangle (4,5);
    \end{tikzpicture} 
    \otimes 
    \begin{tikzpicture}[baseline=(current bounding box.center), scale=0.5]
    \draw[xstep=1cm, ystep=1cm, black] (0,0) grid (1,5);
    \fill[black] (0,1) rectangle (1,2);
    \end{tikzpicture} \right)
    \cdot 
    \left( \sum_{x \in \{0,1\}} \alpha_x  
    \begin{tikzpicture}[baseline=(current bounding box.center), scale=0.5]
    \draw[step=1cm, black] (0,0) grid (5,5);
    \node[fill=white, align=center,
          minimum width=2.3cm, 
          minimum height=0.5cm, 
          inner sep=0pt] 
          at (2.5, 2.5) 
          { $\encc^{\otimes n} $  \\ $ \circ \encq(|x\rangle)$ };
\end{tikzpicture}  \ \otimes \ \begin{tikzpicture}[baseline=(current bounding box.center), scale=0.5]
        \draw[xstep=2cm, ystep=1cm, black] (0,0) grid (2,5);
        \node[fill=white, rotate=-90] at (1, 2.5) { $\encq(\ket{\psi_x})$ };
    \end{tikzpicture} \right),
    \end{align}
    where the grids with blacked-out regions represent the Pauli errors propagated from $P$. Whereas, the actual state in a faulty execution is
\begin{align}
        \ket{\Psi_1}= 
  E_1 S_1
    \cdot 
    \left( \sum_{x \in \{0,1\}} \alpha_x  
    \begin{tikzpicture}[baseline=(current bounding box.center), scale=0.5]
    \draw[step=1cm, black] (0,0) grid (5,5);
    \node[fill=white, align=center,
          minimum width=2.3cm, 
          minimum height=0.5cm, 
          inner sep=0pt] 
          at (2.5, 2.5) 
          { $\encc^{\otimes n} $  \\ $ \circ \encq(|x\rangle)$ };
\end{tikzpicture}  \ \otimes \
  \begin{tikzpicture}[baseline=(current bounding box.center), scale=0.5]
        \draw[xstep=2cm, ystep=1cm, black] (0,0) grid (2,5);
        \node[fill=white, rotate=-90] at (1, 2.5) { $\encq(\ket{\psi_x})$ };
    \end{tikzpicture} \right),
    \end{align}
    where $E_1$ is a Pauli operator\footnote{Here
    we use the fact that step 1 is Clifford and thus initial Pauli errors and Pauli corruptions during step 1 can be combined, although this fact is not important.} error combining the propagation of the initial error $P$ and the Pauli errors introduced by the `bad' players in $C'$ during step 1, and $S_1$ denotes the Pauli errors introduced by the `good' players (those in $[n] \backslash C'$). Continuing our grid visualization with $n=5$, and suppose that $C'=\{2\}$ and $C=\{4\}$, then the support of $E_1$ can be visualized as
    \begin{align}
    \operatorname{supp}(E_1) \subseteq
    \begin{tikzpicture}[baseline=(current bounding box.center), scale=0.5]
    \draw[step=1cm, gray, thick] (0,0) grid (5,5);
    \fill[black] (3,0) rectangle (4,5);
    \fill[black] (1,0) rectangle (2,5);
    \fill[black] (0,3) rectangle (5,4);
    \end{tikzpicture} 
    \otimes 
    \begin{tikzpicture}[baseline=(current bounding box.center), scale=0.5]
    \draw[xstep=1cm, ystep=1cm, black] (0,0) grid (1,5);
    \fill[black] (0,1) rectangle (1,2);
    \fill[black] (0,3) rectangle (1,4);
    \end{tikzpicture}
    \end{align}
   Whereas, $S_1$ is a Pauli operator on the $n \times n$ grid that is $\delta n/2$-sparse on each good player ($S_1$ does not act on the second register) representing sparse edge transmission faults of the good players. I.e., in each row in $[n] \backslash C'$, there is at most $\delta n/2$ blacked-out cells (corresponding to corrupted edges).

    \paragraph{Step 2:} On each good row (vertex) of $\ket{\Psi_1}$, i.e., those in $[n] \backslash C'$, the number of Pauli errors on the codeword $\encc^{\otimes n} \circ \encq (\ket{x})$ is at most $cn + \delta n/2 + \sqrt{\eps } n < \delta n$, the decoding radius of $\mathcal{Q}$ (\Cref{assumption:nice-qLTC}). Thus, these rows are decoded correctly to $f(x)$. On the other hand, the bad rows $C'$ could be arbitrarily corrupted and and since the gates in decoding and computing $f$ may be non-Clifford, we could have a non-Pauli residual error on them.
    \begin{align}
        \ket{\Psi_2}= 
        \underbrace{\left(
             \begin{tikzpicture}[baseline=(current bounding box.center), scale=0.5]
    \draw[step=1cm, gray, thick] (0,0) grid (5,5);
    \fill[black] (0,3) rectangle (5,4);
    \end{tikzpicture}
    \ 
    \begin{tikzpicture}[baseline=(current bounding box.center), scale=0.5]
    \draw[xstep=1cm, ystep=1cm, black] (0,0) grid (1,5);
    \end{tikzpicture} 
    \ 
    \begin{tikzpicture}[baseline=(current bounding box.center), scale=0.5]
    \draw[xstep=1cm, ystep=1cm, black] (0,0) grid (1,5);
    \fill[black] (0,3) rectangle (1,4);
    \end{tikzpicture} 
    \
    \begin{tikzpicture}[baseline=(current bounding box.center), scale=0.5]
        \draw[xstep=2.4cm, ystep=5cm, black] (0,0) grid (2.4,5);
        \fill[black] (0,0) rectangle (2.4,5);
    \end{tikzpicture}
    \right)}_{\text{Non-Pauli errors $E_2$ from step 2}}
    \cdot
    E_1 S_1 \cdot 
    \left( \sum_{x \in \{0,1\}} \alpha_x 
    \begin{tikzpicture}[baseline=(current bounding box.center), scale=0.5]
    \draw[step=1cm, black] (0,0) grid (5,5);
    \node[fill=white, align=center,
          minimum width=2.3cm, 
          minimum height=0.5cm, 
          inner sep=0pt] 
          at (2.5, 2.5) 
          { $\encc^{\otimes n} $  \\ $ \circ \encq(|x\rangle)$ };
\end{tikzpicture}  \
  \begin{tikzpicture}[baseline=(current bounding box.center), scale=0.5]
        \draw[xstep=2cm, ystep=1cm, black] (0,0) grid (2,5);
        \node[fill=white, rotate=-90] at (1, 2.5) { $\encq(\ket{\psi_x})$ };
    \end{tikzpicture} 
\begin{tikzpicture}[baseline=(current bounding box.center), scale=0.5]
        \draw[step=1cm, black] (0,0) grid (1,5);
        \node at (0.5, 0.5) {$|x\rangle$};
        \node at (0.5, 2.5) {$\vdots$};
        \node at (0.5, 4.5) {$|x\rangle$}; 
    \end{tikzpicture} 
    \begin{tikzpicture}[baseline=(current bounding box.center), scale=0.5]
        \draw[xstep=2.4cm, ystep=5cm, black] (0,0) grid (2.4,5);
        \node[fill=white, align=center,
          inner sep=0pt] 
          at (1.2, 2.5) 
          { $\ket{\mathrm{rest}_x}$ };
    \end{tikzpicture}  \right),
    \end{align}
    where the third register is supposed to hold the decoded logical value $x$, $\ket{\mathrm{rest}}_x$ denotes the workspace for computing $f$, and $E_2$ may act globally on the the workspace part and depend on $E_1,S_1$.

    \paragraph{Step 3:}    Finally, each player applies $g_i$ on its share of the remaining codeblock if it computed $x=1$. The state after step 3 is of similar form to $\ket{\Psi_2}$, with the difference being $\encq(\ket{\psi_x})$ replaced by $\encq(G^x \ket{\psi_x})$ in the second block.
    
        \begin{align}
        \ket{\Psi_3}= \sum_{x \in \{0,1\}} 
        \alpha_{x} \underbrace{\left(
    \begin{tikzpicture}[baseline=(current bounding box.center), scale=0.5]
    \draw[step=1cm, gray, thick] (0,0) grid (5,5);
    \fill[black] (3,0) rectangle (4,5);
    \fill[black] (1,0) rectangle (2,5);
    \fill[black] (0,3) rectangle (5,4);
    \end{tikzpicture} 
    \
    \begin{tikzpicture}[baseline=(current bounding box.center), scale=0.5]
    \draw[xstep=1cm, ystep=1cm, black] (0,0) grid (1,5);
    \fill[black] (0,1) rectangle (1,2);
    \fill[black] (0,3) rectangle (1,4);
    \end{tikzpicture} \
    \begin{tikzpicture}[baseline=(current bounding box.center), scale=0.5]
    \end{tikzpicture}\
    \begin{tikzpicture}[baseline=(current bounding box.center), scale=0.5]
    \draw[xstep=1cm, ystep=1cm, black] (0,0) grid (1,5);
    \fill[black] (0,3) rectangle (1,4);
    \end{tikzpicture}
    \
    \begin{tikzpicture}[baseline=(current bounding box.center), scale=0.5]
        \draw[xstep=2.4cm, ystep=5cm, black] (0,0) grid (2.4,5);
        \fill[black] (0,0) rectangle (2.4,5);
    \end{tikzpicture}
    \right)}_{\text{(Non-Pauli) Error } E_x}
    \cdot S_1 \cdot
    \left(
    \begin{tikzpicture}[baseline=(current bounding box.center), scale=0.5]
    \draw[step=1cm, black] (0,0) grid (5,5);
    \node[fill=white, align=center,
          minimum width=2.3cm, 
          minimum height=0.5cm, 
          inner sep=0pt] 
          at (2.5, 2.5) 
          { $\encc^{\otimes n} $  \\ $ \circ \encq(|x\rangle)$ };
\end{tikzpicture}  \
  \begin{tikzpicture}[baseline=(current bounding box.center), scale=0.5]
        \draw[xstep=2cm, ystep=1.1cm, black] (0,0) grid (2,5.5);
        \node[fill=white, rotate=-90] at (1, 2.75) { $\encq(G^x \ket{\psi_x})$ };
    \end{tikzpicture} \
\begin{tikzpicture}[baseline=(current bounding box.center), scale=0.5]
        \draw[step=1cm, black] (0,0) grid (1,5);
        \node at (0.5, 0.5) {$|x\rangle$};
        \node at (0.5, 2.5) {$\vdots$};
        \node at (0.5, 4.5) {$|x\rangle$}; 
    \end{tikzpicture}
    \begin{tikzpicture}[baseline=(current bounding box.center), scale=0.5]
        \draw[xstep=2.4cm, ystep=5cm, black] (0,0) grid (2.4,5);
        \node[fill=white, align=center,
          inner sep=0pt] 
          at (1.2, 2.5) 
          { $\ket{\mathrm{rest}_x}$ };
    \end{tikzpicture}
    \right),
    \end{align}
    where the errors $\{E_x\}$ only depend on $\ket{x}$ and the input error $P$ but not on the states $\ket{\psi_x}$ and recall that $S_1$ is a sparse error acting on the $n \times n$ grid only.
    
    Notice that $E_x$ only acts nontrivially on the rows $C \cup C'$ of the codeblock $\encq(G^x \ket{\psi_x})$. We need to choose $c_\mathrm{tr}(\delta) \ll \delta$ and $ \eps^*_\mathrm{tr} \ll \delta^2$ to ensure $cn + \delta n/2 + \sqrt{\eps} n < \delta n$ (logical decoding succeeds) and $|C'|\leq (\sqrt{\eps}  + 2 \eps/\delta)  < c'(\sqrt{\eps})$ (residual error is small).

    \paragraph{Step 4: Post-suppression.}
    Let \(B_1,\dots,B_m\) denote the \(m\) output codeblocks. By Step~3, for each branch \(x\), the operator \(E_x\) acts nontrivially on register \(B\) only on the rows indexed by \(C\cup C'\). Hence, on each output block \(B_j\), the induced error is supported on at most \(|C|+|C'|\le c_\mathrm{tr}(\delta)n + c_0\sqrt{\eps}\,n\) qubits, for some absolute constant \(c_0\), where we use the bound on \(|C'|\) obtained earlier.

    Expanding \(E_x\) in the Pauli basis on register \(B\), we may write \(E_x=\sum_\nu P^{(B)}_{x,\nu}\otimes M^{(\mathrm{rest})}_{x,\nu}\), where each \(P^{(B)}_{x,\nu}\) is a Pauli operator on \(B\) whose restriction to each block \(B_j\) is supported on at most \(c_\mathrm{tr}(\delta)n + c_0\sqrt{\eps}\,n\) qubits. Therefore the state after Step~3 can be written as \(\sum_{x,\nu} P^{(B)}_{x,\nu}\bigl(\encq(G_{f(x)}\ket{\psi_x})_B\bigr)\otimes \ket{\zeta_{x,\nu}}\) for suitable states \(\ket{\zeta_{x,\nu}}\) on all other registers. In other words, after Step~3 the output register \(B\) is a \([c_\mathrm{tr}(\delta)+c_0\sqrt{\eps}]^m\)-deviation of the ideal branchwise output.

    Now choose \(c_\mathrm{tr}(\delta)\) sufficiently small and then choose \(\eps^*_\mathrm{tr}\) sufficiently small so that \(c_\mathrm{tr}(\delta)+c_0\sqrt{\eps}\le c_{\mathrm{es}}(\delta)\) for every \(\eps\le \eps^*_\mathrm{tr}\), where \(c_{\mathrm{es}}(\delta)\) is the input radius from~\Cref{prop:error-suppression}. It follows that \Cref{prop:error-suppression} applies to each output block \(B_j\), with all other registers regarded as untouched reference systems.

    We now combine these \(m\) error-suppression gadgets. Here we use parallel composability, namely \Cref{lem:parallel-composability}, because the \(m\) output blocks are corrected independently and simultaneously, and we want a single fault-tolerant statement for the tensor product of those \(m\) suppression gadgets acting on the whole register \(B\). This yields a gadget that maps the \([c_\mathrm{tr}(\delta)+c_0\sqrt{\eps}]^m\)-deviated output of Step~3 to a \([T_{\mathrm{es}}\sqrt{\eps}]^m\)-deviation of the ideal branchwise output on \(B\).

    Finally, we compose this parallel post-suppression layer with the adaptive-gates procedure up through Step~3. For this we use sequential composability, namely \Cref{def:seq-composability}, since the output guarantee of Step~3 exactly matches the input promise required by the post-suppression layer. Therefore, after Step~4 the final output is of the form \(\sum_{x,i} P_{x,i}\bigl(\encq(G_{f(x)}\ket{\psi_x})_B \otimes \ket{\mathrm{rest}_{x,i}}\bigr)\), where each \(P_{x,i}\) is a Pauli operator on \(B\) whose restriction to each output codeblock has support size at most \(T_{\mathrm{es}}\sqrt{\eps}\,n\). Renaming constants completes the proof.
\end{proof}

\subsection{Clifford Gates}
In this section, we show how to implement logical Clifford gates on the distributed codeblocks. We will use gate teleportation~\cite{gottesman1999demonstrating} to implement the logical Hadamard and CZ gates, which requires preparing certain stabilizer states as resource states.

\paragraph{Stabilizer states.}
We first show how to prepare the stabilizer states $\ket{\overline{0}}_\mathcal{Q}, \ket{\overline{+}}_\mathcal{Q}, \ket{\overline{\mathrm{CZ}}}_\mathcal{Q}$. The $\ket{0}$ state is needed to start the computation in BQP or QMA protocols. The $\ket{\mathrm{CZ}}$ state is used to perform CZ gate, and the $\ket{+}$ state is used to perform Hadamard gate using gate teleportation.

\begin{mdframed}
\begin{protocol}[Distributed stabilizer state preparation]\mbox{}\label{algo:distributed-state-prep}\\
\textbf{Ideal operation:} Distributed preparation of a logical stabilizer state in $\{\ket{\overline{0}}_\mathcal{Q}, \ket{\overline{+}}_\mathcal{Q}, \ket{\overline{\mathrm{CZ}}}_\mathcal{Q}\}$, promised that for each $\psi \in \{0,+,\mathrm{CZ}\}$, there exists an instance of the $\ket{\psi}$-distillation algorithm~\Cref{lem:stabilizer-distill} with instance size $\Omega(n) \leq n_\psi \leq n/2$ \footnote{This assumption is made only for simplicity and suffices for our use in later sections. Otherwise, this can be relaxed and the only changes to the protocol is that each player needs to prepares $\lceil n_\psi/n \rceil$ copies in step 1.}.\\
\textbf{Input:} None. \\
\textbf{Output:} An $n$-qubit (or $2n$-qubit) state jointly held by the $n$ players.
\begin{enumerate}
\item \textbf{(Preparing input states)} Each player $i$ prepares a $\ket{\psi}_\mathcal{Q}$ and distributes the codeword.
\item \textbf{(Filtering inputs)} Players perform the verification~\Cref{algo:vqss} on the distributed code blocks.
\item \textbf{(Logical distillation)} Using a pre-decided order, the players pick the first $n_\psi$ VQSS-accepted code blocks to run the distillation algorithm in~\Cref{lem:stabilizer-distill} at the logical level of the code $\mathcal{Q}$ with these code blocks as input. This uses transversal gates~\Cref{algo:adaptive-gates} and error suppression gadget~\Cref{algo:error-correction}. (If there are not enough accepted code blocks, abort). 
\end{enumerate}
\end{protocol}
\end{mdframed}

The state preparation procedures rely on distillation protocols using linear-distance quantum codes with transversal Clifford gates. We will use the following stabilizer state distillation protocol from~\cite{nguyen2025quantum}.
\begin{lemma}[Stabilizer state distillation, Proposition 4.9 in~\cite{nguyen2025quantum}]\label{lem:stabilizer-distill} Let $\mathrm{\psi} \in \{0, +, \mathrm{CZ}\}$.  There exist a constant $\delta_\mathrm{\psi} >0$ and infinitely and geometrically dense\footnote{This means that the protocol exists for an infinite increasing sequence of values $\{n_\psi^i\}_{i \in \mathbb{N}}$, such that $n_\psi^{i+1} = \Theta(n_\psi^i)$.} many $n_\psi$ such that the following holds. For each such $n_\psi$, there exists distillation protocol that takes as input a register that is $\delta_\mathrm{\psi}n_\psi$-deviated (not necessarily Pauli deviated) from $\ket{\psi}^{\otimes n_\mathrm{\psi}}$ and outputs $\ket{\psi}$. Furthermore, the protocol only uses $\mathrm{poly}(n_\psi)$ many CNOT gates, destructive measurements in computational/Hadamard basis, and adaptive Pauli gates, with the adaptivity following a $\mathrm{poly}(n_\psi)$-sized auxiliary classical computation on the measurement outcomes.
\end{lemma}

\begin{proposition}[Logical stabilizer state preparation]\label{prop:distributed-state-prep} There exist threshold $\eps^*_{ss}(\delta)=\mathrm{poly}(\delta)$ and absolute constant $c'_{ss}$ such that, for any $\eps \leq \eps^*_{ss}$, ~\Cref{algo:distributed-state-prep} is a $(\mathcal{Q},c'_{ss} \sqrt{\eps}, K_{\eps})$-fault-tolerant protocol for the initialization operation of $\ket{\psi} \in \{\ket{0},\ket{+}, \ket{\mathrm{CZ}} \}$.
\end{proposition}

\begin{proof}
We prove the claim for \(\ket{\psi}=\ket{0}\). The cases \(\ket{+}\) and \(\ket{\mathrm{CZ}}\) are identical, using the corresponding distillation instance from \Cref{lem:stabilizer-distill}.

By the promise in \Cref{algo:distributed-state-prep}, the chosen distillation instance has size \(n_0=\Omega(n)\), so there is an absolute constant \(\gamma_0>0\) such that \(n_0\ge \gamma_0 n\). After Step~1, it is convenient to view the physical qubits as an \(n\times n\) array: row \(i\) consists of all qubits held by player \(i\), and column \(j\) consists of the \(j\)-th distributed code block, namely the block prepared by player \(j\) and then shared among all \(n\) players.

Fix any absolute constant \(c_0\ge 2\). Let \(\kappa\) be the constant from \Cref{prop:vqss}, and let \(\epsilon_{\mathrm{ver}}^*,c'_{\mathrm{ver}}\) be the corresponding parameters from \Cref{prop:vqss} for this choice of \(c_0\), so that the threshold used in VQSS is \((w c_0+2\kappa)\sqrt{\eps}\,m\). Let \(\epsilon_{\mathrm{tr}}^*(\delta),c_{\mathrm{tr}}(\delta),c'_{\mathrm{tr}}\) be the parameters from \Cref{prop:adaptive}, let \(\epsilon_{\mathrm{es}}^*(\delta),c_{\mathrm{es}}(\delta),T_{\mathrm{es}}\) be the parameters from \Cref{prop:error-suppression}, and let \(\delta_0>0\) be the distillation radius from \Cref{lem:stabilizer-distill} for \(\psi=0\). Write
\[
D_{\mathrm{ver}}(\delta):=\frac{w c_0+4\kappa}{\delta}+c'_{\mathrm{ver}}.
\]
Choose
\[
\epsilon_{ss}^*(\delta):=
\min\!\left\{
\epsilon_{\mathrm{ver}}^*,
\epsilon_{\mathrm{tr}}^*(\delta),
\epsilon_{\mathrm{es}}^*(\delta),
\frac14,
(\gamma_0\delta_0)^2,
\left(\frac{c_{\mathrm{es}}(\delta)}{D_{\mathrm{ver}}(\delta)}\right)^2,
\left(\frac{c_{\mathrm{tr}}(\delta)}{3T_{\mathrm{es}}}\right)^2,
\left(\frac{c_{\mathrm{es}}(\delta)}{3T_{\mathrm{es}}+c'_{\mathrm{tr}}}\right)^2
\right\}.
\]
Since \(\epsilon_{\mathrm{ver}}^*=\poly(\delta)\), \(\epsilon_{\mathrm{tr}}^*(\delta)=\poly(\delta)\), \(\epsilon_{\mathrm{es}}^*(\delta)=\poly(\delta)\), and \(c_{\mathrm{tr}}(\delta),c_{\mathrm{es}}(\delta)=\poly(\delta)\), we have \(\epsilon_{ss}^*(\delta)=\poly(\delta)\).

Fix \(\eps\le \epsilon_{ss}^*(\delta)\), and fix an \(F\)-edge faulty execution of \Cref{algo:distributed-state-prep} with \(|F|\le \eps n^2/2\).

\medskip
\noindent\textbf{Step 1: preparing the \(n\) input code blocks.}
For a player \(v\), let \(\deg_F(v)\) be the number of faulty edges incident to \(v\), and define
\[
C_{\mathrm{bad}}:=\{v\in[n]: \deg_F(v)>\sqrt{\eps}\,n\}.
\]
Since \(\sum_v \deg_F(v)=2|F|\le \eps n^2\), we have \(|C_{\mathrm{bad}}|<\sqrt{\eps}\,n\). Call players in \(C_{\mathrm{bad}}\) bad and the others good.

Consider a good column \(j\notin C_{\mathrm{bad}}\), i.e.\ the code block prepared by player \(j\). In the ideal execution this block is \(\encq(\ket{0})\). In the actual execution, a share of this block can be corrupted only if either it is held by a bad player, or it is transmitted across one of the at most \(\sqrt{\eps}\,n\) faulty edges incident to \(j\). Thus the \(j\)-th block differs from \(\encq(\ket{0})\) on at most \(2\sqrt{\eps}\,n\) players. Since \(c_0\ge 2\), every good column is \(c_0\sqrt{\eps}\)-deviated from \(\encq(\ket{0})\).

\medskip
\noindent\textbf{Step 2: VQSS filtering.}
In Step~2 the players run \(n\) copies of \Cref{algo:vqss}, one on each column, with threshold \((w c_0+2\kappa)\sqrt{\eps}\,m\). Fix this faulty execution. As in \Cref{prop:vqss}, there is a set \(C'_{\mathrm{ver}}\subseteq[n]\), depending only on the faulty execution, such that
\[
|C'_{\mathrm{ver}}|\le c'_{\mathrm{ver}}\sqrt{\eps}\,n,
\]
and outside \(C'_{\mathrm{ver}}\) every player sees the same accept/reject outcome for every column.

Indeed, for each column the output is, by \Cref{prop:vqss}, in one of the two canonical forms with repetition register \(\ket{0^n}\) or \(\ket{1^n}\) outside \(C'_{\mathrm{ver}}\). Hence all players outside \(C'_{\mathrm{ver}}\) agree on whether each given column is accepted or rejected. In particular, they all agree on which columns are the first \(n_0\) accepted columns in the pre-decided order.

For every good column \(j\notin C_{\mathrm{bad}}\), completeness of \Cref{prop:vqss} applies, so that column is accepted by every player outside \(C'_{\mathrm{ver}}\). Since \(|C_{\mathrm{bad}}|<\sqrt{\eps}\,n<n/2\) and \(n_0\le n/2\), there are at least \(n_0\) accepted columns, and the protocol therefore does not abort.

Let \(J\subseteq[n]\) be the set of the first \(n_0\) accepted columns in the pre-decided order. Let \(B\subseteq[n_0]\) be the subset of positions among these selected columns that come from bad preparers. Then
\[
|B|\le |C_{\mathrm{bad}}|<\sqrt{\eps}\,n\le \delta_0 n_0,
\]
where the last inequality follows from \(n_0\ge \gamma_0 n\) and \(\eps\le (\gamma_0\delta_0)^2\).

We now use the VQSS guarantees to describe the selected columns. For each selected column \(t\in[n_0]\), let \(P_t\) be a minimal-weight Pauli and \(\ket{\phi_t}\) a logical state such that the corresponding verified block equals \(P_t\encq(\ket{\phi_t})\) outside the players in \(C'_{\mathrm{ver}}\). Since this column is accepted, the soundness part of \Cref{prop:vqss} implies \(|\supp(P_t)|\le \frac{w c_0+4\kappa}{\delta}\sqrt{\eps}\,n\); hence, after the coherent reordering in Step~3a, the \(t\)-th selected code block is a \(D_{\mathrm{ver}}(\delta)\sqrt{\eps}\)-deviation from \(\encq(\ket{\phi_t})\), where \(D_{\mathrm{ver}}(\delta)=\frac{w c_0+4\kappa}{\delta}+c'_{\mathrm{ver}}\). Moreover, if \(t\notin B\), i.e.\ if that selected column came from a good preparer, then by the completeness part of \Cref{prop:vqss} we may take \(\ket{\phi_t}=\ket{0}\).

Let
\[
\ket{\varphi}:=\ket{\phi_1}\otimes\cdots\otimes\ket{\phi_{n_0}}.
\]
Then \(\ket{\varphi}\) is a \(B\)-deviation of \(\ket{0}^{\otimes n_0}\). Since \(|B|\le \delta_0 n_0\), \Cref{lem:stabilizer-distill} implies that the ideal distillation circuit \(\mathcal D_0\) maps \(\ket{\varphi}\) to \(\ket{0}\).

Thus, after the coherent reordering of Step~3a, the first \(n_0\) selected code blocks jointly form a \(D_{\mathrm{ver}}(\delta)\sqrt{\eps}\)-deviation from \(\encq^{\otimes n_0}(\ket{\varphi})\). Since \(D_{\mathrm{ver}}(\delta)\sqrt{\eps}\le c_{\mathrm{es}}(\delta)\) by the choice of \(\epsilon_{ss}^*(\delta)\), \Cref{prop:error-suppression} is applicable to each selected block.

\medskip
\noindent\textbf{Step 3: fault-tolerant simulation of the distillation circuit.}
We simulate \(\mathcal D_0\) at the logical level using \Cref{algo:adaptive-gates} and \Cref{algo:error-correction}. First apply \Cref{algo:error-correction} independently to each selected block. By \Cref{prop:error-suppression}, after this initial suppression round each selected block is a \(T_{\mathrm{es}}\sqrt{\eps}\)-deviation from the corresponding encoded block of \(\encq^{\otimes n_0}(\ket{\varphi})\).

We claim inductively that immediately before each subsequent logical step of the encoded simulation, every currently active code block is a \(T_{\mathrm{es}}\sqrt{\eps}\)-deviation from the corresponding block in the ideal execution of \(\mathcal D_0\) on input \(\ket{\varphi}\).

This holds after the initial suppression round. Assume it holds before some logical step. Since the distillation circuit from \Cref{lem:stabilizer-distill} uses only CNOT gates, destructive \(X/Z\)-basis measurements, and adaptive Pauli corrections, each logical step involves at most three active code blocks. Expanding the current state into Pauli-deviated terms, the joint input to that step is, term by term, a Pauli \(3T_{\mathrm{es}}\sqrt{\eps}\)-deviation from the corresponding ideal encoded input. By the choice of \(\epsilon_{ss}^*(\delta)\), \(3T_{\mathrm{es}}\sqrt{\eps}\le c_{\mathrm{tr}}(\delta)\), so \Cref{prop:adaptive} applies and shows that each surviving output block becomes a \((3T_{\mathrm{es}}+c'_{\mathrm{tr}})\sqrt{\eps}\)-deviation from its ideal encoded output block. Again by the choice of \(\epsilon_{ss}^*(\delta)\), this is at most \(c_{\mathrm{es}}(\delta)\), so \Cref{prop:error-suppression} applies to each surviving output block and resets it to a \(T_{\mathrm{es}}\sqrt{\eps}\)-deviation from its ideal encoded state. This proves the induction.

At the end of the simulation, the ideal logical output block is \(\ket{0}\), because \(\mathcal D_0(\ket{\varphi})=\ket{0}\). Therefore the actual physical output block is a \(T_{\mathrm{es}}\sqrt{\eps}\)-deviation from \(\encq(\ket{0})\). Thus \Cref{algo:distributed-state-prep} is \((\mathcal Q,T_{\mathrm{es}}\sqrt{\eps},K_\eps)\)-fault-tolerant for initialization of \(\ket{0}\).

The same proof applies verbatim to \(\ket{+}\) and \(\ket{\mathrm{CZ}}\), using the corresponding distillation instance from \Cref{lem:stabilizer-distill}. Therefore \Cref{algo:distributed-state-prep} is \((\mathcal Q,c'_{ss}\sqrt{\eps},K_\eps)\)-fault-tolerant for initialization of \(\ket{\psi}\in\{\ket{0},\ket{+},\ket{\mathrm{CZ}}\}\), where one may take \(c'_{ss}:=T_{\mathrm{es}}\).

Finally, Step~1 uses \(O(n)\) local preparations and one distribution per player, Step~2 consists of \(n\) invocations of \Cref{algo:vqss}, and Step~3 simulates a \(\poly(n_0)=\poly(n)\)-size circuit using \Cref{algo:adaptive-gates} and \Cref{algo:error-correction}, each of which has per-player complexity \(\poly(n)\). Hence every player is involved in at most \(\poly(n)\) qubits and gates throughout the protocol.
\end{proof}

We will implement CZ and Hadamard gates using the distributed state preparation procedure and gate teleportation.

\begin{fact}[Clifford gate teleportation]\label{fact:gate-teleportation}
    The CZ gate can be implemented using a resource state $\ket{\mathrm{CZ}}$,CNOT gates, destructive Z-basis measurements, and adaptively controlled Pauli gates.
    The Hadamard gate can be implemented using a resource state $\ket{+}$, a CZ gate, destructive X-basis measurements, and adaptive Pauli X gate. 
\end{fact}

\begin{center}
\begin{tikzpicture}
    \node at (0,0) {
        \begin{quantikz}
            \lstick{$\ket{+}$} & \ctrl{1} \qw & \gate{X^a} & \rstick{$H\ket{\psi}$}  \qw \\
            \lstick{$\ket{\psi}$} & \control{} \qw & \meter{\hspace{0.6cm}  \text{\scriptsize X}} \vcw{-1}
        \end{quantikz}
    };
\end{tikzpicture}
\qquad 
\begin{tikzpicture}\node[scale=0.8]{
\begin{quantikz}
    \lstick{$\ket{+}$} & \ctrl{1} & \ctrl{2} & \qw & \gate[2, nwires={1}]{\text{CZ}(X^a X^b) \text{CZ}} \qw & \rstick[2]{$CZ\ket{\psi}$} \qw \\
    \lstick{$\ket{+}$} & \control{} & \qw & \ctrl{2} & & \qw \\
    \lstick[2]{$\ket{\psi}$} & \qw & \targ{} \qw & \qw & \meter{\hspace{0.6cm}  \text{\scriptsize Z}} \vcw{-1} \vbox{\hrule\hrule} \\
    & \qw & & \targ{} & \meter{\hspace{0.6cm}  \text{\scriptsize Z}} \vcw{-1} 
\end{quantikz}
};
\end{tikzpicture}
\end{center}

\begin{mdframed}
\begin{protocol}[Teleported Clifford gate]\mbox{}\label{algo:distributed-Clifford}\\
\textbf{Input:} An $n$- or $2n$-qubit state $\ket{\Psi}$, a purported qLTC codeword $\encq(\ket{\psi})$, jointly held by the $n$ players.\\
\textbf{Output:} An $n$- or $2n$-qubit state $\Phi$, a purported qLTC codeword $\encq(U \ket{\psi})$, jointly held by the $n$ players, where $U$ is Hadamard or CZ.
\begin{enumerate}
\setcounter{enumi}{0}
\item \textbf{(Preparing resource states)} Run~\Cref{algo:distributed-state-prep} to prepare the required resource states for gate teleportation of $U$ according to~\Cref{fact:gate-teleportation}.
\item \textbf{(Logical teleportation)} Using~\Cref{algo:adaptive-gates} and~\Cref{algo:error-correction}, run the circuits in~\Cref{fact:gate-teleportation} at the qLTC logical level using the input state and the resource states from step 1. (Note: the Hadamard teleportation also requires a CZ gate, which can be done using this protocol recursively.)
\end{enumerate}
\end{protocol}
\end{mdframed}

\begin{prop}[Teleported Clifford gate] \label{prop:distributed-Clifford}
There exist $\eps^*_\mathrm{Cl}= \mathrm{poly}(\delta)$, $c_\mathrm{Cl}(\delta) = \mathrm{poly}(\delta)$ and constant $c'_\mathrm{Cl}$ such that, for any $\eps \leq \eps^*_\mathrm{Cl}$ and any $m \in \{1,2\}$ ($m$ is the size of the Clifford gate),~\Cref{algo:distributed-Clifford} is a $(\mathcal{Q}^{\otimes m}, c_\mathrm{Cl}(\delta)/m, \mathcal{Q}^{\otimes m}, c'_\mathrm{Cl}\sqrt{\eps}, K_{\eps})$-fault-tolerant protocol for the gate $U \in \{H, \mathrm{CZ}\}$ and $m=|\operatorname{supp}(U)|$.
\end{prop}

\begin{proof}
For \(U\in\{H,\mathrm{CZ}\}\), consider the fixed logical gate-teleportation circuit for \(U\) from Step~2 of \Cref{algo:distributed-Clifford}. Replace each logical stabilizer-state preparation box in that circuit by one invocation of \Cref{algo:distributed-state-prep}, each adaptive logical-measurement / transversal-gate box by one invocation of \Cref{algo:adaptive-gates}, and each logical cleanup box by one invocation of \Cref{algo:error-correction}. In the Hadamard circuit, leave the single logical \(\mathrm{CZ}\) box as one additional box. Let \(L_U\) be the maximum number of boxes in any chain of boxes whose outputs successively feed into one another and eventually affect a given logical output code block. Since the teleportation circuits in \Cref{fact:gate-teleportation} have constant size, \(L_H\) and \(L_{\mathrm{CZ}}\) are absolute constants. Define \(L:=\max\{L_H,L_{\mathrm{CZ}}\}\).

Let \(\eps^*_{\mathrm{ss}}(\delta), c'_{\mathrm{ss}}\) be the parameters from \Cref{prop:distributed-state-prep}, let \(\eps^*_{\mathrm{tr}}(\delta), c_{\mathrm{tr}}(\delta), c'_{\mathrm{tr}}\) be the parameters from \Cref{prop:adaptive}, and let \(\eps^*_{\mathrm{es}}(\delta), c_{\mathrm{es}}(\delta), T_{\mathrm{es}}\) be the parameters from \Cref{prop:error-suppression}. Here \(c'_{\mathrm{ss}}, c'_{\mathrm{tr}}\), and \(T_{\mathrm{es}}\) are absolute constants.

Set
\[
M:=\max\{2c'_{\mathrm{ss}},\,c'_{\mathrm{tr}},\,T_{\mathrm{es}}\},\qquad
c'_{\mathrm{Cl}}:=LM,
\]
so \(c'_{\mathrm{Cl}}\) is an absolute constant. Also set
\[
c_{\mathrm{Cl}}(\delta):=\frac14\min\{c_{\mathrm{tr}}(\delta),\,c_{\mathrm{es}}(\delta)\}.
\]
Finally, choose \(\eps^*_{\mathrm{Cl}}(\delta)>0\) such that
\begin{align}
\eps^*_{\mathrm{Cl}}(\delta)
&\le \min\{\eps^*_{\mathrm{ss}}(\delta),\,\eps^*_{\mathrm{tr}}(\delta),\,\eps^*_{\mathrm{es}}(\delta)\},
\label{eq:clifford-eps-choice-1}\\
c'_{\mathrm{Cl}}\sqrt{\eps^*_{\mathrm{Cl}}(\delta)}
&\le \frac12 c_{\mathrm{Cl}}(\delta).
\label{eq:clifford-eps-choice-2}
\end{align}
Such a choice exists because all quantities on the right-hand side are strictly positive.

We first treat the case \(U=\mathrm{CZ}\), for which \(m=|\supp(U)|=2\).

Fix \(\eps\le \eps^*_{\mathrm{Cl}}(\delta)\), let \(F\) be any faulty-edge set with \(|F|\le \eps n^2/2\), and let \(\ket{\Psi_{\mathrm{in}}}\) be an input state that is Pauli \([c_{\mathrm{Cl}}(\delta)/2]^2\)-deviated from \(\encq(\ket{\psi})\), where \(\ket{\psi}\) is the two-logical-qubit input state. By the definition of generalized Pauli deviation, there exists a Pauli operator \(P_{\mathrm{in}}\) supported on a set \(C_{\mathrm{in}}\subseteq [n]\) such that \(\ket{\Psi_{\mathrm{in}}}=P_{\mathrm{in}}\,\encq(\ket{\psi})\) and \(|C_{\mathrm{in}}|\le c_{\mathrm{Cl}}(\delta)n\).

\medskip
\noindent\textbf{Step 1: preparation of the teleportation resource.}
By \Cref{prop:distributed-state-prep}, the output of Step~1 is a \([c'_{\mathrm{ss}}\sqrt{\eps}]^2\)-deviation from \(\encq(\ket{\mathrm{CZ}})\), with respect to the natural partition into the two output code blocks. Hence there exists a decomposition
\begin{equation}
\ket{\Psi_{\mathrm{res}}}
=
\sum_{\alpha}
\Bigl(P_{\alpha}\,\encq(\ket{\mathrm{CZ}})\Bigr)\otimes \ket{\xi_{\alpha}},
\label{eq:clifford-cz-resource}
\end{equation}
where, for each \(\alpha\), the Pauli operator \(P_{\alpha}\) acts on the two resource-state code blocks only through some set \(C^{(\alpha)}_{\mathrm{ss}}\subseteq [n]\) satisfying \(|C^{(\alpha)}_{\mathrm{ss}}|\le 2c'_{\mathrm{ss}}\sqrt{\eps}\,n\le M\sqrt{\eps}\,n\).

Therefore the joint state entering the logical-teleportation part of Step~2 is
\begin{equation}
\ket{\Omega_0}
=
\sum_{\alpha}
\Bigl(
(P_{\alpha}\otimes P_{\mathrm{in}})
\bigl(\encq(\ket{\mathrm{CZ}})\otimes \encq(\ket{\psi})\bigr)
\Bigr)\otimes \ket{\xi_{\alpha}}.
\label{eq:clifford-cz-joint-input}
\end{equation}
For each \(\alpha\), the corresponding Pauli term is supported on \(C_{\mathrm{in}}\cup C^{(\alpha)}_{\mathrm{ss}}\), whose size is at most
\[
c_{\mathrm{Cl}}(\delta)n + M\sqrt{\eps}\,n
\le c_{\mathrm{Cl}}(\delta)n + c'_{\mathrm{Cl}}\sqrt{\eps}\,n
\le \frac32 c_{\mathrm{Cl}}(\delta)n
< c_{\mathrm{tr}}(\delta)n,
\]
where the last inequality uses the definition of \(c_{\mathrm{Cl}}(\delta)\). Hence \Cref{prop:adaptive} applies to every summand of \eqref{eq:clifford-cz-joint-input}.

\medskip
\noindent\textbf{Step 2a: adaptive logical teleportation.}
Applying \Cref{prop:adaptive} term-by-term to \eqref{eq:clifford-cz-joint-input}, for every \(\alpha\) there exists a set \(C^{(\alpha)}_{\mathrm{tr}}\subseteq [n]\) with \(|C^{(\alpha)}_{\mathrm{tr}}|\le c'_{\mathrm{tr}}\sqrt{\eps}\,n\le M\sqrt{\eps}\,n\) such that the output of the adaptive-teleportation subroutine has the form
\begin{align}
\sum_{a,b\in\{0,1\}}
E^{(\alpha)}_{ab}
\Bigl(
\encq\Bigl(
\bigl((\mathrm{CZ}(X^a\otimes X^b)\mathrm{CZ})\otimes \bra{ab}\bigr)
\mathrm{CNOT}^{\otimes 2}(\ket{\mathrm{CZ}}\ket{\psi})
\Bigr)
\otimes \ket{\eta^{(\alpha)}_{ab}}
\Bigr),
\label{eq:clifford-cz-after-adaptive}
\end{align}
where each \(E^{(\alpha)}_{ab}\) acts on the two output code blocks only through players in \(C_{\mathrm{in}}\cup C^{(\alpha)}_{\mathrm{ss}}\cup C^{(\alpha)}_{\mathrm{tr}}\). By the gate-teleportation identity from \Cref{fact:gate-teleportation},
\[
\bigl((\mathrm{CZ}(X^a\otimes X^b)\mathrm{CZ})\otimes \bra{ab}\bigr)
\mathrm{CNOT}^{\otimes 2}(\ket{\mathrm{CZ}}\ket{\psi})
=
\mathrm{CZ}\ket{\psi}
\qquad\text{for all }a,b\in\{0,1\}.
\]
Hence \eqref{eq:clifford-cz-after-adaptive} simplifies to
\begin{equation}
\sum_{a,b\in\{0,1\}}
E^{(\alpha)}_{ab}
\Bigl(
\encq(\mathrm{CZ}\ket{\psi})\otimes \ket{\eta^{(\alpha)}_{ab}}
\Bigr).
\label{eq:clifford-cz-after-adaptive-simplified}
\end{equation}
Moreover,
\[
|C_{\mathrm{in}}\cup C^{(\alpha)}_{\mathrm{ss}}\cup C^{(\alpha)}_{\mathrm{tr}}|
\le c_{\mathrm{Cl}}(\delta)n + 2M\sqrt{\eps}\,n
\le c_{\mathrm{Cl}}(\delta)n + c'_{\mathrm{Cl}}\sqrt{\eps}\,n
\le \frac32 c_{\mathrm{Cl}}(\delta)n
< c_{\mathrm{es}}(\delta)n.
\]
Therefore every Pauli term in the Pauli expansion of every \(E^{(\alpha)}_{ab}\) is a valid input to \Cref{prop:error-suppression} on each output code block.

\medskip
\noindent\textbf{Step 2b: repeated error suppression.}
Fix \(\alpha,a,b\), and expand \(E^{(\alpha)}_{ab}\) in the Pauli basis on the two output code blocks as
\[
E^{(\alpha)}_{ab}
=
\sum_{\beta}
\bigl(P^{(\alpha,ab)}_{\beta}\otimes R^{(\alpha,ab)}_{\beta}\bigr),
\]
where \(P^{(\alpha,ab)}_{\beta}\) is a Pauli operator on the two output code blocks supported on \(C_{\mathrm{in}}\cup C^{(\alpha)}_{\mathrm{ss}}\cup C^{(\alpha)}_{\mathrm{tr}}\), and \(R^{(\alpha,ab)}_{\beta}\) acts on the remaining registers. By the previous paragraph, each such Pauli term is within the admissible radius for \Cref{prop:error-suppression} on each output block. Applying \Cref{prop:error-suppression} independently to the two output blocks, each summand is mapped to a state that is a \([T_{\mathrm{es}}\sqrt{\eps}]^2\)-deviation from \(\encq(\mathrm{CZ}\ket{\psi})\). Summing back over \(\beta\), then over \(a,b\), and finally over \(\alpha\), linearity implies that the final output of \Cref{algo:distributed-Clifford} is a \([T_{\mathrm{es}}\sqrt{\eps}]^2\)-deviation from \(\encq(\mathrm{CZ}\ket{\psi})\). Since \(T_{\mathrm{es}}\le M\le c'_{\mathrm{Cl}}\), this proves the claimed bound for \(U=\mathrm{CZ}\).

\medskip
\noindent\textbf{The case \(U=H\).}
Now \(m=|\supp(U)|=1\). We analyze the fixed Hadamard-teleportation circuit from \Cref{fact:gate-teleportation}, with each logical resource preparation implemented by \Cref{algo:distributed-state-prep}, each adaptive logical step by \Cref{algo:adaptive-gates}, each cleanup step by \Cref{algo:error-correction}, and the single logical \(\mathrm{CZ}\) box implemented by the already-proved \(\mathrm{CZ}\) case above.

Consider any chain of boxes in this circuit whose outputs successively feed into one another and eventually affect the final logical output block. By definition of \(L\), such a chain contains at most \(L\) boxes. At the start of the chain, the input block is Pauli \(c_{\mathrm{Cl}}(\delta)\)-deviated. Each time a fresh resource block is introduced, \Cref{prop:distributed-state-prep} contributes a general deviation whose Pauli terms are supported on at most \(2c'_{\mathrm{ss}}\sqrt{\eps}\,n\le M\sqrt{\eps}\,n\) players. Each time \Cref{prop:adaptive} is applied to a Pauli term, the support of the resulting Pauli operator on the relevant output block enlarges by at most \(c'_{\mathrm{tr}}\sqrt{\eps}\,n\le M\sqrt{\eps}\,n\). Each time \Cref{prop:error-suppression} is applied, the affected block is reset to a \(T_{\mathrm{es}}\sqrt{\eps}\)-deviation, and hence to a general deviation whose Pauli terms are supported on at most \(M\sqrt{\eps}\,n\) players. Finally, the already-proved \(\mathrm{CZ}\) case outputs a \([T_{\mathrm{es}}\sqrt{\eps}]^2\)-deviation on its two output blocks, so in particular on any one output block the Pauli support size is at most \(M\sqrt{\eps}\,n\).

Thus each box in the chain contributes at most \(M\sqrt{\eps}\,n\) additional corrupted players to the relevant Pauli terms, and since there are at most \(L\) boxes in the chain, the total additional support accumulated along the chain is at most \(LM\sqrt{\eps}\,n = c'_{\mathrm{Cl}}\sqrt{\eps}\,n\). By \eqref{eq:clifford-eps-choice-2}, \(c'_{\mathrm{Cl}}\sqrt{\eps}\le \frac12 c_{\mathrm{Cl}}(\delta)\), so every active block entering any application of \Cref{prop:adaptive} or \Cref{prop:error-suppression} remains within the corresponding admissible radii \(c_{\mathrm{tr}}(\delta)\) and \(c_{\mathrm{es}}(\delta)\). Likewise, the input to the single recursive \(\mathrm{CZ}\) invocation in the Hadamard circuit is within the input radius \(c_{\mathrm{Cl}}(\delta)/2\) required in the already-proved \(\mathrm{CZ}\) case. Therefore all propositions invoked in the Hadamard construction are applicable throughout.

Repeating the same term-by-term expansion and linearity argument used in the \(\mathrm{CZ}\) case for each such chain, we conclude that the final output block of the Hadamard teleportation protocol is a \(c'_{\mathrm{Cl}}\sqrt{\eps}\)-deviation from \(\encq(H\ket{\psi})\).

Combining the two cases, we conclude that for every \(U\in\{H,\mathrm{CZ}\}\), with \(m=|\supp(U)|\), the protocol \Cref{algo:distributed-Clifford} is a
\[
(\mathcal{Q}^{\otimes m},\, c_{\mathrm{Cl}}(\delta)/m,\,
\mathcal{Q}^{\otimes m},\, c'_{\mathrm{Cl}}\sqrt{\eps},\, K_\eps)\text{-fault-tolerant}
\]
protocol for \(U\).
\end{proof}

\subsection{Non-Clifford Gates}
In this section, we show how to implement a logical CCZ gate using gate teleportation, which requires preparing logical CCZ state as resource states. The preparation of CCZ states relies on the CCZ state distillation protocol described in page 2 of~\cite{nguyen2025good}.

\begin{fact}[CCZ gate teleportation]\label{fact:CCZ-teleportation}
    The CCZ gate can be implemented using a resource state $\ket{\mathrm{CCZ}}:= \mathrm{CCZ}\ket{+++}$, CNOT gates, destructive Z-basis measurements, and adaptively controlled Pauli and CZ gates.
    \begin{center}
    \begin{tikzpicture}\node[scale=0.8]{
\begin{quantikz}
    \lstick{$\ket{+}$} & \ctrl{2} & \ctrl{3} & \qw & \qw & \gate[3, nwires={1}]{\mathrm{CCZ}(X^a X^b X^c) \mathrm{CCZ}} \qw & \rstick[3]{$\mathrm{CCZ}\ket{\psi}$} \qw \\
    \lstick{$\ket{+}$} & \control{} & \qw & \ctrl{3} & \qw & \qw & \\
    \lstick{$\ket{+}$} & \control{} & \qw & \qw & \ctrl{3} & \qw & \\
    \lstick[3]{$\ket{\psi}$} & \qw & \targ{} \qw & \qw &\qw & \meter{\hspace{0.6cm}  \text{\scriptsize Z}} \vcw{-1} \vbox{\hrule\hrule} \\
    & \qw & \qw & \targ{} &  & \meter{\hspace{0.6cm}  \text{\scriptsize Z}} \vcw{-1} \\
     & \qw & & \qw & \targ{} & \meter{\hspace{0.6cm}  \text{\scriptsize Z}} \vcw{-1} 
\end{quantikz}
};
\end{tikzpicture}
\end{center}
\end{fact}

\begin{lemma}[Magic state distillation~\cite{nguyen2025good}] \label{lem:magic-distill} There exist a constant $\delta_\mathrm{CCZ} >0$ and infinitely and geometrically dense sequence of integers $n_\mathrm{CCZ}$ such that the following holds.
    There exists a distillation protocol that takes as input a register that is $\delta_\mathrm{CCZ}n_\mathrm{CCZ}$-deviated (not necessarily Pauli deviated) from $\ket{\mathrm{CCZ}}^{\otimes n_\mathrm{CCZ}}$ and  outputs the state $\ket{\mathrm{CCZ}}$. Furthermore, the protocol only uses $\poly(n_\mathrm{CCZ})$ many CNOT gates, destructive measurements in computational/Hadamard basis, and adaptive Pauli and CZ gates, with the adaptivity following a $\mathrm{poly}(n_\mathrm{CCZ})$-sized auxiliary classical computation on the measurement outcomes.
\end{lemma}

In combination with~\Cref{algo:distributed-state-prep}, we obtain a protocol to prepare logical CCZ states.

\begin{mdframed}
\begin{protocol}[Distributed magic state preparation]\mbox{}\label{algo:distributed-magic-prep}\\
\textbf{Ideal operation:} Distributed preparation of a logical CCZ state $\ket{\mathrm{CCZ}}_\mathcal{Q}$\\
\textbf{Promise:} There exists $\Omega(n) \leq n_\mathrm{CCZ} \leq n/2$ such that the distillation algorithm~\Cref{lem:magic-distill} can be run.\\
\textbf{Output:} An $n$-qubit state jointly held by the $n$ players
\begin{enumerate}
\item \textbf{(Preparing input states)} For each input register required in the distillation algorithm~\Cref{lem:magic-distill}, each player $i$ prepares a qLTC codeword encoding the ideal state $\ket{\mathrm{CCZ}}_\mathcal{Q}$ of the input register, then distributes each codeword to all players. Now there are $n$ many code blocks being shared corresponding to this input register.
\item \textbf{(Filtering inputs)} Players perform the qLTC verification~\Cref{algo:vqss} on the distributed code blocks.

\item \textbf{(Logical distillation)} Using a pre-decided order and~\Cref{algo:adaptive-gates},~\Cref{algo:distributed-Clifford}, the players pick the first $n_\mathrm{CCZ}$ accepted code blocks to run the distillation algorithm in~\Cref{lem:stabilizer-distill}, at the logical level of the qLTC $\mathcal{Q}$ with these code blocks as input. (If there are not enough accepted code blocks, abort).
\end{enumerate}
\end{protocol}
\end{mdframed}

\begin{mdframed}
\begin{protocol}[Teleported CCZ gate]\mbox{}\label{algo:teleported-CCZ}\\
\textbf{Input:} A $3n$-qubit state $\ket{\Psi}$, a purported qLTC codeword $\mathrm{Enc}(\ket{\psi})$ jointly held by the $n$ players.\\
\textbf{Output:} An $3n$-qubit state $\Phi$, a purported qLTC codeword $\mathrm{Enc}(\mathrm{CCZ}\ket{\psi})$ jointly held by the $n$ players.
\begin{enumerate}
\setcounter{enumi}{0}
\item \textbf{(Preparing resource states)} Run~\Cref{algo:distributed-magic-prep} to prepare the required resource state $\ket{\mathrm{CCZ}}_\mathcal{Q}$.
\item \textbf{(Logical teleportation)} Using adaptive transversal gates (\Cref{algo:adaptive-gates}) and Clifford gates (\Cref{algo:distributed-Clifford}), run the teleportation circuit in~\Cref{fact:CCZ-teleportation} at the qLTC logical level using the input state and the resource states from step 1.
\end{enumerate}
\end{protocol}
\end{mdframed}

\begin{prop}[Teleported CCZ gate]
\label{prop:teleported-CCZ}
There exist $\eps^*_\mathrm{CCZ}= \mathrm{poly}(\delta)$, $c_\mathrm{CCZ}(\delta) = \poly(\delta)$ and constant $c'_\mathrm{CCZ}$ such that, for any $\eps \leq \eps^*_\mathrm{CCZ}$,~\Cref{algo:teleported-CCZ} is a $(\mathcal{Q}^{\otimes 3},c_\mathrm{CCZ}/10, \mathcal{Q}^{\otimes 3}, c'_\mathrm{CCZ}\sqrt{\eps}, K_{\eps})$-fault-tolerant protocol for the gate $\mathrm{CCZ}$.
\end{prop}
\begin{proof}
    Similar to that of~\Cref{prop:distributed-Clifford}. Note that the (adaptive) CZ gates inside the CCZ protocol can be done using~\Cref{algo:distributed-Clifford}.
\end{proof}

\subsection{Combining Everything}
We now combine the above protocols for state preparation and gate teleportation to obtain a distributed protocol for universal quantum computation (with input).

\paragraph{Gate set.} We will assume the compution is composed of gates from $\mathcal{G}= \{\mathrm{X}, \mathrm{Y}, \mathrm{Z}, \mathrm{CNOT}, \mathrm{H}, \mathrm{CZ}, \mathrm{CCZ}\}$, as well as the $\mathrm{Init}$ operation which initializes the $\ket{0}$ state. Note that this set is universal, but not minimal, and we choose it for simplicity of presentation.

Since in later sections, we will consider MPC protocols whose input state may be arbitrary, it will be convenient to consider a modification of the VQSS~\Cref{algo:vqss} such that when the input state is too corrupted and thus rejected, the compution can still proceed in a well-defined manner and outputs some logical codeword. This is done by simply performing the state preparation procedure of a fixed logical state, e.g. $\encq(\ket{0})$ using \Cref{algo:distributed-state-prep} (which in turn uses~\Cref{algo:vqss} as a subroutine), when~\Cref{algo:vqss} rejects.

\begin{mdframed}
\begin{protocol}[VQSS with state replacement]\mbox{}\label{algo:vqss-state-replacement}\\
\textbf{Input:} An $n$-qubit pure state $\ket{\Psi}$, purported to be a codeword of $\mathcal{Q}$, jointly held by the $n$ players. \\
\textbf{Output:} The original codeblock register, and another $n$-qubit register that is ideally close to $\ket{0^n}$ (accept) or $\ket{1^n}$ (reject), both registers jointly held by the players.

\begin{enumerate}
\setcounter{enumi}{0}
\item \textbf{(Verification)} Run~\Cref{algo:vqss} to verify the state.
\item \textbf{(Replacement)} If the verification in step 1 rejects, run~\Cref{algo:distributed-state-prep} to prepare $\encq(\ket{0})$ and replace the input state with this state.
\end{enumerate}
\end{protocol}
\end{mdframed}

The guarantees of this modified VQSS protocol follow directly from the guarantees of~\Cref{algo:vqss} and~\Cref{algo:distributed-state-prep} so we will not spell out the guarantees here. Instead, we directly use this modified protocol as a subroutine in the proof of the main theorem below, which includes the VQSS with replacement protocol as a special case (by taking the computation $U$ below to be identity).

\begin{theorem}[Distributed universal computation]\label{thm:distributed-FTQC}
There exist \(\eps^*_\mathrm{comp}=\poly(\delta)\) and an absolute constant \(c_\mathrm{comp}>0\) such that the following holds. For any quantum circuit \(U\) mapping from \(k\) to \(k'\) qubits using \(T\) operations from the gate set \(\{\mathcal{G},\mathrm{Init}\}\), and for any \(\eps \leq \eps^*_\mathrm{comp}\), we can efficiently construct an MPC protocol \(\mathcal P\) that is a \((\mathcal{Q}^{\otimes k}, c_\mathrm{comp}\sqrt{\eps}, \mathcal{Q}^{\otimes k'}, c_\mathrm{comp}\sqrt{\eps}, K_{\eps})\)-fault-tolerant and friendly protocol for \(U\).
\end{theorem}

\begin{proof}
Let \(w\) and \(\kappa\) be the constants from \Cref{prop:vqss}. Let \(c'_\mathrm{ver}\) and \(\eps^*_\mathrm{ver}\) be the constants from \Cref{prop:vqss}, let \(T_{\mathrm{es}}\), \(c_{\mathrm{es}}(\delta)\), and \(\eps^*_{\mathrm{es}}\) be as in \Cref{prop:error-suppression}, let \(c'_\mathrm{ss}\) and \(\eps^*_{\mathrm{ss}}\) be as in \Cref{prop:distributed-state-prep}, let \(c_{\mathrm{tr}}(\delta),c'_{\mathrm{tr}},\eps^*_{\mathrm{tr}}\) be as in \Cref{prop:adaptive}, let \(c_{\mathrm{Cl}}(\delta),c'_{\mathrm{Cl}},\eps^*_{\mathrm{Cl}}\) be as in \Cref{prop:distributed-Clifford}, and let \(c_{\mathrm{CCZ}}(\delta),c'_{\mathrm{CCZ}},\eps^*_{\mathrm{CCZ}}\) be as in \Cref{prop:teleported-CCZ}.

Choose an absolute constant
\[
c_{\mathrm{comp}}
\ge
\max\{1,\,T_{\mathrm{es}},\,c'_{\mathrm{ss}},\,c'_{\mathrm{tr}},\,c'_{\mathrm{Cl}},\,c'_{\mathrm{CCZ}}\}.
\]
Next choose \(\eps^*_{\mathrm{comp}}=\poly(\delta)\) small enough that, for every \(\eps\le \eps^*_{\mathrm{comp}}\), all of the following hold:
\begin{enumerate}
    \item \(\eps\le \min\{\eps^*_{\mathrm{ver}},\eps^*_{\mathrm{ss}},\eps^*_{\mathrm{es}},\eps^*_{\mathrm{tr}},\eps^*_{\mathrm{Cl}},\eps^*_{\mathrm{CCZ}}\}\);
    \item \((c_{\mathrm{comp}}+c'_{\mathrm{ver}})\sqrt{\eps}\le c_{\mathrm{es}}(\delta)\);
    \item \(c'_{\mathrm{ss}}\sqrt{\eps}\le c_{\mathrm{es}}(\delta)\);
    \item \(\bigl(\frac{w c_{\mathrm{comp}}+4\kappa}{\delta}+c'_{\mathrm{ver}}\bigr)\sqrt{\eps}\le c_{\mathrm{es}}(\delta)\);
    \item \(c_{\mathrm{comp}}\sqrt{\eps}\le \min\{c_{\mathrm{tr}}(\delta)/2,\ c_{\mathrm{Cl}}(\delta)/2,\ c_{\mathrm{CCZ}}(\delta)/10\}\).
\end{enumerate}
This is possible because \(c_{\mathrm{comp}}\), \(w\), \(\kappa\), \(c'_{\mathrm{ver}}\), \(T_{\mathrm{es}}\), \(c'_{\mathrm{ss}}\), \(c'_{\mathrm{tr}}\), \(c'_{\mathrm{Cl}}\), and \(c'_{\mathrm{CCZ}}\) are absolute constants, while \(c_{\mathrm{es}}(\delta)\), \(c_{\mathrm{tr}}(\delta)\), \(c_{\mathrm{Cl}}(\delta)\), and \(c_{\mathrm{CCZ}}(\delta)\) are all \(\poly(\delta)\). In particular, item 4 only requires \(\sqrt{\eps}\le \delta\,c_{\mathrm{es}}(\delta)/O(1)\), which is still a \(\poly(\delta)\) condition.

We now build the protocol \(\mathcal P\).

\paragraph{Input sanitization.}
For a single input codeblock, define a gadget \(\mathcal S\) as follows:
\begin{enumerate}
    \item run the VQSS-with-replacement protocol \Cref{algo:vqss-state-replacement}, where in the invocation of~\Cref{algo:vqss} the threshold in step 2 is set to \((w c_{\mathrm{comp}}+2\kappa)\sqrt{\eps}\,m\);
    \item then run the repeated error-suppression protocol \Cref{algo:error-correction} on the resulting codeblock.
\end{enumerate}
We claim that \(\mathcal S\) is a friendly \((\mathcal Q,c_{\mathrm{comp}}\sqrt{\eps},\mathcal Q,c_{\mathrm{comp}}\sqrt{\eps},K_\eps)\)-fault-tolerant gadget for the identity on one logical qubit, and that the same statement remains valid when arbitrary untouched reference registers are adjoined.

First consider the ordinary fault-tolerance guarantee. Let the input block be a Pauli \(c_{\mathrm{comp}}\sqrt{\eps}\)-deviation of \(\encq(\ket{\phi})\). Apply \Cref{prop:vqss} with \(c_0=c_{\mathrm{comp}}\) (note that this also governs the setting of the threshold in \Cref{prop:vqss}). The completeness part of \Cref{prop:vqss} then implies that the verification stage accepts, and the data block after verification is a Pauli \((c_{\mathrm{comp}}+c'_{\mathrm{ver}})\sqrt{\eps}\)-deviation of \(\encq(\ket{\phi})\). By item 2 above, this lies within the input radius of \Cref{prop:error-suppression}, so after the error-suppression stage the output block is a \(T_{\mathrm{es}}\sqrt{\eps}\)-deviation of \(\encq(\ket{\phi})\), hence a \(c_{\mathrm{comp}}\sqrt{\eps}\)-deviation.

Next we prove friendliness. Let the input block be an arbitrary Pauli-deviation \(P\,\encq(\ket{\phi})\), with \(P\) chosen minimal-weight. Again applying \Cref{prop:vqss} with \(c_0=c_{\mathrm{comp}}\), the verification stage outputs either the accept form or the reject form. If it outputs the reject form, then the replacement stage prepares a fresh encoded \(\ket{0}\), and by \Cref{prop:distributed-state-prep} the resulting block is a \(c'_{\mathrm{ss}}\sqrt{\eps}\)-deviation of \(\encq(\ket{0})\). By item 3 above, the subsequent error-suppression stage is applicable, and the final block becomes a \(T_{\mathrm{es}}\sqrt{\eps}\)-deviation of \(\encq(\ket{0})\). If instead the verification stage outputs the accept form, then \(|\supp(P)|<\frac{w c_{\mathrm{comp}}+4\kappa}{\delta}\sqrt{\eps}\,n\), since the soundness part of \Cref{prop:vqss} would otherwise force rejection. In that case, the data block after verification is a Pauli \(\bigl(\frac{w c_{\mathrm{comp}}+4\kappa}{\delta}+c'_{\mathrm{ver}}\bigr)\sqrt{\eps}\)-deviation of \(\encq(\ket{\phi})\). By item 4 above, the error-suppression stage is again applicable, and the final block becomes a \(T_{\mathrm{es}}\sqrt{\eps}\)-deviation of \(\encq(\ket{\phi})\). Therefore, in all cases the output of \(\mathcal S\) is a \(c_{\mathrm{comp}}\sqrt{\eps}\)-deviation of \(\encq(\ket{\xi})\) for some one-qubit state \(\ket{\xi}\in\{\ket{\phi},\ket{0}\}\). This proves friendliness of \(\mathcal S\).

Now apply \(\mathcal S\) successively to the \(k\) input codeblocks, always treating the untouched codeblocks and all accumulated workspace as untouched reference registers. Since each invocation has the same input and output threshold \(c_{\mathrm{comp}}\sqrt{\eps}\), repeated application of sequential composability (\Cref{def:seq-composability}) yields an input-sanitization stage \(\mathcal S_{\mathrm{in}}\) that is a \((\mathcal Q^{\otimes k},c_{\mathrm{comp}}\sqrt{\eps},\mathcal Q^{\otimes k},c_{\mathrm{comp}}\sqrt{\eps},K_\eps)\)-fault-tolerant gadget for the identity on the full \(k\)-qubit input register. The same inductive argument also proves friendliness: on an arbitrary Pauli-deviated encoded \(k\)-qubit input, \(\mathcal S_{\mathrm{in}}\) outputs a \(c_{\mathrm{comp}}\sqrt{\eps}\)-deviation of \(\Enc_{\mathcal Q^{\otimes k}}(\ket{\xi})\) for some \(k\)-qubit state \(\ket{\xi}\).

\paragraph{Elementary computation gadgets.}
Write the ideal circuit as \(U=g_T\circ \cdots \circ g_1\), where each \(g_t\) is one of \(\mathrm{Init},X,Y,Z,\mathrm{CNOT},H,\mathrm{CZ},\mathrm{CCZ}\). For each \(t\), let \(q_t\) denote the number of logical qubits present after the first \(t\) ideal steps, so \(q_0=k\) and \(q_T=k'\).

For each gate occurrence \(g_t\), we define a whole-register gadget \(\widetilde g_t\) acting on the current collection of codeblocks, always regarding the untouched codeblocks and all workspace registers as untouched reference systems.

\begin{itemize}
    \item If \(g_t\in\{X,Y,Z,\mathrm{CNOT}\}\), we use the transversal-gate case of \Cref{remark:transversal} on the codeblocks touched by \(g_t\). By item 5 above, the current input deviation \(c_{\mathrm{comp}}\sqrt{\eps}\) is within its input radius. Its output is a \(c'_{\mathrm{tr}}\sqrt{\eps}\)-deviation, hence a \(c_{\mathrm{comp}}\sqrt{\eps}\)-deviation.

    \item If \(g_t\in\{H,\mathrm{CZ}\}\), we use \Cref{prop:distributed-Clifford} on the codeblocks touched by \(g_t\). Again item 5 guarantees that the current input deviation \(c_{\mathrm{comp}}\sqrt{\eps}\) is admissible. Its output is a \(c'_{\mathrm{Cl}}\sqrt{\eps}\)-deviation, hence a \(c_{\mathrm{comp}}\sqrt{\eps}\)-deviation.

    \item If \(g_t=\mathrm{CCZ}\), we use \Cref{prop:teleported-CCZ} on the three touched codeblocks. By item 5 the current input deviation \(c_{\mathrm{comp}}\sqrt{\eps}\) is within its input radius, and its output is a \(c'_{\mathrm{CCZ}}\sqrt{\eps}\)-deviation, hence a \(c_{\mathrm{comp}}\sqrt{\eps}\)-deviation.

    \item If \(g_t=\mathrm{Init}\), we append a fresh codeblock using \Cref{prop:distributed-state-prep}. The new block is a \(c'_{\mathrm{ss}}\sqrt{\eps}\)-deviation of \(\encq(\ket{0})\), hence a \(c_{\mathrm{comp}}\sqrt{\eps}\)-deviation.
\end{itemize}

Thus every \(\widetilde g_t\) is a \((\mathcal Q^{\otimes q_{t-1}},c_{\mathrm{comp}}\sqrt{\eps},\mathcal Q^{\otimes q_t},c_{\mathrm{comp}}\sqrt{\eps},K_\eps)\)-fault-tolerant gadget for the corresponding ideal logical operation \(g_t\).

\paragraph{Composition.}
Define the overall protocol by
\[
\mathcal P:=\widetilde g_T\circ \cdots \circ \widetilde g_1\circ \mathcal S_{\mathrm{in}}.
\]
All consecutive gadgets have matching code families and matching threshold \(c_{\mathrm{comp}}\sqrt{\eps}\). Therefore repeated application of sequential composability, i.e. \Cref{def:seq-composability}, shows that \(\mathcal P\) is a \((\mathcal Q^{\otimes k},c_{\mathrm{comp}}\sqrt{\eps},\mathcal Q^{\otimes k'},c_{\mathrm{comp}}\sqrt{\eps},K_\eps)\)-fault-tolerant protocol for the ideal circuit \(U\).

\paragraph{Friendliness.}
Let the initial input be arbitrarily Pauli-deviated from an encoded \(k\)-qubit state. By friendliness of \(\mathcal S_{\mathrm{in}}\), after the input-sanitization stage the state is a \(c_{\mathrm{comp}}\sqrt{\eps}\)-deviation of \(\Enc_{\mathcal Q^{\otimes k}}(\ket{\xi})\) for some \(k\)-qubit state \(\ket{\xi}\). Since each subsequent gadget \(\widetilde g_t\) is fault-tolerant and its input promise is satisfied, repeated application of the same sequential-composition argument shows that the final output of \(\mathcal P\) is a \(c_{\mathrm{comp}}\sqrt{\eps}\)-deviation of \(\Enc_{\mathcal Q^{\otimes k'}}(U\ket{\xi})\). Hence \(\mathcal P\) is friendly.

\paragraph{Efficiency.}
Each primitive gadget used above has per-player work complexity \(\poly(n)\) by the cited propositions, and the total number of invoked gadgets is \(O(T+k)\). Therefore \(\mathcal P\) can be constructed efficiently from the circuit \(U\), and each player's total work is polynomial in \(n\) and \(T\).

This proves the theorem.
\end{proof}

\subsection{Converting to Repetition Code}
The following protocol will be used after MPC is concluded so that every participating player can recover the classical outcome of the computation.

\paragraph{Repetition-code notation.}
For a single-qubit state \(\ket{\psi}=\alpha\ket{0}+\beta\ket{1}\), define
\[
\Enc_{\mathrm{rep}}(\ket{\psi})
~:=~
\alpha\ket{0}^{\otimes n}+\beta\ket{1}^{\otimes n}.
\]
This is the \(n\)-qubit repetition encoding of \(\ket{\psi}\), with one output qubit intended for each player.

\begin{mdframed}
\begin{protocol}[Repetition-code readout]\mbox{}\label{algo:rep-code-output}\\
\textbf{Input:} A distributed \(\mathcal{Q}\)-codeblock \(\encq(\ket{\psi})\) jointly held by the \(n\) players.\\
\textbf{Ideal output:} The repetition-code encoding \(\Enc_{\mathrm{rep}}(\ket{\psi})\), with one qubit held by each player.
\begin{enumerate}
\item Let \(B_1\) denote the input block. The players run \Cref{algo:distributed-state-prep} independently \(n-1\) times to prepare fresh blocks \(B_2,\ldots,B_n\), each ideally equal to \(\encq(\ket{0})\).

\item For each \(j=2,\ldots,n\), do the following:
\begin{enumerate}
    \item Apply a logical CNOT from \(B_1\) to \(B_j\). Since no measurements are involved, this is the \(\ell=0\) special case of \Cref{algo:adaptive-gates}, namely transversal physical CNOT gates between the two code blocks.
    \item Run \Cref{algo:error-correction} on \(B_1\).
    \item Run \Cref{algo:error-correction} on \(B_j\).
\end{enumerate}

\item For each \(i\in[n]\), player \(i\) receives the block \(B_i\) and unencodes it to a single qubit.
\end{enumerate}
\end{protocol}
\end{mdframed}

\begin{lemma}\label{lem:rep-code}
Let \(\delta=\delta(n)\) be the parameter from \Cref{assumption:nice-qLTC}, so that the decoder for \(\mathcal{Q}\) corrects up to \(n\delta\) errors on one code block. Then there exist \(\beta^*_{\mathrm{out}}(\delta)=\poly(\delta)\) and \(\eps^*_{\mathrm{out}}(\delta)=\poly(\delta)\) such that, for any \(\beta\le \beta^*_{\mathrm{out}}(\delta)\), for any \(\eps\le \eps^*_{\mathrm{out}}(\delta)\), and for any set of faulty communication edges $F$ with $|F|\le \eps \frac{n^2}{2}$, the following holds.

Suppose the input state is a Pauli \(\beta n\)-deviation of \(\encq(\ket{\psi})\). Call a vertex \emph{good} if it is incident to at most \(\delta n/2\) faulty edges, and let
\[
W(F):=\{v\in [n]:\deg_F(v)>\delta n/2\}
\]
be the set of bad vertices. Then the output of \Cref{algo:rep-code-output} is deviated from \(\Enc_{\mathrm{rep}}(\ket{\psi})\) only on qubits held by players in \(W(F)\). 
This implies that \Cref{algo:rep-code-output} is a \((\mathcal{Q},\beta,\mathrm{Rep},2\eps/\delta,F)\)-fault-tolerant protocol for the identity operation, where \(\mathrm{Rep}\) is the repetition code.

Finally, each player is involved in at most \(\poly(n)\) qubits and gates during the protocol.
\end{lemma}

\begin{proof}
Let \(\eps^*_{ss}\) and \(c'_{ss}\) be as in \Cref{prop:distributed-state-prep}, and let \(\eps^*_{\mathrm{es}}\), \(c_{\mathrm{es}}(\delta)\), and \(T_{\mathrm{es}}\) be as in \Cref{prop:error-suppression}. Define $\beta^*_{\mathrm{out}}(\delta):=\frac{c_{\mathrm{es}}(\delta)}{4}.$
Choose \(\eps^*_{\mathrm{out}}(\delta)\) so that
\[
\eps^*_{\mathrm{out}}
~\le~
\min\!\left\{
\eps^*_{ss},
\eps^*_{\mathrm{es}},
\left(\frac{\beta^*_{\mathrm{out}}(\delta)}{c'_{ss}}\right)^2,
\left(\frac{\beta^*_{\mathrm{out}}(\delta)}{T_{\mathrm{es}}}\right)^2,
\left(\frac{\delta}{4T_{\mathrm{es}}}\right)^2
\right\}.
\]
This is possible with \(\beta^*_{\mathrm{out}}(\delta)=\poly(\delta)\) and \(\eps^*_{\mathrm{out}}(\delta)=\poly(\delta)\).

Now fix \(\beta\le \beta^*_{\mathrm{out}}(\delta)\), \(\eps\le \eps^*_{\mathrm{out}}(\delta)\), and a faulty edge set \(F\) with $|F|\le \eps \frac{n^2}{2}.$
Let
\[
s:=\sqrt{2|F|},
\qquad
t_{ss}:=c'_{ss}s,
\qquad
t_{\mathrm{es}}:=T_{\mathrm{es}}s.
\]
By the choice of \(\eps^*_{\mathrm{out}}\),
\[
t_{ss}\le \beta^*_{\mathrm{out}}(\delta)n,
\qquad
t_{\mathrm{es}}\le \beta^*_{\mathrm{out}}(\delta)n,
\qquad
t_{\mathrm{es}}\le \delta n/4.
\]

Let \(B_1\) be the input block and \(B_2,\ldots,B_n\) the fresh blocks from Step \(1\). By hypothesis, the input block is $E_1\,\encq(\ket{\psi})$
for some Pauli \(E_1\) with $\operatorname{wt}(E_1)\le \beta n.$
Also, since \(\eps^*_{\mathrm{out}}\le \eps^*_{ss}\), \Cref{prop:distributed-state-prep} applies to each of the \(n-1\) preparations in Step \(1\). Hence each fresh block \(B_j\) for \(j=2,\ldots,n\) is a Pauli \(t_{ss}\)-deviation of \(\encq(\ket{0})\).

We now analyze Step \(2\). For \(j=2,\ldots,n\), let \(R_j\) denote the \(j\)-th round: apply the logical CNOT from \(B_1\) to \(B_j\), then run \Cref{algo:error-correction} on \(B_1\) and on \(B_j\).

\begin{claim}
After round \(R_j\), the control block \(B_1\) and every processed target block \(B_2,\ldots,B_j\) are each Pauli \(t_{\mathrm{es}}\)-deviations of their ideal code blocks, while every unprocessed target block \(B_{j+1},\ldots,B_n\) remains a Pauli \(t_{ss}\)-deviation of \(\encq(\ket{0})\).
\end{claim}
\begin{proof}
We argue by induction on \(j\).

For \(j=2\), before the CNOT, the control block \(B_1\) is a Pauli \(\beta n\)-deviation of \(\encq(\ket{\psi})\), and the target block \(B_2\) is a Pauli \(t_{ss}\)-deviation of \(\encq(\ket{0})\). The logical CNOT is implemented transversally by local physical CNOT gates at each player, so it uses no communication edges and introduces no additional edge faults. Conjugating Pauli errors through the transversal CNOT shows that after the CNOT, each of \(B_1\) and \(B_2\) is a Pauli \((\beta n+t_{ss})\)-deviation of its ideal code block. Since \(\beta\le \beta^*_{\mathrm{out}}(\delta)\) and \(t_{ss}\le \beta^*_{\mathrm{out}}(\delta)n\),
\[
\beta n+t_{ss}\le 2\beta^*_{\mathrm{out}}(\delta)n=\frac{c_{\mathrm{es}}(\delta)}{2}n<c_{\mathrm{es}}(\delta)n.
\]
Therefore \Cref{prop:error-suppression} applies to both \(B_1\) and \(B_2\), and each becomes a Pauli \(t_{\mathrm{es}}\)-deviation of its ideal code block. This proves the base case.

Now suppose the claim holds after round \(R_{j-1}\), for some \(j\ge 3\). Then before round \(R_j\), the control block \(B_1\) is a Pauli \(t_{\mathrm{es}}\)-deviation of its ideal code block, and the fresh target block \(B_j\) is a Pauli \(t_{ss}\)-deviation of \(\encq(\ket{0})\). Again, the transversal CNOT introduces no additional edge faults, and after the CNOT each of \(B_1\) and \(B_j\) is a Pauli \((t_{\mathrm{es}}+t_{ss})\)-deviation of its ideal code block. Since \(t_{\mathrm{es}},t_{ss}\le \beta^*_{\mathrm{out}}(\delta)n\),
\[
t_{\mathrm{es}}+t_{ss}\le 2\beta^*_{\mathrm{out}}(\delta)n=\frac{c_{\mathrm{es}}(\delta)}{2}n<c_{\mathrm{es}}(\delta)n.
\]
So \Cref{prop:error-suppression} applies to both \(B_1\) and \(B_j\), and each becomes a Pauli \(t_{\mathrm{es}}\)-deviation of its ideal code block. All earlier processed targets are untouched in round \(R_j\), so they remain \(t_{\mathrm{es}}\)-deviated, and all later targets remain \(t_{ss}\)-deviated. This proves the claim.
\end{proof}

\medskip
After all \(n-1\) rounds of Step \(2\), every block \(B_i\), \(i\in[n]\), is a Pauli \(t_{\mathrm{es}}\)-deviation of its ideal code block.

Finally, consider Step \(3\). Fix a good player \(i\notin W(F)\). Before routing, the block \(B_i\) is a Pauli \(t_{\mathrm{es}}\)-deviation of its ideal code block. During routing, each physical qubit of \(B_i\) is transmitted over one communication edge incident to player \(i\), so the number of additional corrupted qubits introduced during routing is at most
\[
\deg_F(i)\le \delta n/2.
\]
Hence the total number of corrupted qubits on \(B_i\) at decoding time is at most
\[
t_{\mathrm{es}}+\delta n/2\le \delta n/4+\delta n/2<\delta n.
\]
By the decoding guarantee from \Cref{assumption:nice-qLTC}, player \(i\)'s decoder returns the correct logical qubit. Since this holds for every good player \(i\notin W(F)\), the final output differs from \(\Enc_{\mathrm{rep}}(\ket{\psi})\) only on qubits held by players in \(W(F)\).

For the size of \(W(F)\), the handshake lemma gives
\[
\sum_{v\in [n]}\deg_F(v)=2|F|.
\]
Since every \(v\in W(F)\) satisfies \(\deg_F(v)>\delta n/2\),
\[
|W(F)|\cdot \frac{\delta n}{2}
<
\sum_{v\in W(F)}\deg_F(v)
\le
2|F|.
\]
Therefore
\[
|W(F)|
<
\frac{4|F|}{\delta n}
\le
\frac{2\eps}{\delta}\,n.
\]

Finally, each player uses only \(\poly(n)\) resources. Step \(1\) invokes \Cref{algo:distributed-state-prep} \(n-1\) times. Step \(2\) performs \(n-1\) transversal CNOTs together with \(2(n-1)\) invocations of \Cref{algo:error-correction}. By \Cref{prop:error-suppression}, each invocation of \Cref{algo:error-correction} uses \(\poly(n)\) qubits and gates per player. Step \(3\) asks each player to decode one code block, which is polynomial-time by \Cref{assumption:nice-qLTC}. Hence each player is involved in at most \(\poly(n)\) qubits and gates throughout the protocol.
\end{proof}

%% file: routing.tex
\section{Fault-Tolerant Quantum Routing on an HDX}\label{sec:routing}

In this section, we construct a fault-tolerant protocol for routing encoded qubits between cliques of an HDX according to an arbitrary permutation. We first recall the classical clique-to-clique routing construction in~\Cref{subsec:classical-transfer}. Then in~\Cref{subsec:quantum-transfer}, we replace its majority transfers with quantum transfers using an inner quantum code concatenation and the error-suppression procedure from~\Cref{sec:mpc}. We introduce a `partial' fault tolerance framework to allow a small fraction of codeblocks to fail in \Cref{sec:prelims-partial-ft}. Finally, in~\Cref{subsec:quantum-routing}, we combine the clique-to-clique transfer guarantee with HDX concentration bounds to prove the quantum routing result in~\Cref{prop:ft-routing}.

\subsection{The Classical Clique-to-Clique Routing Protocol}\label{subsec:classical-transfer}
We start by stating the clique-to-clique routing statement in~\cite{BafnaMV25} and its version by~\cite{BT} in the small gap regime. The proof of the latter result is the same as~\cite{BafnaMV25} except that it uses~\Cref{lem:concentration-improved} as the sampling bound instead of~\Cref{lem:chernoff-hdx} and it uses a derandomization of~\cite{Nenadov23} to set up the clique-to-clique paths over an expander. Below is a figure that illustrates the routing protocol,

\begin{figure}[H]
    \centering
    \includegraphics[width=0.5\linewidth]{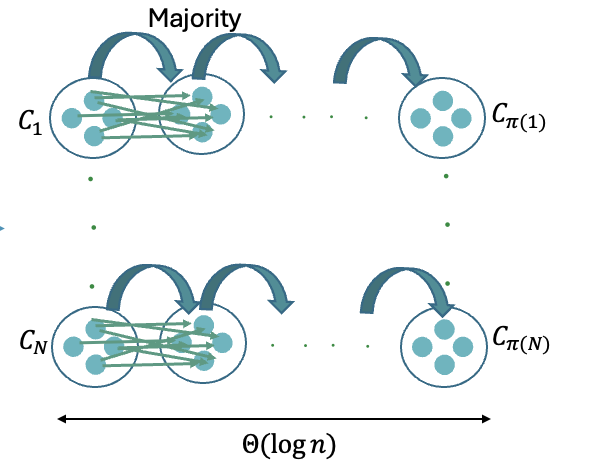}
    \caption{Clique-to-Clique Routing Protocol where in each vertical layer of cliques every clique in $G$ appears exactly once, and the clique-to-clique transfers proceed via majority votes.
    }
    \label{fig:clique-routing}
\end{figure}

We first give a sketch of~\cite{BafnaMV25}'s proof for the sake of completeness, and then generalize it to the quantum setting (with slightly different parameters).

\begin{lemma}\label{lem:clique-routing}
Let $X$ be a $d$-partite $\lambda$-HDX with $|X(1)|=n$. Let $\pi:X(d)\rightarrow X(d)$ be any permutation. Then we can efficiently construct a routing protocol on $G=(X(1),X(2))$ with round complexity $T=O(\log n)$ and work complexity $(1/\lambda)^{O(d^2)}\log |\Sigma|$ such that for all initial functions $A_0:X(d)\rightarrow \Sigma^d$ satisfying
\[\Pr_{D\sim X(d)}[\maj_{0.99}(A_0(D)) \neq \perp] \geq 1-\eta,\]
and for all possible adversaries that corrupt at most $\mu$-fraction of edges we have the following two statements:
\begin{enumerate}
\item \textbf{(Low corruption regime)} There exists $d_0 \in \N$ such that for all $c \in \N$ and for large enough $n \in \N$, when $d > d_0$ and $\mu,\lambda \lll 1/\log^{c+1} n$, the protocol computes the function $A_T:X(d)\rightarrow \Sigma^d$ satisfying
\[\Pr_{D\sim X(d)}[\maj_{0.99}(A_T(\pi(D)))=\maj_{0.99}(A_0(D))]\geq 1-\eta-\frac{\mu}{\log^c n}.\]
\item \textbf{(Constant corruption regime)} There exist $\mu_0 \in (0,1)$ such that for all $c \in \N$ and for large enough $n \in \N$, when $\mu < \mu_0$ and  $d = \Theta(\log\log^2 n)$ and $\lambda < 1/\poly(d)$, the protocol computes the function $A_T:X(d)\rightarrow \Sigma^d$ satisfying
\[\Pr_{D\sim X(d)}[\maj_{0.99}(A_T(\pi(D)))=\maj_{0.99}(A_0(D))]\geq 1-\eta-\frac{1}{\log^c n}.\]
\end{enumerate}
\end{lemma}

\begin{proof}
Fix an ordering on the cliques in $X(d)$, a permutation $\pi$ on $[N]$ (recall that $|X(d)|=N$), and an initial  function $A_0$. 
\vspace{-2ex}
\paragraph{Setting up paths:} First consider the graph $H$ defined on the vertex set $X(d)$ whose edges are given by the down-up random walk $U^d_{\alpha d}D^d_{\alpha d}$ for some $\alpha \in (0,1)$ chosen to be a small enough constant. Since $\pi_d$ is the uniform distribution over $X(d)$, $H$ is a regular graph with multiedges. By \Cref{lem:spectral_gap_of_graphs_from_HDX} we know that $H$ has $\sigma_2(H) \leq \alpha+\poly(d)\lambda \leq 2\alpha$. By~\cite{Nenadov23}'s result and its derandomization presented in~\cite{BT}, provided that $2\alpha$ is a small enough constant, we can efficiently set up paths of length $O(\log|H|)=O(\log N)$ where each vertex appears in at most $O(\log N)$ paths; thus, we fix a routing protocol $\cR$ that has $N$ paths each of length $T' \leq O(\log N)$. Let the $j^{th}$ path in $\cR$ be the path from $C_j \rightarrow C_{\pi(j)}$.

Now let $T:=2T'$, and expand each path above to also include the $\alpha d$-cliques associated with each edge in $H$. That is, for even $t \in [0,T]$ let $C_{j,t}$ denote the $d$-face that occurs in the $(t/2)^{th}$-time-step of the $j^{th}$ path in $\cR$, and for odd $t\in [T]$ let $C_{j,t}$ denote the $\alpha d$-face that is associated to the edge $(C_{j,t-1},C_{j,t+1})\in E(H)$. Let $C_{j,0}$ denote $C_j$ and $C_{j,T}$ denote $C_{\pi(j)}$. 

\vspace{-2ex}
\paragraph{Routing protocol:}
The routing protocol on $G$ proceeds as follows. For each path $j$ and time-step $t$, we perform a clique-to-clique transfer from a $d$-clique $C_1$ to an $\alpha d$-clique $C_2$ and back to a $d$-clique $C_3$ using majority votes in the obvious way. We remark that this clique-to-clique transfer can be seen as a special case of~\Cref{lem:clique-to-clique-transfer} where the message is encoded in a repetition code, and the lemma we provided is for the general setting of quantum codes.

\vspace{-2ex}
\paragraph{Bad cliques and bounding them:}
We now begin the analysis of the protocol, and for that we need to introduce a few notions. 
First, at most $\mu$-fraction of the edges are corrupted, and we denote the set of corrupted edges by $\cE \subseteq X(2)$. Define the set $\cB \subseteq X(d)$ of bad $d$-faces $D$ as those that contain at least $\alpha/100$-fraction of bad edges. Now there are two regimes of parameters, one where $\mu$ is smaller than $1/\pl n$ and the other when $\mu$ is a small constant (independent of $n$). We can use the concentration bounds on HDX from~\Cref{sec:concentration-prelims} to bound the fraction of bad cliques in both regimes:
\begin{enumerate}
\item If $d \gtrsim 1/\alpha^2$, using~\Cref{lem:concentration-improved} we get,
\begin{equation}
\mu(\cB)=\Pr_{D \sim X(d)}\left[\mu(\cE|_D) \geq \frac{\alpha}{100}\right]\leq \frac{2}{\alpha^2}\mu(\mu+4\lambda) \leq \frac{\mu}{\log^{c+1} n},     
\end{equation}
since $\mu, \lambda\lll 1/\log^{c+1} n$.
\item Using~\Cref{lem:chernoff-hdx} we get,
\begin{equation}\label{claim:a_few_bad_cliques}
\mu(\cB) =\Pr_{D \sim X(d)}\left[\mu(\cE|_D) \geq \frac{\alpha}{100}\right] \leq 
\Pr_{D \sim X(d)}\left[\mu(\cE|_D) \geq \mu(\cE)+ \left(\frac{\alpha}{100}-\mu\right)\right] \leq 
2^{-\Omega(\sqrt{d})} \lll \frac{1}{\log^{c+1} n}, 
\end{equation}
where we used that
$d \geq \Omega(\log\log^2 n)$ and $\mu \lll \alpha$. 
\end{enumerate}

\paragraph{Error Analysis of the Protocol:} One can check if the cliques $C_1$ and $C_3$ are good, i.e. they contain at most $\alpha/100$-fraction of bad edges, then the clique-to-clique transfer (using majority votes) from $C_1 \rightarrow C_1 \cap C_3=C_2 \rightarrow C_3$ succeeds. Thus using this fact and a union bound over the path length, we get that at most $\eta+\mu/\log^c n$-paths fail to transmit the required message in the first regime, and $\eta+1/\log^c n$-fraction of paths in the second regime.
\end{proof}

\subsection{One Clique-to-Clique Transfer in the Quantum Setting}\label{subsec:quantum-transfer}
The clique-to-clique routing protocol in~\cite{BafnaMV25} uses as its basic primitive a majority transfer between a clique $C_i$ to a clique $C_j$, where these cliques share an $\alpha d$-sized clique. We prove an analogue of the clique-to-clique routing protocol where one wants to route qubits. In words, every $d$-dimensional clique $C_i$ of $X$ holds an encoding of a qubit, that needs to be transmitted to $C_{\pi(i)}$, for some permutation $\pi$. In their setting, this encoding was a repetition code of a classical bit, but in our setting since the edges of the graph can introduce quantum errors, we need to use a quantum code. We prove an analogue of their statement for routing qubits. To replace their use of classical majority gates, we use the following protocol.

\begin{mdframed}
\begin{protocol}Clique-to-clique transfer on an HDX\mbox{}\label{algo:transfer}\\
\textbf{Input:} A state $\ket{\Psi}$ purported to be a length-$d$ codeword of $\mathcal{Q}$, jointly held by vertices in a $d$-clique $C_1$. \\
\textbf{Output:} The original codeword is now jointly held by the players in another $d$-clique $C_3$. \\
\textbf{Promise:} There is a smaller clique $C_2$ of size $\alpha d$ that is contained in both $C_1$ and $C_3$ for $\alpha <1$.
\begin{enumerate}
\setcounter{enumi}{0}
\item Each player in $C_1$ encodes their qubit into a linear-distance quantum code (`inner code') of size $\alpha d$ and sends qubit $j$ to player $j$ in $C_2$.
\item Each player in $C_2$ forwards the received qubits to $C_3$: qubits received by player $i$ of $C_1$ is sent to player $i$ of $C_3$.
\item The players in $C_3$ first locally error-correct and decode the inner code. Then they run the the error suppression~\Cref{algo:error-correction}.
\end{enumerate}
\end{protocol}
\end{mdframed}

\begin{lemma}[Clique-to-clique transfer]\label{lem:clique-to-clique-transfer}
Let \(\mathcal Q\) be the code of blocklength \(d\) from \Cref{assumption:nice-qLTC}, and let \(\delta=\delta(d)\) be the associated parameter from that assumption.

Then there exist \(\beta^*_{\mathrm{transfer}}>0\) and a universal constant \(c_{\mathrm{transfer}}>0\) such that, for any \(\beta\le \beta^*_{\mathrm{transfer}}\), and any set of bad edges \(F\subseteq E(C_1)\cup E(C_3)\) such that, for each \(r\in\{1,3\}\), \(|F\cap E(C_r)|\le \beta \binom{d}{2}\), the following holds.

\Cref{algo:transfer} is a \((\mathcal Q,c_{\mathrm{transfer}}\sqrt{\beta},F)\)-fault-tolerant protocol for the identity operation in the sense of \Cref{def:fault-tolerant-protocol}. Equivalently, for every state \(\ket{\psi}\), every input state that is \(c_{\mathrm{transfer}}\sqrt{\beta}\)-deviated from \(\mathrm{Enc}_{\mathcal Q}(\ket{\psi})\), and every Pauli fault string on the communications of \Cref{algo:transfer} supported only on communications whose underlying edges lie in \(F\), the output state held by \(C_3\) is \(c_{\mathrm{transfer}}\sqrt{\beta}\)-deviated from \(\mathrm{Enc}_{\mathcal Q}(\ket{\psi})\).

Furthermore, each player is involved in at most \(\poly(d)\) qubits and gates during the protocol.
\end{lemma}

\begin{proof}
Let \(C_1\) and \(C_3\) be \(d\)-cliques, and let \(C_2\subseteq C_1\cap C_3\) be a clique of size \(\alpha d\), where \(\alpha\in(0,1)\). Fix bijections \(\iota_1:[d]\to C_1\), \(\iota_3:[d]\to C_3\), and \(\kappa:[\alpha d]\to C_2\), corresponding to the labeling used in \Cref{algo:transfer}.

Let \(\rho_{\mathrm{in}}>0\) be a relative decoding radius for the inner code used in Steps~1--3 of \Cref{algo:transfer}, so that this inner code has blocklength \(\alpha d\) and corrects any Pauli error on at most \(\rho_{\mathrm{in}}\alpha d\) qubits. Let \(\epsilon_{\mathrm{es}}^*(\delta)\), \(c_{\mathrm{es}}(\delta)\), and \(T_{\mathrm{es}}\) be the quantities from \Cref{prop:error-suppression}.

Choose a universal constant \(c_{\mathrm{transfer}}\ge T_{\mathrm{es}}\), and define
\[
\beta^*_{\mathrm{transfer}}
:=
\min\!\left\{
\left(\frac{\rho_{\mathrm{in}}\alpha c_{\mathrm{transfer}}}{2}\right)^2,\,
\left(\frac{c_{\mathrm{es}}(\delta)}{2c_{\mathrm{transfer}}}\right)^2,\,
\epsilon_{\mathrm{es}}^*(\delta)
\right\}.
\]

Now fix an arbitrary state \(\ket{\psi}\), let \(\ket{\Phi}:=\mathrm{Enc}_{\mathcal Q}(\ket{\psi})\), let \(\beta\le \beta^*_{\mathrm{transfer}}\), and let \(\mathsf P\) be an arbitrary Pauli string on the communications of \Cref{algo:transfer} such that \(\mathsf P_\ell=I\) whenever the \(\ell\)-th communication uses an edge not in \(F\).

We may write the designated input block as \(P_{S_{\mathrm{in}}}\ket{\Phi}\), where \(S_{\mathrm{in}}\subseteq[d]\), \(|S_{\mathrm{in}}|\le c_{\mathrm{transfer}}\sqrt{\beta}\,d\), and \(P_{S_{\mathrm{in}}}\) is a Pauli operator supported on \(S_{\mathrm{in}}\).

We first analyze Steps~1 and~2 of \Cref{algo:transfer}. For each outer coordinate \(i\in[d]\), Step~1 encodes the \(i\)-th qubit into an inner codeword of length \(\alpha d\), whose positions are indexed by \(j\in[\alpha d]\), and sends the \(j\)-th inner-code qubit along the edge \(\{\iota_1(i),\kappa(j)\}\in E(C_1)\). Step~2 forwards that same qubit along the edge \(\{\kappa(j),\iota_3(i)\}\in E(C_3)\).

For each \(i\in[d]\), define \(e_1(i):=\bigl|\{\,j\in[\alpha d] : \{\iota_1(i),\kappa(j)\}\in F\,\}\bigr|\) and \(e_3(i):=\bigl|\{\,j\in[\alpha d] : \{\kappa(j),\iota_3(i)\}\in F\,\}\bigr|\). Thus \(e_1(i)\) is the number of inner-code positions of outer coordinate \(i\) whose Step~1 transmission uses a bad edge, and \(e_3(i)\) is the analogous quantity for Step~2.

Write \(F_1:=F\cap E(C_1)\) and \(F_3:=F\cap E(C_3)\). By hypothesis, \(|F_1|,|F_3|\le \beta\binom{d}{2}\le \frac{\beta d^2}{2}\).

Each edge of \(F_1\) is used by at most two Step~1 transmissions: if exactly one endpoint lies in \(C_2\), it is used once, and if both endpoints lie in \(C_2\), it is used once in each direction. Therefore \(\sum_{i=1}^d e_1(i)\le 2|F_1|\le \beta d^2\). Similarly, each edge of \(F_3\) is used by at most two Step~2 transmissions, so \(\sum_{i=1}^d e_3(i)\le 2|F_3|\le \beta d^2\).

Define \(B(F):=\{\,i\in[d] : e_1(i)+e_3(i)>\rho_{\mathrm{in}}\alpha d\,\}\). For each \(i\in B(F)\), by definition \(e_1(i)+e_3(i)>\rho_{\mathrm{in}}\alpha d\). Summing over \(i\in B(F)\) yields
\[
|B(F)|\rho_{\mathrm{in}}\alpha d
<
\sum_{i\in B(F)}\bigl(e_1(i)+e_3(i)\bigr)
\le
\sum_{i=1}^d\bigl(e_1(i)+e_3(i)\bigr)
\le 2\beta d^2,
\]
and hence \(|B(F)|\le \frac{2\beta}{\rho_{\mathrm{in}}\alpha}\,d\).

Now let \(B_{\mathrm{tot}}:=S_{\mathrm{in}}\cup B(F)\). Then
\[
|B_{\mathrm{tot}}|
\le
\left(
c_{\mathrm{transfer}}\sqrt{\beta}
+
\frac{2\beta}{\rho_{\mathrm{in}}\alpha}
\right)d.
\]
Because \(\beta \le \left(\frac{\rho_{\mathrm{in}}\alpha c_{\mathrm{transfer}}}{2}\right)^2\), we have \(\frac{2\beta}{\rho_{\mathrm{in}}\alpha}\le c_{\mathrm{transfer}}\sqrt{\beta}\), and therefore \(|B_{\mathrm{tot}}|\le 2c_{\mathrm{transfer}}\sqrt{\beta}\,d\).

We claim that after the local inner-code correction and decoding in Step~3, the designated \(d\)-qubit block held by \(C_3\) is a Pauli \((|B_{\mathrm{tot}}|/d)\)-deviation of \(\ket{\Phi}\). Indeed, fix any \(i\in[d]\setminus B_{\mathrm{tot}}\). Since \(i\notin S_{\mathrm{in}}\), the \(i\)-th outer-code qubit is initially correct. Since \(i\notin B(F)\), we have \(e_1(i)+e_3(i)\le \rho_{\mathrm{in}}\alpha d\). Because \(\mathsf P\) is supported only on communications using edges in \(F\), the \(i\)-th inner-code block experiences Pauli faults on at most \(e_1(i)+e_3(i)\) of its \(\alpha d\) physical qubits during Steps~1 and~2. This is within the decoding radius of the inner code, so the local inner-code correction and decoding performed at the vertex \(\iota_3(i)\in C_3\) recovers exactly the \(i\)-th qubit of \(\ket{\Phi}\).

Thus, after the local inner decoding step, only the coordinates in \(B_{\mathrm{tot}}\) can differ from \(\ket{\Phi}\). Hence the designated \(d\)-qubit block is a Pauli \((|B_{\mathrm{tot}}|/d)\)-deviation of \(\ket{\Phi}\). In particular, just before the final call to \Cref{algo:error-correction} on \(C_3\), the designated block is a Pauli \(2c_{\mathrm{transfer}}\sqrt{\beta}\)-deviation of \(\ket{\Phi}\).

We now apply \Cref{prop:error-suppression} to the final repeated error-suppression stage on \(C_3\), with bad edge set \(F_3\). Since \(\beta \le \epsilon_{\mathrm{es}}^*(\delta)\), the edge-fault hypothesis of \Cref{prop:error-suppression} is satisfied. Also, because \(\beta \le \left(\frac{c_{\mathrm{es}}(\delta)}{2c_{\mathrm{transfer}}}\right)^2\), we have \(2c_{\mathrm{transfer}}\sqrt{\beta}\le c_{\mathrm{es}}(\delta)\), so the input-deviation hypothesis of \Cref{prop:error-suppression} is also satisfied. Therefore that proposition implies that the final designated output block is a Pauli \(T_{\mathrm{es}}\sqrt{\beta}\)-deviation of \(\ket{\Phi}\). Since \(c_{\mathrm{transfer}}\ge T_{\mathrm{es}}\), this is in particular a Pauli \(c_{\mathrm{transfer}}\sqrt{\beta}\)-deviation of \(\ket{\Phi}=\mathrm{Enc}_{\mathcal Q}(\ket{\psi})\), as required.

This proves that \Cref{algo:transfer} is a \((\mathcal Q,c_{\mathrm{transfer}}\sqrt{\beta},F)\)-fault-tolerant protocol for the identity operation in the sense of \Cref{def:fault-tolerant-protocol}.

For the resource bound, each player in \(C_1\) encodes one qubit into an inner block of size \(\alpha d\), so it handles \(O(d)\) qubits and gates in Step~1. Each player in \(C_2\) forwards one qubit for each outer coordinate, so it handles \(O(d)\) qubits and gates in Step~2. Each player in \(C_3\) decodes one inner block of size \(\alpha d\), and then participates in \Cref{algo:error-correction}, which uses \(\poly(d)\) qubits and gates per player by \Cref{prop:error-suppression}. Therefore every player is involved in at most \(\poly(d)\) qubits and gates during the protocol.
\end{proof}

We now prove a quantum analogue of this statement for routing qubits on an HDX. The proof remains the same, except that we use the clique-to-clique transfer in~\Cref{lem:clique-to-clique-transfer} and we set the parameters accordingly. Before we do that we provide a few definitions that extend the fault-tolerance notions in~\Cref{sec:prelims-ft} to the setting of multiple cliques, since we have to give up on some of the codeblocks in the output state, so we also consider a relaxed notion of fault-tolerant protocol where we only require the output state to be deviated from an encoding of $C\ket{\psi}$ on a subset of the output codeblocks.

\subsection{Partial Fault-Tolerance Preliminaries}\label{sec:prelims-partial-ft}

\begin{definition}[Partial Deviation]
For $0 \leq \alpha \leq 1$, we say that $\ket{\psi}$ is a \emph{$(\alpha,\vec{t})$-deviation} of $\ket{\phi}_{R_1,\ldots, R_m}$ if the requirements of $\vec{t}$-deviation hold true on at least a $(1-\alpha)$ fraction of the registers (for the same registers across the terms in the decomposition in~\Cref{eq:deviation-decomp}). In other words, $(0,\vec{t})$-deviation is the same as $\vec{t}$-deviation.    
\end{definition}

\begin{definition}[Partially FT protocols]
    \label{def:partial-FT}
    Consider the same setting as in~\Cref{def:fault-tolerant-protocol}. For $F \subseteq E$ and $0 \leq \beta \leq 1$, the protocol $\cP$ is said to be a \emph{$(\mathcal{Q},\vec{t},\mathcal{Q'},\vec{t'},F, \beta)$-partially fault-tolerant protocol} for $C$ if there exists a constant $c > 0$ such that the following statement holds:

    For any $0 \leq \alpha \leq 1$ and $k$-qubit pure state $\ket{\psi}$, upon any input state $\ket{\Psi}$ that is Pauli \((\alpha, \vec{t})\)-deviated from \(\mathrm{Enc}_\mathcal{Q}(\ket{\psi})\), an $F$-edge faulty execution of $\cP$ specified by Pauli faults $\mathsf{F}
$ outputs a state \(\cP_\mathsf{F}(\ket{\Psi})\) is \((\min \{c \alpha+\beta, 1\},\vec{t'})\)-deviated from \(\mathrm{Enc}_{\mathcal{Q}'}(C\ket{\psi})\). We also assume that the the bad output block set \(B'\) can be chosen as a function only of the bad input blocks and the Pauli fault pattern \(\mathsf F\), and not of the
particular Pauli deviation on the good blocks. 

If the above statement holds for all $\mu(F) \leq \epsilon $, then we say $\cP$ is a $(\mathcal{Q},\vec{t},\mathcal{Q'},\vec{t'}, G_\epsilon, \beta)$-FT protocol.
\end{definition}

\begin{lemma}[Sequential composability of partially fault-tolerant gadgets]
\label{lem:partial-ft-sequential}
Let \(\mathcal{Q}_1,\mathcal{Q}_2,\mathcal{Q}_3\) be codes with associated block partitions of sizes
\(m_1,m_2,m_3\), respectively, and let \(G_1,G_2\) be operations mapping \(k_1\to k_2\) and
\(k_2\to k_3\) qubits.

Suppose \(\mathcal{P}_1\) is a
\((\mathcal{Q}_1,\vec t_1,\mathcal{Q}_2,\vec t_2,F_1,\beta_1)\)-partially fault-tolerant protocol
for \(G_1\), witnessed by a constant \(c_1>0\), and \(\mathcal{P}_2\) is a
\((\mathcal{Q}_2,\vec t_2,\mathcal{Q}_3,\vec t_3,F_2,\beta_2)\)-partially fault-tolerant protocol
for \(G_2\), witnessed by a constant \(c_2>0\). Then \(\mathcal{P}_2\circ \mathcal{P}_1\) is a
\((\mathcal{Q}_1,\vec t_1,\mathcal{Q}_3,\vec t_3,F_1\cup F_2,\beta_{12})\)-partially
fault-tolerant protocol for \(G_2\circ G_1\), witnessed by
$c_{12}=c_2c_1$ and $\beta_{12}=\min\{c_2\beta_1+\beta_2,1\}.
$
\end{lemma}

\begin{proof}
Fix \(0\le \alpha\le 1\), a \(k_1\)-qubit pure state \(\ket{\psi}\), a Pauli
\((\alpha,\vec t_1)\)-deviation \(\ket{\Psi}\) of \(\Enc_{\mathcal Q_1}(\ket{\psi})\), and an
\((F_1\cup F_2)\)-edge faulty execution of \(\mathcal P_2\circ \mathcal P_1\), specified by a
Pauli fault pattern \(\mathsf F\). Let \(\mathsf F_1,\mathsf F_2\) be the restrictions of
\(\mathsf F\) to \(\mathcal P_1\) and \(\mathcal P_2\), respectively.

By definition of \((\alpha,\vec t_1)\)-deviation, there is a set $B_0\subseteq [m_1]$ of bad input blocks with
$|B_0|\le \alpha m_1$
such that \(\ket{\Psi}\) is \(\vec t_1\)-deviated on every block outside \(B_0\).

Apply \(\mathcal P_1\). By the partial fault-tolerance guarantee for \(\mathcal P_1\),
there exists a set \(B_1\subseteq [m_2]\), depending only on \(B_0\) and \(\mathsf F_1\), such that
\[
\frac{|B_1|}{m_2}
\le
\min\!\left\{c_1\frac{|B_0|}{m_1}+\beta_1,\,1\right\}
\le
\min\{c_1\alpha+\beta_1,1\}
=: \alpha_1,
\]
and the intermediate state
$\ket{\Omega}:=(\mathcal P_1)_{\mathsf F_1}\ket{\Psi}$ 
is \((B_1,\vec t_2)\)-deviated from \(\Enc_{\mathcal Q_2}(G_1\ket{\psi})\), i.e. it is
\(\vec t_2\)-deviated on the output block outside \(B_1\).

Now write
\[
\ket{\Omega}=\sum_i \ket{\Omega_i}\otimes \ket{\xi_i},
\]
where each \(\ket{\Omega_i}\) is a Pauli \((B_1,\vec t_2)\)-deviation of
\(\Enc_{\mathcal Q_2}(G_1\ket{\psi})\). Since the bad-block set is the same \(B_1\) for every term,
the fixed-set guarantee for \(\mathcal P_2\) applies with the same input bad set \(B_1\) to every
\(\ket{\Omega_i}\). Hence there exists a set \(B_2\subseteq [m_3]\), depending only on \(B_1\) and
\(\mathsf F_2\), such that
\[
\frac{|B_2|}{m_3}
\le
\min\!\left\{c_2\frac{|B_1|}{m_2}+\beta_2,\,1\right\}
\le
\min\{c_2(c_1\alpha+\beta_1)+\beta_2,1\},
\]
and for every \(i\),
$(\mathcal P_2)_{\mathsf F_2}\ket{\Omega_i}$
is \((B_2,\vec t_3)\)-deviated from \(\Enc_{\mathcal Q_3}(G_2G_1\ket{\psi})\).

Therefore, by linearity,
\[
(\mathcal P_2\circ \mathcal P_1)_{\mathsf F}\ket{\Psi}
=
\sum_i (\mathcal P_2)_{\mathsf F_2}\ket{\Omega_i}\otimes \ket{\xi_i},
\]
so the final state is itself \((B_2,\vec t_3)\)-deviated from
\(\Enc_{\mathcal Q_3}(G_2G_1\ket{\psi})\). Thus it is
\((\alpha_2,\vec t_3)\)-deviated, where
\[
\alpha_2
=
\min\{c_2c_1\alpha+c_2\beta_1+\beta_2,1\}
=
\min\{c_{12}\alpha+\beta_{12},1\}.
\]
This is exactly the required partial fault-tolerance guarantee for
\(\mathcal P_2\circ \mathcal P_1\).
\end{proof}

\subsection{Clique-to-Clique Routing in the Quantum Setting}\label{subsec:quantum-routing}

\begin{prop}[Fault-tolerant quantum routing on an HDX]
\label{prop:ft-routing}
Let $X$ be the $d$-partite $\lambda$-HDX from~\Cref{lem:clique-routing} with $|X(1)|=n$ and for large enough $n \in \N$ and $|X(d)|=N$. Let $\pi:X(d)\rightarrow X(d)$ be any permutation and $c \in \N$ be any universal constant. Then we can efficiently construct a quantum routing protocol $\mathcal{R}$ on $G=(X(1),X(2))$ with round complexity $T=O(\log n)$ and work complexity $(1/\lambda)^{O(d^2)}$ and a constant $c_\mathrm{route} >0$
 such that $\mathcal{R}$ is a $(\mathcal{Q}^{\otimes N}, c_\mathrm{route} \delta , G_{\mu},  \frac{\mu}{\log^c n})$-partially FT protocol (with the universal constant $1$) for the $N$-qubit permutation unitary $U_\pi$ in the sense of~\Cref{def:partial-FT} with the following choice of HDX parameters in different regimes of $\mu$:
\begin{enumerate}
    \item \textbf{(Low corruption)} When $\mu \lll 1/\log^{c+1} n$, set $d > d_0$ for an absolute constant $d_0 \in \N$ and $\lambda \lll 1/\log^{c+1} n$.
    \item \textbf{(Almost-constant corruption)} When $ \mu \lll 1/\mathrm{polylogloglog} \ n$, set $d = \Theta(\log\log^3 n)$ and $\lambda = 1/\log^c n$.
\item \textbf{(Constant corruption)} When $\mu < \mu_0$ where $\mu_0 \in (0,1)$ is an absolute constant, set $d = \Theta(\log\log^2 n)$ and $\lambda = 1/\log^c n$. This assumes the existence of an asymptotically good quantum LTC. 
\end{enumerate}
\end{prop}

\begin{proof}
\textbf{Construction of $\mathcal{R}$:} We set up the paths as in the proof of~\Cref{lem:clique-routing}, and we use the clique-to-clique transfer protocol from~\Cref{lem:clique-to-clique-transfer} to transfer the qubits along the paths.
\\
\textbf{Error analysis:} The main difference is that~\Cref{lem:clique-routing} only guarantees that the majority value on the target clique is correct, i.e., the underlying codeword held by each clique is the length-$d$ repetition code, and the innder code used in the clique-to-clique transfer is also a repetition code of length $\alpha d$. In our case, the underlying code is a length-$d$ MPC-supporting quantum code $\mathcal{Q}$ with decoding radius $\delta= 1/\mathrm{polylog}(d)$ (\Cref{assumption:nice-qLTC}), while the inner code is another length-$\alpha d$ quantum code with decoding radius $\rho_\mathrm{in}=\Theta(1)$ in our clique-to-clique transfer~\Cref{lem:clique-to-clique-transfer}. Thus, we need to set the parameters accordingly in the HDX concentration inequalities. In particular, let $\mu$-fraction of the edges be corrupted
and define the set $\mathcal{B} \subseteq X(d)$ of bad $d$-faces as those that contain at least $\beta \lesssim (\alpha \delta)^2 = 1/\mathrm{polylog}(d)$-fraction of bad edges. 
\begin{enumerate}
    \item In the low corruption regime, we pick a large enough constant $d$ such that $d \ggg 1/\beta^2$ and $\beta \lesssim 1/\pl d$, and we pick $\lambda \lll 1/\log^{c+1} n$. Then for any $\mu \leq \lambda$ we apply~\Cref{lem:concentration-improved} (with $i=2$) and find that the fraction of bad cliques is upper bounded by
    $$ \mu(\mathcal{B}) \leq  \frac{ 1 + O\left(\frac{1}{d} \right)}{\beta^2 - O\left(\frac{1}{d} \right)} \cdot \mu \cdot (\mu  + 4\lambda) \leq  O_d(1)\mu \lambda.$$
    \item In the almost-constant corruption regime, we use~\Cref{lem:chernoff-hdx} (in particular,~\Cref{eq:hdx-chernoff-2}) with $i=2$, $\varepsilon= \beta/2$ and $\mu \leq \beta/2$ to get that the fraction of bad cliques is upper bounded by
    \begin{align}
        \mu(\mathcal{B}) \leq \exp\left(-c'\varepsilon^2 \frac{\sqrt{d}}{2}\right) \leq \exp\left(-\Theta\left(\frac{\log\log^{3/2} n}{\mathrm{polylogloglog}\ n} \right)  \right) \leq \frac{\mu}{\log^{c+1} n}.
    \end{align}
    \item In the constant corruption regime, we assume a good quantum LTC satisfying~\Cref{assumption:nice-qLTC} with decoding radius $\delta_0= \Theta(1)$. Then we can again use~\Cref{lem:chernoff-hdx} with $\beta, \varepsilon \lll \delta_0^2 $ and bound
    \begin{align}
        \mu(\mathcal{B}) \leq \exp\left(-c'\varepsilon^2 \frac{\sqrt{d}}{2}\right) \leq \exp\left(-\Theta\left(\log\log n \right)  \right) \leq \frac{\mu}{\log^{c+1} n}.
    \end{align}
\end{enumerate}
Each clique participates in at most $O(\log N)$ paths, so we can guarantee the fraction of paths that fail is at most $\mu \lambda \cdot O(\log N) \leq \mu/\log^c n$ for large enough $n$ in all cases. 

Now we can apply~\Cref{lem:clique-to-clique-transfer} to each path and conclude that at the end of the protocol, for at least $1-\mu/\log^c n$ fraction of the paths, the underlying codeword held by the target clique is a $c_\mathrm{route} \delta$-deviation of the codeword held by the source clique, where $c_\mathrm{route}$ is a constant that depends on the parameters of the inner and outer codes used in~\Cref{lem:clique-to-clique-transfer}. This proves that $\mathcal{R}$ is a $(\mathcal{Q}^{\otimes N}, c_\mathrm{route} \delta , G_{\mu},  \frac{\mu}{\log^c n})$-partially FT protocol for $U_\pi$.

\end{proof}

%% file: ckt-to-protocol-v2.tex
\section{Distributed Computing Protocol  for Commuting LH}\label{sec:protocol-for-circuit}
In~\Cref{sec:overview}, we described a template for our gap amplification. Here we formalize it in the setting where the input local Hamiltonian is commuting. This provides many of the ideas required in the analysis of our basic protocol, so we present it before moving to the more general case of non-commuting local Hamiltonians in~\Cref{sec:nc-protocol}.

We start by considering the sequential measurement circuit of the given LH, and then describe a distributed computing protocol to implement this circuit robustly on an HDX. Let us start by defining this circuit more formally,

\begin{definition}[Sequential measurement circuit]\label{def:sequential-measurement-circuit-appendix}
Fix constants \(k,g\). We consider Hamiltonians
\(H=\sum_{i=1}^m h_i\) on \(n\) qubits such that each $h_i$ is a projector on $k$ qubits that can be described using $O(1)$ bits, and each term \(h_i\) intersects the support of at most \(g=\Theta(1)\) other terms. Since each term intersects only \(O(1)\) other terms, the interaction graph of the terms has bounded degree. Hence the terms can be partitioned into \(\ell=\ell(g)=O(1)\) layers
\(H=\sum_{j=1}^\ell H_j\), where \(H_j=\sum_{a=1}^{m_j} h_{j,a}\) and the terms inside each \(H_j\) have pairwise disjoint supports.

For a layer \(j\) and a bit string \(y_j\in\{0,1\}^{m_j}\), define
\(Q^{(j)}_{y_j}:=\prod_{a=1}^{m_j}(y_{j,a}h_{j,a}+(1-y_{j,a})(I-h_{j,a}))\).
Because the layer terms are disjoint, the family \(\{Q^{(j)}_{y_j}\}_{y_j}\) consists of pairwise orthogonal projections satisfying \(\sum_{y_j}Q^{(j)}_{y_j}=I\). For \(\mathbf y=(y_1,\dots,y_L)\), set \(Q_{\mathbf y}:=Q^{(L)}_{y_L}\cdots Q^{(1)}_{y_1}\).

The \emph{sequential measurement circuit} \(C\) is the circuit that coherently measures every layer \(H_j\) into fresh ancillas, so that on any input state \(\ket{\psi}\),
\[
C\ket{\psi}=\sum_{\mathbf y}\ket{\mathbf y}_A\otimes Q_{\mathbf y}\ket{\psi}_S.
\]
Thus the output distribution is \(\Pr[\mathbf y]=\|Q_{\mathbf y}\ket{\psi}\|^2\).
\end{definition}

We will give a natural distributed computing protocol $\cP$ on an HDX $G=(X(1),X(2))$ to implement $C$. We will use a $[d,1,\Delta]$-MPC code $\mathcal{Q}$ and a $[d,1,d]$-repetition code $\mathsf{Rep}$, on each $d$-clique in the HDX. The protocol $\cP$ receives the encoded version of $N$ witness qubits, each encoded in the MPC code from~\Cref{sec:mpc} and distributed to one of the cliques. Additionally, each vertex starts out with a lot of ancillas, which are used in $\cP$ to prepare encoded ancillas required for executing $C$, as well as for other purposes throughout the protocol. At the end of the protocol, the $m$ measurement qubit blocks are converted to $m$ repetition code blocks, i.e. $\encq \rightarrow \encc$.

\begin{mdframed}
\begin{protocol}\mbox{}\label{algo:overall-protocol} Distributed Protocol $\cP$ for Sequential Measurement Circuit $C$\\
\textbf{Input:}  Every clique $i \in [N]$ in the HDX $G=(X(1), X(2))$ receives $k$ code blocks; each block is distributed to the $d$ players.\\
\textbf{Output:}
\begin{enumerate}
\item For each time-step $t \geq 1$ of circuit $C$ with gate fan-in $k$,
\begin{enumerate}
\item Clique $i$ runs the $i^{th}$ gate $G_{i,t}$ in $C$ using the quantum-multiparty computation protocol~\Cref{thm:distributed-FTQC}. This results in $k$ code blocks supported on the vertices of clique $i$, supposedly encoding the $k$ output wires in a quantum code, and the measurement wire in a classical repetition code (using the repetition code readout protocol~\Cref{lem:rep-code}).
\item Decompose the perfect matching formed by the outputs of the $t^{th}$-layer of $C$ and the inputs of the $(t+1)^{th}$ layer of $C$, into $k$ permutations $\pi_1,\ldots, \pi_k$\footnote{Note that we can always do this by Hall's theorem.}.
\item Route each permutation $\pi_i$ using the clique-to-clique routing in~\Cref{prop:ft-routing}. This results in clique $i$ holding $k$ code blocks, which are supposed to be encoding the input wires of $G_{i,t+1}$ in a quantum code. 
\end{enumerate}
\end{enumerate}
\end{protocol}
\end{mdframed}

In the protocol above, each measurement output in $C$, is now encoded in a repetition code and distributed across the vertices of the clique that computes this output. Thus in effect each vertex in $G$ actually contains a lot of measurement outputs -- one for every clique that it is part of, and it can take an OR over all the bits it sees and declare that as its output (with violation being denoted by a `1' output). The completeness of this protocol is quite straightforward -- given a perfectly encoded witness for an LH with zero (or inverse-doubly-exponential) energy, one can check that the protocol above run without any faults correctly computes all the measurements. Thus all the vertices would see zero violations in their local views, and output a `0'. The main technical work goes into analyzing the soundness. Here we show that even a faulty-run of the protocol (as dictated by the noise model that arises in~\Cref{sec:ckt-to-H}) ends up with a lot of vertices outputting `1' in expectation: $\E_v[\Tr(\Pi_v\ketbra{\Psi}{\Psi})$ is large, where $\ket{\Psi}$ is the input to the protocol $\cP$ and each $\Pi_v$ is the measurement operator that corresponds to taking an OR over all measurements that the vertex $v$ is part of (\Cref{def:vertex-check-operator}).

\begin{theorem}[Soundness of~\Cref{algo:overall-protocol}]\label{thm:commuting-main}
Suppose the input $k$-local Hamiltonian has (unnormalized) ground energy $\epsilon M$ and is commuting. Fix \(\cE\subseteq X(2)\) with \(\mu(\cE)  \lll \min(\epsilon \pl n,\ \eps_0)\).
    Consider an adversarially faulty execution of the distributed protocol $\cP$ for the sequential measurement circuit $C$ of $H$, so that the computation state can be expressed as
    \begin{align}
        \ket{\Psi} = \sum_{\sF\in \mathsf{PAULI}^{\cE}} c_{\sF} \cP_\sF \ket{\phi_\sF}.
    \end{align}
    and is a normalized state.
    Then,
    \begin{align}
        \mathbb{E}_{v} \Tr(\Pi_v \ketbra{\Psi}) \geq \epsilon {\pl n},
    \end{align}
where $\Pi_v$ is the vertex OR operator from~\Cref{def:vertex-check-operator}.
\end{theorem}

The rest of this section is organized as follows. In~\Cref{subsec:commuting-completeness}, we prove the completeness of the protocol.
The proof of soundness begins by recording some facts about commuting Hamiltonians in~\Cref{sec:robust-ckt}, which help us avoid the postselection attacks discussed in the overview. Next in~\Cref{sec:serial-repetition} we show a `serial repetition' type of statement, i.e. the soundness of the OR checks if performed on a non-faulty run of $C$ (in the logical space). Then in the rest of the section, we show how to interpret the state $\ket{\Psi}$ above as the encoded output of the circuit $C$ run on an appropriate valid state~\Cref{subsec:FT-commuting-protocol} and~\Cref{subsec:physical-to-logical-commuting-protocol}. After doing this, we derive the soundness statement for $\ket{\Psi}$ using an application of the serial repetition lemma in~\Cref{subsec:soundness-commuting-protocol}.

\subsection{Completeness}\label{subsec:commuting-completeness}
We first prove the completeness of the protocol for zero-energy LH.
\begin{claim}\label{claim:completeness-commuting}
    Let the input Hamiltonian $H$ have a zero-energy ground state $\ket{\xi_0}$, and consider the state
    $
        \ket{\Psi_0} = \cP  \encq(\ket{\xi_0}).
    $
    Then,
    $
        \mathbb{E}_{v} \Tr(\Pi_v \ketbra{\Psi_0}) =0.
    $
\end{claim}
\begin{proof} By construction, the output energy measurement qubits of the protocol are a repetition code encoding of the logical values $0$. Thus all local views contain only $0$'s.
\end{proof}

\subsection{Robust Measurement Circuit for a Gapped Commuting Hamiltonian}\label{sec:robust-ckt}

Below is a standard fact about commuting Hamiltonians,
\begin{fact}[Simultaneous diagonalization of commuting Hamiltonian]\label{fact:clh-eigenspace}
Let $H_1,\dots,H_m$ be Hermitian operators on  the Hilbert space $\mathcal H$ (each $H_i$ may be local, but is
viewed as an operator on all of $\mathcal H$ by tensoring with identities), such that they commute pairwise: $[H_i,H_j]=0, \forall\, i,j.$
Then there exists an orthonormal basis $\{\,|\psi_\alpha\rangle\,\}_\alpha$ of $\mathcal H$
and real numbers $\lambda_{i,\alpha}$ such that for every $i$ and $\alpha$,
\[
H_i|\psi_\alpha\rangle=\lambda_{i,\alpha}|\psi_\alpha\rangle.
\]
\end{fact}

\begin{fact}\label{lem:clh-circuit}
Consider a $k$-local commuting Hamiltonian $H=\sum_i \Pi_i$ with common eigenbasis $\cB$ (from~\Cref{fact:clh-eigenspace}) and energy $\delta$. Let $C$ be the sequential measurement circuit for $H$. Given any input $\ket{\phi} = \sum_{\ket{v}\in \cB} c_v \ket{v}$ to $C$ we have that,
\[C(\ket{\phi}\otimes \ket{0}^m) = \sum_{\ket{v}\in \cB} c_v\ket{v}\otimes \ket{b_v},\]
where each $b_v \in \{0,1\}^m$ with Hamming weight at least $\delta m$.
\end{fact}

\begin{proof}
Since $\ket{v} \in \cB$, $\Pi_i \ket{v}$ is either $0$ or $1\cdot\ket{v}$, and the corresponding gate in the sequential measurement circuit $C$ outputs $\ket{v}\otimes$ $\ket{0}$ or $\ket{v}\otimes \ket{1}$ respectively. So $C(\ket{v} \otimes \ket{0}^m)$ would be equal to $\ket{v}\otimes \ket{b_v}$ for some length $m$ bit string $b_v$. The energy of $H$ further implies that the Hamming weight of $b_v$ must be at least $\delta m$. Thus, given a general input $\ket{\phi}\otimes\ket{0}^m$, applying the above and using linearity completes the proof.
\end{proof}

We will use this fact to prove a version of the robust detectability lemma for the commuting case later on in the section. This turns out to be simpler than the full detectability lemma, so we include the proof in~\Cref{lem:soundness-of-rho}.

\subsection{Soundness of Serial Repetition}\label{sec:serial-repetition}
For each vertex \(v\in X(1)\), let
\[
S_v:=\{(i,t)\in [N]\times [\ell] : v\in D_i\}.
\]
The ideal logical vertex check is the OR of all measurement wires whose clique contains \(v\).

\begin{definition}[Logical vertex OR-check]\label{def:logical-vertex-check}
Define the projector \(\pi_v\) on the logical measurement wires of \(C\) by
\[
\pi_v
:=
\sum_{b\in \{0,1\}^{S_v}} \mathrm{OR}_{S_v}(b)\,\ket{b}\!\bra{b},
\]
with the identity on all remaining output registers of \(C\). Equivalently, \(\pi_v\) measures the bits \(\{M_{i,t} : (i,t)\in S_v\}\) in the computational basis and outputs their OR.
\end{definition}

Using the \(\pi_v\)'s defined above, one can show that for all valid mixed states \(\rho\) on the input of \(C\), if $C$ measures a local Hamiltonian with energy $\mu$ then,
\begin{equation}\label{eq:serial-rep}
\E_v\left[\Tr[\pi_v C\rho C^\dagger]\right] \gtrsim \min(\mu\pl n, c_0),
\end{equation}
for some universal constant $c_0$. The idea is that for a local Hamiltonian with energy \(\mu\), then on every input $\rho$, running the circuit and then measuring all the outcomes of $C^\dagger \rho C$, yields a string of bits $\Vec{b}$ that contains at least a \(\mu\)-fraction of ones. This is easy to show this for a commuting local Hamiltonian, but for a non-commuting one this follows from a stronger version of the detectability lemma proved in~\Cref{sec:DL-new}. Now fix such a string $\Vec{b}$ with $\mu$-fraction of $1$s. Given a fixed $\Vec{b}$, each trace with respect to $\pi_v$ simplifies to just an OR taken over a $\pl n$-sized subset of $\Vec{b}$. Now observe that, if one had taken all possible $\pl n$-sized ORs, then at least a $\mu\pl n$-fraction of them would turn out to be $1$. The set of ORs given by $\{\pi_v\}$, namely the family $\{S_v\}$, does not include every possible OR, instead it forms a ``pseudorandom'' family. Indeed, on every such $\Vec{b}$, using the concentration properties of an HDX, one can show that at least $\mu\pl n$-fraction of these ORs would output $1$. It is straightforward to make this argument formal, but instead of doing so, we move on to proving a harder version of~\eqref{eq:serial-rep} in~\Cref{lem:serial-rep}. This is the version we later reduce to in our fault-tolerance analysis of~\Cref{algo:overall-protocol}. Fortunately, the proof of~\Cref{lem:serial-rep} follows most of the intuition given here.

We start by defining modified versions of the $\pi_v$ operators, where we allow ourselves to ignore a small family of bad measurements that might arise due to corruptions in the protocol.
Concretely, fix an edge-corruption pattern \(\cE\subseteq X(2)\) that the distributed computing protocol is run with. Let $\cE$ have measure $\eps$ that is smaller than some universal constant $\eps_0$ to be chosen later. Also let us fix the quantum code $\cQ$ used for MPC throughout, and let $\delta$ be its distance parameter (\Cref{assumption:nice-qLTC}), with $\eps_0 <\delta/10$. We can thus define a set of bad cliques and bad vertices with respect to $\cE$.

\begin{definition}[Bad cliques and bad vertices]\label{def:bad-cliques}
For each clique \(D_i\), we call it \emph{\(\cE\)-bad} if it contains at least $\eps_0$-fraction of corrupted edges 
and \emph{\(\cE\)-good} otherwise; since $\cE$ is fixed throughout the section, we will henceforth drop it from the notation. Let
$\cB$ denote the set of bad cliques.

For each clique \(D_i\), let $\cE_i$ denote the set of corrupted edges inside it and define the set of bad vertices inside it as,
\[
W_i
:=
\{v\in D_i:\deg_{\cE_i}(v)>\delta d/2\},
\]
that is, the set of vertices in $D_i$ that are incident to too many corrupted edges in $D_i$. We say that \(v\) is good in $D_i$ if \(v\notin W_i\). 
\end{definition}

We can now define a modified version of the $\pi_v$ operators above; these ignore certain bad wires while taking the OR. 
\begin{definition}[Modified logical vertex OR-check]\label{def:logical-mod-measurement}
Let $J$ be any set of cliques. For \(v\in X(1)\), define
\[
S'_v(J)
:=
\{(i,t)\in [N]\times [\ell] : v\in D_i,\ i\notin J,\ v\notin W_i\}.
\]
The corresponding modified logical operator is
\[
\pi'_v(J)
:=
\sum_{b\in \{0,1\}^{S'_v(J)}} \mathrm{OR}_{S'_v(J)}(b)\,\ket{b}\!\bra{b},
\]
with identity on all other output registers of \(C\). When the set $J$ is clear from context, we simply write \(\pi'_v\).
\end{definition}

Given these operators, we now prove that the stronger version of~\eqref{eq:serial-rep} holds, with $\pi_v$ replaced by $\pi'_v$. 

\begin{lemma}\label{lem:serial-rep}
Let $c \geq 1$ be any constant and let \(X\) be a \(d\)-partite \(\lambda\)-HDX on $n$ vertices with $\lambda=\Theta(1/\log^c n)$. Let \(C\) be a sequential measurement circuit for a commuting local Hamiltonian\footnote{Even though we state the lemma here for commuting local Hamiltonians, the proof works almost as is for non-commuting ones too, except one needs to use the detectability lemma instead of~\Cref{lem:clh-circuit}.} with energy at least \(\mu\). Fix a set of corrupted edges $\cE$ of measure $\eps$, and the set of bad cliques $\cB$ \& bad vertices with respect to it with parameters $\eps_0 < \delta/2$ (see~\Cref{def:bad-cliques}). Let \(J\subseteq [N]\) be any set of cliques with \(\cB \subseteq J\) and \(|J|\le \eta N\), where \(\eta<\mu\). Then, for every valid input mixed state \(\rho\) to the circuit \(C\), we get bounds on the soundness of $\rho$ in two different regimes,
\begin{enumerate}
\item When $\mu-\eta \leq \lambda^2$, then we set $d$ to be a large enough constant and get that $
\E_v\left[\Tr[\pi'_v C\rho C^{\dagger}]\right] \geq \frac{\mu-\eta}{\lambda}$.
\item When $\mu-\eta \geq \lambda^2$, we set $d$ to be $\Theta(\mathrm{loglog}^2 n)$ to get that $
\E_v\left[\Tr[\pi'_v C\rho C^{\dagger}]\right] \geq 1/100$.
\end{enumerate}
\end{lemma}

\begin{proof}
We start by reducing the problem to a classical serial repetition problem as alluded to for proving~\eqref{eq:serial-rep}. 

By writing \(\rho\) as a mixture over pure states of the form \(\ket{\xi}\otimes\ket{0^m}\), it suffices to show that the expected trace is large for each such pure state. Further let us expand this out using~\Cref{lem:clh-circuit},
\[
C\ket{\xi}\otimes\ket{0^m}=\sum_{w\in \mathcal V}c_w\ket{w}\otimes\ket{b_w},
\]
where \(\mathcal V\) is the common eigenspace basis of \(H\) and we know that each vector \(\ket{b_w}\) has at least \(\mu\)-fraction of \(1\)'s in it. Thus the expected trace simplifies to,
\[
\E_v\left[\Tr[\pi'_v C \ketbra{\xi}{\xi}\otimes \ketbra{0^m}{0^m} C^{\dagger}]\right]
=
\E_v\!\left[\E_{w\sim \mathcal V}\!\left[\Tr[\pi'_v \ketbra{b_w}{b_w}]\right]\right],
\]
where \(w\) is sampled from \(\mathcal V\) with probability \(|c_w|^2\). Thus it suffices to show that for a fixed $\ket{b_w}$, \(\E_v\Tr[\pi'_v\ketbra{b_w}{b_w}]\) is large. Now note that $b_w$ needs to have $\geq \mu$-fractional Hamming weight in at least one layer; call this $N$-length string $b$. Let $\cS$ be the set of cliques that $\supp(b)\setminus J$. In this layer, each $\pi'_v$ only takes an OR over cliques that are outside of $J$, which forms an $\leq \eta$-fraction of all cliques, and where it is good. So,
\[\E_v\Tr[\pi'_v\ketbra{b_w}{b_w}] \geq \E_v[\bigvee_{i\in [N]\setminus J: D_i\ni v} b_i] = \E_v[\exists i \in \cS: v\in D_i].\]

Hence we have reduced the problem to the following: We have a set \(\cS\) of \((\mu-\eta)\)-fraction of cliques. We know by assumption that each clique in $\cS$ is good, since it excludes cliques from $J \supset \cB$. We want to conclude that a large fraction of vertices are a part of at least one clique from \(\cS\) where they are good with respect to it too.

\paragraph{Proof of the sampling claim:}   For each good clique \(C\), by Markov's inequality it has at most $2\eps_0/\delta$-fraction of bad vertices, which is at most $1/2$ by assumption.
Our proof goes via defining the following set of vertices,
\[
T = \{v: \exists C \in \cS, \text{ such that } v \text{ is a good vertex inside } C\}.
\]
Now since each clique in \(\cS\) contains at least a $1/2$-fraction of good vertices and by the definition of \(T\), they all belong to \(T\). So using~\Cref{lem:concentration-improved} we get that,
\[
(\mu-\eta)\leq \Pr_{C \sim X(d)}[\mu(T|_C) \geq 1/2] \lll \mu(T)(\mu(T)+\lambda),
\]
thereby implying that \(\mu(T) \gtrsim \min(\sqrt{\mu-\eta},\frac{\mu-\eta}{\lambda})\).
When \((\mu-\eta) \leq \lambda^2\) this expression becomes \(\frac{\mu-\eta}{\lambda}\), as required. 

Otherwise, assume that $\mu-\eta \geq \lambda^2$, and that \(\mu(T)\leq 1/100\) for the sake of contradiction. Then we can apply the Chernoff bound in~\Cref{lem:chernoff-hdx} instead to get that,
\[
\mu-\eta \leq \Pr_{C \sim X(d)}[\mu(T|_C) \geq 1/2] \leq \exp(-\sqrt{d}) \leq 1/\log^{2c} n,
\]
which is a contradiction to the assumption that $\mu-\eta \geq \lambda^2$. Therefore \(\mu(T)\) must be at least \(1/100\), and in either case the claimed lower bound follows.
\end{proof}

\subsection{Fault Tolerance of the Protocol under Adversarial Edge-Corruptions}\label{subsec:FT-commuting-protocol}
In this section we analyze the distributed protocol \(\cP\) under adversarial edge-Pauli faults, using the partial fault-tolerance notions from \Cref{sec:prelims-partial-ft}. Recall that \(m=N\ell\) is the number of measurement wires of \(C\). We write
\[
\Enc_{\mathrm{out}}
\;:=\;
\encq^{\otimes Nk}\otimes \encc^{\otimes m}
\]
for the final output encoding map, namely \(N\) data blocks in the MPC code \(\mathcal Q\) and \(m\) measurement blocks in the repetition code. A faulty execution of \(\cP\) is specified by a Pauli fault pattern \(\mathsf F\) with \(\operatorname{supp}_{\cP}(\mathsf F)\subseteq \cE\), in the sense of \Cref{def:edge-pauli-fault}. The protocol \(\cP\) is built by composing two kinds of gadgets:
\begin{enumerate}
    \item For each layer \(t\in[\ell]\), a parallel family of \(N\) $k$-local gate gadgets, one on each clique \(D_i\), implemented using \Cref{thm:distributed-FTQC} and then converted to a repetition block on the measurement output using \Cref{lem:rep-code};
    \item Between consecutive layers, \(k\) routing gadgets implementing the \(k\) permutations induced by the circuit wiring, covered by \Cref{prop:ft-routing}.
\end{enumerate}

Fix \(\cE\subseteq X(2)\) with \(\mu(\cE)\le \eps_0\), and let \(\cB \subseteq [N]\) be the set of bad cliques, namely those cliques whose local corrupted-edge density exceeds \(\eps_0\). Let us now choose the parameter $\eps_0$ as: 
\[
\eps_0:=\min\!\left\{\eps^*_{\mathrm{comp}},\,\eps^*_{\mathrm{out}},\,\frac{\delta}{10}\right\},
\qquad
\tau:=\min\!\left\{c_{\mathrm{comp}}\sqrt{\eps_0},\,c_{\mathrm{route}}\delta,\,\beta^*_{\mathrm{out}}(\delta)\right\}.
\]
Here \(\eps^*_{\mathrm{comp}},c_{\mathrm{comp}}\) come from \Cref{thm:distributed-FTQC}, \(\eps^*_{\mathrm{out}},\beta^*_{\mathrm{out}}(\delta)\) come from \Cref{lem:rep-code}, and \(c_{\mathrm{route}}\) comes from \Cref{prop:ft-routing}.

Let $\eta_{\mathrm{bad}}$ be the measure of the bad cliques $\cB$. We can bound this by using the same HDX concentration bounds from  earlier, \(\eta_{\mathrm{bad}}\le \mu(\cE)/\log^{c+1} n\). For each routing stage, let \(\beta_{\mathrm{route}}=O(\mu(\cE)/\log^c n)\) denote the additive loss parameter coming from \Cref{prop:ft-routing}; we use the standing assumption that the coefficient multiplying \(\alpha\) in the routing statement is \(1\). 

We now define exceptional clique sets \(L_t\subseteq [N]\) recursively. Set
$L_1:=\cB.$
Given \(L_t\), let \(R_t\) be the bad destination set produced by the layer-\(t\) routing stage, and set
\[
L_{t+1}:=\cB_{\cE}\cup R_t.
\]
Writing \(\alpha_t\) as the measure of $L_t$, this gives
\[
\alpha_1=\eta_{\mathrm{bad}}
\qquad\text{and}\qquad
\alpha_{t+1}\le \alpha_t+\eta_{\mathrm{bad}}+\beta_{\mathrm{route}}.
\]

Finally, let \(\LC(L_1,\ldots,L_\ell)\) denote the forward light cone in \(C\) of the exceptional gates
\[
\{G_{i,t}: i\in L_t,\ t\in[\ell]\}.
\]
Informally, the proof will show that the protocol is fault-tolerant outside this light cone. The only additional point is the first computation layer: there we use the friendliness of the local MPC gadget to replace the arbitrary initial deviation by some valid encoded witness \(\ket{\widetilde{\phi}}\), and all subsequent layers are analyzed relative to the honest execution of \(C\) on \(\ket{\widetilde{\phi}}\).

\begin{prop}[Fault tolerance of \(\cP\)]\label{lem:ft}
Fix \(\cE\subseteq X(2)\) with \(\mu(\cE)\le \eps\), and let \(L_1,\ldots,L_\ell\) be the exceptional clique sets defined above for the given fault pattern \(\mathsf F\). Then the following holds.

For every pure input state \(\ket{\Phi}\) to the protocol, and every Pauli fault pattern \(\mathsf F\) with \(\supp_{\cP}(\mathsf F)\subseteq \cE\), there exists a state \(\ket{\widetilde{\phi}}\) such that, the final state \(\cP_{\mathsf F}(\ket{\Phi})\) is a $\tau$-deviation of
$\Enc_{\mathrm{out}}\bigl(C(\ket{\widetilde{\phi}}\otimes \ket{0^m})\bigr)$ outside all output blocks lying in \(\LC(L_1,\ldots,L_\ell)\):
\[\cP_\sF\ket{\Phi}= \sum_r a_r\,E_r\,\Enc_{\mathrm{out}}\bigl(C(\ket{\widetilde{\phi}}\otimes \ket{0^m})\bigr)\otimes \ket{\mathrm{rest}_r},\]
where we have the following support guarantees for each error $E_r$,
\begin{enumerate}
    \item The restriction of $E_r$ to every data block outside \(\LC(L_1,\ldots,L_\ell)\) has weight at most $\tau d$.
    \item If \(M_{i,t}\) is a measurement wire outside \(\LC(L_1,\ldots,L_\ell)\), then the restriction of $E_r$ on the corresponding repetition block is contained in the bad vertex set \(W_i(\cE)\) of $D_i$. In particular, every \(v\in D_i\setminus W_i(\cE)\) stores the correct measurement bit.
\end{enumerate}
\end{prop}

\begin{proof}
The proof mirrors the protocol stage by stage. The formal content is that the computation stages are handled by parallel composition of FT gadgets, while the routing stages are handled by sequential composition of partially FT gadgets. The only place where friendliness enters is the first computation stage.

For \(t\in[\ell]\), let \(\Gamma_t\) denote the ideal tensor product of all layer-\(t\) local gates together with the repetition-code conversion on their measurement outputs, and for \(t\in[\ell-1]\), let \(\Pi_t\) denote the ideal inter-layer wiring map between layers \(t\) and \(t+1\). Also, for \(t\ge 2\), let \(\ket{\psi_t^{\mathrm{in}}(\widetilde{\phi})}\) be the logical state of the \(kN\) input wires entering layer \(t\) in the honest execution of \(C\) on \(\ket{\widetilde{\phi}}\).

We will prove two assertions:

\smallskip
\noindent\textbf{A\(_t\).} Just before computation layer \(t\) (for \(t\ge 2\)), the collection of active data blocks is an \((\alpha_t,\tau)\)-deviation of \(\Enc_{\mathcal Q}^{\otimes kN}(\ket{\psi_t^{\mathrm{in}}(\widetilde{\phi})})\), with all bad blocks contained in \(L_t\).

\smallskip
\noindent\textbf{B\(_t\).} Just after computation layer \(t\), every data output block of every gate \(G_{i,t}\) with \(i\notin L_t\) is a \(\tau\)-deviation of the corresponding honest \(\mathcal Q\)-encoded output wire of \(C\) on \(\ket{\widetilde{\phi}}\), and the measurement block \(M_{i,t}\) has Pauli support contained in \(W_i(\cE)\).

\smallskip
Once these assertions are established, the proposition follows immediately: an output wire lies outside \(\LC(L_1,\ldots,L_\ell)\) exactly when its backward light cone avoids all exceptional gates \(G_{i,t}\) with \(i\in L_t\), and for such a wire the corresponding FT guarantee from the relevant assertion \(B_t\) applies. Absorbing all output blocks in \(\LC(L_1,\ldots,L_\ell)\), together with the work registers, into the auxiliary register then gives exactly the claimed deviation statement.

We now prove the assertions.

\smallskip
\noindent\emph{Base step: the first computation layer.}
Consider the family of gadgets \(\{\mathcal M_{i,1}: i\notin L_1\}\). Since \(L_1=\cB_{\cE}\), every clique in this family is good, and the corresponding layer-1 gates are disjoint. As observed above, each \(\mathcal M_{i,1}\) is friendly and FT, so by the same tensor-product argument underlying parallel composability, the combined first-layer gadget
\[
\mathcal M_1^{\mathrm{good}}:=\bigotimes_{i\notin L_1}\mathcal M_{i,1}
\]
is again friendly. Applying this to the arbitrary deviation \(\ket{\Phi}\) of \(\Enc_{\mathrm{in}}(\ket{\phi})\), friendliness yields a state \(\ket{\widetilde{\phi}}\) such that the outputs on all cliques \(i\notin L_1\) are exactly as if \(\mathcal M_1^{\mathrm{good}}\) had been applied to a valid \(\tau\)-deviation of the encoded layer-1 inputs for \(\ket{\widetilde{\phi}}\). Since \(\mathcal M_1^{\mathrm{good}}\) is FT, assertion \(B_1\) follows: every data output block of every gate \(G_{i,1}\) with \(i\notin L_1\) is \(\tau\)-deviated from the correct \(\mathcal Q\)-encoded output block, and the measurement block \(M_{i,1}\) has support contained in \(W_i(\cE)\).

If \(\ell=1\), we are done. Otherwise we now apply the first routing stage \(\mathcal R_1\). By \(B_1\), the collection of routed data blocks is an \((\alpha_1,\tau)\)-deviation of the correct encoded input to layer \(2\), with bad blocks contained in \(L_1\). Since \(\mathcal R_1\) is partially FT, there exists a destination bad set \(R_1\) with \(|R_1|/N\le \alpha_1+\beta_{\mathrm{route}}\) such that all routed blocks outside \(R_1\) are \(\tau\)-deviations of the correct inputs to layer \(2\). Setting \(L_2:=\cB_{\cE}\cup R_1\), we obtain assertion \(A_2\).

\smallskip
\noindent\emph{Inductive step.}
Fix \(t\in\{2,\ldots,\ell\}\), and assume \(A_t\). We first prove \(B_t\), and then, if \(t<\ell\), we deduce \(A_{t+1}\).

To prove \(B_t\), consider the family \(\{\mathcal M_{i,t}: i\notin L_t\}\). For every such \(i\), all \(k\) input blocks to the gate \(G_{i,t}\) are good by \(A_t\), since the bad blocks are contained in \(L_t\). Moreover, \(i\notin L_t\) implies in particular \(i\notin \cB_{\cE}\), so the local clique gadget satisfies the hypotheses of \Cref{thm:distributed-FTQC}. Therefore each \(\mathcal M_{i,t}\) may be applied to a \(\tau\)-deviation of the correct encoded input, and by its FT guarantee it outputs the correct layer-\(t\) data wires up to \(\tau\)-deviation, together with a repetition block whose Pauli support is contained in \(W_i(\cE)\). Since these gadgets act on disjoint registers, parallel composability gives the same conclusion simultaneously for all \(i\notin L_t\). This proves \(B_t\).

Now assume \(t<\ell\). By \(B_t\), the collection of data wires leaving layer \(t\) is an \((\alpha_t,\tau)\)-deviation of the correct encoded input to the routing stage \(\Pi_t\), with bad blocks contained in \(L_t\). Apply the partially FT routing gadget \(\mathcal R_t\). Since the coefficient multiplying \(\alpha\) is \(1\), the output is an \((\alpha_t+\beta_{\mathrm{route}},\tau)\)-deviation of the correctly routed encoded state, with bad destination set \(R_t\) satisfying \(|R_t|/N\le \alpha_t+\beta_{\mathrm{route}}\). We then define \(L_{t+1}:=\cB_{\cE}\cup R_t\). This yields assertion \(A_{t+1}\), and at the same time gives the recursive bound
\[
\alpha_{t+1}=\frac{|L_{t+1}|}{N}\le \alpha_t+\eta_{\mathrm{bad}}+\beta_{\mathrm{route}}.
\]

This completes the induction.

\smallskip
We now conclude. Assertion \(B_\ell\) tells us that every final data block outside the exceptional layer-\(\ell\) cliques is a \(\tau\)-deviation of the corresponding honest encoded output wire of \(C\) on \(\ket{\widetilde{\phi}}\), and every final measurement block \(M_{i,\ell}\) with \(i\notin L_\ell\) has support contained in \(W_i(\cE)\). More generally, by tracing the induction through the circuit, the same is true for every output wire whose backward light cone avoids the exceptional gates \(\{G_{i,t}: i\in L_t\}\), i.e. for every output wire outside \(\LC(L_1,\ldots,L_\ell)\). Finally, we absorb all work registers and all output blocks in \(\LC(L_1,\ldots,L_\ell)\) into the auxiliary register in the sense of \Cref{eq:deviation-decomp}. What remains is exactly a deviation of
\(\Enc_{\mathrm{out}}(C(\ket{\widetilde{\phi}}\otimes \ket{0^m}))\)
with the two stated support guarantees.
\end{proof}

The proposition above is the only fault tolerance guarantee needed later: it isolates a small family of exceptional gates \(L_t\), and outside those gates it says that the protocol output differs from the ideal encoded output only by correctable errors on data blocks and by errors on the measurement blocks that are confined to the bad vertices of the corresponding clique.

\subsection{Moving from Physical Measurements to Logical Measurements}\label{subsec:physical-to-logical-commuting-protocol}
In this section we will consider a pure state of the form, 
\[
\ket{\Psi} = \sum_{\sF\in \mathsf{Pauli}^{\cE}}c_{\sF} \cP_\sF \ket{\phi_\sF},
\]
where each error pattern $\sF$ only touches the edges in $\cE$. Note that one summand here is just a state that is obtained by running the faulty protocol $\cP$ with Pauli-edge faults $\sF$ contained in $\cE$ and this is something one would analyze in standard multiparty computation. It turns out though that we need to analyze this more complicated state where the Pauli-faults on $\cE$, as well as the input witness, can vary across the summands, and the only guarantee we have is that the state $\ket{\Psi}$ has norm $1$. As we will show in~\Cref{sec:ckt-to-H} this is the noise model that arises out of the circuit to Hamiltonian reduction, thus we call it the tensor-network state. Proving the soundness of such a state (as opposed to $\cP_\sF\ket{\phi}$) turns out to be quite non-trivial. Our high-level strategy is to reduce physical measurements on this state to logical measurements on a certain mixed state that is obtained from $\ket{\Psi}$.

Thus, in this section we will construct a mixed state $\rho'$ in the logical space and show that all ``valid'' measurements on $\ket{\Psi}$, those of the form $\REnc(\pi)$ (with a few additional caveats), correspond to performing the measurement $\pi$ on $C\rho' C^\dagger$. This state $\rho$ will have a nice structure: it is composed of pure states of the form $\ket{\xi}\otimes\ket{0^m}$, exactly the set of valid inputs to $C$. Once we have this, we can hope to analyze the soundness of $\ket{\Psi}$, using the soundness of $\rho'$. The latter can then be analyzed using serial repetition, \Cref{lem:serial-rep}.

We first need a simple claim about quantum codes which will allow us to relate measurements on the encoded protocol output to measurements on the logical output of \(C\).

\begin{definition}[Encoded observables]\label{def:section7-renc}
Recall~\Cref{def:reversible-decoding}. Let \(U_\mathcal{Q}\) be the reversible decoder for the MPC code and let \(U_{\mathsf{Rep}}\) be the standard reversible majority decoder for the repetition code. Define
\[
U:=U_\mathcal{Q}^{\otimes kN}\otimes U_{\mathsf{Rep}}^{\otimes m}.
\]
For every projector \(\pi\) on the logical output space of \(C\), define the corresponding physical observable by
\[
\REnc(\pi):=U^\dagger(\pi\otimes \Id)U.
\]
\end{definition}

For later use, define
\[
K_{\mathrm{meas}}:=\{M_{i,t}: i\in L_t,\ t\in[\ell]\},
\]
namely the exceptional measurement wires, and let \(K_{\mathrm{data}}\) be the set of all non-measurement wires lying in the forward light-cone of the exceptional gate set \(\{G_{i,t}: i\in L_t\}\).

\begin{claim}\label{claim:KL}
Let
\[
T
=
E\,\Enc_{\mathrm{out}}
\bigl(C(\ketbra{\psi}{\psi'}\otimes \ketbra{0^m}{0^m})C^\dagger\bigr)E'
\]
where \(E,E'\) satisfy the support conditions from \Cref{lem:ft}. Then for every projector \(\pi\) on the logical output space there exist a coefficient \(\beta_{E,E'}\) and Pauli operators \(P,P'\) on the logical output space such that
\begin{equation}\label{eq:KL-main}
\Tr(\REnc(\pi)T)
=
\beta_{E,E'}\Tr\!\left(\pi\,P\,C(\ketbra{\psi}{\psi'}\otimes \ketbra{0^m}{0^m})C^\dagger P'\right),
\end{equation}
where the support of \(P,P'\) on the data wires is contained in \(K_{\mathrm{data}}\), and on the measurement wires their \(X\)-support is contained in \(K_{\mathrm{meas}}\) (while their \(Z\)-support may be arbitrary).
\end{claim}

\begin{proof}
Write the RHS in~\eqref{eq:KL-main}, $\Tr(U^\dagger \pi U T)$. This equals $\Tr(\pi U T U^\dagger)$, so let us consider the ket-side of $UTU^\dagger$: $UE\Enc_{\mathrm{out}}(C\ket{\psi}\otimes\ket{0^m})$. We will first simplify this using the decoding properties of $\cQ$ and the repetition code.

On every non-measurement output block outside the forward light-cone of the exceptional gates, the restriction of \(E\) has weight at most \(\tau d\), so by the reversible-decoding property of the code $\cQ$ it contributes only a syndrome register and no logical action. Therefore any nontrivial logical Pauli on the decoded data output can only come from the wires in \(K_{\mathrm{data}}\).

Now consider a measurement block \(M_{i,t}\) with \(i\notin L_t\). By \Cref{lem:ft}, the corresponding physical error is supported only on the coordinates in \(W_i(\cE)\). Since \(|W_i(\cE)|\le \tau_{\mathrm{rep}} d < d/2\), majority decoding returns the correct logical bit on that block, so no logical \(X\)-error can appear there. However, Pauli \(Z\)-operators supported on \(W_i(\cE)\) may still contribute a phase, hence after decoding they can induce a logical \(Z\)-operator on the measurement wire together with syndrome terms. Thus a logical \(X\)-error on a decoded measurement wire can occur only on the exceptional set \(K_{\mathrm{meas}}\), while $Z$-errors might be arbitrary.

Formally, using claims~\ref{claim:reversible-decoding-properties} and~\ref{claim:dec-rep} we get that there exist a syndrome state \(\ket{S_E}\), a coefficient \(\alpha_E\), and a Pauli operator \(P\) with the claimed support properties such that
\[
U E \Enc_{\mathrm{out}} C(\ket{\psi}\otimes \ket{0^m})
=
\alpha_E\,P\,C(\ket{\psi}\otimes \ket{0^m})\otimes \ket{S_E}.
\]
The same statement holds for the bra-side of $UTU^\dagger$, with corresponding data \(\alpha_{E'},P',\ket{S_{E'}}\). Therefore
\begin{align*}
\Tr(\REnc(\pi)T)
&=
\Tr\!\left((\pi\otimes \Id)U T U^\dagger\right) \\
&=
\alpha_E\overline{\alpha_{E'}}\,
\Tr\!\left(\pi\,P\,C(\ketbra{\psi}{\psi'}\otimes \ketbra{0^m}{0^m})C^\dagger P' \otimes \ketbra{S_E}{S_{E'}}\right) \\
&=
\beta_{E,E'}\Tr\!\left(\pi\,P\,C(\ketbra{\psi}{\psi'}\otimes \ketbra{0^m}{0^m})C^\dagger P'\right),
\end{align*}
where
\[
\beta_{E,E'}:=\alpha_E\overline{\alpha_{E'}}\braket{S_{E'}}{S_E}.
\]
This is exactly \eqref{eq:KL-main}.
\end{proof}

\begin{lemma}[Equivalence to a logical output state]
\label{lem:rho}
Consider the setting above, and let \(K_{\mathrm{data}},K_{\mathrm{meas}}\) be the wire sets defined just before \Cref{claim:KL}. For every valid pure state
\[
\ket{\Psi} = \sum_{\sF\in \cP} c_\sF \cP_\sF \ket{\phi_\sF},
\]
there exists a valid density matrix of the form
\begin{align}\label{eq:rho}
    \rho = \sum_{i} p_{i} \ketbra{v_i}{v_i},
\end{align}
with \(\{p_i\}\) a probability distribution and each \(\ket{v_i}\in \mathrm{span}(PC(\ket{\psi}\otimes\ket{0^m}))\), where \(P\) is a Pauli operator such that its support on the data wires lies in \(K_{\mathrm{data}}\), while on the measurement wires its \(X\)-support lies in \(K_{\mathrm{meas}}\) and its \(Z\)-support is arbitrary. This mixed state \(\rho\) satisfies, for every projector \(\pi\),
\[
\Tr(\REnc(\pi)\ketbra{\Psi}{\Psi})=\Tr(\pi \rho).
\]
\end{lemma}

\begin{proof}
Start by expanding \(\ketbra{\Psi}{\Psi}\) and tracing out the work ancillas:
\begin{align*}
\Tr_{\mathrm{ancillas}}(\ketbra{\Psi}{\Psi})
&=
\sum_{\sF,\sF'} c_\sF \overline{c_{\sF'}}\,
\Tr_{\mathrm{ancillas}}\!\left(\cP_\sF\ketbra{\phi_\sF}{\phi_{\sF'}}\cP_{\sF'}^\dagger\right).
\end{align*}
For each \(\sF\), apply \Cref{lem:ft} to the ket vector \(\cP_\sF\ket{\phi_\sF}\), and similarly for each \(\sF'\). Writing 
\[
\cP_\sF\ket{\phi_\sF}
=
\sum_i a_{\sF,i}\,E_{\sF,i}\,\Enc_{\mathrm{out}}\bigl(C(\ket{\widetilde{\phi_{\sF}}}\otimes \ket{0^m})\bigr)\otimes \ket{\mathrm{rest}_{\sF,i}},
\]
and the analogous expansion for \(\sF'\), we obtain
\begin{align}\label{eq:ft-two-sided}
\Tr_{\mathrm{ancillas}}(\ketbra{\Psi}{\Psi})
=
\sum_{\sF,\sF',i,j}
\gamma_{\sF,\sF',i,j}\,
E_{\sF,i}\,\Enc_{\mathrm{out}}
\bigl(C(\ketbra{\widetilde{\phi_{\sF}}}{\widetilde{\phi_{\sF'}}}\otimes \ketbra{0^m}{0^m})C^\dagger\bigr)E'_{\sF',j},
\end{align}
where
$\gamma_{\sF,\sF',i,j}
=
c_\sF\overline{c_{\sF'}}\,a_{\sF,i}\overline{a_{\sF',j}}\,
\braket{\mathrm{rest}_{\sF',j}}{\mathrm{rest}_{\sF,i}}$ and the support structure of the errors $E_{\sF,i}, E'_{\sF',j}$ are as specified in~\Cref{lem:ft}.

Let us bundle the indices \((\sF,\sF',i,j)\) into \(k\) for brevity of notation. Now apply \Cref{claim:KL} term-by-term to \eqref{eq:ft-two-sided}. This gives
\begin{align*}
\Tr(\REnc(\pi)\Tr_{\mathrm{ancillas}}(\ketbra{\Psi}{\Psi}))
&=
\sum_{k}\beta'_k \Tr\!\left(\pi P_k(C\ketbra{\psi_k}{\psi'_k}\otimes \ketbra{0^m}{0^m}C^\dagger)P'_k\right) \\
&=
\Tr\!\left(\pi \underbrace{\sum_k \beta'_k P_k(C\ketbra{\psi_k}{\psi'_k}\otimes \ketbra{0^m}{0^m}C^\dagger)P'_k}_{:=\rho}\right),
\end{align*}
where every \(P_k,P'_k\) satisfies the support conditions from \Cref{claim:KL}.

By construction, \(\Tr(\REnc(\pi)\ketbra{\Psi}{\Psi}) = \Tr(\pi\rho)\) for every projector \(\pi\). It remains to verify that \(\rho\) is a valid density matrix and has the claimed support structure.
Positivity follows because, for every projector \(\pi\),
\[
\Tr(\pi\rho)=\Tr(\REnc(\pi)\ketbra{\Psi}{\Psi})\ge 0.
\]
Also,
\[
\Tr(\rho)
=
\Tr(\REnc(\Id)\ketbra{\Psi}{\Psi})
=
\Tr(\ketbra{\Psi}{\Psi})
=
1,
\]
so \(\rho\) has trace \(1\). Hence \(\rho\) is a density matrix.

Finally, \(\rho\) is supported on the span of states of the form
\[
P\,C(\ket{\psi}\otimes \ket{0^m}),
\]
with \(P\) obeying exactly the support conditions stated in the lemma. Therefore \(\rho\) can be diagonalized inside that span, yielding a decomposition of the form \eqref{eq:rho}.
\end{proof}

Next, we further reduce measurements on the density matrix $\rho$ above to those of a valid execution of the logical circuit \(C\) on some valid input density matrix. Given a set $S$ of input wires of $C$ let $\FL(S)$ denote the set of output wires in the forward lightcone of $S$. Similarly, given a set of output wires $S$, let $\BL(S)$ denote the set of input wires in the backward lightcone of $S$. The following lemma crucially uses that $C$ is a circuit that measures a commuting local Hamiltonian.

\begin{lemma}[Equivalence to a logical input state]\label{lem:soundness-of-rho}
Let
\[
K':=\FL(\BL(K_{\mathrm{data}}\cup K_{\mathrm{meas}})),
\]
and let \(\rho\) be the state from \Cref{lem:rho}. Then there exists a mixed state \(\rho'\), each of whose pure states is of the form \(\ket{\xi}\otimes\ket{0^m}\), such that for every projector \(\pi\) acting on the measurement qubits contained in \(\overline{K'}\), 
\[
\Tr[(\pi\otimes \Id)\rho]=\Tr[(\pi\otimes \Id)C\rho' C^\dagger].
\]
\end{lemma}

\begin{proof}
Write
\[
K:=\BL(K_{\mathrm{data}}\cup K_{\mathrm{meas}}),
\qquad\text{so that}\qquad
K'=\FL(K).
\]

By \Cref{lem:rho}, \(\rho\) is a mixture of pure states each of which lies in the span of the set  \(\cS =\{P\,C(\ket{\phi}\otimes\ket{0^m}\}\), where on the data register \(P\) is supported on \(K_{\mathrm{data}}\), on the measurement register its \(X\)-support is contained in \(K_{\mathrm{meas}}\), and its \(Z\)-part is arbitrary. 

Let us start by simplifying the states in $\cS$. Consider any state $\ket{w} = PC\ket{\phi}\otimes\ket{0}^m$ in $\cS$.
Write \(P=P_1\otimes X^J Z^I\), where \(\supp(P_1)\subseteq K_{\mathrm{data}}\) and \(J\subseteq K_{\mathrm{meas}}\). Using \Cref{lem:clh-circuit}, if
\[
C(\ket{\phi}\otimes\ket{0^m})=\sum_{v\in\mathcal V} c_v \ket{v}\otimes\ket{b_v},
\]
then each \(\ket{b_v}\) is a computational basis state, so \(Z^I\ket{b_v}=\lambda_v\ket{b_v}\) for some \(\lambda_v\in\{\pm1\}\). Hence
\[
Z^I C(\ket{\phi}\otimes\ket{0^m})=C(\ket{\phi'}\otimes\ket{0^m})
\]
for
\[
\ket{\phi'}:=\sum_v \lambda_v c_v\ket{v}.
\]
Therefore
\[
\ket{w}
=
P_1\otimes X^J\, C(\ket{\phi'}\otimes\ket{0^m})
=
C\underbrace{\bigl(C^\dagger(P_1\otimes X^J)C\bigr)}_{:=E}(\ket{\phi'}\otimes\ket{0^m}).
\]
It is easy to check that the support of the operator $E$ is contained within \(K\). Thus,
\[\cS = \{C E\ket{\phi'}\otimes\ket{0^m}: \supp(E)\subseteq K\}.\]

Now set \(\sigma:=C^\dagger\rho C\). This mixture is composed of pure states in the span of, 
\[C^\dagger\cS = \{E(\ket{\phi'}\otimes\ket{0^m}): \supp(E)\subseteq K\} = \left\{\ket{\phi''}_{K\cup D}\otimes\ket{0}_{A\setminus K}\right\},\] 
where \(A\) and $D$ denote the ancilla and the witness-input register respectively. 

Next define \(\omega:=\Tr_A(\sigma)\). Since every pure state in the support of \(\sigma\) lies in $C^\dagger \cS$, all its ancillas outside of $K$ remain fixed to $0$. Thus, if we set
\[
\rho':=\omega\otimes\ketbra{0^m}{0^m},
\]
then each pure state in a mixture for \(\rho'\) is of the form \(\ket{\xi}\otimes\ket{0^m}\), and moreover one can check that $\Tr_K(\rho')=\Tr_K(\sigma)$. 

Now since $\rho'$ and $\sigma$ have equal marginals outside $K$, the circuit $C$ run on them has equal marginals outside the forward lightcone of $K$. Thus one can check that,
\begin{equation}
\Tr_{K'}(C\rho' C^\dagger) =\Tr_{K'}(C\sigma C^\dagger) = \Tr_{K'}(\rho),    
\end{equation}
where the first equality follows since \(K'=\FL(K)\) and second follows by the definition of $\sigma$.

Finally, let \(\pi\) be a projector on the output qubits in \(\overline{K'}\). Then,
\[
\Tr((\pi\otimes\Id)\rho)
=
\Tr(\pi\,\Tr_{K'}(\rho))
=
\Tr(\pi\,\Tr_{K'}(C\rho' C^\dagger)))
=
\Tr((\pi\otimes\Id)C\rho' C^\dagger),
\]
as required.
\end{proof}

\subsection{Soundness of the Protocol for the Tensor-Network State}\label{subsec:soundness-commuting-protocol}
Recall the picture: for each layer \(t\in[\ell]\) and each clique \(D_i\), the protocol implements the gate \(G_{i,t}\) on clique \(D_i\), and the measurement output \(M_{i,t}\) is written as a length-\(d\) repetition block on the vertices of \(D_i\). Hence each vertex \(v\in X(1)\) stores one physical measurement bit for every pair \((i,t)\) with \(v\in D_i\). Thus we can define the measurement operators $\Pi_v$ in the physical (protocol) space analogous to the operators $\pi_v$ we defined earlier.

\begin{definition}[Vertex OR-checks in the Protocol Space]\label{def:vertex-check-operator}
Fix \(v\in X(1)\) and let
\[
S_v:=\{(i,t)\in [N]\times [\ell]: v\in D_i\}.
\]
Let \(\mathcal R_v\) be the tensor product of the physical measurement-output qubits stored at \(v\), one qubit for each \((i,t)\in S_v\). We define \(\Pi_v\) to be the projector acting on \(\mathcal R_v\) by
\[
\Pi_v
:=
\sum_{b\in \{0,1\}^{S_v}} \mathrm{OR}_{S_v}(b)\,\ket{b}\!\bra{b},
\]
and as the identity on all other registers.
Equivalently, on a computational-basis state of the local measurement registers at \(v\), the operator \(\Pi_v\) outputs the OR of those bits.
\end{definition}

Similarly, we can also define the ``modified'' measurement operators $\Pi'_v$ analogous to the logical operators $\pi'_v$ defined earlier.

\begin{definition}[Modified OR-checks in the Protocol Space]\label{def:mod-measurement}
Fix the edge-corruption pattern \(\cE\). For each layer \(t\in[\ell]\), let \(L_t\subseteq [N]\) be the exceptional set of gates supplied by \Cref{lem:ft}, and for each clique \(D_i\) let \(W_i(\cE)\) be the bad-vertex set. For \(v\in X(1)\), define
\[
S'_v:=\{(i,t)\in [N]\times [\ell]: v\in D_i,\ (i,t) \notin K',\ \text{and \(v\) is good in } D_i\},
\]
where $K' = \FL(\BL(K_{meas}\cup K_{data}))$ as defined in~\Cref{lem:soundness-of-rho}.
Recall that the modified logical operator \(\pi'_v\) is the projector on the logical measurement registers indexed by \(S'_v\) given by
\[
\pi'_v
:=
\sum_{b\in \{0,1\}^{S'_v}} \mathrm{OR}_{S'_v}(b)\,\ket{b}\!\bra{b},
\]
with identity on all other output registers of \(C\).
The modified physical operator \(\Pi'_v\) is defined analogously on the local measurement registers stored at \(v\):
\[
\Pi'_v
:=
\sum_{b\in \{0,1\}^{S'_v}} \mathrm{OR}_{S'_v}(b)\,\ket{b}\!\bra{b},
\]
with identity on all other protocol registers.
\end{definition}

Let $\rho(\Psi):=\Tr_{\mathrm{ancillas}}(\ketbra{\Psi}{\Psi}).$ In the next lemma we will show that the projectors $\Pi'_v$ and $\REnc(\pi'_v)$ are equal on $\rho(\Psi)$. To do so, we will use the fault-tolerance property of the protocol, \Cref{lem:ft}, which says that the measurement bits seen by $v$ on the cliques in which it is good, agree with the correct outcome.

\begin{lemma}\label{lem:Pi-to-pi}
For every vertex \(v\in X(1)\),
\[
\Tr(\Pi'_v\,\rho(\Psi))=\Tr(\REnc(\pi'_v)\,\rho(\Psi)).
\]
\end{lemma}

\begin{proof}
Fix $v$, and abbreviate $S:=S'_v$. Consider one summand in the expansion \eqref{eq:ft-two-sided} of $\rho_{\Psi}$:
\[
T
=
E(\Enc_{\cQ}\otimes \Rep)\bigl(C\ketbra{\psi}{\psi'}\otimes \ketbra{0^m}{0^m}C^{\dagger}\bigr)E' := \ketbra{w}{w'}.
\]
We will show that by virtue of $E, E'$,
\[\Pi'_v\,\ket{w} = \REnc(\pi'_v)\,\ket{w}, \quad \Pi'_v\,\ket{w'} = \REnc(\pi'_v)\,\ket{w'}.\]
This then immediately gives the desired equality for $T$, thus implying the equality for $\rho(\Psi)$.

Write the output of $C$ by separating the logical measurement bits indexed by $S$:
\[
C(\ket{\psi}\otimes \ket{0^N})
=
\sum_{b\in \{0,1\}^{S}} \ket{\phi_b}\otimes \ket{b}_{S},
\]
where $\ket{\phi_b}$ contains all remaining output registers. After applying $(\Enc_{\cQ}\otimes \Rep)$ and then $E$, we obtain
\[
\ket{w}
=
\sum_{b\in \{0,1\}^{S}} E\bigl(\ket{\widetilde\phi_b}\otimes \encc^{\otimes |S|}\ket{b}_{S}\bigr),
\]
for suitable states $\ket{\widetilde\phi_b}$.

Now fix $(i,t)\in S$. By definition of $S$, the measurement output of the gate $(i,t)$ does not belong to $K'$, which is a superset of $L_t$, and the vertex $v$ is good in $D_i$. Hence the fault-tolerance statement implies two facts about the error $E$ restricted to the repetition block of $M_{i,t}$: first, there is no $X$-error on the coordinate stored at $v$; second, the total $X$-weight on that block is $< d/2$. Therefore, on the term indexed by $b$, the physical bit stored at $v$ in that block is exactly $b_{i,t}$, and majority decoding of that block also returns $b_{i,t}$.

Consequently, both operators $\Pi'_v$ and $\REnc(\pi'_v)$ act on the $b$-term by the same scalar, namely $\mathrm{OR}_S(b)$. Thus
\[
\Pi'_v\,\ket{w}
=
\sum_{b\in \{0,1\}^{S}} \mathrm{OR}_S(b)\, E\bigl(\ket{\widetilde\phi_b}\otimes \encc^{\otimes |S|}\ket{b}_{S}\bigr)
=
\REnc(\pi'_v)\,\ket{w}.
\]
The same proof implies the equality for $\ket{w'}$. This proves that $\Tr(\Pi'_v T)=\Tr(\REnc(\pi'_v)T)$ for each summand $T$ in \eqref{eq:ft-two-sided}. Summing over all summands gives
\[
\Tr(\Pi'_v\,\rho(\Psi))=\Tr(\REnc(\pi'_v)\,\rho(\Psi)),
\]
as claimed.
\end{proof}

We can now put everything together and thus finish the proof of~\Cref{thm:commuting-main} via the following lemma,
\begin{lemma}\label{lem:main7}
Consider the setting of parameters  in~\Cref{prop:ft-routing}. Suppose the measure of bad edges is at most \(\eps\), $\lambda \leq 1/\log^{c+1} n$ and assume that \(\eps \lll \mu\log^c n\). Then,
\[
\E_{v\in X(1)}\bigl[\Tr(\Pi_v\ketbra{\Psi}{\Psi})\bigr] \gtrsim \mu/\lambda.
\]
\end{lemma}

\begin{proof}
Since $\Pi_v$ acts only on the non-ancilla output registers it suffices to show that 
\[
\E_{v\in X(1)}[\Tr(\Pi_v\,\rho_\Psi)] \gtrsim \mu/\lambda.
\]

For each $v$, the set $S'_v$ is a subset of $S_v$. Since both $\Pi_v$ and $\Pi'_v$ are diagonal in the computational basis on the local measurement registers, and the OR over a larger set of bits dominates the OR over a subset, we have
$\Pi_v \succeq \Pi'_v$, therefore,
\begin{equation}\label{eq:main7-1}
\E_v\bigl[\Tr(\Pi_v\,\rho_{\Psi})\bigr]
\geq
\E_v\bigl[\Tr(\Pi'_v\,\rho_{\Psi})\bigr].
\end{equation}
By \Cref{lem:Pi-to-pi},
\begin{equation}\label{eq:main7-2}
\Tr(\Pi'_v\,\rho_{\Psi}) = \Tr(\REnc(\pi'_v)\,\rho_{\Psi})
\qquad\text{for every }v.
\end{equation}
Applying \Cref{lem:rho} and then \Cref{lem:soundness-of-rho}, we obtain a mixed state $\rho'$ whose pure states are of the form $\ket{\xi}\otimes \ket{0^m}$ and such that
\begin{equation}\label{eq:main7-3}
\Tr(\REnc(\pi'_v)\,\rho_{\Psi}) = \Tr(\pi'_v C\rho' C^{\dagger})
\qquad\text{for every }v.
\end{equation}

It remains to lower-bound the average of the logical modified checks. By construction, \(\pi'_v\) satisfies the hypothesis of~\Cref{lem:serial-rep} with the set of omitted gates there being $K'$.  Hence we get,
\[
\E_v\bigl[\Tr(\pi'_v C\rho' C^{\dagger})\bigr]
\geq
\frac{1}{\lambda}(\mu-\eta),
\]
where $\eta$ the maximum fraction of gates $K'$ touches in any layer of $C$. Since $K = \FL(\BL(\mathrm{LC}(L_1,\ldots, L_\ell)))$ and $C$ is a constant depth circuit, $\eta$ is at most $O_\ell(\eps/\log^c n)$ by the HDX concentration lemmas (applied in the three relevant regimes like in~\Cref{prop:ft-routing}). By the assumption in the lemma, this is at most $\mu/2$ giving us that,
\[
\E_v\bigl[\Tr(\pi'_v C\rho' C^{\dagger})\bigr]
\gtrsim
\mu/\lambda.
\]
Combining this with \eqref{eq:main7-1}, \eqref{eq:main7-2}, and \eqref{eq:main7-3} yields
\[
\E_v\bigl[\Tr(\Pi_v\,\ketbra{\Psi}{\Psi})\bigr] \gtrsim \mu/\lambda,
\]
as required.
\end{proof}

This allows us to complete the proof of~\Cref{thm:commuting-main},
\begin{proof}[Proof of~\Cref{thm:commuting-main}]
To conclude, apply~\Cref{lem:Pi-to-pi} with parameters \(\mu=\epsilon\) and  \(\mu(\mathcal E)\). In the low-energy regime \(\epsilon\leq\lambda^2\), the assumed corruption bound gives \(\mathbb E_v\operatorname{Tr}(\Pi_v|\Psi\rangle\langle\Psi|)\gtrsim\epsilon/\lambda\). Choosing \(\lambda\) to be  \(1/\pl n\) completes the proof.
\end{proof}

%% file: nc-protocol.tex
\section{Modified Distributed Computing Protocol for Non-Commuting LH}\label{sec:nc-protocol}

In \Cref{sec:protocol-for-circuit}, the analysis of the protocol $\cP$'s soundness against the `tensor-network' noise model relies crucially on the commuting Hamiltonian assumption. To recall, $\cP$ outputs the constraint measurements in the form of a repetition codeword on each clique, and we established that the vertex checks will be violated if the original Hamiltonian has large energy, even with adversarially corrupted execution of $\cP$ under the `tensor-network' noise model:
\begin{align}
    \sum_{\sF\in \mathsf{PAULI}^{\cE}} c_{\sF} \cP_\sF \ket{\phi_\sF},
\end{align}
where $\sF$ are Pauli errors supported on a corrupted edge set $\cE$.
While the fault tolerance of $\cP$ guarantees that most cliques will have small residual errors on their output codeword, the subtlety is the repetition code cannot tolerate Z errors. A single physical Z error leads to a logical Z error; thus, at the logical level of the sequential measurement circuit $C$, we may incur logical Z errors on all measurement qubits. Intuitively, this should not be an issue because we are measuring out these qubits to obtain a classical distribution. However, the challenge is that a coherent superposition of different Z error patterns on the measurement qubits may lead to a postselection gate: $(I+Z) \propto \ketbra{0}$. In particular, if the ideal logical distribution over the measurement bits have nonzero measure on the all-zero string, then the adversary can collapse the entire logical distribution to the all-zero string. This is a serious problem in the soundness case.

For commuting Hamiltonians, this challenge is resolved essentially by the observation that the distribution on the measurement bits of the sequential circuit is only supported on strings with a large number of 1's, owing to the fact that the measurements have a mutual eigenbasis (\Cref{fact:clh-eigenspace}). For non-commuting Hamiltonians, we do not have this property, and thus we need to modify the protocol $\cP$ to make it robust against such postselection attacks. The natural question is why can't we use a different code, say a quantum code, for the readout of the measurement bits that can tolerate Z errors? The issue is that the vertices will then not be able to read out the logical measurement values held by the cliques, due to the well-known local indistinguishability of quantum codewords. This suggests an apparent dichotomy in the readout step. On one hand, we need some sort of local decodability on each vertex. On other hand, we need some noise tolerance in both X and Z bases. We resolve this challenge by two ideas.

First, we  switch to a \emph{noisy} repetition code:
\begin{align}
    \NRep: \ket{0} \mapsto ((1-\gamma)\ketbra{0}+\gamma \ketbra{1})^{\otimes d}, \qquad \ket{1} \mapsto (\gamma \ketbra{0} + (1-\gamma) \ketbra{1})^{\otimes d},
\end{align}
where $0< \gamma < 1/2$. This means that the ideal logical measurement value of `$1$' is encoded into a superposition of strings with a large number of 1's, but also some 0's. This prevents the adversary from postselecting on the all-zero logicalstring, and thus, in a sense, introduces a small tolerance against Z errors while approximately retaining the local decodability of the repetition code. Concretely, we modify the distributed protocol $\cP$ (\Cref{algo:overall-protocol}) to a new protocol $\cP_\gamma$ paramaterized by $\gamma$ as follows. At the end of each energy measurement in step 2(a) of~\Cref{algo:overall-protocol}, instead of running the repetition code conversion~\Cref{algo:rep-code-output} for readout, the players in a clique perform a $\gamma$-noisy version of it:

\begin{mdframed}
\begin{protocol}[Noisy repetition-code readout]\mbox{}\label{algo:noisy-rep-code-output}\\
\textbf{Input:} A distributed \(\mathcal{Q}\)-codeblock \(\encq(a \ket{0} + b \ket{1})\) jointly held by the \(n\) players.\\
\textbf{Ideal output:} \(  a^2 \NRep(\ketbra{0})+ b^2 \NRep(\ketbra{1}) \).
\begin{enumerate}
\item Let \(B_0\) denote the input block. The players performs MPC (\Cref{thm:distributed-FTQC}) to prepare fresh blocks \(B_1,\ldots,B_n\), each supposedly equal to \(\encq(\sqrt{1-\gamma} \ket{00} + \sqrt{\gamma} \ket{11}) \).

\item For each \(j=1,\ldots,n\), do the following:
\begin{enumerate}
    \item Apply a logical CNOT from \(B_0\) to the first register of \(B_{j}\).
    \item Run error suppression \Cref{algo:error-correction} on \(B_1\).
    \item Run error suppression \Cref{algo:error-correction} on \(B_j\).
\end{enumerate}

\item For each \(i\in[n]\), player \(i\) receives the first register of \(B_{i}\) and unencodes it to a single qubit.
\end{enumerate}
\end{protocol}
\end{mdframed}

The remaining steps of $\cP_\gamma$ then follow the same as $\cP$.

\begin{claim}[Noiseless execution of $\cP_\gamma$]\
    \label{claim:ideal-protocol-nc}
    Let $\ket{\psi}$ be any input state to the sequential measurement circuit and $C \ket{\psi}= \sum_{y \in \{0,1\}^{M}} \alpha_y \ket{y}\otimes \ket{\xi_y}$, where $y$ are the measurement outcomes. Let $p_\gamma(x)$ be the multivariable Bernoulli distribution with bias $\gamma$ on $Md$ bits. Then, the output of the modified protocol $\cP_\gamma$ on input $\encq(\ket{\psi})$ is
    \begin{align}
        \cP_\gamma(\encq(\ket{\psi})) &= \sum_{y \in \{0,1\}^{M} \atop x \in \{0,1\}^{Md}}   \alpha_y  \sqrt{p_\gamma(x)} \ket{\Rep(y) \oplus x}_\mathrm{meas} \otimes \encq(\ket{y}\ket{\xi_y}\ket{x}) \otimes \ket{\mathrm{rest}}, \\
        &:= \mathrm{NEnc}_\mathrm{out} (C\ket{\psi}) \otimes \ket{\mathrm{rest}},
    \end{align}
    where we implicitly generalize $\Rep$ and $\encq$ multi-qubit strings by applying them bitwise.
\end{claim}
\begin{proof}
    Recall that in the original protocol $\cP$ (\Cref{algo:overall-protocol}),the state before the repetition code readout step (step 2(a) of \Cref{algo:overall-protocol}) upon input $\encq(\ket{\psi})$ is
    \begin{align}
        \sum_{y \in \{0,1\}^{M}}  \alpha_y \encq(\ket{y} \ket{\xi_y}) \otimes \ket{\mathrm{rest}}.
    \end{align}
    Next, we perform random noise on the repetition code readout
    \begin{align}
        \sum_{y \in \{0,1\}^{M} \atop x \in \{0,1\}^{Md}}   \alpha_y \sqrt{p_\gamma(x)} \encq(\ket{\Rep(y) \oplus x}_\mathrm{meas}) \otimes \encq(\ket{y}\ket{\xi_y}\ket{x}) \otimes \ket{\mathrm{rest}}.
    \end{align}
    Finally, we unencode $\mathcal{Q}$ on the first register to obtain the desired form.
\end{proof}

It turns out that this noisy repetition encoding doesn't quite work on its own because the global number of adversarial postselection errors in the soundness case may still be too large (see~\Cref{subsec:distribution-APgamma}). Thus we also need to introduce a second idea which is a modification of the circuit-to-Hamiltonian mapping in~\Cref{sec:ckt-to-H}. The details of this is not important at this point, but the effect we get is that at the end of protocol $\cP_\gamma$, the players will further apply a `weak postselection' gate that favors `$1$' on the measurement bits to counter against the adversary's noisy postselection on `$0$' (which is no longer perfect due to the noise introduced by idea 1). This is implemented by applying a local operator
\begin{align}
    \label{eq:38453829}
    A_\eta = \sqrt{1-\eta} \ketbra{1} + \sqrt{\eta} \ketbra{0}
\end{align}
on each measurement bit, where $\eta \leq 1/2$ is a constant that depends on $\gamma$. The vanilla case corresponds to $\eta=1/2$ and $A_\eta \propto I$. It will turn out that $\eta$ can be chosen to be significantly larger than $\gamma$ and in some regime to be $1/2 -o(1)$ (see~\Cref{def:gamma'}).

In summary, the overall modified protocol is described below in \Cref{algo:nc-protocol}. Recall that we consider an input Hamiltonian $H_\mathrm{in}$ consisting of $M$ many $k$-local terms such that its sequential measurement circuit $C$ measures $N$ non-overlappingterms per layer and has $\ell$ layers.

\begin{mdframed}
\begin{protocol}\mbox{}\label{algo:nc-protocol} Modified Protocol $\cP_\gamma$ with Postselection Gates $A_\eta$ for Sequential Measurement Circuit $C$\\
\textbf{Input:}  On every $d$-clique $i \in [N]$ in the HDX $G=(X(1), X(2))$ we have $k$ distributed input registers, which are expected to be $k$ quantum MPC codeblocks.\\
\textbf{Output:} Each vertex $v$ holds shares of the violation values of the Hamiltonian terms assigned to cliques containing $v$.
\\

\noindent For each time-step $t \in [\ell]$, let $G_{1,t},\ldots, G_{N,t}$ be the $N$ gates in the $t^{th}$ layer of $C$:
\begin{enumerate}
\item Clique $i$ runs the $i^{th}$ gate $G_{i,t}$ in $C$ using the quantum-multiparty computation protocol~\Cref{thm:distributed-FTQC}. This results in $k+1$ code blocks supported on the vertices of clique $i$, supposedly encoding the $k$ output wires and the measurement value wire in the quantum MPC code
\begin{enumerate}
    \item Encode the measurement wire into a noisy repetition code using~\Cref{algo:noisy-rep-code-output}, so that each vertex holds a single qubit of the noisy repetition encoding.
    \item Each vertex applies the postselection gate $A_\eta$ (see~\Cref{eq:38453829}) on its local qubit of the noisy repetition encoding.
\end{enumerate}
\item Decompose the perfect matching formed by the outputs of the $t^{th}$-layer of $C$ and the inputs of the $(t+1)^{th}$ layer of $C$, into $k$ permutations $\pi_1,\ldots, \pi_k$.
\item Route each permutation $\pi_i$ using the clique-to-clique routing in~\Cref{prop:ft-routing}. This results in clique $i$ holding $k$ code blocks, which are supposed to be encoding the input wires of $G_{i,t+1}$ in a quantum code. 
\end{enumerate}
\end{protocol}
\end{mdframed}

The ideal output state of the protocol is obtained by applying $A_{\eta}^{\otimes Md}$ on the output of $\cP_\gamma$, and for soundness analysis in this section we will analyze states of the form
\begin{align}
    \sum_{\sF\in \mathsf{PAULI}^{\cE}} c_{\sF} A_\eta^{\otimes Md} (\cP_\gamma)_\sF \ket{\phi_\sF}.
    \label{eq:state-overview-sec7}
\end{align}
In our final Hamiltonian we will have a term per vertex that checks if the measurement bits in its local view have at least some fraction of 1's. We will show that if the original Hamiltonian $H_\mathrm{in}$ has large energy, then this state will violate a large fraction of these vertex checks whenever the adversarially corrupted edge set $\cE$ is below a fraction.

\begin{definition}[Vertex threshold checks]\label{def:vertex-threshold-check}
Fix \(v\in X(1)\) and let the set of measurement bits in the local view of \(v\) be
\[
M_v:=\{(i,t)\in [N]\times [\ell]: v\in D_i\}.
\]
Let $0\leq \tau_\mathrm{out}\leq 1$ and $\mathrm{Th}_{\tau_\mathrm{out}}$ be the threshold function,
we define \(\Pi_v^{\tau_\mathrm{out}}\) to be the projector acting on these qubits
\[
\Pi_v^{\tau_\mathrm{out}}
:=
\sum_{b\in \{0,1\}^{M_v}} \mathrm{Th}_{\tau_\mathrm{out}}(b)\,\ket{b}\!\bra{b}.
\]
We also denote the \emph{vertex OR checks} $\Pi_v
:=
\sum_{b\in \{0,1\}^{M_v}} \mathrm{OR}_{M_v}(b)\,\ket{b}\!\bra{b}$ as in~\Cref{def:vertex-check-operator}.
\end{definition}

The rest of this section is organized as follows. First, we will show in~\Cref{subsec:completeness-nc-protocol} the completeness of the modified protocol (\Cref{claim:completeness-nc-protocol}). Then we show in~\Cref{subsec:nc-robust-measurement-circuit} that the logical measurement circuit is robust to adversarial errors if the original Hamiltonian has large energy, by proving a version of the detectability lemma that applies to the sequential measurement circuit. Next, in~\Cref{subsec:distribution-APgamma}, we analyze the output distribution of the modified protocol for the corrupted computation state in~\Cref{eq:state-overview-sec7}. Finally, in~\Cref{sec:soundness-nc-protocol}, we use the sampling properties of the HDX to conclude the soundness of the modified protocol (\Cref{thm:soundness-nc-protocol}).

\subsection{Completeness}\label{subsec:completeness-nc-protocol}
Here we show that the modified protocol has near-perfect completeness and that its output state can be prepared efficiently when the input Hamiltonian has zero ground energy.

\begin{claim}\label{claim:completeness-nc-protocol}
    Let the input Hamiltonian $H_\mathrm{in}$ have zero-energy ground state $\ket{\xi_0}$, and consider the normalized state
    \begin{align}
        \ket{\Psi_0} \propto A_\eta^{\otimes Md} \cP_\gamma  \encq(\ket{\xi_0}).
        \label{eq:nc-completeness-1}
    \end{align}
    Then, the measurement bits give an i.i.d. Bernoulli distribution with parameter $ \gamma ' = \frac{\gamma (1-\eta) }{ \gamma (1-\eta) + (1-\gamma) \eta}$. As a corollary, the expected value of the vertex OR checks is
    \begin{align}
        \mathbb{E}_{v} \Tr(\Pi_v \ketbra{\Psi_0}) \leq \gamma' C_{\max}, \qquad \text{where } C_{\max} := \max_v |M_v|
        \label{eq:nc-completeness-2}
    \end{align}
    And in the regime $(\tau_\mathrm{out} -\gamma') \cdot \sqrt{C_{\min}} \geq \pl n$, where $C_{\min}:=\min_v |M_v|$, the expected value of the threshold checks
    \begin{align}
        \mathbb{E}_{v} \Tr(\Pi_v^{\tau_\mathrm{out}} \ketbra{\Psi_0}) \leq 2^{- {\pl n}}.
        \label{eq:nc-completeness-3}
    \end{align}
    Furthermore, the state $\ket{\Psi_0}$ is efficiently preparable given $\ket{\xi_0}$.
\end{claim}
\begin{proof} Let $\ket{\xi_0}$ be a ground state of $H_\mathrm{in}$ and consider the execution of $\cP_\gamma$ on input $\encq(\ket{\xi_0})$, where $\cQ$-blocks are distributed to the cliques as specified by the protocol.
    It follows from~\Cref{claim:ideal-protocol-nc} that the output of $\cP_\gamma$ is 
    \begin{align} \cP_\gamma(\encq(\ket{\xi_0})) = \sum_{x \in \{0,1\}^{M d} } \sqrt{p_\gamma(x)} \ket{x}_\mathrm{meas} \otimes \encq(\ket{x}) \otimes \ket{\mathrm{rest}'},
    \end{align}
    where $p_{\gamma}(x)$ denotes the multivariable Bernoulli distribution with bias $\gamma$,
    whose reduced state on the measurement bits is
    \begin{align}
        \sum_{x \in \{0,1\}^{Md}}  p_{\gamma}(x) \ketbra{x}{x}.
    \end{align}
    Applying the $A_\eta$ operators then gives
    \begin{align}
        A_\eta^{\otimes Md} \sum_{x \in \{0,1\}^{Md}}  p_{\gamma}(x) \ketbra{x}{x} A_\eta^{\otimes Md}  \propto \sum_{x \in \{0,1\}^{Md}}  p_{\gamma}(x) (1-\eta)^{|x|} \eta^{Md-|x|} \ketbra{x}{x} \propto \sum_{x \in \{0,1\}^{Md}}  p_{\gamma'}(x) \ketbra{x}{x}.
    \end{align}
    The OR check in~\Cref{eq:nc-completeness-2} follows from the union bound over the measurement bits in the link of a vertex. For the threshold checks in~\Cref{eq:nc-completeness-3},
    by the additive Chernoff bound, each local view of a vertex $v$ will have more than $\tau_\mathrm{out}$ fraction of $1$'s with probability $e^{- \Omega((\tau_\mathrm{out} -\gamma')\sqrt{|M_v|} )}$.

    The only nontrivial place to confirm that $\ket{\Psi_0}$ is efficiently prepareble from $\ket{\xi_0}$ is in applying $A_\eta^{\otimes Md}$, which seems problematic because $A_\eta$ is non-unitary. However, as seen above, the effect of $A_\eta$'s is simply to change the bias of the Bernoulli distribution on the measurement bits from $\gamma$ to $\gamma'$. Thus, we can simply prepare the noisy repetition code with bias $\gamma'$ instead of $\gamma$ in~\Cref{algo:noisy-rep-code-output}, and then we do not need to apply $A_\eta$'s at all.
\end{proof}

\subsection{Robust Measurement Circuit from the Robust Detectability Lemma}\label{subsec:nc-robust-measurement-circuit}
In the noncommuting case, while intuitive, it is a priori far from trivial why the sequential measurement circuit should output a large number of violations if the input Hamiltonian has large energy. This question is closely related to the detectability lemma~\cite{AharonovALV09} with a subtle difference. In~\cite{AharonovALV09}, the authors give a `proof checking' statement: for any state, there exists a layer of non-overlapping terms that incur high violations. However, the mere existence of such layer is not sufficient for our purposes. Because we sequentially measure the layers in the circuit $C$, the state changes as we move on to the next layers. Nevertheless, we show that their proof can be adapted to establish an operational variant of the detectability lemma for sequential measurement circuits. As a corollary, we show that if the Hamiltonian has energy at least $\epsilon$, then the sequential measurement circuit is robust to $\Theta(\epsilon)$-fraction of adversarial errors. To state this lemma we need a few mild assumptions on our input Hamiltonian as similarly made in~\cite{AharonovALV09}.

\begin{assumption}\label{assumption:DL-Hamiltonian}
    Fix constants $k,g$ and a fixed set $\mathfrak{h}$ of $k$-qubit projectors that can be specified by a constant number of bits. We will assume that $H = \sum_{i=1}^{M} h_i$ on $n$ qubits, such that each $h_i$ is a projector that belongs to the set $\mathfrak{h}$ and only intersects with $g=\Theta(1)$ other terms in $H$.
\end{assumption}

\begin{definition}[Robust measurement circuits]
Consider a quantum circuit $C$ with $n$ inputs and $M$ measurement outputs. 
We say that $C$ is $(\eta, \epsilon, \delta)$-robust to adversarial faults if the following holds. For any subset $F$ of locations in $C$ of size at most $\eta |C|$, the adversary corrupts the gates in $F$ and is allowed to take a coherent superposition of different fault patterns on different input states, 
\begin{align}
    \sum_{\mathsf{F} \in \mathsf{PAULI}^{F}} c_\mathsf{F} C_\mathsf{F} \ket{\psi_\mathsf{F}}.
\end{align}
Let the corresponding distribution on the measurement qubits $q(y)$, then
\begin{align}
    \Pr_{y \sim q}[|y| \leq \epsilon M] \leq \delta.
\end{align}

\end{definition}

In~\Cref{sec:DL-new}, we will prove the following version of the detectability lemma that applies to the sequential measurement circuit.

\begin{lemma}[Robust detectability lemma]\label{lem:dl} Let $H$ be a Hamiltonian as in~\Cref{assumption:DL-Hamiltonian}, whose (unnormalized) minimum energy $ \geq \epsilon M$.
Then its sequential measurement circuit is a $(\epsilon/r, \epsilon/r', 2^{- r'' \epsilon M })$-robust circuit, where $r, r',r''$ are constants depending only on the constants in~\Cref{assumption:DL-Hamiltonian}.
\end{lemma}

\subsection{Fault Tolerance of the Measurement Distribution under Adversarial Edge-Corruptions}\label{subsec:distribution-APgamma}

We now analyze the output distribution of~\Cref{algo:nc-protocol} under tensor-network adversarial edge-corruption model:
\begin{align}
    \ket{\Psi} = \sum_{\sF\in \mathsf{PAULI}^{\cE}} c_{\sF} A_\eta^{\otimes Md} (\cP_\gamma)_\sF \ket{\phi_\sF}.
\end{align}

\begin{definition}\label{def:nc-bad-cliques}
    Let us recall the definitions of good and bad cliques from~\Cref{def:bad-cliques}. Fix a corrupted edge set $\cE\subseteq X(2)$. For each clique \(D_i\), we call it \emph{bad} if it contains at least $\eps_0$-fraction of corrupted edges 
and \emph{good} otherwise. Let $\cB_{\cE}$ denote the set of bad cliques. We will work with the following parameters
\begin{align}
    \label{eq:params-section7}
\eps_0 \leq \min\!\left\{\eps^*_{\mathrm{comp}},\,\eps^*_{\mathrm{out}},\,\frac{\delta}{10}\right\},
\qquad
\tau:=\min\!\left\{c_{\mathrm{comp}}\sqrt{\eps_0},\,c_{\mathrm{route}}\delta,\,\beta^*_{\mathrm{out}}(\delta)\right\}.
\end{align}
where $\delta$ is the parameter of the quantum MPC code from~\Cref{sec:mpc}, \(\eps^*_{\mathrm{comp}},c_{\mathrm{comp}}\) come from \Cref{thm:distributed-FTQC}, \(\eps^*_{\mathrm{out}},\beta^*_{\mathrm{out}}(\delta)\) come from \Cref{lem:rep-code}, and \(c_{\mathrm{route}}\) comes from \Cref{prop:ft-routing}. Intuitively, $\varepsilon_0$ is the noise threshold of the quantum MPC scheme in~\Cref{sec:mpc}, and $\tau$ is the residual error guarantee of quantum MPC on a good (below-threshold) clique.

For each good clique \(D_i\), let $\cE_i$ denote the set of corrupted edges inside it and define the set of bad vertices inside it as $W_i :=\{v\in D_i:\deg_{\cE_i}(v)>\delta d/2\}$, where $d$ is the clique size. Note that $|W_i| \leq (2\eps_0/\delta) \cdot d$ .
We say that \(v\) is good in $D_i$ if \(v\notin W_i\). 
\end{definition}

Consider first a single summand in the above expression. We can obtain a fault tolerance guarantee similar to~\Cref{lem:ft}.

\begin{prop}[Fault-tolerance of \(\cP_\gamma\)]\label{lem:ft-Pgamma}
Fix a corrupted edge set \(\cE\subseteq X(2)\) with \(\mu(\cE)\le \eps_0\). Then there exists \(\kappa\le \mu(\cE)/\pl n\) and exceptional clique subsets
$L_t\subseteq [N], t\in[\ell]$ \footnote{Recall that $N\cdot \ell=M$ is the total number of constraints in the input Hamiltonian.}, 
depending only on \(\cE\), with \(|L_t|\le \kappa N\) for every \(t\), such that the following holds.

For every pure input state \(\ket{\phi_{\mathsf{F}}}\) and every Pauli fault pattern \(\mathsf F\) with
\(\operatorname{supp}_{\cP_\gamma}(\mathsf F)\subseteq \cE\), there exist a set of indices $r$ and corresponding coefficients \(c_{r, \mathsf F}\), normalized logical witness states \(\ket{\psi_{r, \mathsf{F}}}\) to the logical circuit $C$, Pauli operators \(E_{r, \mathsf F}\), and work-register states \(\ket{\mathrm{rest}_{r,\mathsf{F}}}\) such that
\begin{align}\label{eq:2dsaaw21}
(\cP_\gamma)_{\mathsf F}\ket{\phi_{\sF}}
&=\sum_r c_{r,\sF}\,E_{r,\sF}\ \mathrm{NEnc}_\mathrm{out} ( C \ket{\psi_{r,\sF}} ) \otimes \ket{\mathrm{rest}_{r,\sF}}\\
&=\sum_r c_{r,\sF}\,E_{r,\sF}\ \sum_{y \in \{0,1\}^{M} \atop x \in \{0,1\}^{Md}}  \alpha_{r,y,\sF} \sqrt{p_\gamma(x)} \ket{\Rep(y)\oplus x} \otimes \encq(\ket{y}\ket{\xi_{r,y,\sF}} \ket{x}) \otimes \ket{\mathrm{rest}_{r,\sF}} \notag{}
\end{align}
where $\mathrm{NEnc}_\mathrm{out}$ is defined in~\Cref{claim:ideal-protocol-nc} and the second line expands the logical output state 
\begin{align}
    C \ket{\psi_{r, \mathsf F}} = \sum_{y \in \{0,1\}^{M} }\alpha_{r,y, \mathsf F} \ket{y} \ket{\xi_{r,y, \mathsf F}}.
\end{align}
Moreover, every summand in \eqref{eq:2dsaaw21} satisfies:

\begin{enumerate}
    \item The restriction of $E_{r,\mathsf F}$ to every $\cQ$-block outside the lightcone \(\LC(L_1,\ldots,L_\ell)\) has weight at most $\tau d$.
    \item If \(M_{i,t}\) is a measurement wire outside the lightcone \(\LC(L_1,\ldots,L_\ell)\), then the restriction of $E_{r, \mathsf F }$ to the repetition block corresponding to \(M_{i,t}\) is supported only on the coordinates indexed by $W_i(\cE)$.
\end{enumerate}
\end{prop}
\begin{proof}
Straightforward imitation of the proof of~\Cref{lem:ft}.
\end{proof}

Using the fault tolerance guarantee of $\cP_\gamma$ in \Cref{lem:ft-Pgamma}, we will show in~\Cref{lem:output-distribution-Pgamma} below that the measurement bit distribution of $\ket{\Psi}$, in the soundness case, has high Hamming weight on good cliques.

\begin{definition}\label{def:gamma'}
    Let $p_\gamma (x)$ denotes the multivariable Bernoulli distribution with bias $\gamma$. Let $0\leq \eta \leq 1/2$, define the parameters
    \begin{align}
        \gamma ' = \frac{\gamma (1-\eta) }{ \gamma (1-\eta) + (1-\gamma) \eta}, \qquad \gamma '' = \frac{(1-\gamma)(1-\eta) }{ (1-\gamma)(1-\eta) + \gamma \eta}.
    \end{align}
    Later we will choose
    \begin{align}
    \eta = \max\left(\frac{1}{2} - \Theta(\log(1/\gamma)\eps_0/\delta), \Theta(\gamma^{0.3})\right),
\end{align}
    in which case we have
\begin{align}
    \gamma' = \frac{\gamma(1-\eta)}{\gamma(1-\eta) + \eta (1-\gamma)} \lll \gamma^{0.5}, \qquad \text{and} \qquad \gamma'' = \frac{(1-\gamma)(1-\eta)}{(1-\gamma)(1-\eta) + \gamma \eta} \geq 1 - O(\gamma ).
\end{align}
\end{definition}

\begin{lemma}[Output distribution of faulty $\cP_\gamma$]\label{lem:output-distribution-Pgamma}  
    Suppose the input Hamiltonian has (unnormalized) ground energy $\epsilon M$. Fix \(\cE\subseteq X(2)\) with \(\mu(\cE)  \le \eps_0\).
    Consider an adversarially faulty execution of the distributed protocol $\cP_\gamma$, including the final postselection gates $A_\eta$, with the parameter choices as in~\Cref{def:gamma'} for the sequential measurement circuit $C$ of $H_\mathrm{in}$ with faults restricted to edges $\cE$.

    If
    $
        \mu(\cE)/{\pl n} \lll \epsilon 
    $
    then the
    distribution of the physical measurement bits held by the good vertices of the good cliques\footnote{See~\Cref{def:nc-bad-cliques}.}, denoted $\cW^c$, has the form
     \begin{align}
        P(z) = \sum_{y_G} q(y_G)  P_{\Rep(y_G)_{\cW^c}}(z),
     \end{align}
     where $q(y_G)$ is a distribution over the logical measurement bits held by the cliques outside of $LC(L_1,...,L_\ell)$ \footnote{The lightcone of the corrupted/exceptional cliques, see~\Cref{lem:ft-Pgamma}.} satisfying $\Pr_{y_G \sim q}[|y_G| \ggg \epsilon M] \geq 1-2^{-\poly(M)}$, and
     \begin{align}
         P_{\Rep(y_G)_{\cW^c}}(z) =  \bigotimes_{j \in \cW^c: \Rep(y_G)_j=1} p_{\gamma''}(z_j) \bigotimes_{j \in \cW^c: \Rep(y_G)_j=0} p_{\gamma'}(z_j).
     \end{align}
\end{lemma}

In other words, the distribution of the good measurement bits in the good cliques is a noisy repetition encoding of an underlying logical distribution $q(y_G)$ that is mostly supported on $|y_G| \ggg \epsilon M$.

\paragraph{Proof overview:} Similarly to the commuting case proof in~\Cref{sec:protocol-for-circuit}, we will show that $P(z)$ can be reduced to analyzing the logical distributions of a faulty sequential measurement circuit. First, we separate the noise into a contribution from good cliques and another from bad cliques. If there were only the bad-cliques noise, we could directly obtain guarantees on the physical measurement distribution by relating it to a logical distribution of a faulty logical energy measurement circuit~\Cref{claim:Lambda-high}. This logical distribution is mostly supported on high-weight strings by the robust detectability lemma (\Cref{lem:dl}). Next, \Cref{claim:physical-ditribution-each-component} shows that this property is maintained when the good-clique noise is added \footnote{This step is where we need the noisy postselection operator $A_\eta$.}. All together, the final distribution over $z$ is mostly supported on high-weight strings.

\begin{proof}[Proof of~\Cref{lem:output-distribution-Pgamma}]
    Applying the fault tolerance of the protocol~\Cref{lem:ft-Pgamma}, we express the output state as of an adversially noisy execution of $\cP_\gamma$ as
\begin{align}
    \ket{\Psi} &= \sum_{\sF\in \mathsf{PAULI}^{\cE}} c_{\sF} A_\eta^{\otimes Md} (\cP_\gamma)_\sF \ket{\phi_\sF}\notag  \\
    &= \sum_{\sF} \sum_r c_{r,\sF} A_\eta^{\otimes Md} \,E_{r,\sF}\ \mathrm{NEnc}_\mathrm{out} ( C \ket{\psi_{r,\sF}} ) \otimes \ket{\mathrm{rest}_{r,\sF}}.
\end{align}
Crucially,~\Cref{lem:ft-Pgamma} guarantees: On at least $1 - \mu(\cE) /\pl n$ cliques, the restriction of \(E_{r, \mathsf F}\) to the corresponding $\cQ$-codeblock(s) has weight at most \(\tau d\), and the restriction of $E_{r, \mathsf F}$ to the measurement register is supported only on the bad vertices with respect to that clique.

Let $\cW \subseteq [Md]$ be the set of corrupted measurement bits (first register in the above state) and $\cW^c = [Md]\backslash \cW$ be its complement. Now decomposing each $E_{r,\sF}$ into two parts: computational basis states $\ket{o},\ket{o'}$ on the corrupted measurement bits $\cW$,  and Paulis $\mathsf{E}$ on the corrupted data and workspace qubits (denoted by $\cR$), we can regroup terms and rewrite the state as 
\begin{align}
\ket{\Psi} &= \sum_{\substack{o,o' \in \{0,1\}^{\cW} \\  \mathsf{E} \in \mathsf{Pauli}^{\cR}}} A_\eta^{\otimes Md} c_{o,o',\mathsf{E}} (\ketbra{o}{o'} \otimes \mathsf{E})  \cdot \mathrm{NEnc}_\mathrm{out} ( C \ket{\psi_{o,o',\mathsf{E}}} ) \otimes \ket{\mathrm{rest}_{o,o',\mathsf{E}}} \tag*{} \\
 & = \sum_{\substack{o,o' \in \{0,1\}^{\cW} \\  \mathsf{E} \in \mathsf{Pauli}^{\cR}}} A_\eta^{\otimes Md} c_{o,o',\mathsf{E}} (\ketbra{o}{o'} \otimes \mathsf{E})  \sum_{\substack{y \in \{0,1\}^{M} \atop x \in \{0,1\}^{Md} }} \alpha_{o,o',\mathsf{E},y} \ket{\Rep(y) \oplus x} \otimes \encq(\ket{y}\ket{\xi_{o,o',\mathsf{E},y}} \ket{x}) \otimes \ket{\mathrm{rest}_{o,o',\mathsf{E}}},
    \label{eq:Psi-noncommuting-proof}
\end{align}
where $\alpha_{o,o',\mathsf{E},y}$ are the coefficients in an ideal execution of the logical measurement circuit $C$ on some normalized input state $\ket{\phi_{o,o',\mathsf{E}}}$,
\begin{align}
    C \ket{\phi_{o,o',\mathsf{E}}} = \sum_{y} \alpha_{o,o',\mathsf{E},y} \ket{y}\ket{\xi_{o,o',\mathsf{E},y}}.
    \label{eq:Cphi}
\end{align}

Our goal is to compute 
\begin{align}
    P(z) = \Tr(\ketbra{z}_{\cW^c} \ketbra{\Psi}{\Psi} ).
\end{align}

We denote the abstract logical measurement bits held by the cliques inside the lightcone of the bad cliques\footnote{Each clique may be in charge of constantly many logical measurement qubits throughout the protocol.} as $B \subseteq [M]$ and the remaining measurement bits by $G = [M]\backslash B$. Note that $|B| = M  \cdot \mu(\cE) /\pl n \lll \epsilon M$. Abusing notation, let $\cW_B$ be the restriction of $\cW$ to the bad cliques and similarly define $\cW_G$, so that $\cW = \cW_B \cup \cW_G$. We assume without loss of generality that $\cW_B$ equals the full set of bad-clique physical measurement bits, i.e.\ $|\cW| = |B|\cdot d$. For $\cR$, we split $\cR = \cR_B \cup \cR_G \cup \cR_w$, where $\cR_B$ ($\cR_G$) denotes the restriction onto the $\mathcal{Q}$-code blocks on $B$-cliques ($G$-cliques), and $\cR_w$ is the workspace qubits (collecting over all cliques). Similarly we assume wlog $\cR_B$ equals the full set of $\cQ$-code physical qubits on $B$-cliques. We will use the decompositions $y=(y_B,y_G)$, $x = (x_B, x_G)$, $o'=(o'_B,o'_G)$, $\mathsf{E}=(\mathsf{E}_B, \mathsf{E}_G, \mathsf{E}_{w})$.

We now state the three main claims and prove the lemma.

First, \Cref{claim:reduction-to-bad-cliques} shows that it suffices to perform the analysis for fixed values $o,o'_G,\mathsf{E}_G$.

\begin{claim}[Decohering good-clique noise]\label{claim:reduction-to-bad-cliques}
     Decompose the sum $o',E$ in~\Cref{eq:Psi-noncommuting-proof} into $o'_G,o'_B,E_B,E_G$, where $E_G$ ranges over non-equivalent Pauli errors (with respect to the stabilizers of the MPC code $\mathcal{Q}$). Let 
    \begin{align}
        \ket{\Psi_{o,o'_G,\mathsf{E}_G}} = \sum_{\substack{o'_B \in \{0,1\}^{\cW_B} \\  \mathsf{E}_{Bw} \in \mathsf{Pauli}^{\cR_{Bw}}}} A_\eta^{\otimes Md} c_{o,o',\mathsf{E}} (\ketbra{o}{o'} \otimes \mathsf{E})  \cdot \mathrm{NEnc}_\mathrm{out} ( C \ket{\psi_{o,o',\mathsf{E}}} ) \otimes \ket{\mathrm{rest}_{o,o',\mathsf{E}}}.
    \end{align}
     It holds that 
    \begin{align}
         P(z) = \sum_{o,o'_G,\mathsf{E}_G} \left\| \bra{z} \ket{\Psi_{o,o'_G, \mathsf{E}_G}} \right\|^2.
    \end{align}
\end{claim}

This claim asserts that it suffices to analyze the physical measurement distribution of each $\ket{\Psi_{o,o'_G,\mathsf{E}_G}}$. We will relate this to an (un-normalized) probability distribution of the logical measurement bits $y_G$ associated with  good cliques  of a faulty logical sequential measurement circuit, which in turn is guaranteed to have most mass on high-weight $y_G$ by the robust detectability lemma.

Consider fixed $o, o'_G, \mathsf{E}_G$. We separate the noise contribution in $\ket{\Psi_{o,o'_G,\mathsf{E}_G}}$ into good-clique and bad-clique parts. The bad-clique noise defines the following state
\begin{align}
\ket{\Phi_{o, o'_G, \mathsf{E}_G}} &= \sum_{\substack{o'_B \in \{0,1\}^{\cW_B} \\  \mathsf{E}_{Bw} \in \mathsf{Pauli}^{\cR_{Bw}}}} c_{o,o',\mathsf{E}} (\bra{o'_B} \otimes \mathsf{E}_{Bw})  \cdot \mathrm{NEnc}_\mathrm{out} ( C \ket{\psi_{o,o',\mathsf{E}}} ) \otimes \ket{\mathrm{rest}_{o,o',\mathsf{E}}}
\end{align}
so that applying the good-clique noise yields $\ket{\Psi_{o,o'_G,\mathsf{E}_G}}= A_\eta^{\otimes Md} (\ketbra{o}{o'_G} \otimes \mathsf{E}_G)  \cdot \ket{\Phi_{o, o'_G, \mathsf{E}_G}}$.

The next claim establishes that the damage of the bad-clique noise can be controlled.

\begin{claim}[Reducing physical distribtion to logical distribution]\label{claim:Lambda-high} Suppose that $|B| = M  \cdot (\mu(\cE) /\pl n) \lll \epsilon M$, then
\begin{align}
    \ket{\Phi_{o, o'_G, \mathsf{E}_G}} = \sum_{y_G \in \{0,1\}^{G}} \sum_{x_G \in \{0,1\}^{Gd} } \sqrt{p_\gamma(x_G)} \Lambda_{y_G} \ket{\Rep(y_G) \oplus x_G}  \otimes \ket{Y_{y_G,x_G}},
\end{align}
where $\ket{Y_{y_G,x_G}}$'s are normalized states that are orthogonal for distinct $x_G,y_G$, and the coefficients $\Lambda_{y_G}$ satisfy the following condition: There exist constants $r''$, and $c' \ll c$ depending only on the Hamiltonian family (\Cref{assumption:DL-Hamiltonian}) such that 
        \begin{align}
            \Lambda_{\mathrm{low}} \leq 2^{-r'' \epsilon M} \Lambda_{\mathrm{high}}, \qquad \text{where }\Lambda_{\mathrm{high}} := \sum_{y_G: |y_G|\geq c\epsilon M} \Lambda_{y_G}^2, \qquad \Lambda_{\mathrm{low}} := \sum_{y_G: |y_G|< c' \epsilon M} \Lambda_{y_G}^2.
        \end{align}
\end{claim}

In other words, the above claim says that, if there were only bad-clique noise, then up to the Bernoulli noise $\gamma$, the physical measurement distribution would essentially be a repetition encoding of strings of large Hamming weight $\geq c\epsilon M$. The actual state has the good-clique noise and the postselection gates on $\ket{\Psi_{o,o'_G,\mathsf{E}_G}}$. The following claim shows that these essentially do not affect the aforementioned guarantee.

\begin{claim}[Effect of good-clique noise on physical distribution]\label{claim:physical-ditribution-each-component} Under the same conditions in~\Cref{claim:physical-ditribution-each-component},
    it holds that
    \begin{align}
    \Tr(\ketbra{z} \Psi_{o,o'_G,\mathsf{E}_G} ) = \sum_{y_G} P_{\Rep(y_G)_{\cW^c}}(z) \Gamma^{ |\Rep(y_G)_{\cW^c}|}  \beta^{|\Rep(y_G)_{\cW_G} \oplus o'_{G}|}  |\Lambda'_{g,y_G}|^2.
\end{align}
      where the coefficients $\Lambda'_{g,y_G}$ satisfy the same properties as in~\Cref{claim:Lambda-high} and
    \begin{align}
        P_{\Rep(y)_{\cW^c}}(z) &:=  \bigotimes_{i: \Rep(y)_i=1} p_{\gamma''}(z_i) \bigotimes_{i: \Rep(y)_i=0} p_{\gamma'}(z_i)   \qquad \text{with } \gamma'' ,\gamma' \text{ defined in~\Cref{def:gamma'}},
    \end{align}
    and $\Gamma := \frac{(1-\gamma)(1-\eta) + \gamma \eta}{\gamma(1-\eta) + \eta (1-\gamma)}$, $\beta := \frac{\gamma}{1-\gamma}.$
\end{claim}

We defer the proofs of the above claims to later. Let us first combine~\Cref{claim:reduction-to-bad-cliques} and~\Cref{claim:physical-ditribution-each-component} (\Cref{claim:Lambda-high} is only used within the proof of~\Cref{claim:physical-ditribution-each-component}) to prove the lemma.

Indeed, combining the two lemmas we obtain
\begin{align}
    P(z) = \sum_{\substack{g=(o,o'_G, \mathsf{E}_G)}} \sum_{y_G} P_{\Rep(y_G)_{\cW^c}}(z) \Gamma^{ |\Rep(y_G)_{\cW^c}|}  \beta^{|\Rep(y_G)_{\cW_G} \oplus o'_{G}|}  |\Lambda'_{g,y_G}|^2.
\end{align}
with constants $r''$, and $c' \ll c$ such that, for every $g$, 
        \begin{align}
            \Lambda'_{g,\mathrm{low}} \leq 2^{-r'' \epsilon M} \Lambda'_{g,\mathrm{high}}, \qquad \text{where }\Lambda'_{g,\mathrm{high}} := \sum_{y_G: |y_G|\geq c\epsilon M} |\Lambda'_{g,y_G}|^2, \qquad \Lambda'_{g,\mathrm{low}} := \sum_{y_G: |y_G|< c' \epsilon M} |\Lambda'_{g,y_G}|^2.
        \end{align}
The key is that the `tilting' factor $\Gamma^{ |\Rep(y_G)_{\cW^c}|}  \beta^{|\Rep(y_G)_{\cW_G} \oplus o'_{G}|}$ does not significantly skew the weight on low-weight $y_G$. Let $\tilde{\Lambda}'_{g,y_G}= \Gamma^{|\Rep(y_G)_{\cW^c}|/2}\beta^{|\Rep(y_G)_{\cW_G}\oplus o'_G|/2}\Lambda'_{g,y_G}$ be the tilted amplitude, so that
    
    \begin{align}
        P(z) = \sum_{\substack{g=(o,o'_G, \mathsf{E}_G)}} \sum_{y_G} P_{\Rep(y_G)_{\cW^c}}(z) |\tilde{\Lambda}'_{g,y_G}|^2.
    \end{align}
    Let $\tilde{\Lambda}'_{g,\mathrm{high}}$ and $\tilde{\Lambda}'_{g,\mathrm{low}}$ denote the corresponding sums of squared tilted amplitudes:
    \begin{align}
    \tilde{\Lambda}'_{g,\mathrm{high}}
    &= \sum_{|y_G|\geq c\eps M} \Gamma^{|\Rep(y_G)_{\cW^c}|} \beta^{|\Rep(y_G)_{\cW_G}\oplus o'_G|} |\Lambda'_{g,y_G}|^2 \tag*{(using $|\Rep(y_G)_{\cW^c}| \geq (1-2\eps_0/\delta)|y_G|d \geq 0.99|y_G|d$)}\\
    &\geq \sum_{|y_G|\geq c\eps M} \Gamma^{0.99|y_G|d}  \beta^{|\Rep(y_G)_{\cW_G}\oplus o'_G|}  |\Lambda'_{g,y_G}|^2
    , \label{eq:Lambda high}
    \\
    \tilde{\Lambda}'_{g,\mathrm{low}}
    &= \sum_{|y_G| < c'\eps M} \Gamma^{|\Rep(y_G)_{\cW^c}|}\beta^{|\Rep(y_G)_{\cW_G}\oplus o'_G|} |\Lambda'_{g,y_G}|^2 \tag*{($|\Rep(y_G)_{\cW^c}|\leq |y_G|d$)} \\
    &\leq \sum_{|y_G| < c'\eps M} \Gamma^{|y_G|d} \beta^{|\Rep(y_G)_{\cW_G}\oplus o'_G|}  |\Lambda'_{g,y_G}|^2
    \label{eq:Lambda low}
\end{align}
We analyze the ratio between two nonzero summands~\eqref{eq:Lambda high} and~\eqref{eq:Lambda low}. For any $y_G$ with $|y_G| \geq  c\epsilon M$ and any $y'_G$ with $|y'_G| < c'\epsilon M \ll c\epsilon M/100$, the ratio is lower bounded by
    \begin{align}
        &\geq \frac{|\Lambda'_{g,y_G}|^2}{|\Lambda'_{g,y'_G}|^2} \cdot \Gamma^{(0.99|y_G|-|y'_G|)d} \cdot \beta^{|\Rep(y_G)_{\cW_G}  \oplus o'_{G}| - |\Rep(y'_G)_{\cW_G}  \oplus o'_{G}|}  \notag
        \\
        &\geq \frac{|\Lambda'_{g,y_G}|^2}{|\Lambda'_{g,y'_G}|^2} \cdot \Gamma^{0.98|y_G| d} \cdot \beta^{|\Rep(y_G)_{\cW_G}| + |o'_G| - ( |o'_G| - | \Rep(y'_G)_{\cW_G}| ) }   \tag*{(triangle inequaltiy)}\\
        &\geq \frac{|\Lambda'_{g,y_G}|^2}{|\Lambda'_{g,y'_G}|^2}  \cdot \Gamma^{0.98|y_G|d} \cdot \beta^{2.02(\eps_0/\delta) |y_G|d } \tag*{(bad-vertex fraction $\leq 2\eps_0/\delta$ per good clique)} \\
        &\geq \frac{|\Lambda'_{g,y_G}|^2}{|\Lambda'_{g,y'_G}|^2}  (\Gamma^{0.98} \beta^{2.02 \eps_0/\delta})^{|y_G|d}.
    \end{align}
    Assuming $\Gamma^{0.98}\beta^{2.02 \eps_0/\delta} > 1 $, then
    \begin{align}
        \frac{\tilde{\Lambda}'_{g,\mathrm{low}} }{\tilde{\Lambda}'_{g,\mathrm{high}}}  \leq \frac{\Lambda'_{g,\mathrm{low}} }{\Lambda'_{g,\mathrm{high}}} (\Gamma^{0.98} \beta^{2.02 \eps_0/\delta})^{-c\epsilon Md} \leq 2^{-r'' \epsilon M} (\Gamma^{0.98} \beta^{2.02 \eps_0/\delta})^{-c\epsilon Md} \leq o(1).
        \label{eq:final Lambda bound}
    \end{align}

    Because this holds for all $g$, we conclude that
    \begin{align}
        P(z) = \sum_{y_G} q(y_G) P_{\Rep(y_G)_{\cW^c}}(z),
     \end{align}
     where $q(y_G)$ is a distribution satisfying $\Pr_{y_G \sim q}[|y_G| < c'\epsilon M] \leq o(1)$, and thus $P(z)$ is mostly supported on high-weight strings.

     Finally, we derive the condition of $\eta$ required to ensure that $\Gamma^{0.98}\beta^{2.02 \eps_0/\delta} > 1 $, assuming $\gamma \ll 1/2$ and $\iota:= 3\eps_0/\delta \leq 0.3$ (by out choice of $\eps_0$ in~\eqref{eq:params-section7}). Let $a = \frac{1-\gamma}{\gamma}$ and $b = \frac{1-\eta}{\eta}$ and rewrite $\Gamma = \frac{ab+1}{a+b}$. The choice $\Gamma\geq a^\iota$ still suffices, since $\Gamma^{0.98}\beta^{2.02\eps_0/\delta}\geq a^{(0.98\cdot 3-2.02)\eps_0/\delta}=a^{0.92\eps_0/\delta}>1$. Requiring that $\frac{ab+1}{a+b} \geq a^\iota$ yields $b \geq \frac{a^{\iota+1} - 1}{a - a^\iota} \geq a^\iota$. This gives $\eta \lll \gamma^\iota$. However, observe the choice of $\eta$ can be tightened when $\iota \log(1/\gamma)$ is small. In this specific regime, expand $a^\iota = \exp(\iota \log a) \approx 1 + \iota \log(1/\gamma)$ and we instead obtain $b \geq 1 + \iota \log(1/\gamma)$. Substituting $b = \frac{1}{\eta} - 1$ yields $\frac{1}{\eta} \geq 2 + \iota \log(1/\gamma)$, which gives $\eta \leq \frac{1}{2} - \frac{\iota}{4} \log(1/\gamma)$. In summary, we need $\eta = \max(1/2 - \Theta(\log(1/\gamma)\eps_0/\delta), \Theta(\gamma^{3\eps_0/\delta}))$, precisely the choice in the lemma statement. This concludes the proof of~\Cref{lem:output-distribution-Pgamma}.
\end{proof}

    \begin{remark}\label{remark:why-need-A}
        \Cref{eq:final Lambda bound} is the only place where we need the noisy postselection operator $A_\eta$. If we do not have $A_\eta$, then $\Gamma=1$ and it is unclear if~\eqref{eq:final Lambda bound} could still be bounded (note that even the first inequality already uses the fact that $\Gamma^{0.98} \beta^{2.02 \eps_0/\delta} > 1$). In other words, without $A_\eta$, the tilting factor (originating from adversarial faults on good cliques) might increase the weight on low-weight strings. It is an interesting question whether one could control the logical distribution more tightly to avoid this issue.
    \end{remark}

\subsection{Proof of Omitted Claims}
We now prove the three claims used above.

\begin{proof}[Proof of~\Cref{claim:reduction-to-bad-cliques}] Define
\begin{align}\label{eq:Y}
    \ket{Y_{o,o',\mathsf{E},y,x}}:=\mathsf{E} \left( \encq(\ket{y}\ket{\xi_{o,o',\mathsf{E},y}} \ket{x}) \otimes \ket{\mathrm{rest}_{o,o',\mathsf{E}}} \right).
\end{align}
    By the assumption on the good cliques and the Knill-Laflamme condition on the code $\cQ$ it holds that 
    \begin{align}
         \braket{Y_{\tilde{o},\tilde{o}',\mathsf{E}',y',x'}}{Y_{o,o',\mathsf{E},y,x}}&\neq 0 \qquad \text{ only if } (y_G,x_G, \mathsf{E}_G) = (y'_G,x'_G, \mathsf{E}'_G).
    \end{align}
Since only the terms in~\eqref{eq:Psi-noncommuting-proof} with $o'= (\Rep(y)\oplus x)_{\cW}$ are non-vanishing, the above condition also enforces there are no cross-terms between different $o_G$ and $o'_G$. Therefore,
\begin{align}
    P(z) = \sum_{o,o'_G,\mathsf{E}_G} \left\| \bra{z} \ket{\Psi_{o,o'_G, \mathsf{E}_G}} \right\|^2.
\end{align}
\end{proof}

\begin{proof}[Proof of~\Cref{claim:Lambda-high}]
Suppress the fixed $o,o'_G,\mathsf{E}_G$ throughout. Expanding $\mathrm{NEnc}_{\mathrm{out}}$ in the definition of $\ket{\Phi}$ gives
\begin{align*}
    \ket{\Phi}
    &=\sum_{o'_B,\mathsf{E}_{Bw},y,x}
        c_{o'_B,\mathsf{E}_{Bw}}\alpha_{o'_B,\mathsf{E}_{Bw},y}
        \sqrt{p_\gamma(x_B)p_\gamma(x_G)}\\
    &\qquad{}\times\braket{o'_B}{\Rep(y_B)\oplus x_B}
        \ket{\Rep(y_G)\oplus x_G} \otimes\mathsf{E}_{Bw}
        \left(\encq(\ket{y}\ket{\xi_{o'_B,\mathsf{E}_{Bw},y}}\ket{x})
        \otimes\ket{\mathrm{rest}_{o'_B,\mathsf{E}_{Bw}}}\right),
\end{align*}
where $y=(y_B,y_G)$, $x=(x_B,x_G)$, and $\alpha_{o'_B,\mathsf{E}_{Bw},y}\ket{\xi_{o'_B,\mathsf{E}_{Bw},y}}=\bra{y}C\ket{\psi_{o'_B,\mathsf{E}_{Bw}}}$. The inner product enforces $o'_B=\Rep(y_B)\oplus x_B$. Grouping by the untouched encoded labels $(y_G,x_G)$ and comparing with the claim's decomposition, whose noise factor is $\sqrt{p_\gamma(x_G)}$, we obtain
\begin{align*}
    \Lambda_{y_G}\ket{Y_{y_G,x_G}}
    &=\sum_{\substack{o'_B,\mathsf{E}_{Bw},y_B,x_B\\o'_B=\Rep(y_B)\oplus x_B}}
        c_{o'_B,\mathsf{E}_{Bw}}\sqrt{p_\gamma(x_B)}
        \,\alpha_{o'_B,\mathsf{E}_{Bw},y}\\
    &\qquad{}\times\mathsf{E}_{Bw}
        \left(\encq(\ket{y}\ket{\xi_{o'_B,\mathsf{E}_{Bw},y}}\ket{x_B,x_G})
        \otimes\ket{\mathrm{rest}_{o'_B,\mathsf{E}_{Bw}}}\right).
\end{align*}
The right-hand side has an untouched normalized factor $\encq(\ket{x_G})$, so its norm is independent of $x_G$. Choosing $\Lambda_{y_G}\geq 0$ and normalizing the nonzero $\ket{Y_{y_G,x_G}}$, so that $\Lambda_{y_G}$ equal the norm of RHS. First, note that the states $\ket{Y_{y_G,x_G}}$ are orthogonal for distinct pairs $(y_G,x_G)$, since every $\mathsf{E}_{Bw}$ leaves the factor $\encq(\ket{y_G}\ket{x_G})$ untouched, and encoding preserves orthogonality. Next, let us calculate $|\Lambda_{y_G}|^2$.

We recall the reversible decoding unitary $U_{\cQ}$ from~\Cref{def:reversible-decoding}. Recall from~\Cref{claim:reversible-decoding-properties} that, for any state $\ket{\psi}$ and Pauli $\mathsf{E}$, $U_\cQ \mathsf{E}\Enc_\cQ(\ket{\psi}) = \omega_\mathsf{E}\, \mathsf{P}_\mathsf{E} \ket{\psi}\ket{\syn(\mathsf{E})}$ for Pauli $\mathsf{P}_\mathsf{E}$,  phase $\omega_\mathsf{E}$ and $\syn(\mathsf{E})$. Inserting $U_\mathcal{Q}$ inside the norm:
\begingroup
\small
\begin{align*}
    |\Lambda_{y_G}|^2
    &=\left\|\sum_{\substack{o'_B,\mathsf{E}_{Bw},y_B,x_B\\o'_B=\Rep(y_B)\oplus x_B}} c_{o'_B,\mathsf{E}_{Bw}}\sqrt{p_\gamma(x_B)}\,\alpha_{o'_B,\mathsf{E}_{Bw},y}\,U_{\cQ}\mathsf{E}_{Bw}\left(\encq(\ket{y}\ket{\xi_{o'_B,\mathsf{E}_{Bw},y}}\ket{x_B})\otimes\ket{\mathrm{rest}_{o'_B,\mathsf{E}_{Bw}}}\right)\right\|^2\\
    &=\left\|\sum_{\substack{o'_B,\mathsf{E}_{Bw},y_B,x_B\\o'_B=\Rep(y_B)\oplus x_B}} c_{o'_B,\mathsf{E}_{Bw}}\sqrt{p_\gamma(x_B)}\,\alpha_{o'_B,\mathsf{E}_{Bw},y}\,\mathsf{P}_{\mathsf{E}_{Bw}}\left(\ket{y}\ket{\xi_{o'_B,\mathsf{E}_{Bw},y}}\ket{x_B}\otimes\ket{\mathrm{rest}'_{o'_B,\mathsf{E}_{Bw}}}\right)\right\|^2,
\end{align*}
\endgroup
where we absorb the phase, syndrome registers, and workspace error into $\ket{\mathrm{rest}'_{o'_B,\mathsf{E}_{Bw}}}$. The operator $\mathsf{P}_{\mathsf{E}_{Bw}}$ is a Pauli acting only on $\ket{x_B}$, on $\ket{y_B}$, and on a small subset $F$ of qubits in $\ket{\xi_{o'_B,\mathsf{E}_{Bw},y}}$, namely those whose code blocks are held by bad cliques. It acts trivially on $\ket{y_G}$ and on the remaining data qubits. Recall that each clique holds $O(1)$ data code blocks of the logical circuit $C$, so $|B\cup F|=O(|B|)=O(M\mu(\cE)/\pl n)$.

Now we can relate the above expression to a faulty execution of the logical sequential measurement circuit $C$ with faults on the set $B\cup F$. Define $O_{o'_B,\mathsf{E}_{Bw}}$ as the following linear map acting on the qubits of $\ket{y_B}$ and the $F$-qubits of $\ket{\xi_{o'_B,\mathsf{E}_{Bw},y}}$, and appending the noise and residual registers:
\begin{align*}
    O_{o'_B,\mathsf{E}_{Bw}}\ket{y_B}\ket{f}_F
    :=c_{o'_B,\mathsf{E}_{Bw}}\mathsf{P}_{\mathsf{E}_{Bw}}
    \left(\ket{y_B}\ket{f}_F
    \otimes\sum_{x_B:\ o'_B=\Rep(y_B)\oplus x_B}
        \sqrt{p_\gamma(x_B)}\ket{x_B}\otimes\ket{\mathrm{rest}'_{o'_B,\mathsf{E}_{Bw}}}\right).
\end{align*}
Extending $O_{o'_B,\mathsf{E}_{Bw}}$ by the identity on the remaining logical registers, we obtain
\begin{align*}
    |\Lambda_{y_G}|^2
    &=\left\|\sum_{y_B}\sum_{o'_B,\mathsf{E}_{Bw}}
        O_{o'_B,\mathsf{E}_{Bw}}\alpha_{o'_B,\mathsf{E}_{Bw},y}
        \ket{y}\ket{\xi_{o'_B,\mathsf{E}_{Bw},y}}\right\|^2\\
    &=\left\|\bra{y_G}\sum_{o'_B,\mathsf{E}_{Bw}}
        O_{o'_B,\mathsf{E}_{Bw}}C\ket{\psi_{o'_B,\mathsf{E}_{Bw}}}\right\|^2.
\end{align*}
This is precisely the (unnormalized) probability of obtaining $y_G$ as the measurement outcome of the good logical qubits in a coherent faulty execution of the logical sequential measurement circuit $C$ on the input states $\ket{\psi_{o'_B,\mathsf{E}_{Bw}}}$ and faults on the set $B\cup F$. Since $|B\cup F|=O(M\mu(\cE)/\pl n)\lll\epsilon M$, we can apply~\Cref{lem:dl} to conclude that there exist constants $c'\ll c$ (depending on $r,r',r''$ in~\Cref{lem:dl}) such that
\begin{align*}
    \Lambda_{\mathrm{high}}
    &=\sum_{y_G:\ |y_G|\geq c\epsilon M}|\Lambda_{y_G}|^2
    \geq (1-2^{-r''\epsilon M})
    \sum_{y_G\in\{0,1\}^G}|\Lambda_{y_G}|^2,\\
    \Lambda_{\mathrm{low}}
    &=\sum_{y_G:\ |y_G|<c'\epsilon M}|\Lambda_{y_G}|^2
    \leq 2^{-r''\epsilon M}
    \sum_{y_G\in\{0,1\}^G}|\Lambda_{y_G}|^2.
\end{align*}
\end{proof}

\begin{proof}[Proof of~\Cref{claim:physical-ditribution-each-component}] Let $g=(o,o'_G, \mathsf{E}_G)$. According to~\Cref{claim:Lambda-high}, we can write
\begin{align}
\ket{\Phi_{g}} = \sum_{y_G \in \{0,1\}^{G}} \sum_{x_G \in \{0,1\}^{Gd} } \sqrt{p_\gamma(x_G)} \Lambda_{g,y_G} \ket{\Rep(y_G) \oplus x_G} \otimes \ket{Y_{g, y_G,x_G}}
\end{align}
so that $\ket{\Psi_{g}}= A_\eta^{\otimes Md} (\ketbra{o}{o'_G} \otimes \mathsf{E}_G)  \cdot \ket{\Phi_{g}}$.

Hitting $\ket{\Psi_{g}}$ with $\ketbra{z}_{\cW^c}$ from the left yields
\begin{align}
    \ketbra{z} \ket{\Psi_{g}}&= \sum_{y_G,x_G}  \sqrt{p_{1-\eta}(o,z)} \ketbra{z}_{\cW^c} \ketbra{o}{o'_G}  \cdot \Lambda_{g,y_G} \sqrt{p_\gamma(x_G)} \ket{\Rep(y_G) \oplus x_G} \otimes \mathsf{E}_G \ket{Y_{g,y_G,x_G}} \tag*{\hphantom{1cm}} \\
    &= \sum_{y_G}  \sqrt{p_{1-\eta}(o,z)}  \sqrt{p_\gamma(\Rep(y_G)_{\cW_G} \oplus o'_G) p_\gamma(\Rep(y_G)_{\cW^c} \oplus z) }  \Lambda_{g,y_G} \ket{o, z} \otimes \mathsf{E}_G \ket{Y_{g,y_G,x_G}} \tag*{\hphantom{1cm}} \\
    &=  \sum_{y_G} \beta^{|\Rep(y_G)_{\cW_G} \oplus o'_G|/2}   \tilde{\Lambda}_{g,y_G} \sqrt{p_{1-\eta}(z) p_\gamma(\Rep(y_G)_{\cW^c} \oplus z) } \ket{o, z} \otimes \mathsf{E}_G \ket{Y_{g,y_G,x_G}},\label{eq:single-ket-term}
\end{align}
where
\begin{align}
    \beta &:= \frac{\gamma}{1-\gamma}, \qquad \text{and} \qquad \tilde{\Lambda}_{g,y_G} := \Lambda_{g,y_G} \sqrt{p_{1-\eta}(o)} (1-\gamma)^{|\cW_G|/2}.
\end{align}
Notice that the distribution of $z$ is product for a fixed $y_G$. Furthermore, at location $i \in \cW^c$ where $\Rep(y)_i = 1$, the product factor is
 $$p_{1-\eta}(z_i) p_{1-\gamma}(z_i)= ((1-\gamma)(1-\eta) + \gamma \eta ) p_{\gamma''}(z_i),$$
and at location $i$ where $\Rep(y)_i = 0$, the product factor is 
$$p_{1-\eta}(z_i) p_{\gamma}(z_i)= (\gamma(1-\eta) + \eta (1-\gamma) ) p_{\gamma'}(z_i),$$ 
with $\gamma' ,\gamma''$ defined in~\Cref{def:gamma'}. Thus, we can further simplify~\eqref{eq:single-ket-term} to
\begin{align}
   \sum_{y_G}  \sqrt{P_{\Rep(y_G)_{\cW^c}}(z)} \Gamma^{ |\Rep(y_G)_{\cW^c}|/2}  \beta^{|\Rep(y_G)_{\cW_G} \oplus o'_G|/2}   \Lambda'_{g,y_G}  \ket{o, z} \otimes \mathsf{E}_G \ket{Y_{g,y_G,x_G}},
\end{align}
where
 \begin{align}
        \Lambda'_{g,y_G} &:=  \tilde{\Lambda}_{g, y_G}  (\gamma(1-\eta) + \eta (1-\gamma))^{|\cW^c|/2} \qquad \text{and} \qquad \Gamma := \frac{(1-\gamma)(1-\eta) + \gamma \eta}{\gamma(1-\eta) + \eta (1-\gamma)},\\
        P_{\Rep(y_G)_{\cW^c}}(z) &:=  \bigotimes_{i: \Rep(y_G)_i=1} p_{\gamma''}(z_i) \bigotimes_{i: \Rep(y_G)_i=0} p_{\gamma'}(z_i)   \qquad \text{with } \gamma'' ,\gamma' \text{ defined in~\Cref{def:gamma'}}.
    \end{align}

    Therefore,
    \begin{align}
        \left\| \bra{z} \ket{\Psi_{o,o'_G, \mathsf{E}_G}} \right\|^2 = \sum_{y_G} P_{\Rep(y_G)_{\cW^c}}(z) \Gamma^{ |\Rep(y_G)_{\cW^c}|}  \beta^{|\Rep(y_G)_{\cW_G} \oplus o'_{G}|}  |\Lambda'_{g,y_G}|^2,
    \end{align}
    where we used $\ket{Y_{g,y_G,x_G}}$'s are orthogonal (\Cref{claim:Lambda-high}).
\end{proof}

\subsection{Soundness of the Tensor-Network State}\label{sec:soundness-nc-protocol}

\Cref{lem:output-distribution-Pgamma} asserts that the distribution of the good measurement bits in the good cliques is essentially a noisy repetition encoding of an underlying logical distribution $q(y_G)$ that is mostly supported on $|y_G| \ggg \epsilon M$. Thus, we can follow the intuition outlined in~\Cref{sec:serial-repetition} to show that a large fraction of the vertex threshold checks (\Cref{def:vertex-threshold-check}) will be violated with high probability. It suffices to consider a fixed string $y_G$ of Hamming weight $\ggg \epsilon M$. We need to analyze the violations of the vertex checks when the Bernoulli noise is added on top of the string $\Rep(y_G)$.

An easy solution to deal with the Bernoulli noise is to pick $\gamma \lll 1/\mathrm{poly}(C_{\max})$, where $C_{\max}:=\max_v |M_v| = \pl n$ is the maximum number of measurement bit held by a vertex. Accordingly we pick $\eta = \gamma^{0.3}$, so that $\gamma'\leq  O(\gamma^{0.5}) \leq 1/\mathrm{poly}(C_{\max})$ and $\gamma'' \geq 1 - O(\gamma) \geq 1 - 1/\mathrm{poly}(C_{\max})$ (see~\Cref{def:gamma'}). Then we effectively recover the setting of~\Cref{sec:serial-repetition} where the good vertices of the good cliques display the correct logical measurement value, and simply using the vertex OR checks and the analysis of serial repetition there would suffice to establish the soundness. However, this choice of $\eta$ means the $A_\eta$ operator is enforcing a rather strong postselections ($\eta = 1/\pl n \rightarrow 0$), which seems undesirable.

Instead, in the high soundness regime $\epsilon \geq 1/\pl n$, it turns out that we can pick more `friendly' parameters: $\gamma, \gamma', 1-\gamma''=\Theta(1)$ and $\eta \approx 1/2$. In this regime, the Bernoulli noise is significant and we need to use the vertex threshold checks. We will use the following lemma, which is a strengthening of the sampling argument in the proof of~\Cref{lem:serial-rep}.

\begin{lemma}[HDX sampling lemma]
    \label{lem:hdx-sampling-nc-protocol}
Let $X$ be a $d$-partite $\lambda$-HDX with $\lambda \leq 1/\poly(d)$ and $\cE \subseteq X(2)$ with $\mu(\cE) \leq \eps_0$. Let $\cB_{\cE}$ be the set of bad cliques, i.e., those containing more than $\eps_0$ fraction of edges from $\cE$. Let $\cS$ be any set of cliques with measure $\mu$, and furthermore assume that each of these is good (i.e. contains at most $\eps_0$-fraction bad edges). If $\mu \gtrsim \Omega(1/\log^c n)$ and $d= \Omega(c^2\log\log^2 n)$, then at least an $\Omega(1)$-fraction of vertices $v$ satisfy both the properties below,
\begin{enumerate}
\item The link of $v$ contains at least $\mu/10$-fraction of cliques that belong to $\cS$ and where $v$ is also good with respect to them.
\item The link of $v$ has at most $O(1/\log^c n)$-fraction of cliques that are bad overall, or are bad with respect to $v$.
\end{enumerate}
\end{lemma}

\begin{proof}
We first prove a simpler claim to lower bound the fraction of vertices satisfying property (1).
\begin{claim}
\label{lem:hdx-sampling-good-cliques}
In the setting of the lemma, consider the bipartite graph between $H=(X(d), X(1))$. Let $T$ be the set of vertices $v$ such that, it is part of a lot of cliques in $\cS$ where it is also good:
\[T = \left\{v: \frac{|\{i \in [N]: v\in D_i\setminus W_i \cap i\in \cS\}|}{|\{i \in [N]: v\in D_i\}|} \geq \mu/10\right\}.\]
Then $\mu(T) \geq 1/100$.
\end{claim}

\begin{proof}
We start by showing that $\cS$ has a lot of edges inside $T$ (in the bipartite graph $H$) by a simple counting argument. First note that since each clique in $\cS$ is good, it has at most $2\eps_0/\delta\leq 1/10$-fraction of bad vertices in it. Let $\gamma$ denote the average fractional-degree of a vertex in $\cS$ inside $\overline{T}$. Then,
\[\mu \gamma = \frac{|E(S,\overline{T})|}{|E(H)|} \leq \frac{\mu}{10}(1-\mu(T))+\frac{2\eps_0}{\delta}\mu \leq \mu/5,\]
implying that $\gamma \leq 1/5$. So by Markov's inequality we get that at least $2/5$-fraction of cliques in $\cS$ have more than $1/2$-fraction of their neighbors in $\overline{T}$. Therefore at least $3/5$-fraction of $\cS$ has at least $1/2$-fraction of their neighbors inside $T$.

We are now in a position to apply the HDX concentration lemma to get that $T$ itself must be large. Assume for the sake of contradiction that $\mu(T) \leq 1/100$. Then~\Cref{lem:chernoff-hdx} implies that,
\[
3\mu/5 \leq \Pr_{C \sim X(d)}[\mu(T|_C) \geq 1/2] \lll \exp(-\sqrt{d}) \leq 1/\log^c n,
\]
which is a contradiction to the fact that $\mu \geq \Omega(1/\log^c n)$.
\end{proof}

To bound the other properties of $v$, we define three events on the vertices of the HDX.
\begin{enumerate}
\item Event 1: A vertex has at least $\mu/10$-fraction of cliques from $\cS$ in their local view and where the vertex is itself good too. By the sampling claim above we know that this holds for at least $1/100$-fraction of vertices.
\item Event 2: A vertex $v$ touches less than $1000\mu(\cE)$-fraction bad edges. By Markov's inequality there are only a $1/1000$-fraction that don't satisfy this.
\item Event 3: A vertex $v$ has less than $1000\mu(\cB)$-fraction of bad cliques in their local view. Again by Markov's inequality only a $1/1000$-fraction of vertices don't satisfy this.
\end{enumerate}

Taking a union bound over events (1), (2) and (3) gives that there are at least $\frac{1}{100}-\frac{2}{1000} = \Omega(1)$-fraction of vertices that satisfy all three  events. Fix any such vertex $v$. Event (1) directly implies property 1. Event (2) implies that there are at most $O(\eps)$-fraction bad vertices (corresponding to the bad edges) in its link. Apply the HDX-Chernoff~\Cref{lem:chernoff-hdx} in the link of $v$ to get that at most $1/\log^c n$-fraction of cliques in its link contain more than $\delta/2$-fraction of these bad vertices (since $\mu(\cE)\lll \eps_0 \lll \delta$). Reinterpreted, this just means that only an $O(1/\log^c n)$-fraction of the cliques in the link of $v$ are bad with respect to $v$ (i.e. $v$ has too many bad edges incident on it inside the clique). Event (3) says that only an $O(\mu(\cB))$-fraction of cliques in the link of $v$ are bad overall (have too many bad edges), and again using the HDX Chernoff lemma we get that this is at most $O(1/\log^c n)$. Thus we get that every such vertex satisfies the properties we want in its local view, thus proving the lemma.

\end{proof}

We can now state and prove the soundness statement for the tensor-network state $\ket{\Psi}$.

\begin{theorem}[Soundness of $\cP_\gamma$]
    \label{thm:soundness-nc-protocol}
    Suppose the input Hamiltonian has unnormalized ground energy $\epsilon M$. Fix \(\cE\subseteq X(2)\) with \(\mu(\cE)  \lll \min(\epsilon \pl n,\ \eps_0)\).
    Consider an adversarially faulty execution of the distributed protocol $\cP_\gamma$ for the sequential measurement circuit $C$ of $H_\mathrm{in}$, so that the computation state can be expressed as
   $
        \ket{\Psi} = \sum_{\sF\in \mathsf{PAULI}^{\cE}} c_{\sF} A_\eta^{\otimes Md} (\cP_\gamma)_\sF \ket{\phi_\sF}
    $
    and is a normalized state.
    Then,
    \begin{align}
        \mathbb{E}_{v} \Tr(\Pi_v \ketbra{\Psi}) \geq (1-(1-\gamma'') C_{\max}) \epsilon {\pl n},
    \end{align}
    where $\Pi_v$ is the vertex OR operator from~\Cref{def:vertex-threshold-check}, $C_{\max} = \max_v |M_v|$ is the maximum number of measurement bits held by a vertex.

    In the regime $\epsilon \geq  1/\pl n$, where $C_{\min}:=\min_v |M_v|$, there exist choices $\gamma, \gamma', 1-\gamma''= \Theta(1)$, $\eta= 1/2 - \Theta(\eps_0/\delta)$, and $\tau_\mathrm{out} = \gamma' + \Theta(\epsilon)$ such that
    \begin{align}
        \mathbb{E}_{v} \Tr(\Pi_v^{\tau_\mathrm{out}} \ketbra{\Psi}) \geq 1/200.
    \end{align}
\end{theorem}

\begin{proof}
    By~\Cref{lem:output-distribution-Pgamma}, the distribution of the measurement bits held by the good vertices of the good cliques has the form
     \begin{align}
        P(z) = \sum_{y_G} q(y_G)  P_{\Rep(y_G)_{\cW^c}}(z),
     \end{align}
     where $q(y_G)$ is a distribution over the logical measurement bits held by the cliques outside of $LC(L_1,...,L_\ell)$, satisfying $\Pr_{y_G \sim q}[|y_G| \ggg \epsilon M] \geq 1-2^{-\poly(M)}$, and
     \begin{align}
         P_{\Rep(y_G)_{\cW^c}}(z) =  \bigotimes_{j \in \cW^c: \Rep(y_G)_j=1} p_{\gamma''}(z_j) \bigotimes_{j \in \cW^c: \Rep(y_G)_j=0} p_{\gamma'}(z_j).
     \end{align}
    It suffices to show that for any fixed $y_G$ with $|y_G| \ggg \epsilon M$, the distribution $P_{\Rep(y_G)_{\cW^c}}(z)$ violates the claimed fraction of vertex checks in expectation. Let $\cS$ be the set of cliques in the support of $\Rep(y_G)$ that computes a $1$, then $\mu(\cS) \ggg \epsilon$.
    
    Let us first consider the OR check version. Applying~\Cref{lem:serial-rep} we find that the set $T$ of vertices that see at least one 1 in their local view has measure at least $\mu(T)=\epsilon/\lambda = \epsilon \pl n$. Since the vertex OR check is violated if a vertex sees at least one 1, we conclude that $\mathbb{E}_v \Tr(\Pi_v \ketbra{\Psi}) \geq (1-(1-\gamma'') C_{\max})\epsilon \pl n $, where the union bound factor $(1-\gamma'') C_{\max}$ accounts for the noise.

    Next we consider the threshold version. We instead apply~\Cref{lem:hdx-sampling-nc-protocol} to conclude that the set $T$ of vertices who sees at least $\kappa \ggg \epsilon$ fraction of 1's in their local view has measure $\mu(T) \geq 1/100$. We analyze the energy of the checks on these vertices. For each vertex $v \in T$, the expected fraction of 1's in its local view is at least $(1- 1/\pl n )\gamma' + \kappa  \gamma'' \geq \gamma' + \kappa \gamma''/2$ . Suppose $\gamma', 1-\gamma'' =\Theta(1)$, then by the additive Chernoff bound, the probability that $v$ sees less than $\tau_\mathrm{out}=\gamma' + \kappa/10$ fraction of 1's is at most $\exp(-\Omega(\epsilon^2 C_{\min}))$, where $C_{\min}$ is the minimum number of measurement bits held by a vertex. Thus, if $\epsilon \ggg 1/\sqrt{C_{\min}}$, then the probability that $v$ violates the threshold check is at least $1/2$. Our HDX~\Cref{thm:hdx-construction} has $C_\mathrm{min}= (1/\lambda)^{O(d^2)} > 1/\pl n$, thus this condition holds in the regime $\epsilon \geq 1/\pl n$. Since this holds for at least $1/100$-fraction of vertices, we conclude that $\mathbb{E}_v \Tr(\Pi_v^{\tau_\mathrm{out}} \ketbra{\Psi}) \geq 1/200$.
     Finally, note that when $\gamma=\Theta(1)$, we can pick $\eta$ to be a constant arbitrarily close to $1/2$ by picking $\eps_0/\delta$ to be small, see~\Cref{def:gamma'}.
\end{proof}

%% file: protocol-to-ham.tex
\section{Converting a Distributed Computing Protocol into an LH}\label{sec:ckt-to-H}
In this section we describe the mapping from a distributed computing protocol to a local Hamiltonian. For pedagogical clarity, we begin by briefly illustrating the circuit-to-Hamiltonian procedure from~\cite{Anshu_2024} with a toy example circuit on one qubit in~\Cref{subsec:1-qubit-ABN}. Then,~\Cref{sec:comm-ckt-to-H} adapts the procedure to distributed computing protocols. 
Next, we prove the `combinatorial soundness' of the protocol-to-Hamiltonian mapping in~\Cref{subsec:comb-soundness}. 
Finally,~\Cref{sec:nc-ckt-to-H} modifies the mapping to incorporate the postselection gates used in the non-commuting case from~\Cref{sec:nc-protocol}.

\subsection{One-Qubit Toy Example}\label{subsec:1-qubit-ABN}
The Hamiltonian that the circuit-to-Hamiltonian procedure from~\cite{Anshu_2024} produces can be thought of as a way to enforce a computation by \emph{gate teleportation}. To illustrate the idea, suppose the input circuit is
\[
W = U_D \cdots U_2 U_1
\]
acting on a single-qubit input state $\ket{\psi}$. For each gate $U_t$, prepare the two-qubit resource state
\[
\ket{\Phi_{U_t}}_{a_t b_t} = (\Id \otimes U_t)\ket{\Phi_I},
\qquad
\ket{\Phi_I}=\frac{1}{\sqrt2}(\ket{00}+\ket{11}).
\]
Arrange these resource states in a line, with $\ket{\psi}$ at the left end. If, for each $t$, one projects the neighboring pair of qubits in registers $b_{t-1}$ and $a_t$ onto $\ket{\Phi_I}$, then the input state is teleported from left to right through the chain, and the gate $U_t$ is applied at the $t$th step. After all projections have been applied, the final qubit contains $W\ket{\psi}$ (up to the normalization factor coming from successful postselection on $\ket{\Phi_I}$ during each Bell basis measurement).

This suggests a parent-Hamiltonian construction. Let
\[
\ket{\Omega}
=
\ket{\psi}_{b_0}\otimes \bigotimes_{t=1}^D \ket{\Phi_{U_t}}_{a_t b_t},
\qquad
\ket{\Phi_{U_t}}=(\Id\otimes U_t)\ket{\Phi_I},
\]
where $b_0$ is the input qubit. Before any teleportation constraints are imposed, $\ket{\Omega}$ is a product state, and hence the unique ground state of the frustration-free Hamiltonian
\[
H_0
=
(\Id-\ket{\psi}\bra{\psi})_{b_0}
+
\sum_{t=1}^D
(\Id-\ket{\Phi_{U_t}}\bra{\Phi_{U_t}})_{a_t b_t}.
\]
If the Bell-pair projections on the pairs $b_{t-1}a_t$ were replaced by invertible local maps $L_t$, then the computation state would be proportional to $L_D\cdots L_1\ket{\Omega}$, and a parent Hamiltonian for that state could be obtained by conjugating each term of $H_0$ by the inverse maps acting on the neighboring pairs that overlap that term.

The obstacle is that projection onto $\ket{\Phi_I}$ is not invertible. ABN resolve this issue by replacing the exact EPR projection with the injective perturbation
\[
Q
=
\ket{\Phi_I}\bra{\Phi_I}
+
\delta \sum_{P\in\{X,Z,XZ\}} \ket{\Phi_P}\bra{\Phi_P},
\qquad
\ket{\Phi_P}:=(\Id\otimes P)\ket{\Phi_I}.
\]
This map is invertible,\footnote{It is invertible but not trace-preserving. Nevertheless, ground states of the resulting Hamiltonians can still be prepared efficiently \cite[Section~2]{Anshu_2024}.} so the parent-Hamiltonian construction just proposed now goes through. Operationally, replacing the ideal Bell projection by $Q$ means that each teleportation step is allowed, with amplitude controlled by $\delta$, to produce one of the three nontrivial Bell outcomes as well as the ideal one. Equivalently, each step introduces a small amplitude for a Pauli fault. Thus the ABN ground state encodes not a perfectly noiseless execution of $W$, but a slightly faulty execution that can be handled by fault tolerance.

For clarity of nomenclature, we remark that the `resource states' we have been referring to are often called \emph{Choi states}, and that the process of computation by gate teleportation can be recast graphically as a \emph{tensor network}. More specifically, a computation by gate teleportation can be represented as a tensor network in which entangled resource states and measurements are tensors and carrying out the teleportation corresponds to contracting those tensors~\cite{biamonte2017tensornetworksnutshell}.

\subsection{The Tensor-Network Circuit-to-Hamiltonian Mapping}\label{sec:comm-ckt-to-H}
We give a modification of the circuit-to-Hamiltonian mapping of~\cite{Anshu_2024} tailored to distributed computing protocol on an unweighted multigraph $G=(V,E)$.

\paragraph{Tensor networks.}
We will encode the distributed computing protocol $\mathcal{P}$ into a tensor network. The qubits are grouped into columns, where each column corresponds to a layer of the circuit $\mathcal{P}$. The first column contains the input state, and each subsequent column contains the Choi states of the gates applied at that layer. Consider a $k$-qubit gate $U$, the Choi state is a $2k$-qubit state that encodes the action of the gate
\begin{equation}
    \ket{\Phi_U}_{AB} = (I_{A}\otimes U_B)  \left( \bigotimes_{i=1}^k \frac{\ket{00}_{A_i B_i} + \ket{11}_{A_i B_i}}{\sqrt{2}}  \right) .
\end{equation}
We will often denote $\ket{\Phi_U} = (I \otimes U)\ket{\Phi_I}^{\otimes k}$ for brevity, where $\ket{\Phi_I}=\left( \ket{00} + \ket{11} \right)/\sqrt{2}$. See \Cref{fig:state} for a diagrammatic representation of this state.

\begin{figure}[h]
    \centering
    \begin{tikzpicture}[scale=0.6]
      \path [draw=blue,snake it, thick] (-0.5,2.5) -- (1.5,2.5); %
      \path [draw=blue,snake it, thick] (-0.5,0.5) -- (1.5,0.5); %
      \draw[ultra thick] (0.9,0) rectangle (2.1,3) node[right] {$U$};
      \draw [fill=black] (-0.5,0.5) circle [radius=0.1] node[above] {$A_2$}; %
      \draw [fill=black] (-0.5,2.5) circle [radius=0.1] node[below] {$A_1$}; %
      \draw [fill=black] (1.5,0.5) circle [radius=0.1] node[above] {$B_2$}; %
      \draw [fill=black] (1.5,2.5) circle [radius=0.1] node[below] {$B_1$}; %
    \end{tikzpicture} 
    \caption{
        The Choi state $\ket{\Phi_U}$ associated with a 2-qubit gate.
        The registers $A_iB_i$, for $i=1,2$, are in the EPR state $\ket{\Phi_I}$, indicated by the blue wavy lines.
        The unitary $U$ is applied to qubits $B_i$ (black box).
    }
    \label{fig:state}
\end{figure}
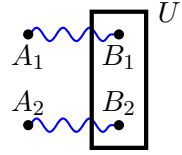

The state $\ket{\Phi_U}$ allows implementing the gate $U$ via teleportation and postselection. For example, the application of a 1-qubit gate $U$ on an input state $\ket{0}$ is simulated by projecting qubits 1, 2 of the joint system $\ket{0}_1\otimes \ket{\Phi_U}_{2,3}$ onto the EPR state $\ket{\Phi_I}$, obtaining $(U\ket{0})_3$. Extending the previous idea to the entire circuit, we can encode any quantum circuit into a tensor network. See~\Cref{fig:tensor-network} for an illustration.

\begin{figure}
    \centering
\begin{tikzpicture}[scale=0.4]
  \clip (-0.25,-0.5) rectangle (16.73,7.6);
    
      \draw [fill=black] (1.5,0.5) circle [radius=0.1] ; %
      \draw [fill=black] (1.5,2.5) circle [radius=0.1] ; %

       \draw[black,rounded corners=5pt] (1,-0.25) rectangle (2,3.25) node[midway,left] {$\ket{\xi}\ $}; %

      \draw [fill=black] (1.5,0.5+4) circle [radius=0.1] node[left] {$\ket{0}$}; %
      \draw [fill=black] (1.5,2.5+4) circle [radius=0.1] node[left] {$\ket{0}$}; %

      \draw[gray,rounded corners,fill=gray, opacity=0.1] (0.75,-0.25) rectangle (4,1.25); %
      \draw[gray,rounded corners,fill=gray, opacity=0.1] (0.75,1.75) rectangle (4,3.25); %
      \draw[gray,rounded corners,fill=gray, opacity=0.1] (0.75,-0.25 +4) rectangle (4,1.25 + 4); %
      \draw[gray,rounded corners,fill=gray, opacity=0.1] (0.75,1.75+4) rectangle (4,3.25+4); %

  \foreach \x in {8} {
    \foreach \y in {0,4,8} {
      \path [draw=blue,snake it, thick] (\x-0.5,\y+2.5) -- (\x+1.5,\y+2.5); %
      \path [draw=blue,snake it, thick] (\x-0.5,\y+0.5) -- (\x+1.5,\y+0.5); %
      \draw[ultra thick] (\x+1,\y) rectangle (\x+2,\y+3);
      \draw [fill=black] (\x-0.5,\y+0.5) circle [radius=0.1]; %
      \draw [fill=black] (\x-0.5,\y+2.5) circle [radius=0.1]; %
      \draw [fill=black] (\x+1.5,\y+0.5) circle [radius=0.1]; %
      \draw [fill=black] (\x+1.5,\y+2.5) circle [radius=0.1]; %
      \draw[gray,rounded corners,fill=gray, opacity=0.1] (\x+0.75,\y-0.25) rectangle (\x+4,\y+1.25); %
      \draw[gray,rounded corners,fill=gray, opacity=0.1] (\x+0.75,\y+1.75) rectangle (\x+4,\y+3.25); %
    }
  }
  \foreach \y in {-2,2,6,10} {
    \foreach \x in {4,12} {
      \path [draw=blue,snake it, thick] (\x-0.5,\y+2.5) -- (\x+1.5,\y+2.5); %
      \path [draw=blue,snake it, thick] (\x-0.5,\y+0.5) -- (\x+1.5,\y+0.5); %
      \draw[ultra thick] (\x+1,\y) rectangle (\x+2,\y+3);
      \draw [fill=black] (\x-0.5,\y+0.5) circle [radius=0.1]; %
      \draw [fill=black] (\x-0.5,\y+2.5) circle [radius=0.1]; %
      \draw [fill=black] (\x+1.5,\y+0.5) circle [radius=0.1]; %
      \draw [fill=black] (\x+1.5,\y+2.5) circle [radius=0.1]; %

    }
    
    \draw[gray,rounded corners,fill=gray, opacity=0.1] (4+0.75,\y-0.25) rectangle (4+4,\y+1.25); %
    \draw[gray,rounded corners,fill=gray, opacity=0.1] (4+0.75,\y+1.75) rectangle (4+4,\y+3.25); %
    
  }
\end{tikzpicture}
    \caption{
        Tensor network representation of a quantum circuit with depth $D=3$ with brickwork structure. 
        \textbf{Input:}
        The first column (black dots) represents the input qubits, which could be restricted to be in some subspace (e.g., $\ketbra{0}{0}$), or to be arbitrary (any input state $\ket{\xi}$).
        \textbf{Operations:} The subsequent columns contain the Choi states of the gates (black dots connected by blue wavy lines) applied at the corresponding circuit locations. The EPR projector $\ketbra{\Phi_I}{\Phi_I}$ or its $\delta$-noisy version $Q = \ketbra{\Phi_I}{\Phi_I} + \delta\sum_{P \in \{X,Y,Z\}} \ketbra{\Phi_P}{\Phi_P}$ (gray boxes) is applied on pairs of qubits to perform the gate. In addition to Choi states, we may also have $\ket{0}$'s, representing mid-circuit initializations.
        \textbf{Output:}
        The output includes the qubits at the last column and any qubits that are `discarded' throughout the circuit.
    }
    \label{fig:tensor-network}
\end{figure}
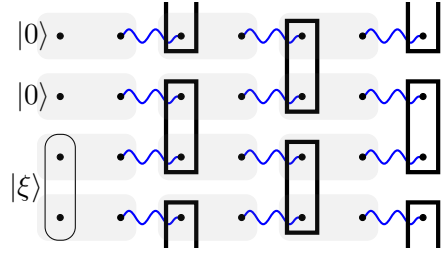

We now specify to distributed computing protocol $\mathcal{P}$ as in~\Cref{def:protocol} that takes in a $w$-qubit input state $\ket{\xi}$ (think of as a QMA witness). For each vertex $v$, let $\mathsf{init}(v)$ be the total number of $\ket{0}$ qubits initialized by $v$ throughout the protocol. Define the following product state
\begin{equation}
    \ket{\Phi_{\mathcal{P}, \xi}}= \ket{\xi}_\mathrm{in} \otimes \left(\bigotimes_{v \in V} \ket{0}^{\mathsf{init}(v)} \right) \left( \bigotimes_{U \in \mathcal{P}} \ket{\Phi_U} \right),
    \label{eq:tensor-network-state}
\end{equation}
where $U$ are gates acting within a vertex or an edge of $G=(V,E)$. So the total number of qubits is $N_{TN} = O(nm)$.

    The tensor network state encoding the execution of $\mathcal{P}$ is obtained by projecting appropriate pairs of qubits of $\ket{\Phi_{\mathcal{P}, \xi}}$ onto $\ketbra{\Phi_I}{\Phi_I}$ (gray boxes in~\Cref{fig:tensor-network}). Specifically, for each edge $e=(u,v) \in E$, let $\mathsf{supp}(e)$ denote the support of all gray boxes corresponding to the Choi states crossing $e$, and we define
\begin{align}
    P_e = \ketbra{\Phi_I}{\Phi_I}^{\mathsf{supp}(e)}.
\end{align}
Intuitively, this `edge' projector enforces qubit transmissions between $u$ and $v$. Additionally, for each vertex $v \in V$, let $\operatorname{supp}(v)$ denote the support of all gray boxes intersecting with the Choi states fully contained within $v$, \emph{except} the ones that are already included in the edge projectors. We define the `vertex' projector as
\begin{align}
    P_v = \ketbra{\Phi_I}{\Phi_I}^{\mathsf{supp}(v)}.
\end{align}
Intuitvely, a `vertex' projector enforces local computations within $v$. Note that the supports of the edge and vertex projectors are disjoint.

Then our tensor network encoding of $\mathcal{P}$ upon input $\ket{\xi}$ is 
\begin{align}
    \left( \bigotimes_{e \in E} P_e \right)  \left( \bigotimes_{v \in V} P_v \right) \ket{\Phi_{\mathcal{P},\xi}} = \ket{\Phi_I}^{\otimes N_\mathrm{EPR}} \mathcal{P} (\ket{\xi}).
\end{align}

\paragraph{Parent Hamiltonian.} The key idea of~\cite{Anshu_2024} is to obtain a local `parent Hamiltonian' corresponding to a perturbed version of the described tensor network. Here, we consider a modification of their Hamiltonian that makes use of the grouping of qubits and gates according to the graph $G$ described earlier. For each edge $e \in E$, define a perturbed version of the projector $P_e$
\begin{align}
    Q_e := P_e + \frac{\delta_\mathrm{edge}}{\sqrt{4^{|\operatorname{supp}(e)|}-1}} ( I - P_e) = P_e + \frac{\delta_\mathrm{edge}}{\sqrt{4^{|\operatorname{supp}(e)|}-1}} \sum_{P \in  \{X,Y,Z\}^{\operatorname{supp}(e)} \backslash I } \ketbra{\Phi_{P}}{\Phi_P}
\end{align}
This perturbed projection can be interpreted as performing the qubit transmissions on the edge $e$, but allowing for Pauli errors to occur with probability proportional to $O(\delta_\mathrm{edge}^2)$. Similarly, for each vertex $v \in V$, we define a perturbed version of the projector $P_v$ as
\begin{align}
    Q_v = P_v + \frac{\delta_\mathrm{vertex}}{\sqrt{4^{|\operatorname{supp}(v)|}-1}} \sum_{P \in  \{X,Y,Z\}^{\operatorname{supp}(v)} \backslash I } \ketbra{\Phi_{P}}{\Phi_P}.
\end{align}

We can now define the Hamiltonian. For $v$, let $\Lambda_v = Q_v^{-1}$ and $\mathsf{comp}(v)$ be the set of gates done locally at vertex $v$. For $e=(u,v)$, let $\Lambda_e =  Q_e^{-1}$ and $\mathsf{trans}(e)$ be the set of Choi states corresponding to qubit transmissions on edge $e$.  Then Hamiltonian term corresponding to the edge $e$ is
\begin{align}
    h_{e} = \mathrm{span} \left( \Lambda_u \Lambda_v \Lambda_e \left(I -\left( \ketbra{0}{0}^{\mathrm{init}(u,v) } \bigotimes_{U \in \mathrm{comp}(u,v)} \ketbra{\Phi_U}{\Phi_U}  \right) \otimes \ketbra{\Phi_I}{\Phi_I}^{\otimes \mathrm{trans}(e)}  \right) \Lambda_u \Lambda_v \Lambda_e \right).
\end{align}
That is, each term checks the local computations and communications at the corresponding vertex or edge, but with a `sandwiching' by the inverse of the perturbed projectors. The reason for requiring the edge terms to check the local computations is for later convenience in the analysis of combinatorial states.

Similarly, for $v$, let $\Lambda_v = Q_v^{-1}$ and $\mathsf{comp}(v)$ be the set of gates done locally at vertex $v$. Then the Hamiltonian term corresponding to the vertex $v$ is
\begin{align}
    h_{v} = \mathrm{span} \left(  \Lambda_v \left(I - \left( \ketbra{0}{0}^{\mathrm{init}(v)} \bigotimes_{U \in \mathrm{comp}(v)} \ketbra{\Phi_U}{\Phi_U}  \right)  \right) \Lambda_v \right).
\end{align}

The overall Hamiltonian is
\begin{align}
    H_\mathrm{TN} (\mathcal{P}, \delta_\mathrm{vertex}, \delta_\mathrm{edge}) = \frac{1}{2|V|} \sum_{v \in V} h_v  + \frac{1}{2|E|} \sum_{e \in E}   h_e .
\end{align}

\paragraph{Notations.}
\begin{itemize}
    \item $N_\mathrm{EPR} = O(nm)$ denotes the number of EPR projectors in the tensor network (gray boxes in~\Cref{fig:tensor-network}).
    \item $N_\mathrm{anc} = \sum_{v \in V} |\mathsf{init}(v)| = O(nm)$ denotes the number of qubits initialized to $\ket{0}$ in the protocol.
    \item $\mathsf{Pauli}^{S}$ is the set of all Pauli strings supported on the EPR locations in $S$. With a slight abuse of notation we will also use $\mathsf{Pauli}^{\otimes N_\mathrm{EPR}}$ to denote the set of all Pauli strings supported on the EPR locations in the tensor network, and similarly $\mathsf{Pauli}^{\otimes N_\mathrm{anc}}$ for the ancilla qubits.
    \item $|\mathsf{F}|$ denotes the weight of the Pauli string $\mathsf{F}$.
    \item $|\mathsf{F}|_E$ denotes the `edge' weight, i.e., how many edges on which $\mathsf{F}$ acts nontrivially.
\end{itemize}

\begin{claim}[Ground states]\label{claim:gs}
     The tensor network Hamiltonian $H_\mathrm{TN}(\mathcal{P}, \delta_\mathrm{vertex}, \delta_\mathrm{edge})$ is frustration-free and has a ground space $\operatorname{span}\{\ket{\Psi_{\mathcal{P},\xi}} : \text{$w$-qubit } \ket{\xi}\}$, with
    \begin{align}
        \ket{\Psi_{\mathcal{P},\xi}} \propto \sum_{ \mathsf{F} \in \mathsf{Pauli}^{\otimes N_\mathrm{EPR}}} \delta^{\mathsf{F}} \ket{\Phi_{\mathsf{F}}} \otimes  \big(\mathcal{P}_{\mathsf{F}} \ket{\xi} \big)_\mathrm{out} ,
    \end{align}
where $\mathcal{P}_{\mathsf{F}}$ denotes the noisy implementation of $\mathcal{P}$ with Pauli errors $\mathsf{F}$, and
\begin{align}
    \delta^{\mathsf{F}}  := \prod_{e \in E} \prod_{v \in V} \big(\frac{\delta_\mathrm{edge}}{\sqrt{4^{|\operatorname{supp}(e)|}-1}} \big)^{|\mathsf{F}|_e} \big(\frac{\delta_\mathrm{vertex}}{\sqrt{4^{|\operatorname{supp}(v)|}-1}} \big)^{|\mathsf{F}|_v} .
\end{align}
\end{claim}
\begin{proof}
    The `sandwiched' Hamiltonian $\sum_{e \in E} Q_e h_e Q_e + \sum_{v \in V} Q_v h_v$ is non-interacting, frustration-free, and has ground states of the form $\ket{\Phi_{\mathcal{P}, \xi}}= \ket{\xi}_\mathrm{in} \otimes \left(\bigotimes_{v \in V} \ket{0}^{\mathsf{init}(v)} \right) \left( \bigotimes_{U \in \mathcal{P}} \ket{\Phi_U} \right)$, see~\Cref{eq:tensor-network-state}. Thus, the ground states of $H_\mathrm{TN}$ are
        $\left( \bigotimes_{e \in E} Q_e \right)  \left( \bigotimes_{v \in V} Q_v \right) \ket{\Phi_{\mathcal{P}, \xi}}.
    $
\end{proof}

\subsection{Combinatorial Soundness}\label{subsec:comb-soundness}

We will show that any state that violates at most an $\varepsilon$ fraction of vertex and edge terms in $H_\mathrm{TN}$ is close to a noisy execution of $\mathcal{P}$ with noise rate $\Theta(\varepsilon)$. There is a slight difference with~\Cref{def:edge-pauli-fault}. Because of the random noise coming from the tensor network injectivity parameters, we need to keep track of both vertex and edge faults.

\begin{definition}[Vertex and edge faults]
    \label{def:fault-ckt-to-H}
    Let $\cP$ be a distributed protocol on a unweighted multigraph $G=(V,E)$. We say that the execution is faulty on vertex $v$ if there is a non-trivial fault touching a qubit in player $v$'s register, and faulty on edge $(i,j)$ if there is a non-trivial Pauli fault occurring in a communication between players $i$ and $j$. If the fraction of faulty vertices is at most $\delta_1$ and the fraction of faulty edges is at most $\delta_2$, we say the execution is a \emph{$(\delta_1, \delta_2)$-vertex and edge faulty} execution of $\cP$.
\end{definition}

\begin{proposition}[Combinatorial robustness]\label{prop:combinatorial-soundness} There exists choices of the injectivity parameters
    $\delta_\mathrm{edge},\delta_\mathrm{vertex} \leq e^{-\mathrm{poly}(m,\log |V|)}$
     such that for any $\varepsilon \geq 0$ and for any $\varepsilon$-combinatorial state $\ket{\psi}$ of $H_\mathrm{TN}(\mathcal{P}, \delta_\mathrm{vertex}, \delta_\mathrm{edge})$, the reduced state $\psi_\mathrm{out}$ on the output qubits is $e^{-\mathrm{polylog}(m, |V|) \max(\varepsilon|V|, 1)}\text{-close}$ in trace distance to a mixture of outputs of $(\Theta(\varepsilon),\Theta(\varepsilon))$-vertex and edge faulty executions of $\mathcal{P}$, where an fault execution is in the sense of~\Cref{def:fault-ckt-to-H}.
\end{proposition}

\begin{proof} 
    The case $\varepsilon = 0$ is a straightforward union bound over the iid random noise in the ground state. So let us consider the nontrivial case $\varepsilon \neq 0$, and assume wlog $\varepsilon > 1/\mathrm{poly}(|V|)$.

    An $\varepsilon/2$-combinatorial $\ket{\psi}$ violates at most an $\varepsilon$ fraction of the vertex and edge terms in $H_\mathrm{TN}$.
    Consider the unnormalized state $\ket{\psi'}=  \left( \bigotimes_{e \in E} Q_e \right)^{-1}  \left( \bigotimes_{v \in V} Q_v \right)^{-1}   \ket{\psi}$ and let $F$ be the set of locations within violated edges and $W$ within violated vertices. Note that $|F| \leq 2\varepsilon |E| m$ and $|W| \leq 2\varepsilon |V| m$.
    For the satisfied Hamiltonain terms, the state $\ket{\psi'}$ locally equals the correct Choi states and $\ket{0}$ states. Thus, $\ket{\psi'}$ lies in the span of all Pauli errors occurring only at locations in $S$:
    \begin{align}\label{eq:psi'-combin-proof}
        \ket{\psi'}  = \sum_{\mathsf{G} \in \mathsf{Pauli}^{F \cup W}} c_{\mathsf{G}} \ket{\xi_{\mathsf{G}}}_\mathrm{in} \otimes   \mathsf{G} \left(\bigotimes_{v \in V} \ket{0}^{\mathsf{init}(v)}  \bigotimes_{U} \ket{\Phi_U} \right),
    \end{align} 
    where, wlog, $c_{\mathsf{G}}$ is nonzero only when $\mathsf{G} $ is composed of only the Pauli operators $I$ and $X$ on the initialization locations $\ket{0}$.

    Next, we apply the $Q$ maps to obtain the original combinatorial state:

    \begin{align}
    \ket{\psi} & =  \left( \bigotimes_{e \in E} Q_e \right)  \left( \bigotimes_{v \in V} Q_v \right) \ket{\psi'} \\ & = \sum_{\mathsf{F} \in \mathsf{Pauli}^{\otimes N_\mathrm{EPR}}}
    \delta^{\mathsf{F}}  \ket{\Phi_{\mathsf{F}}} \bigotimes \left(\sum_{\mathsf{G} \in \mathsf{Pauli}^{F \cup W}} c_{\mathsf{G}}   \mathcal{P}_{\mathsf{G}\mathsf{F}} \ket{0^{N_\mathrm{anc}}}\ket{\xi_{\mathsf{G}}}\right)_\mathrm{out},
    \label{eq:chi}
\end{align}
where $\delta^{\mathsf{F}} $ is defined in~\Cref{claim:gs}, $\mathsf{G}\mathsf{F}$ is the `location-wise' multiplication of the Pauli strings and $\mathcal{P}_{\mathsf{G}\mathsf{F}}$ denotes the noisy implementation of $\mathcal{P}$ with Pauli errors $\mathsf{G}\mathsf{F}$.

We now show that $\psi_\mathrm{out} = \Tr_{\overline{\mathrm{out}}} \ketbra{\psi}{\psi} $ is very close to a mixed state which is obtained by removing from~\Cref{eq:chi} the summands $\mathsf{F}$ whose edge-weight is larger than $\alpha |E|$ or vertex-weight is larger than $\alpha |V|$ and renormalizing properly, for some $\alpha$ to be specified. Specifically, define
\begin{equation}
    \Pi = \sum_{\mathsf{F} \in \mathsf{Pauli}^{\otimes N_\mathrm{EPR}}: \  |\mathsf{F}|_E \geq \alpha |E| \text{ or } |\mathsf{F}|_V \geq \alpha |V|}   \ketbra{\Phi_{\mathsf{F}}}{\Phi_{\mathsf{F}}}.
\end{equation}

Let $S$ denote the EPR locations on violated edges, violated vertices, and (possibly non-violated) edges adjacent to violated vertices. Consider the natural definition of $Q^S$ and $Q^{S^c}$, where $Q^{S^c}$ does not intersect with the EPR locations $S$, let the `partially undone' state be
\begin{align}
    \ket{\chi} := Q^{S^c}\ket{\psi'}, \qquad \text{such that } \ket{\psi} = Q^S \ket{\chi}.
\end{align}
For $\alpha > 1000 \varepsilon$ to be chosen, let $\alpha' >0$ such that 
\begin{align}
    \label{eq:alpha-combin-soundness}
    \alpha' := (\alpha -\varepsilon) > \varepsilon.
\end{align}
We can upper bound $\Pi$ with the following operator:
\begin{align}
    \Pi \preceq \Pi' = \left(\sum_{\mathsf{F} \in \mathsf{Pauli}^{S^c}:\ |\mathsf{F}|_E \geq \alpha' |E| \text{ or } |\mathsf{F}|_V \geq \alpha' |V| } \ketbra{\Phi_{\mathsf{F}}}{\Phi_{\mathsf{F}}} \right) \otimes I_{S},
\end{align}

Note that $4^{-m} \delta_\mathrm{vertex} \cdot I \preceq Q_v \preceq I$ and $4^{-m} \delta_\mathrm{edge} \cdot I \preceq Q_e \preceq I$. Furthermore, $\Pi'$ and $Q$ commute. So we have
\begin{align}\label{eq:psi-Pi-psi}
    \bra{\psi} \Pi \ket{\psi} & \leq \bra{\psi} \Pi' \ket{\psi} = \frac{\bra{\chi} Q^S \Pi' Q^S \ket{\chi} }{ \bra{\chi} Q^S Q^S \ket{\chi}} \\
    &\leq (4^{-m} \delta_\mathrm{vertex})^{-2 \varepsilon |V| } (4^{-m} \delta_\mathrm{edge})^{- 2\varepsilon |E| } \frac{\bra{\chi}  \Pi' \ket{\chi} }{ \bra{\chi} \ket{\chi}}.
\end{align}

Let us bound the RHS:
\begin{align}
    \frac{\bra{\chi}  \Pi' \ket{\chi} }{ \bra{\chi} \ket{\chi}} 
    &= \frac{ \bra{\chi} \left( \sum_{\mathsf{F} \in \mathsf{Pauli}^{S^c}: \ |\mathsf{F}|_E \geq \alpha' |E| \text{ or } |\mathsf{F}|_V \geq \alpha' |V|} \ketbra{\Phi_{\mathsf{F}}}{\Phi_{\mathsf{F}}} \right) \ket{\chi} }{ \bra{\chi} \ket{\chi}}\\
    & = \frac{ \bra{\psi'} \left( \sum_{\mathsf{F} \in \mathsf{Pauli}^{S^c}: \ |\mathsf{F}|_E \geq \alpha' |E| \text{ or } |\mathsf{F}|_V \geq \alpha'  |V| } (\delta^{\mathsf{F}})^2 \ketbra{\Phi_{\mathsf{F}}}{\Phi_{\mathsf{F}}} \right) \ket{\psi'} }{ \bra{\psi'} \left( \sum_{\mathsf{F} \in \mathsf{Pauli}^{S^c}} (\delta^{\mathsf{F}})^2 \ketbra{\Phi_{\mathsf{F}}}{\Phi_{\mathsf{F}}} \right) \ket{\psi'} }\\
    & = \frac{  \sum_{\mathsf{F} \in \mathsf{Pauli}^{S^c}: \ |\mathsf{F}|_E \geq \alpha' |E|  \text{ or } |\mathsf{F}|_V \geq \alpha' |V| } (\delta^{\mathsf{F}})^2 }{\sum_{\mathsf{F} \in \mathsf{Pauli}^{S^c}} (\delta^{\mathsf{F}})^2  }, \label{eq:randomweight}
\end{align}
where the last equality follows from the fact that $\ket{\psi'}$ on $S^c$ is a product $\ket{0}$'s and Choi states $\ket{\Phi_U}$ and thus the norm $\|\bra{\psi'} (\mathrm{I} \otimes \ket{\Phi_{\mathsf{F}}})\|$ is independent of $\mathsf{F}$.

Recall from~\Cref{claim:gs} that $\delta^{\mathsf{F}}  := \prod_{e \in E} \prod_{v \in V} \big(\frac{\delta_\mathrm{edge}}{\sqrt{4^{|\operatorname{supp}(e)|}-1}} \big)^{|\mathsf{F}|_e} \big(\frac{\delta_\mathrm{vertex}}{\sqrt{4^{|\operatorname{supp}(v)|}-1}} \big)^{|\mathsf{F}|_v} $. So the above expression can be interpreted as the probability that a random Pauli string $\mathsf{F}$, where each edge (vertex) location is independently chosen to be faulty with probability $\mu_\mathrm{edge} = \delta_\mathrm{edge}^2/(1+ \delta_ \mathrm{edge}^2)$ ($\mu_\mathrm{vertex}= \delta_\mathrm{vertex}^2/(1+ \delta_\mathrm{vertex}^2)$), has edge-weight at least $\alpha' |E|$ or vertex-weight at least $\alpha' |V|$. Therefore,~\Cref{eq:randomweight} can be upper bounded with the (multiplicative) Chernoff bound.

\begin{fact}\label{fact:chernoff} Let $X= X_1 + \hdots + X_N $ where $X_i \in \{0,1\}$ are i.i.d. binary random variables with $\mathbb{E}[X_i] \leq \mu$. Let, $\eta > 10$, then $\Pr[ X \geq \eta \mu N ] \leq e^{-\eta\log(\eta)\mu N/3}.$
\end{fact}

We apply the Chernoff bound with two sets of parameters. The first set is to bound the number of faulty edges: mean $\mu_\mathrm{edge} = \delta_\mathrm{edge}^2/(1+ \delta_ \mathrm{edge}^2)$, number of edges $N_\mathrm{edge} \geq |E| -\alpha' |E| > |E|/2$ and deviation $\eta_\mathrm{edge} =\alpha'/\mu_\mathrm{edge} $. The second set is to bound the number of faulty vertices: $\mu_\mathrm{vertex}= \delta_\mathrm{vertex}^2/(1+ \delta_\mathrm{vertex}^2)$, $N_\mathrm{vertex} \geq |V| -\alpha' |V| > |V|/2$ and $\eta_\mathrm{vertex} = \alpha'/\mu_\mathrm{vertex}$. Direct substitution into~\Cref{fact:chernoff} gives

\begin{align}
    \eqref{eq:randomweight} \leq e^{ - \log(\alpha'/\mu_\mathrm{edge}) \alpha' |E|/6 } +  e^{ - \log(\alpha'/\mu_\mathrm{vertex}) \alpha' |V|/6 } .
\end{align}

Combining everything above we obtain
\begin{align}
    \bra{\psi} \Pi \ket{\psi}  &\leq (4^{-m} \delta_\mathrm{vertex})^{- 2\varepsilon |V| } (4^{-m} \delta_\mathrm{edge})^{- 2\varepsilon |E|} (e^{ - \log(\alpha'/\mu_\mathrm{edge}) \alpha' |E|/6 } +  e^{ - \log(\alpha'/\mu_\mathrm{vertex}) \alpha' |V|/6 } ).\\
    & \leq \underbrace{e^{ (4m + \log (1/\delta_\mathrm{vertex}) ) 2\varepsilon |V| } e^{ (4m + \log (1/\delta_\mathrm{edge})) 2\varepsilon |E| }  e^{ - \log(\alpha'/\mu_\mathrm{vertex}) \alpha' |V|/6 } }_{\text{`vertex term'}} + \\
    & \qquad \underbrace{e^{ (4m + \log (1/\delta_\mathrm{vertex}) ) 2\varepsilon |V| } e^{ (4m + \log (1/\delta_\mathrm{edge})) 2\varepsilon |E| } e^{ - \log(\alpha'/\mu_\mathrm{edge}) \alpha' |E|/6 } }_{\text{`edge term'}}.
\end{align}
Assuming
\begin{itemize}
    \item (i) $\delta_\mathrm{edge}, \delta_\mathrm{vertex} \leq e^{-\mathrm{poly}(m, \log |V|)}$,
    \item (ii) $\alpha' >  990 \varepsilon \geq 1/\mathrm{poly}(|V|)$ (see~\Cref{eq:alpha-combin-soundness} for the first inequality, and the second is wlog),
\end{itemize}
let us bound each term separately. For the vertex term, the negative exponent dominates if
\begin{align}
    \frac{1}{12} \log(\alpha'/\mu_\mathrm{vertex}) \alpha' |V| & \geq  \frac{1}{12} \log(1/\delta_\mathrm{vertex}) \alpha' |V|  \tag*{(true by assumption (i))}\\
    & \geq 4 \log (1/\delta_\mathrm{vertex}) \varepsilon |V|  +  4 \log (1/\delta_\mathrm{edge})\varepsilon |E| \label{eq:vertex-upper-bound} \\
    & \geq 2 (4m + \log (1/\delta_\mathrm{vertex}) ) \varepsilon |V|  + 2(4m + \log (1/\delta_\mathrm{edge})) \varepsilon |E| \tag*{(true by assumption (ii))}
\end{align}
For~\Cref{eq:vertex-upper-bound} to hold, it suffices to have
\begin{align}
    \log(1/\delta_\mathrm{vertex}) (\alpha'/48\varepsilon - 1) & \geq  \log(1/\delta_\mathrm{edge}) \cdot |E|/|V|\\
    \text{and it suffices to choose } \delta_\mathrm{vertex} & \leq   (\delta_\mathrm{edge})^{\frac{|E|/|V|}{\alpha' /200 \varepsilon }} .
\end{align}
For the edge term, we similarly require
\begin{align}
     \frac{1}{12} \log(1/\delta_\mathrm{edge}) \alpha' |E|  &\geq 4 \log (1/\delta_\mathrm{vertex}) \varepsilon |V|  +  4 \log (1/\delta_\mathrm{edge})\varepsilon |E|  \\
    \Leftrightarrow \log(1/\delta_\mathrm{edge}) \left(\frac{\alpha'}{48\varepsilon} - 1\right) \frac{|E|}{|V|}& \geq  \log(1/\delta_\mathrm{vertex}) \\
    \text{and we choose } \delta_\mathrm{vertex} & \geq   (\delta_\mathrm{edge})^{\frac{|E|}{|V|} \frac{\alpha'}{200\varepsilon} }.
\end{align}

In summary, under the following conditions:
\begin{align}\label{eq:end-of-combin-proof}
   \delta_\mathrm{edge}, \delta_\mathrm{vertex} &\leq e^{-\mathrm{poly}(m, \log |V|)}\\
    (\delta_\mathrm{edge})^{\frac{|E|}{|V|} \frac{\alpha'}{200\varepsilon} } & \leq \delta_\mathrm{vertex} \leq   (\delta_\mathrm{edge})^{\frac{|E|/|V|}{\alpha' /200 \varepsilon }},\\
    \alpha & >  \max (1000 \varepsilon ,  100/|V|),
\end{align} 
we obtain the upper bound, with $\alpha' = \alpha -  \varepsilon$,
\begin{align}
     \bra{\psi} \Pi \ket{\psi}  &\leq e^{-\frac{1}{12} \log(1/\delta_\mathrm{vertex}) \alpha' |V|}   + e^{-\frac{1}{12} \log(1/\delta_\mathrm{edge}) \alpha' |E|} \\
     & \leq 2 e^{- \min (\log 1/\delta_\mathrm{vertex}, \log 1/\delta_\mathrm{edge}) \alpha |V|/24}\\
     &\leq e^{- \mathrm{poly}(m, \log |V|) \max(\varepsilon|V|, 1)}.
\end{align}
In other words, $\ket{\psi}_\mathrm{out}$ is $e^{-\mathrm{poly}(m, \log |V|) \max(\varepsilon|V|, 1)}$-close in fidelity to a state $\ket{\psi'}$ which is a superposition of states corresponding executions of $\mathcal{P}$ with at most $(\alpha +\varepsilon) |E|$ faulty edges and $(\alpha +\varepsilon) |V|$ faulty vertices.
Picking $\alpha = 1000 \varepsilon = \Theta(1)$ gives the desired result.

\end{proof}

\subsection{A Modified Construction for Handling Postselection Gates}\label{sec:nc-ckt-to-H}
We now introduce a modified construction that can handle protocols with postselection gates. The main motivation is to be able to handle protocols for the non-commuting Hamiltonian measurement circuit. See the discussion at the beginning of~\Cref{sec:nc-protocol}.

Let $M$ be the number of measurement output qubits.
We introduce the following modifications to the circuit-to-Hamiltonian construction from~\Cref{sec:comm-ckt-to-H}.

Given a protocol $\cP_\gamma$, our modified tensor network encoding of $\mathcal{P}_\gamma$ upon input $\ket{\xi}$ is 
\begin{align}
    \left( \bigotimes_{e \in E} Q_e \right)  \left( \bigotimes_{v \in V} Q_v \right) \left( \bigotimes_{b \in [M]} A_b \right) \ket{\Phi_{\mathcal{P}_\gamma,\xi}},
\end{align}
where 
\begin{align}
    A := \sqrt{1-\eta}\ketbra{1}{1} + \sqrt{\eta} \ketbra{0}{0}
\end{align}
can be thought of as a noisy postselection operator that biases the measurement bits towards $1$ when $\eta <1/2$. Our final choice will be $\eta = 1/2 - o(1)$. This small amount of postselection is crucial for our fault tolerance analysis. See~\Cref{remark:why-need-A}.

Associated to this we modify the vertex Hamiltonian terms as follows:
\begin{align}
    h'_v = \operatorname{span} \left( (A^{-1})^{\otimes M_v} h_{v} (A^{-1})^{\otimes M_v} \right),
\end{align}
where $M_v \subset [M]$ is the number of measurement qubits shared in the vertex $v$.

The overall Hamiltonian is
\begin{align}
    H_\mathrm{TN} (\mathcal{P}_\gamma, \delta_\mathrm{vertex}, \delta_\mathrm{edge},\eta) =  \frac{1}{2|E|} \sum_{e \in E}  h_e + \frac{1}{2|V|} \sum_{v \in V} h'_v
\end{align}

\begin{claim}[Ground state of modified Hamiltonian]\label{claim:modified-ground-state}
     The tensor network Hamiltonian $H_\mathrm{TN}(\mathcal{P}_\gamma, \delta_\mathrm{vertex}, \delta_\mathrm{edge}, \eta)$ is frustration-free and has a ground space $\operatorname{span}\{\ket{\Psi_{\mathcal{P}_\gamma,\xi}} : \text{$w$-qubit } \ket{\xi}\}$, with
    \begin{align}
        \ket{\Psi_{\mathcal{P}_\gamma,\xi}} \propto \sum_{ \mathsf{F} \in \mathsf{Pauli}^{\otimes N_\mathrm{EPR}}}  \delta^{\mathsf{F}} \ket{\Phi_{\mathsf{F}}} \otimes  A^{\otimes M} \big((\mathcal{P}_\gamma)_{\mathsf{F}} \ket{\xi} \big)_\mathrm{out} ,
    \end{align}
where $(\mathcal{P}_\gamma)_{\mathsf{F}}$ denotes the noisy implementation of $\mathcal{P}_\gamma$ with Pauli errors $\mathsf{F}$, and
\begin{align}
    \delta^{\mathsf{F}}  := \prod_{e \in E} \prod_{v \in V} \big(\frac{\delta_\mathrm{edge}}{\sqrt{4^{|\operatorname{supp}(e)|}-1}} \big)^{|\mathsf{F}|_e} \big(\frac{\delta_\mathrm{vertex}}{\sqrt{4^{|\operatorname{supp}(v)|}-1}} \big)^{|\mathsf{F}|_v} .
\end{align}
\end{claim}
\begin{proof}
    Similar to the proof of~\Cref{claim:gs}, but we additionally apply $A$ on the measurement bits.
\end{proof}

\begin{proposition}[Combinatorial Soundness of Modified Hamiltonian]\label{prop:modified-combinatorial-soundness} There exists choices of the injectivity parameters such that $$\delta_\mathrm{edge},\delta_\mathrm{vertex} \leq e^{-   \mathrm{poly}(\frac{m}{\varepsilon}, \log |V|) \log(1/\eta)}$$ such that for any $\varepsilon > 0$ and for any $\varepsilon$-combinatorial state $\ket{\psi}$ of $H_\mathrm{TN}(\mathcal{P}_\gamma, \delta_\mathrm{vertex}, \delta_\mathrm{edge}, \eta)$, the reduced state $\psi_\mathrm{out}$ on the output qubits is
    $$e^{- \mathrm{poly}(\frac{m}{\varepsilon}, \log |V|) \log(1/\eta) \cdot |V|}\text{-close }$$ 
in trace distance to a mixture of outputs of $(\Theta(\varepsilon), \Theta(\eps))$-vertex and edge faulty executions of $\mathcal{P}_\gamma$.\\
And when $\varepsilon=0$ (ground state), the same statement holds with $m/\varepsilon$ in the exponent replaced with $m$.
\end{proposition}
\begin{proof}
The proof is a straightforward modification of the proof of~\Cref{prop:combinatorial-soundness}. 
Again, consider an $\varepsilon$-combinatorial (normalized) state $\ket{\psi}$.

In the (non-orthonormal) basis $\left( \bigotimes_{e \in E} Q_e \right)^{-1}  \left( \bigotimes_{v \in V} Q_v \right)^{-1}  \left( \bigotimes_{b \in [M]} A_b \right)^{-1}$ the state has the form (unnormalized)
\begin{align}
        \ket{\psi'}  = \sum_{\mathsf{G} \in \mathsf{Pauli}^{F \cup W}} c_{\mathsf{G}} \ket{\xi_{\mathsf{G}}}_\mathrm{in} \otimes   \mathsf{G} \left(\bigotimes_{v \in V} \ket{0}^{\mathsf{init}(v)}  \bigotimes_{U} \ket{\Phi_U} \right),
\end{align} 
where $F$ is the set of locations within violated edges and $W$ within violated vertices, with $|F| \leq 2\varepsilon |E| m$ and $|W| \leq 2\varepsilon |V| m$. Thus,
\begin{align}
    \ket{\psi} & =  \left( \bigotimes_{e \in E} Q_e \right)  \left( \bigotimes_{v \in V} Q_v \right)\left( \bigotimes_{b \in [M]} A_b \right) \ket{\psi'} \\ 
    & = \sum_{\mathsf{F} \in \mathsf{Pauli}^{\otimes N_\mathrm{EPR}}}
    \delta^{\mathsf{F}}  \ket{\Phi_{\mathsf{F}}} \bigotimes A^{\otimes M} \left(\sum_{\mathsf{G} \in \mathsf{Pauli}^{F \cup W}} c_{\mathsf{G}}   \mathcal{P}_{\mathsf{G}\mathsf{F}} \ket{0^{N_\mathrm{anc}}}\ket{\xi_{\mathsf{G}}}\right)_\mathrm{out}\\
    &:= A^{\otimes M} \ket{\psi''}.
\end{align}

Again, we would like to upper bound
$ \bra{\psi} \Pi \ket{\psi} $
where
\begin{equation}
    \Pi = \sum_{\mathsf{F} \in \mathsf{Pauli}^{\otimes N_\mathrm{EPR}}: \  |\mathsf{F}|_E \geq \alpha |E| \text{ or } |\mathsf{F}|_V \geq \alpha |V|}   \ketbra{\Phi_{\mathsf{F}}}{\Phi_{\mathsf{F}}},
\end{equation}
and $\alpha$ is to be specified.

Note that the unnormalized state $\ket{\psi''}$ defined above would have played the role of $\ket{\psi}$ in the proof of~\Cref{prop:combinatorial-soundness}, and our now $\ket{\psi}$ is obtained by additionally applying $A^{\otimes M}$ on some output qubits, which we can think of as a noisy postselection. This means that instead of bounding $\bra{\psi''} \Pi \ket{\psi''}$, we need to bound $\bra{\psi''} A^{\otimes M} \Pi A^{\otimes M} \ket{\psi''}$. Using the fact that $A$ commutes with $\Pi$ (they do not overlap) and that $ \eta^{M} \cdot I \prec A \preceq I$
\begin{align}
    \bra{\psi} \Pi \ket{\psi}  = \frac{\bra{\psi''} A^{\otimes M} \Pi A^{\otimes M} \ket{\psi''}}{\bra{\psi''} A^{\otimes M} A^{\otimes M} \ket{\psi''}} \leq (1/\eta)^{2M} \frac{ \bra{\psi''} \Pi \ket{\psi''}}{ \braket{\psi''}}.
\end{align}
The quantity $\frac{ \bra{\psi''} \Pi \ket{\psi''}}{ \braket{\psi''}}$ is exactly the quantity that we bounded in the proof of~\Cref{prop:combinatorial-soundness}. 

In particular, consider first the case $\varepsilon=0$ (ground state). Then we choose $\alpha=0$ and simply perform a union bound
\begin{align}
    \bra{\psi} \Pi \ket{\psi} \leq e^{\log(1/\eta) 2M} (\delta_\mathrm{vertex} |V| + \delta_\mathrm{edge}|E|).
\end{align}
It then suffices to choose $\delta_\mathrm{edge},\delta_\mathrm{vertex} \leq e^{-   \mathrm{poly}(m, \log |V|) \log(1/\eta)}$ to obtain the claimed bound.

Next we consider the nontrivial case $\varepsilon > 0$ and assume wlog $\varepsilon > 1/\mathrm{poly}(|V|)$.
At the end of \Cref{prop:combinatorial-soundness}'s proof (\Cref{eq:end-of-combin-proof}), we had the conclusion that when
\begin{align}
    \alpha & >  \max (1000\varepsilon ,  100/|V|) \\
   \delta_\mathrm{edge}, \delta_\mathrm{vertex} &\leq e^{-\mathrm{poly}(m, \log |V|)}\\
    (\delta_\mathrm{edge})^{\frac{|E|}{|V|} \frac{\alpha'}{200\varepsilon} } & \leq \delta_\mathrm{vertex} \leq   (\delta_\mathrm{edge})^{\frac{|E|/|V|}{\alpha'/200\varepsilon}}, \tag*{where $\alpha' = \alpha - \varepsilon$}
\end{align}
then 
\begin{align}
     \frac{ \bra{\psi''} \Pi \ket{\psi''}}{ \braket{\psi''}} 
     & \leq 2 e^{- \min (\log 1/\delta_\mathrm{vertex}, \log 1/\delta_\mathrm{edge}) \alpha |V|/24}.
\end{align}
Thus, 
\begin{align}
    \bra{\psi} \Pi \ket{\psi}  &\leq 2 (1/ \eta)^{2M} \cdot e^{- \min (\log 1/\delta_\mathrm{vertex}, \log 1/\delta_\mathrm{edge}) \alpha |V|/24}\\
     &\leq 2 e^{2 \log (1/ \eta) |V| m} \cdot e^{- \min (\log 1/\delta_\mathrm{vertex}, \log 1/\delta_\mathrm{edge}) \alpha |V|/24}. \tag*{($M \leq |V| m$)}
\end{align}
We can choose $\log 1/\delta_\mathrm{vertex}, \log 1/\delta_\mathrm{edge} \geq 200 \log(1/\eta) m/\alpha$ to ensure 
\begin{align}
    \bra{\psi} \Pi \ket{\psi}  \leq e^{- \min (\log 1/\delta_\mathrm{vertex}, \log 1/\delta_\mathrm{edge}) \alpha |V|/50}.
\end{align}
In summary, we can choose $\alpha = 1000\varepsilon$, $\delta_\mathrm{edge}, \delta_\mathrm{vertex} \leq e^{-   \mathrm{poly}(\frac{m}{\varepsilon}, \log |V|) \log(1/\eta)}$ (while satisfying the relation between $\delta_\mathrm{edge}$ and $\delta_\mathrm{vertex}$ above) to get
\begin{align}
    \bra{\psi} \Pi \ket{\psi}  \leq e^{- \mathrm{poly}(\frac{m}{\varepsilon}, \log |V|) \log(1/\eta) \cdot |V|}.
\end{align}
\end{proof}

%% file: reduction.tex
\section{The Final Reduction}
\label{sec:final-reduction}
In this section, we combine all the components to obtain our reduction from the input local Hamiltonian $H_\mathrm{in}$ with gap $\epsilon$ to an output 2-local Hamiltonian $H$ with larger alphabet and combinatorial gap $\epsilon \pl n$. We will begin by describing the reduction procedure, and then analyze the resulting Hamiltonian $H$.

We make the following mild assumptions on the input Hamiltonian $H_\mathrm{in}$, although our results hold more generally for $O(\operatorname{loglog} n)$ locality and interaction degree, and more general input alphabet size. Below, the assumption on the fixed constraint set is only required for the non-commuting case, and is used in the detectability lemma.
\begin{assumption}
    Fix constants $k,D=\Theta(1)$ and a fixed set $\mathfrak{h}$ of $k$-qubit projectors that can be specified by a constant number of bits. We assume that $H_\mathrm{in} = \sum_{i=1}^{M} h_i$ on $n$ qubits, such that each $h_i$ is a projector that belongs to the set $\mathfrak{h}$ and only intersects with $D$ other terms in $H_\mathrm{in}$.
\end{assumption}

\begin{mdframed}
\medskip
\noindent \textbf{Gap Amplification Procedure:}
\begin{enumerate}
\item Preprocessing: Given an input Hamiltonian $H_\mathrm{in}$ on $n$ qubits, construct a HDX $X$ from~\Cref{thm:hdx-construction} with $N=O(n)$ $d$-cliques. Then pad $H_\mathrm{in}$ with dummy qubits and constraints to obtain a local Hamiltonian $H'$ on $O(N)$ qubits with $O(N)$ constraints.
\item Compute the sequential measurement circuit $C$ for $H'$ from~\Cref{def:sequential-measurement-circuit-appendix}, such that each layer of $C$ measures $N$ constraints of $H'$.
\item Compute the distributed protocol on $G=(X(1),X(2))$ that implements $C$: For commuting Hamiltonians, we use the protocol $\cP$ from~\Cref{sec:protocol-for-circuit}. For non-commuting Hamiltonians, we use the protocol $\cP_\gamma$ from~\Cref{sec:nc-protocol}.
\item Unweighting: Convert the protocol into a protocol $\cP'$ on a multigraph with unweighted edges. Specifically, each edge $e=(v,v')$ in $G$ is replaced by $k_e$ parallel multi-edges  done by replacing each weighted edge in $X(2)$, where $k_e$ is the number of $d$-cliques $e$ is part of in $X$. On the new multi-graph $G'=(X(1),E')$, the protocol $\cP'$ emulates communications between $v$ and $v'$ by encoding each communication qubit into a $\Theta(k_e)$-qubit linear-distance quantum code with efficient decoding algorithm (e.g.~\cite{dinur2023good}) and send each qubit of the encoding through a different edge between $v$ and $v'$. The receiving vertex decodes the received qubits to recover the communication qubit.
\item Convert the distributed protocol $\cP'$ to a 2-local Hamiltonian using~\Cref{sec:ckt-to-H}. Specifically, for commuting Hamiltonians, we use the construction from~\Cref{sec:comm-ckt-to-H} to obtain $H_\mathrm{TN}(\cP,\delta_\mathrm{vertex}, \delta_\mathrm{edge})$, and for non-commuting Hamiltonians, we use~\Cref{sec:nc-ckt-to-H} to obtain $H_\mathrm{TN}(\cP_\gamma,\delta_\mathrm{vertex}, \delta_\mathrm{edge}, \eta)$.
\item For each vertex $v$ of $X$, introduce a vertex output constraint $h_v$ (either $\Pi_v$ or $\Pi_v^{\tau_\mathrm{out}}$ from~\Cref{def:vertex-threshold-check}) over measurement outputs held by the vertex.
\end{enumerate}
The final Hamiltonian is on the graph $G'$
\begin{align*}
    H = H_\mathrm{TN} + \E_{v \sim X(1)} h_v.
\end{align*}
\end{mdframed}

We now analyze the completeness and combinatorial soundness of the resulting Hamiltonian $H$.

\begin{definition}[Combinatorial energy]
    We say that a Hamiltonian $H$ has combinatorial gap $\epsilon$ if for any state $\ket{\psi}$, if we sample a constraint from $H$ according to its specified distribution, then the probability that $\ket{\psi}$ violates the constraint is at least $\epsilon$.
\end{definition}

\begin{theorem}[Main theorem]\label{thm:combinatorial-soundness-H} There exist a function $\epsilon_{\max}(n)=1/{\plll n}$ such that the following holds. For all $\epsilon_\mathrm{in} >0$ let $\epsilon = \min(\epsilon_\mathrm{in}  \operatorname{polylog} n, \epsilon_{\max}(n))$. There exists a polynomial-time procedure, given a constant-locality, constant-degree Hamiltonian $H_\mathrm{in}$, it produces a $2$-local Hamiltonian $H$ on $O(n)$ qudits of interaction degree $D$ and alphabet size $\Sigma$ such that:
\begin{enumerate}
    \item (Degree and alphabet)
    \begin{itemize}
        \item If $\epsilon \lll 1/\log n$, then $D,\ \log \Sigma = \operatorname{polylog}(n)$.
        \item  If $\epsilon \ggg 1/\log n$, then $D ,\ \log \Sigma = 2^{\operatorname{polyloglog}(n)}$.
    \end{itemize}
    \item (Completeness) If $H_\mathrm{in}$ has zero ground energy, then $H$ also has ground energy at most $2^{-\operatorname{polylog}(n)} $.
    \item (Combinatorial soundness) If $H_\mathrm{in}$ has normalized ground energy $\epsilon_\mathrm{in} $, then $H$ has $\epsilon$ combinatorial energy.
\end{enumerate}
Furthermore, assuming the existence of asymptotically good quantum locally testable codes in the sense of~\Cref{assumption:nice-qLTC}, the above statement holds with $\epsilon_{\max}(n) = \Theta(1)$.

In the special case when $H_\mathrm{in}$ is promised to be commuting, the same statement holds with the improvement that the output Hamiltonian also has zero energy in the completeness case.
\end{theorem}

\begin{remark}
We remark that the proof of this theorem sets the friendly postselection parameter $\eta$ to be approximately $1/2$ in the high-gap regime ($\geq 1/\pl n$), while in the low input gap regime we need to pick $\eta \ll \epsilon \pl n$. Note that choosing $\eta  \approx 1/2$ could potentially make a general soundness analysis for our output Hamiltonian more tractable: see the discussion in~\Cref{sec:discussion}.   
\end{remark}

\begin{proof} 

    \medskip\noindent \textbf{Polynomial time:} The only non-trivial part in our reduction is to construct the HDX and the fault-tolerant routing protocol on it in polynomial time. This is guaranteed by~\Cref{thm:hdx-construction} and~\Cref{prop:ft-routing}.

\medskip\noindent \textbf{Degree and alphabet:}  In the low-gap regime $\epsilon \lll 1/\log n$, it suffices to use an HDX with $d=O(1)$ (low corruption regime in the routing protocol~\Cref{prop:ft-routing}) and $\lambda= 1/\pl n$. Whereas, in the high-gap regime $\epsilon \ggg 1/\log n$, it suffices to use $d=O(\operatorname{polyloglog} n)$ and $\lambda= 1/\operatorname{polylog} n$ (high corruption regime~\Cref{prop:ft-routing}). In the HDX from~\Cref{thm:hdx-construction}, each vertex is involved in $ (1/\lambda)^{O(d^2)}$ cliques, and the work complexity of each vertex in the distributed protocol is $m=  (1/\lambda)^{O(d^2)} \pl n$. Thus the Hamiltonian interaction degree is $D = (1/\lambda)^{O(d^2)}$ and log-alphabet size $\log \Sigma =  O(m)$. Plugging in the parameters of $d$ and $\lambda$ for the two regimes gives the claimed bounds on $D$ and $\log \Sigma$.

Let us now move to the completeness and soundness. First, the preprocessing step and the construction of the sequential measurement circuit $C$ only affect the soundness by a constant factor, and thus we will ignore this loss in the analysis below.

Let us consider the general, non-commuting case first.
We first take the following choice of parameters for the construction of $H_\mathrm{TN}(\cP_\gamma, \delta_\mathrm{vertex}, \delta_\mathrm{edge}, \eta)$:
\begin{itemize}
    \item $\gamma = 2^{-\operatorname{polylog}(n)}$, and as a result,
    \item $\eta= \Theta(\gamma^{0.3})= 2^{-\operatorname{polylog}(n)}$ (see~\Cref{def:gamma'}),
    \item $\delta_\mathrm{vertex}, \delta_\mathrm{edge} = 2^{2^{ \operatorname{polyloglog}(n) }/\poly(\epsilon)}$ (see~\Cref{prop:modified-combinatorial-soundness})
    \end{itemize}
Then, we can use the OR vertex checks defined in~\Cref{def:vertex-threshold-check} for the $h_v$ terms.

\begin{remark}\label{remark:eta} The above choice of parameters is made for simplicity and can be relaxed in several ways. First, in the low-gap regime it actually suffices to choose $\gamma \ll \epsilon \pl(n)$ if we sacrifice near-perfect completeness: with this choice of $\gamma$ we can create a completeness-soundness gap by comparing the vertex OR check energies from~\Cref{claim:completeness-nc-protocol} and \Cref{thm:soundness-nc-protocol}. Second, as discussed in~\Cref{sec:soundness-nc-protocol} and particularly~\Cref{thm:soundness-nc-protocol}, in the  high-gap case $\epsilon \geq 1/\pl n$, we can choose $\gamma=\Theta(1)$ and $\eta = 1/2 - \Theta(\epsilon_0/\delta) \approx 1/2$ (see~\Cref{def:gamma'}). In this setting, we instead use the vertex threshold checks in~\Cref{def:vertex-threshold-check} and the vertex threshold energy bounds from~\Cref{claim:completeness-nc-protocol} and~\Cref{thm:soundness-nc-protocol} to also obtain completeness $2^{-\operatorname{polylog}(n)} $ and combinatorial soundness $\epsilon \pl n$ with exactly the same calculations as below.
\end{remark}

\medskip\noindent \textbf{Completeness:} Let $\ket{\xi_0}$ be a ground state of $H_\mathrm{in}$, then
    \Cref{claim:completeness-nc-protocol} asserts that the state
    \begin{align}
        \ket{\Psi_0} \propto A^{\otimes Md} \cP_\gamma  \encq(\ket{\xi_0})
    \end{align}
    has energy on the vertex output check terms
    \begin{align}
        \mathbb{E}_{v} \Tr(h_v \ketbra{\Psi_0}) \leq 2^{-\operatorname{polylog}(n)}.
    \end{align}
    On the other hand, according to~\Cref{claim:modified-ground-state}, the state
    \begin{align}
        \ket{\Psi_{\mathcal{P}_\gamma,\ket{\xi_0}}} \propto \sum_{ \mathsf{F} \in \mathsf{Pauli}^{\otimes N_\mathrm{EPR}}}  \delta^{\mathsf{F}} \ket{\Phi_{\mathsf{F}}} \otimes  A^{\otimes Md} \big((\mathcal{P}_\gamma)_{\mathsf{F}} \ket{\xi_0} \big)_\mathrm{out}
    \end{align}
    is a ground state of $H_\mathrm{TN}(\mathcal{P}_\gamma, \delta_\mathrm{vertex}, \delta_\mathrm{edge}, \eta)$. According to~\Cref{prop:modified-combinatorial-soundness}, the reduced state of $\ket{\Psi_{\mathcal{P}_\gamma,\ket{\xi_0}}}$ on the output qubits is $2^{-\operatorname{polylog}(n)}$-close to $\ket{\Psi_0}$.
    Therefore, the ground energy of $H$ is at most $2^{-\operatorname{polylog}(n)}$.

    \medskip\noindent \textbf{Combinatorial soundness:}
Consider any $\mu$-combinatorial state of $\ket{\psi}$ of $H_\mathrm{TN}(\mathcal{P}_\gamma, \delta_\mathrm{vertex}, \delta_\mathrm{edge}, \eta)$, then~\Cref{prop:modified-combinatorial-soundness} asserts that the reduced state on the measurement bits is $2^{-\operatorname{polylog} n}$-close to a mixture of outputs of $\mu'=\Theta(\mu)$-vertex and edge faulty executions $\{\ket{\Psi_i}\}_i$ of $\mathcal{P}'$ on the multi-graph $G'=(X(1), E')$. Next, we claim that the fault tolerance of $\cP'$ is inherited from that of $\cP_\gamma$ and this will allow us to apply~\Cref{thm:soundness-nc-protocol} to each $\ket{\Psi_i}$. 

For each set of multi-edges $E'(e) \subset E'$ that corresponds to an original edge $e$ in $X(2)$, let $w_e$ be the number of corrupted multi-edges in $E'(e)$. We say $E'(e)$ or $e$ is bad if more than $c k_e$ multi-edges in $E'(e)$ are corrupted, where $c=\Theta(1)$ is the decoding radius of the linear-distance quantum code used in the unweighting step. Recall that $\mu_2(e) \propto k_e$, therefore,
\begin{align}
    \E_{e \sim \mu_2} w_e = \mu'.
\end{align}
A Markov bound then gives
\begin{align}
    \E_{e \sim \mu_2} \mathbf{1}_{E'(e) \text{ is bad}} \leq \frac{\mu'}{c}.
\end{align}
In other words, $\mu'$ edge corruption in $\cP'$ corresponds to $\mu'/c$ edge corruption in $\cP_\gamma$. Also, $\mu'$ vertex corruption in $\cP'$ corresponds directly to $\mu'$ vertex corruption in $\cP_\gamma$ because they have the same vertex set. Because the vertices in our HDX are regular (see~\Cref{thm:hdx-construction}), $\mu'$ vertex corruption corresponds to $2 \mu'$ edge corruption by Markov inequality. Thus, the total corruption in $\cP_\gamma$ is at most $2 \mu' + \mu'/c = O(\mu')$. Therefore, the soundness analysis of $\cP_\gamma$ from~\Cref{thm:soundness-nc-protocol} applies to each $\ket{\Psi_i}$. 

If $\mu \lll \epsilon \pl n $, then~\Cref{thm:soundness-nc-protocol} asserts that for each of the states $\ket{\Psi_i}$,
\begin{align}
        \mathbb{E}_{v} \Tr(h_v \ketbra{\Psi_i}) \geq \epsilon {\pl n}.
\end{align}
As a result
\begin{align}
        \mathbb{E}_{v} \Tr(h_v \ketbra{\psi}) \geq \epsilon {\pl n} .
\end{align}
Thus, the combinatorial energy of $H$ is at least $ \ggg \epsilon \pl n$.

Next, the restriction $\epsilon_{\max}$ follows from the corruption tolerance of our local quantum MPC scheme. As shown in~\Cref{prop:ft-routing}, using existing almost-good quantum locally testable codes can unconditionally tolerate up to $1/\plll n$ fraction of global corruptions. This number becomes a constant assuming the existence of asymptotically good quantum locally testable codes in the sense of~\Cref{assumption:nice-qLTC}. 

Finally, in the commuting Hamiltonian case, we instead use the vanilla distributed protocol from~\Cref{sec:protocol-for-circuit} and tensor network Hamiltonian construction from~\Cref{sec:comm-ckt-to-H}, as well as the completeness from~\Cref{claim:completeness-commuting} and soundness guarantee from~\Cref{thm:commuting-main}. The improvement is that the ground energy is exactly zero in the completeness case by~\Cref{claim:completeness-commuting}.

\end{proof}

%% file: dl-new.tex
\section{A Robust Detectability Lemma for Circuits}\label{sec:DL-new}

In this appendix we prove two strengthenings of the detectability lemma of~\cite{AharonovALV09}. The first, \Cref{lem:ideal-detectability}, differs from the detectability lemma of~\cite{AharonovALV09} primarily in the following way. \cite[Lemma 5.1]{AharonovALV09} strictly speaking only guarantees that, for any Hamiltonian with $>\epsilon_0$ ground energy and a decomposition into constantly many layers of non-overlapping terms, \emph{there exists} one ``special" layer in the layer decomposition of the Hamiltonian such that, if the terms in that ``special" layer are simultaneously measured on any state (excluding terms from other layers), it is likely that at least $\ell$ term violations will be observed, for some $\ell$ depending on $\epsilon_0$. This does not technically say anything about the sequential measurement circuit defined in~\Cref{def:sequential-measurement-circuit-appendix2}, whose measurement statistics we would like to control, because the ``special" layer may not be the first layer, and~\cite[Lemma 5.1]{AharonovALV09} does not say anything about what may happen during a measurement of the ``special" layer when other layers have already been measured beforehand. However, a fairly routine modification of the proof of~\cite[Lemma 5.1]{AharonovALV09} is capable of proving our~\Cref{lem:ideal-detectability} as well. We undertake to do this proof, and we remark a little more fully on the differences between our statement and theirs in~\Cref{rem:DL-vs-AALV}.

The second strengthening of~\cite[Lemma 5.1]{AharonovALV09} that we need is~\Cref{lem:dl-appendix}. This is a \emph{postselection-resistant} version of~\Cref{lem:ideal-detectability}: it guarantees that the sequential measurement circuit of a Hamiltonian with large ground energy will register many violations even when the circuit is executed in the presence of adversarial, non-unitary faults, in the sense made more precise in~\Cref{def:robust-meas-circuit-appendix}. We prove~\Cref{lem:dl-appendix} using a lightcone argument that exploits the limited depth of the sequential measurement circuits that we consider.

\begin{assumption}\label{assumption:DL-Hamiltonian-appendix}
Fix constants \(k,g\) and a fixed finite set \(\mathfrak h\) of \(k\)-qubit projectors, each describable by \(O(1)\) bits. We consider Hamiltonians
\(H=\sum_{i=1}^M h_i\) on \(n\) qubits such that each \(h_i\in\mathfrak h\), and each term \(h_i\) intersects the support of at most \(g=\Theta(1)\) other terms.
\end{assumption}

\begin{definition}[Robust measurement circuits]
\label{def:robust-meas-circuit-appendix}
Consider a quantum circuit \(C\) with \(n\) input qubits and \(m\) measurement outputs. We say that \(C\) is \((\eta,\epsilon,\delta)\)-robust to adversarial faults if the following holds.

For every set \(F\) of circuit locations with \(|F|\le \eta |C|\), the adversary may corrupt the gates at locations in \(F\), and may take a coherent superposition of different fault patterns on different input states, of the form
\[
\sum_{\mathsf F\in \mathsf{PAULI}^F} c_{\mathsf F}\, C_{\mathsf F}\ket{\psi_{\mathsf F}}.
\]
Let \(q\) denote the resulting distribution on the measurement outcomes \(y\in\{0,1\}^m\). Then
\[
\Pr_{y\sim q}[\,|y|\le \epsilon m\,]\le \delta.
\]
\end{definition}

Recall the definition of a sequential measurement circuit,
\begin{definition}[Sequential measurement circuit]\label{def:sequential-measurement-circuit-appendix2}
Let \(H=\sum_{i=1}^M h_i\) satisfy \Cref{assumption:DL-Hamiltonian-appendix}. Since each term intersects only \(O(1)\) other terms, the interaction graph of the terms has bounded degree. Hence the terms can be partitioned into \(L=L(g)=O(1)\) layers
\(H=\sum_{j=1}^L H_j\), where \(H_j=\sum_{a=1}^{m_j} h_{j,a}\) and the terms inside each \(H_j\) have pairwise disjoint supports.

For a layer \(j\) and a bit string \(y_j\in\{0,1\}^{m_j}\), define
\(Q^{(j)}_{y_j}:=\prod_{a=1}^{m_j}(y_{j,a}h_{j,a}+(1-y_{j,a})(I-h_{j,a}))\).
Because the layer terms are disjoint, the family \(\{Q^{(j)}_{y_j}\}_{y_j}\) consists of pairwise orthogonal projections satisfying \(\sum_{y_j}Q^{(j)}_{y_j}=I\). For \(\mathbf y=(y_1,\dots,y_L)\), set \(Q_{\mathbf y}:=Q^{(L)}_{y_L}\cdots Q^{(1)}_{y_1}\).

The \emph{sequential measurement circuit} \(C\) is the circuit that coherently measures every layer \(H_j\) into fresh ancillas, so that on any input state \(\ket{\psi}\),
\[
C\ket{\psi}=\sum_{\mathbf y}\ket{\mathbf y}_A\otimes Q_{\mathbf y}\ket{\psi}_S.
\]
Thus the output distribution is \(\Pr[\mathbf y]=\|Q_{\mathbf y}\ket{\psi}\|^2\).
\end{definition}

\paragraph{Proof strategy and terminology}
Before proving \Cref{lem:ideal-detectability}, it is helpful to recall the overall logic of the detectability-lemma argument of Aharonov, Arad, Landau, and Vazirani~\cite{AharonovALV09}, since the proof below follows the same architecture.

The quantity we ultimately want to bound is the total weight of all measurement outcomes \(\mathbf y=(y_1,\dots,y_L)\) for which every layer reports only a small number of violations, namely the total probability of outcomes with \(|y_j|\le \ell\) for all \(j\). The reason the proof is indirect is that these outcomes are not easy to control directly from the ground-energy assumption \(\lambda_{\min}(H)\ge E_0\). The strategy is therefore to compare two bounds on the same quantity: first, a lower bound coming from the energy \(E_0\), and second, an upper bound obtained by decomposing the post-measurement states into carefully chosen local sectors. Once these two bounds are in place, one rearranges them to conclude that the low-violation outcomes must have very small total weight.

The main combinatorial object used for the upper bound is a \emph{pyramid decomposition}. Fix one layer, say layer \(1\). For each term \(t\) in that layer, one defines its \emph{pyramid} by starting from \(t\) and then growing a connected cluster downward through the later layers: at layer \(j+1\), one includes every term that intersects the part of the cluster already present in layers \(1,\dots,j\). Because the Hamiltonian has bounded locality and bounded interaction degree, every pyramid contains only \(O(1)\) constraints and acts on only \(O(1)\) qubits.

A pyramid decomposition is obtained by choosing a family of layer-1 terms whose corresponding pyramids are pairwise disjoint; these chosen layer-1 terms are the \emph{tops} of the pyramids in the decomposition. Equivalently, if one declares two layer-1 terms to conflict when their grown pyramids intersect, then a pyramid decomposition comes from choosing an independent set in this conflict graph.

The point of introducing pyramids is to separate the operator $Q_{\mathbf y}$ associated with a measurement outcome $\mathbf{y}$ into a component built from constraints lying inside the chosen pyramids and a component built from the remaining constraints. Because of how the pyramids are grown through the layers, the outside factors can be commuted past the inside factors, which allows the key step in the analysis to essentially ignore constraints outside the chosen pyramids and proceed as if only the constraints inside the chosen pyramids existed.

The next step is to understand what can happen \emph{inside a single pyramid}. For a fixed pyramid $P$, let $\mathcal H_P$ denote the Hilbert space of the qubits contained in $P$. One then considers the finite-dimensional algebra generated by the local projectors appearing in that pyramid as operators on $\mathcal H_P$, and decomposes $\mathcal H_P$ into a direct sum of invariant blocks for that algebra. Some of these blocks are completely classical in the sense that every projective term included in $P$ (every `included projection') acts on the block as either \(0\) or \(I\); these are the \(X\)-blocks. The remaining part is the genuinely quantum part, where different projectors do not simultaneously reduce to \(0/1\) values; this is the \(Y\)-space. The crucial fact is that any legal product of included projections contracts any vector inside the \(Y\)-space by a constant factor \(\theta<1\). Intuitively, the \(Y\)-space is where incompatibility between the local constraints is visible, and repeated application of the included projections suppresses amplitude there.

After this single-pyramid analysis, one passes to the full collection of pyramids. An \emph{\(XY\)-sector} is a specification, for each pyramid in the chosen family, of whether we are looking at one of its \(X\)-blocks or at its whole \(Y\)-space. Since different pyramids act on disjoint qubit sets, these local choices define a tensor-product subspace of the full Hilbert space: on each pyramid one picks the chosen local \(X\)-block or the local \(Y\)-space, and outside the pyramids one leaves the state unrestricted. The projector \(P_\nu\) onto a sector \(\nu\) is the tensor product of these local projectors. If \(|\nu|=s\), meaning that exactly \(s\) pyramids are placed in their \(Y\)-space, then the contraction property inside each \(Y\)-pyramid yields an overall suppression factor \(\theta^{2s}\) for the weight of that sector. Thus sectors with many \(Y\)-pyramids are exponentially suppressed.

This explains where the proof is heading. One decomposes each post-measurement state \(Q_{\mathbf y}\ket{\psi}\) into its \(XY\)-sectors and studies two complementary features of a sector \(\nu\). On the one hand, if \(|\nu|\) is large, then the sector is rare (has small norm) because of the \(\theta^{2|\nu|}\) suppression. On the other hand, if \(|\nu|\) is small, then most pyramids are in \(X\)-blocks, and on \(X\)-blocks the layer-1 top projectors behave classically: a top term can contribute energy when measured on \(Q_{\mathbf y}\ket{\psi}\) only if the measurement outcome \(\mathbf y\) explicitly marks that top as violated. Since we are restricting to outcomes \(\mathbf y\) with at most \(\ell\) violations in the first layer (indeed in any layer), an \(X\)-pyramid top can contribute only if it is marked violated by $\mathbf{y}$, which can happen for at most \(\ell\) pyramids, while each \(Y\)-pyramid contributes at most \(1\) to the top-layer energy. Hence the top-layer energy in sector \(\nu\) is at most \(|\nu|+\ell\). Thus the proof gets useful information in both regimes: large sectors are improbable, while small sectors cannot carry much top-layer energy.

The last stage of the argument is to turn this sector analysis into an upper bound on the total energy of all low-violation outcomes. Concretely, one wants to bound
\[
\sum_{\mathbf y:\,|y_j|\le \ell\ \forall j}\bra{\Omega_{\mathbf y}}H\ket{\Omega_{\mathbf y}},
\]
where \(\Omega_{\mathbf y}=Q_{\mathbf y}\ket{\psi}\). For one fixed family of disjoint pyramids, the sector analysis controls the part of this sum coming from the \emph{top terms} of those pyramids, that is,
\[
\sum_{\mathbf y:\,|y_j|\le \ell\ \forall j}\bra{\Omega_{\mathbf y}}H^{\mathrm{top}}_{\mathcal P}\ket{\Omega_{\mathbf y}},
\]
where \(H^{\mathrm{top}}_{\mathcal P}\) is the sum of the corresponding top constraints. A single pyramid family does not cover every term in the layer, so one repeats the construction with constantly many different families of disjoint pyramids until every term in that layer appears as a top in at least one family. Since the operators involved are all positive, the resulting bounds can simply be added, and this yields a bound for the whole first layer.

One then repeats the same reasoning for each later layer. For layer \(j\), one applies the same pyramid-and-sector argument not to the entire circuit from layer \(1\), but to the suffix beginning at layer \(j\). This gives a bound for the contribution of \(H_j\) to the same total energy sum over low-violation outcomes. Summing these bounds over all layers produces an upper bound on the full quantity
\[
\sum_{\mathbf y:\,|y_j|\le \ell\ \forall j}\bra{\Omega_{\mathbf y}}H\ket{\Omega_{\mathbf y}}.
\]

Finally, the ground-energy assumption gives a lower bound on exactly the same quantity, because every nonzero post-measurement state has energy at least \(E_0\). Comparing the lower and upper bounds forces the total probability of the low-violation outcomes to be small. This is the basic logic behind both the original detectability lemma and the variant proved here.

\begin{remark}\label{rem:DL-vs-AALV} We remark how \Cref{lem:ideal-detectability} differs from the original DL setting.
The original detectability lemma of~\cite{AharonovALV09} is formulated as a bound on the norm of a coherent product of layer projectors applied to a state. By contrast, \Cref{lem:ideal-detectability} is tailored to the sequential measurement circuit. The quantity we need to control is the total probability of all measurement outcomes in which every layer reports at most \(\ell\) violated constraints, namely the incoherent sum \(\sum_{\mathbf y:\,|y_j|\le \ell\ \forall j}\|Q_{\mathbf y}\ket{\psi}\|^2\). The proof uses the same pyramid and sector machinery, but applies it directly to this incoherent sum. This has two consequences that are important for us. First, it matches exactly the operational quantity produced by the measurement circuit. Second, it keeps an explicit exponentially small tail in \(\ell\), which is what later allows us to prove robustness against adversarial circuit faults.
\end{remark}

\begin{lemma}[Detectability lemma for sequential measurement circuit]\label{lem:ideal-detectability}
Let \(H\) satisfy \Cref{assumption:DL-Hamiltonian-appendix}, let \(L\) be as in \Cref{def:sequential-measurement-circuit-appendix2}, and let \(C\) be the corresponding sequential measurement circuit. Then there exist constants \(a=a(k,g,\mathfrak h)>0\) and \(c=c(k,g,\mathfrak h)>0\) such that the following holds.

If \(\lambda_{\min}(H)\ge E_0>0\), then for every density matrix \(\rho\) on the \(n\) system qubits and every integer \(\ell\ge 0\),
\[
p_\rho(\ell):=\sum_{\mathbf y:\ |y_j|\le \ell\ \forall j}\Tr(Q_{\mathbf y}\rho\,Q_{\mathbf y}^\dagger)
\le \frac{2^{-c\ell}}{E_0-a\ell},
\]
whenever \(E_0>a\ell\). In particular,
\[
\Pr_{y\sim C(\rho)}[\,|y|\le \ell\,]\le \frac{2^{-c\ell}}{E_0-a\ell}.
\]
\end{lemma}

\begin{proof}
By linearity in \(\rho\), it suffices to consider a pure state \(\rho=\ket{\psi}\!\bra{\psi}\). Write \(\ket{\Omega_{\mathbf y}}:=Q_{\mathbf y}\ket{\psi}\) and
\[
x^2(\ell):=\sum_{\mathbf y:\ |y_j|\le \ell\ \forall j}\|\Omega_{\mathbf y}\|^2.
\]

\medskip\noindent\textbf{Step 1: Energy lower bound.}
Since every nonzero normalized post-measurement state \(\ket{\Omega_{\mathbf y}}/\|\Omega_{\mathbf y}\|\) has energy at least \(E_0\), we have \(\bra{\Omega_{\mathbf y}}H\ket{\Omega_{\mathbf y}}\ge E_0\|\Omega_{\mathbf y}\|^2\). Summing over all \(\mathbf y\) with \(|y_j|\le \ell\) for every \(j\) gives
\[
E_0\,x^2(\ell)\le \sum_{\mathbf y:\ |y_j|\le \ell\ \forall j}\bra{\Omega_{\mathbf y}}H\ket{\Omega_{\mathbf y}}.
\]

\medskip\noindent\textbf{Step 2: Pyramids and the inside/outside factorization.}
We now upper bound the right-hand side. Fix a family \(\mathcal P\) of pairwise disjoint pyramids built from layer-1 terms as follows. Starting from a chosen set of pairwise disjoint layer-1 terms, called the \emph{tops}, grow each top downward through the layers by including, at layer \(j+1\), every term of \(H_{j+1}\) that intersects a term already included in layers \(1,\dots,j\). Since \(L\), \(k\), and the interaction degree are all \(O(1)\), each pyramid contains only \(O(1)\) terms and acts on only \(O(1)\) qubits.

For each layer \(j\), every term of \(H_j\) is either inside some pyramid in \(\mathcal P\) or outside all of them. Thus for every outcome \(y_j\), we may factor \(Q^{(j)}_{y_j}=Q^{(j)}_{y_{j,\mathrm{in}}}Q^{(j)}_{y_{j,\mathrm{out}}}\), where the first factor is the product over inside-pyramid terms and the second is the product over outside terms. By construction of pyramids, if \(j'\le j\), then every term in \(Q^{(j')}_{y_{j',\mathrm{in}}}\) is disjoint from every term in \(Q^{(j)}_{y_{j,\mathrm{out}}}\), so these operators commute. Therefore
\[
Q_{\mathbf y}=A_{y_{\mathrm{in}}}B_{y_{\mathrm{out}}},
\]
where
\[
A_{y_{\mathrm{in}}}:=Q^{(L)}_{y_{L,\mathrm{in}}}\cdots Q^{(1)}_{y_{1,\mathrm{in}}}
\qquad\text{and}\qquad
B_{y_{\mathrm{out}}}:=Q^{(L)}_{y_{L,\mathrm{out}}}\cdots Q^{(1)}_{y_{1,\mathrm{out}}}.
\]
Also, by peeling off one layer at a time, \(\sum_{y_{\mathrm{in}}}A_{y_{\mathrm{in}}}^\dagger A_{y_{\mathrm{in}}}=I\) and \(\sum_{y_{\mathrm{out}}}B_{y_{\mathrm{out}}}^\dagger B_{y_{\mathrm{out}}}=I\).

\medskip\noindent\textbf{Step 3: The \(XY\)-decomposition and contraction on the \(Y\)-space.}
Fix one pyramid \(P\in\mathcal P\). Let \(\mathcal{H}_P\) denote the Hilbert space of the qubits acted on by projectors in \(P\). Let \(\mathcal A_P\) be the finite-dimensional algebra generated by the local projectors appearing in \(P\) as operators on the Hilbert space \(\mathcal{H}_P\). Decompose \(\mathcal{H}_P\) into irreducible common invariant blocks for \(\mathcal A_P\), say \(\mathcal H_P=\bigoplus_\alpha V_{P,\alpha}\). Call a block \(V_{P,\alpha}\) an \(X\)-block if every generator of \(\mathcal A_P\) acts on it as either \(0\) or \(I\); let \(Y_P\) be the orthogonal sum of all remaining blocks. Thus \(\mathcal H_P=X_P\oplus Y_P\), and every projector appearing in $P$ is block-diagonal with respect to this decomposition. (Under this decomposition, $X_P$ contains all the common eigenvectors of $\mathcal{A}_P$, and $Y_P$ contains no such common eigenvectors.)

The number of possible pyramids is finite and each pyramid Hilbert space $\mathcal{H}_P$ has constant dimension. As such, there exists a constant \(\theta\in(0,1)\), depending only on \(k,g,\mathfrak h\), such that for every pyramid \(P\) and every possible product of the form \(W=T_L\cdots T_1\), where each \(T_j\) is a product that chooses (for each term $h_{j,a}$ in layer $j$ of the pyramid $P$) either $h_{j,a}$ or $I - h_{j,a}$ and multiplies the choices, one has
\[\|\Pi_{Y_P}W\Pi_{Y_P}\|\le \theta.\]

Indeed, if this norm were \(1\), then some unit vector \(v\in Y_P\) would satisfy \(\|Wv\|=1\). Since each \(h_{j,a}\) (or $I - h_{j,a}$) is an orthogonal projection, norm preservation through the product forces every factor in $W$ to fix \(v\). This implies that every $h_{j,a}$ in the pyramid $P$ acts on \(v\) as either \(0\) or \(I\), contradicting \(v\in Y_P\).

An \(XY\)-sector \(\nu\) chooses, independently for every pyramid \(P\in\mathcal P\), either one \(X\)-block of \(P\) or the whole space \(Y_P\). Let \(P_\nu\) denote the corresponding product projector (that projects on $X_P$ whenever $\nu$ chose the $X$-block for pyramid $P$ and $Y_P$ otherwise---the order of the product does not matter because all pyramids in $\mathcal{P}$ are disjoint), and let \(|\nu|\) be the number of pyramids in which \(\nu\) chose the \(Y\)-sector.

\medskip\noindent\textbf{Step 4: Sector-wise decay.}
Fix \(\mathbf y\) and \(\nu\). Every inside-pyramid factor is block-diagonal with respect to the \(X/Y\)-decomposition, so
\(P_\nu A_{y_{\mathrm{in}}}=A_{y_{\mathrm{in}}}P_\nu=P_\nu A_{y_{\mathrm{in}}}P_\nu\).
On an \(X\)-block, the restriction of \(A_{y_{\mathrm{in}}}\) to that block is either \(0\) or \(I\), while on a \(Y\)-block its norm is at most \(\theta\). Therefore, if on any pyramid where $\nu$ chose an $X$-block, $y_{\mathrm{in}}$ does not match the fixed outcome pattern for the $X$-block $\nu$ chose, $\|P_\nu A_{y_{\mathrm{in}}}P_\nu\| = 0$. Otherwise, we have the inequality
\[
\|P_\nu A_{y_{\mathrm{in}}}P_\nu\|\le \theta^{|\nu|}.
\]
Since \(Q_{\mathbf y}=A_{y_{\mathrm{in}}}B_{y_{\mathrm{out}}}\), it follows that
\[
\|P_\nu\Omega_{\mathbf y}\|^2\le \,\theta^{2|\nu|}\,\|P_\nu B_{y_{\mathrm{out}}}\ket{\psi}\|^2.
\]

\medskip\noindent\textbf{Step 5: Sector-wise energy upper bound.}
Let \(H^{\mathrm{top}}_{\mathcal P}\) be the sum of the layer-1 top terms of the pyramids in \(\mathcal P\). It commutes with every \(P_\nu\). For fixed \(\mathbf y\) and \(\nu\), every \(Y\)-pyramid contributes at most \(1\) to \(H^{\mathrm{top}}_{\mathcal P}\), so all \(Y\)-pyramids together contribute at most \(|\nu|\). For an \(X\)-pyramid with top term \(h_{\mathrm{top}}\), if the corresponding bit of \(y_1\) equals \(0\), then the rightmost inside-pyramid factor in \(A_{y_{\mathrm{in}}}\) for that pyramid is \(I-h_{\mathrm{top}}\), hence \(h_{\mathrm{top}}P_\nu A_{y_{\mathrm{in}}}=0\). Thus an \(X\)-pyramid can contribute only if its top term is violated in \(y_1\), and there are at most \(|y_1|\le \ell\) such pyramids. Therefore
\[
\bra{\Omega_{\mathbf y}}P_\nu H^{\mathrm{top}}_{\mathcal P}P_\nu\ket{\Omega_{\mathbf y}}
\le (|\nu|+\ell)\,\|P_\nu\Omega_{\mathbf y}\|^2.
\]

\medskip\noindent\textbf{Step 6: Summing over inside patterns and sectors.}
Let \(R=R(k,g)\) be a constant upper bound on the number of inside-pyramid terms in one layer of one pyramid. Fix a sector \(\nu\) with \(|\nu|=s\) and an outside outcome pattern \(y_{\mathrm{out}}\). Compatibility with \(\nu\) uniquely determines every inside outcome bit on the \(X\)-pyramids; these fixed bits may include violations. Only the outcome bits on the \(s\) many \(Y\)-pyramids remain free. In each layer, there are at most \(Rs\) such bits, of which at most \(\ell\) can equal \(1\). Violations already specified by the \(X\)-pyramids or by \(y_{\mathrm{out}}\) can only reduce this allowance. Hence the number of compatible inside patterns satisfying \(|y_j|\le\ell\) for all (j) is at most
\[
N_{s,\ell}:=
\left(\sum_{t=0}^{\min{\ell,Rs}}\binom{Rs}{t}\right)^L \le (C_0s)^{L\ell}
\]
for a constant \(C_0=C_0(k,g)\).

Define
\[
\Lambda_\nu^2:=\sum_{\mathbf y:\ |y_j|\le \ell\ \forall j}\|P_\nu\Omega_{\mathbf y}\|^2.
\]
Summing the sector decay bound over all compatible \(y_{\mathrm{in}}\) gives
\[
\Lambda_\nu^2\le N_{s,\ell}\,\theta^{2s}\sum_{y_{\mathrm{out}}}\|P_\nu B_{y_{\mathrm{out}}}\ket{\psi}\|^2.
\]
Now define \(\Lambda_s^2:=\sum_{|\nu|=s}\Lambda_\nu^2\) and
\(\eta_s^2:=\sum_{|\nu|=s}\sum_{y_{\mathrm{out}}}\|P_\nu B_{y_{\mathrm{out}}}\ket{\psi}\|^2\).
Then \(\Lambda_s^2\le N_{s,\ell}\theta^{2s}\eta_s^2\). Since the $XY$-sectors form an orthogonal decomposition of the full Hilbert space (i.e.~they form a projective measurement), $\sum_\nu P_\nu = I$. Then we have
\begin{align}
\sum_{s\ge0}\eta_s^2
&=\sum_{y_{\mathrm{out}}}\sum_\nu
\|P_\nu B_{y_{\mathrm{out}}}\psi\|^2\\
&=\sum_{y_{\mathrm{out}}}\|B_{y_{\mathrm{out}}}\psi\|^2\\
&=\left\langle\psi\left|
\sum_{y_{\mathrm{out}}}B_{y_{\mathrm{out}}}^\dagger
B_{y_{\mathrm{out}}}
\right|\psi\right\rangle\\
&=\langle\psi|\psi\rangle=1.
\end{align}

Since \(\log N_{s,\ell}=O(\ell\log s)\), there exist constants \(r_0=r_0(k,g,\mathfrak h)\) and \(c_0=c_0(k,g,\mathfrak h)>0\) such that for all \(s>r_0\ell\), one has \((s+\ell)N_{s,\ell}\theta^{2s}\le 2^{-2c_0s}\). Therefore
\begin{align*}
\sum_{\mathbf y:\ |y_j|\le \ell}\bra{\Omega_{\mathbf y}}H^{\mathrm{top}}_{\mathcal P}\ket{\Omega_{\mathbf y}}
&=\sum_{s\ge 0}\sum_{\mathbf y:\ |y_j|\le \ell}\bra{\Omega_{\mathbf y}}P_sH^{\mathrm{top}}_{\mathcal P}P_s\ket{\Omega_{\mathbf y}}\\
&\le \sum_{s\ge 0}(s+\ell)\Lambda_s^2\\
&\le (r_0+1)\ell\sum_{s\le r_0\ell}\Lambda_s^2+\sum_{s>r_0\ell}2^{-2c_0s}\eta_s^2\\
&\le (r_0+1)\ell\,x^2(\ell)+2^{-c_0\ell}.
\end{align*}

\medskip\noindent\textbf{Step 7: Covering one layer with constantly many pyramid families.}
We next cover all terms in one fixed layer. For a layer-1 term \(t\), let \(P(t)\) be the pyramid grown from \(t\). Two layer-1 terms conflict if \(P(t)\cap P(t')\neq\emptyset\). Since each pyramid contains only \(O(1)\) terms and each term intersects only \(O(1)\) other terms, the conflict graph on the layer-1 terms has bounded degree. Hence it can be colored using \(T=T(k,g)=O(1)\) colors. For each color class, the corresponding pyramids are pairwise disjoint, so the previous argument applies. Since every layer-1 term is the top of its own pyramid in exactly one color class, summing over the \(T\) color classes yields constants \(a_1,b_1>0\), depending only on \(k,g,\mathfrak h\), such that
\[
\sum_{\mathbf y:\ |y_j|\le \ell}\bra{\Omega_{\mathbf y}}H_1\ket{\Omega_{\mathbf y}}
\le a_1\ell\,x^2(\ell)+b_1\,2^{-c_0\ell}.
\]

\medskip\noindent\textbf{Step 8: Extending the same bound to every layer.}
The same estimate holds for every layer. Fix \(j\in\{1,\dots,L\}\). For every prefix outcome \(y_{<j}=(y_1,\dots,y_{j-1})\), define the unnormalized state
\(\ket{\psi_{y_{<j}}}:=Q^{(j-1)}_{y_{j-1}}\cdots Q^{(1)}_{y_1}\ket{\psi}\).
Applying the preceding layer-1 estimate to the truncated circuit consisting only of layers \(j,j+1,\dots,L\), with input \(\ket{\psi_{y_{<j}}}\), and then summing over all prefixes with \(|y_t|\le \ell\) for \(t<j\), gives
\[
\sum_{\mathbf y:\ |y_t|\le \ell\ \forall t}\bra{\Omega_{\mathbf y}}H_j\ket{\Omega_{\mathbf y}}
\le a_1\ell\,x^2(\ell)+b_1\,2^{-c_0\ell}.
\]
Summing over \(j=1,\dots,L\), we obtain constants \(a:=La_1\) and \(b:=Lb_1\) such that
\[
\sum_{\mathbf y:\ |y_j|\le \ell\ \forall j}\bra{\Omega_{\mathbf y}}H\ket{\Omega_{\mathbf y}}
\le a\ell\,x^2(\ell)+b\,2^{-c_0\ell}.
\]

\medskip\noindent\textbf{Step 9: Final rearrangement.}
Combining this with the lower bound from Step 1, we get
\(E_0x^2(\ell)\le a\ell\,x^2(\ell)+b\,2^{-c_0\ell}\).
Hence, whenever \(E_0>a\ell\),
\[
x^2(\ell)\le \frac{b\,2^{-c_0\ell}}{E_0-a\ell}.
\]
Since \(b\) is an absolute constant, we may decrease the exponent once and write
\[
x^2(\ell)\le \frac{2^{-c\ell}}{E_0-a\ell}
\]
for some \(c=c(k,g,\mathfrak h)>0\), as claimed.
\end{proof}

\subsection{Robust version}
In this section, we undertake to prove~\Cref{lem:dl-appendix}. We firstly prove three utility lemmas (\Cref{lem:zero-outcome-bound}, \Cref{lem:reset-locality}, and \Cref{lem:backward-fault-localization}). We then use these to prove~\Cref{lem:dl-appendix}. The main idea in this section is to use the low depth of the sequential measurement circuit to control the propagation of faults using a lightcone argument, and to observe that the lightcone argument applies even when the faults could be non-unitary coherent sums of unitary faults. The lightcone argument allows us to define a circumscribed set of output measurement bits in the sequential measurement circuit that could be corrupted, and apply~\Cref{lem:ideal-detectability} to reason about the uncorrupted measurement outcomes.

\begin{lemma}[Zero-outcome bound for the uncorrupted sequential measurement circuit]\label{lem:zero-outcome-bound}
Let \(H\) satisfy \Cref{assumption:DL-Hamiltonian-appendix}, and let \(C\) be its sequential measurement circuit. There exists a constant \(\kappa=\kappa(k,g,\mathfrak h)>0\) such that if \(\lambda_{\min}(H)\ge E_0>0\), then for every input state \(\rho\),
\[
\Pr_{y\sim C(\rho)}[y=0^M]\le \frac{1}{1+\kappa E_0}.
\]
Consequently, for every constant \(R_0>0\) there exists \(\kappa'=\kappa'(R_0,k,g,\mathfrak h)>0\) such that whenever \(0<E_0\le R_0\), one has \(\Pr_{y\sim C(\rho)}[y=0^M]\le 2^{-\kappa' E_0}\).
\end{lemma}

\begin{proof}
By linearity in \(\rho\), it suffices to prove the first claim for a pure state \(\rho=\ket{\psi}\!\bra{\psi}\).

Let \(\Pi_0:=Q^{(L)}_{0}\cdots Q^{(1)}_{0}\), let \(x_0:=\|\Pi_0\ket{\psi}\|^2\), and write \(\ket{\Omega_0}:=\Pi_0\ket{\psi}\). Then \(x_0=\Pr[y=0^M]\). Since every nonzero normalized post-measurement state has energy at least \(E_0\), we have \(E_0x_0\le \bra{\Omega_0}H\ket{\Omega_0}\).

Reuse the same pyramid and \(XY\)-sector decomposition as in the proof of~\Cref{lem:ideal-detectability}. For one fixed pyramid decomposition \(\mathcal P\), let \(H^{\mathrm{top}}_{\mathcal P}\) be the sum of the top terms of the pyramids in \(\mathcal P\). Let \(P_\nu\) be the \(XY\)-sector projections, let \(|\nu|\) be the number of \(Y\)-pyramids in \(\nu\), and write \(\Lambda_\nu^2:=\|P_\nu\Omega_0\|^2\). Exactly as in the proof of~\Cref{lem:ideal-detectability}, there exists a constant \(\mu\in(0,1)\), depending only on \(k,g,\mathfrak h\), and nonnegative numbers \(\eta_s^2\) such that \(\Lambda_s^2:=\sum_{|\nu|=s}\Lambda_\nu^2\le \mu^s\eta_s^2\) and \(\sum_{s\ge 0}\eta_s^2\le 1\).

For the zero outcome, no \(X\)-pyramid top can be violated, so the sector-wise top-energy bound simplifies to \(\bra{\Omega_0}P_\nu H^{\mathrm{top}}_{\mathcal P}P_\nu\ket{\Omega_0}\le |\nu|\,\|P_\nu\Omega_0\|^2\). Hence \(\bra{\Omega_0}H^{\mathrm{top}}_{\mathcal P}\ket{\Omega_0}\le \sum_{s\ge 1}s\,\Lambda_s^2\).

Let \(R:=\sup_{s\ge 1}\frac{s\mu^s}{1-\mu^s}<\infty\). Then
\[
\sum_{s\ge 1}s\,\Lambda_s^2
\le \sum_{s\ge 1}s\,\mu^s\eta_s^2
\le R\sum_{s\ge 1}(1-\mu^s)\eta_s^2.
\]
If \(z:=\sum_{s\ge 0}\mu^s\eta_s^2\), then \(x_0=\sum_{s\ge 0}\Lambda_s^2\le z\), and therefore \(\sum_{s\ge 1}s\,\Lambda_s^2\le R(1-z)\le R(1-x_0)\). Thus \(\bra{\Omega_0}H^{\mathrm{top}}_{\mathcal P}\ket{\Omega_0}\le R(1-x_0)\).

As in the proof of \Cref{lem:ideal-detectability}, a constant number of pyramid decompositions suffice to cover all terms in one fixed layer as pyramid tops, and applying the same argument to every suffix circuit shows that there is a constant \(A=A(k,g,\mathfrak h)\) such that \(\bra{\Omega_0}H\ket{\Omega_0}\le A(1-x_0)\). Combining this with \(E_0x_0\le \bra{\Omega_0}H\ket{\Omega_0}\) gives \(E_0x_0\le A(1-x_0)\), hence \(x_0\le \frac{A}{E_0+A}=\frac{1}{1+(E_0/A)}\). This proves the first claim with \(\kappa:=1/A\).

For the second claim, fix \(R_0>0\) and set \(\kappa':=\frac{\log_2(1+\kappa R_0)}{R_0}\). Since the function \(x\mapsto \log(1+\kappa x)/x\) is decreasing on \((0,\infty)\), for every \(0<E_0\le R_0\) we have \(\frac{1}{1+\kappa E_0}\le 2^{-\kappa' E_0}\).
\end{proof}

\begin{lemma}[Forward lightcone argument]\label{lem:reset-locality}
Let \(U\) be any unitary circuit. Let \(K\) be a set of input qubits, and let \(B\) be the forward lightcone of \(K\) at the output of \(U\). If two input states \(\rho,\rho'\) have the same reduced density matrix on the complement of \(K\), then \(U\rho U^\dagger\) and \(U\rho' U^\dagger\) have the same reduced density matrix on the complement of \(B\).
\end{lemma}

\begin{proof}
Let \(O\) be any observable supported on the complement of \(B\). Then \(U^\dagger O U\) is supported on the complement of \(K\), by definition of forward lightcone. Since \(\rho\) and \(\rho'\) have the same marginal outside \(K\), we have
\[
\Tr(O\,U\rho U^\dagger)=\Tr(U^\dagger O U\,\rho)=\Tr(U^\dagger O U\,\rho')=\Tr(O\,U\rho' U^\dagger).
\]
Thus the two output states agree on every observable supported outside \(B\), so their reduced density matrices outside \(B\) are equal.
\end{proof}

\begin{lemma}[Backward lightcone argument]\label{lem:backward-fault-localization}
Fix a unitary implementation \(U=G_T\cdots G_1\) of the sequential measurement circuit before the final computational-basis measurements of its output ancillas. There exists a constant \(\beta_{\mathrm{back}}=\beta_{\mathrm{back}}(k,g,\mathfrak h)\) such that the following holds.

Let \(F\subseteq [T]\) be any set of gate locations. Then there is a set \(K_A\) of input ancillas, with \(|K_A|\le \beta_{\mathrm{back}}|F|\), such that for every family of fault patterns supported on \(F\), every choice of coefficients \(c_\alpha\), and every family of system input states \(\ket{\psi_\alpha}\), the corresponding coherent faulty output state can be written as
\[
\sum_\alpha c_\alpha\, U_\alpha(\ket{\psi_\alpha}_S\otimes \ket{0}_A)
=
U\bigl(\ket{\Phi}_{S,K_A}\otimes \ket{0}_{A\setminus K_A}\bigr)
\]
for some state \(\ket{\Phi}_{S,K_A}\in \mathcal H_S\otimes \mathcal H_{K_A}\).

Here \(U_\alpha\) denotes the circuit obtained from \(U\) by corrupting only the gates at locations in \(F\), with each corrupted gate acting on the same qubits as the corresponding ideal gate.
\end{lemma}

\begin{proof}
Fix \(F=\{t_1<\cdots <t_s\}\subseteq [T]\). Since every ideal gate \(G_t\) is unitary, any corruption of \(G_t\) can be written as a post-gate fault: if the corrupted gate is \(E_t\), write \(E_t=D_tG_t\), where \(D_t:=E_tG_t^\dagger\) acts on the same qubits as \(G_t\). Thus, for every fault pattern \(\alpha\), the corresponding faulty circuit has the form
\[
U_\alpha
=
G_T\cdots G_{t_s+1}\, D_{t_s}^{(\alpha)} G_{t_s}\cdots G_{t_2+1}\, D_{t_2}^{(\alpha)} G_{t_2}\cdots G_{t_1+1}\, D_{t_1}^{(\alpha)} G_{t_1}\cdots G_1,
\]
where each \(D_{t_i}^{(\alpha)}\) acts on the support of \(G_{t_i}\).

For each \(i\), define the conjugated fault operator
\[
\widetilde D_{t_i}^{(\alpha)}:=G_1^\dagger\cdots G_{t_i}^\dagger\, D_{t_i}^{(\alpha)}\, G_{t_i}\cdots G_1.
\]
A direct cancellation computation shows that
\[
U_\alpha = U\, \widetilde D_{t_s}^{(\alpha)}\cdots \widetilde D_{t_1}^{(\alpha)}.
\]
Indeed,
\begin{align*}
U\, \widetilde D_{t_s}^{(\alpha)}\cdots \widetilde D_{t_1}^{(\alpha)}
&=
G_T\cdots G_1
\bigl(G_1^\dagger\cdots G_{t_s}^\dagger D_{t_s}^{(\alpha)} G_{t_s}\cdots G_1\bigr)
\cdots
\bigl(G_1^\dagger\cdots G_{t_1}^\dagger D_{t_1}^{(\alpha)} G_{t_1}\cdots G_1\bigr)\\
&=
G_T\cdots G_{t_s+1} D_{t_s}^{(\alpha)} G_{t_s}\cdots G_{t_2+1} D_{t_2}^{(\alpha)} G_{t_2}\cdots G_{t_1+1} D_{t_1}^{(\alpha)} G_{t_1}\cdots G_1,
\end{align*}
which is exactly \(U_\alpha\).

For a location \(t\), let \(\mathrm{BL}(t)\) denote its backward lightcone in the ideal circuit \(U\), and let \(\partial_{\mathrm{in}}\mathrm{BL}(t)\) be the set of input wires entering that lightcone. By construction, \(\widetilde D_t^{(\alpha)}\) is supported on \(\partial_{\mathrm{in}}\mathrm{BL}(t)\). Let \(K_A\) be the union of the input ancillas contained in \(\partial_{\mathrm{in}}\mathrm{BL}(t)\) over all \(t\in F\).

Because \(U\) is a constant-depth circuit of bounded-size gates, each single-location backward lightcone meets only \(O_{k,g,\mathfrak h}(1)\) input ancillas. Therefore \(|K_A|\le \beta_{\mathrm{back}}|F|\) for some constant \(\beta_{\mathrm{back}}\) depending only on \(k,g,\mathfrak h\).

Now fix a fault pattern \(\alpha\). Since \(\widetilde D_{t_s}^{(\alpha)}\cdots \widetilde D_{t_1}^{(\alpha)}\) acts only on the system register together with the ancillas in \(K_A\), there exists a state \(\ket{\phi_\alpha}_{S,K_A}\in \mathcal H_S\otimes \mathcal H_{K_A}\) such that
\[
\widetilde D_{t_s}^{(\alpha)}\cdots \widetilde D_{t_1}^{(\alpha)}
(\ket{\psi_\alpha}_S\otimes \ket{0}_A)
=
\ket{\phi_\alpha}_{S,K_A}\otimes \ket{0}_{A\setminus K_A}.
\]
Equivalently,
\[
U_\alpha(\ket{\psi_\alpha}_S\otimes \ket{0}_A)
=
U\bigl(\ket{\phi_\alpha}_{S,K_A}\otimes \ket{0}_{A\setminus K_A}\bigr).
\]

Therefore
\[
\sum_\alpha c_\alpha\, U_\alpha(\ket{\psi_\alpha}_S\otimes \ket{0}_A)
=
U\Bigl(\sum_\alpha c_\alpha\, \ket{\phi_\alpha}_{S,K_A}\otimes \ket{0}_{A\setminus K_A}\Bigr).
\]
Since the factor \(\ket{0}_{A\setminus K_A}\) is independent of \(\alpha\), this equals
\[
U\bigl(\ket{\Phi}_{S,K_A}\otimes \ket{0}_{A\setminus K_A}\bigr),
\]
where
\[
\ket{\Phi}_{S,K_A}:=\sum_\alpha c_\alpha\, \ket{\phi_\alpha}_{S,K_A}\in \mathcal H_S\otimes \mathcal H_{K_A}.
\]
This proves the claim.
\end{proof}

\begin{lemma}[Robust detectability lemma]\label{lem:dl-appendix}
Let \(H\) satisfy \Cref{assumption:DL-Hamiltonian-appendix}, and assume \(\lambda_{\min}(H)\ge \epsilon M\). Then the sequential measurement circuit of \(H\) is \((\epsilon/r,\epsilon/r',2^{-r''\epsilon M})\)-robust (in the sense of \Cref{def:robust-meas-circuit-appendix}) for constants \(r,r',r''>0\) depending only on the constants in \Cref{assumption:DL-Hamiltonian-appendix}.
\end{lemma}

\begin{proof}
Let \(U\) be a fixed unitary implementation of the sequential measurement circuit before the final computational-basis measurements of the output ancillas, and let \(C\) denote the resulting measurement circuit. Since \(\mathfrak h\) is finite and every term acts on \(k=O(1)\) qubits, each term-measurement gadget can be implemented using \(O_{k,\mathfrak h}(1)\) elementary gates, so \(|C|\le \gamma M\) for a constant \(\gamma=\gamma(k,\mathfrak h)\).

Let \(\beta_{\mathrm{back}}\) be the constant from \Cref{lem:backward-fault-localization}. Since \(U\) has constant depth and bounded-size gates, there is also a constant \(\beta_{\mathrm{forw}}=\beta_{\mathrm{forw}}(k,g,\mathfrak h)\) such that the forward lightcone of any single input ancilla contains at most \(\beta_{\mathrm{forw}}\) output measurement ancillas. Set \(\beta:=\max\{\beta_{\mathrm{back}},\beta_{\mathrm{forw}}\}\).

Let \(a,c>0\) be the constants from \Cref{lem:ideal-detectability}, and let \(\kappa>0\) be the constant from \Cref{lem:zero-outcome-bound}. Choose a constant \(r'>0\) with \(r'\ge \max\{6a,4\}\), then set \(r:=2\beta r'\). Finally define
\[
r''_1:=\frac{\log_2(1+\kappa r'/2)}{r'/2},\qquad
r''_2:=\frac{c}{2r'},\qquad
r'':=\min\{r''_1,r''_2\}.
\]

Fix any set \(F\) of faulty locations with \(|F|\le \frac{\epsilon}{r}|C|\). By \Cref{lem:backward-fault-localization}, there is a set \(K_A\) of input ancillas with \(|K_A|\le \beta|F|\) such that the corresponding coherent faulty output state is of the form
\[
U\bigl(\ket{\Phi}_{S,K_A}\otimes \ket{0}_{A\setminus K_A}\bigr)
\]
for some \(\ket{\Phi}_{S,K_A}\in \mathcal H_S\otimes \mathcal H_{K_A}\). Let
\[
\rho:=\bigl(\ket{\Phi}\!\bra{\Phi}\bigr)_{S,K_A}\otimes \ket{0}\!\bra{0}_{A\setminus K_A}^{\otimes},
\]
that is, \(\rho\) is the pure input state
\(\ket{\Phi}_{S,K_A}\otimes \ket{0}_{A\setminus K_A}\).

Now define
\[
\rho':=\Tr_{K_A}(\rho)\otimes \ket{0}\!\bra{0}^{\otimes K_A}.
\]
Then \(\rho\) and \(\rho'\) have the same reduced density matrix on the complement of \(K_A\). Let \(B\) be the set of output measurement ancillas in the forward lightcone of \(K_A\). By \Cref{lem:reset-locality}, the output states \(U\rho U^\dagger\) and \(U\rho' U^\dagger\) have the same reduced density matrix on the complement of \(B\). Hence the two output distributions agree on the measurement bits outside \(B\).

Since each ancilla in \(K_A\) influences at most \(\beta\) output measurement ancillas, \(|B|\le \beta |K_A|\le \beta^2 |F|\). Using \(|F|\le (\epsilon/r)|C|\) and \(|C|\le \gamma M\), and absorbing the constant \(\gamma\) into \(\beta\), we may write
\[
|B|\le \beta\frac{\epsilon}{r}M=\frac{\epsilon M}{2r'}.
\]

Let \(x:=\epsilon M\) and \(t:=x/r'\). If \(y\) denotes the full measurement string and \(y_{\overline B}\) its restriction outside \(B\), then
\[
\Pr_{\mathrm{faulty}}[|y|\le t]
\le \Pr_{\mathrm{faulty}}[|y_{\overline B}|\le t]
= \Pr_{C(\rho')}[|y_{\overline B}|\le t]
\le \Pr_{C(\rho')}[|y|\le t+|B|].
\]
Set \(\ell:=t+|B|\). Then \(\ell\le \frac{x}{r'}+\frac{x}{2r'}=\frac{3x}{2r'}\).

We split into two cases.

If \(x\le r'/2\), then \(\ell\le 3/4<1\). Since \(|y|\) is an integer, the event \(\{|y|\le \ell\}\) is exactly the all-zero event \(\{y=0^M\}\). Hence
\[
\Pr_{\mathrm{faulty}}[|y|\le t]
\le \Pr_{C(\rho')}[y=0^M]
\le \frac{1}{1+\kappa x}
\le 2^{-r''_1x}
\le 2^{-r''x}.
\]

If \(x\ge r'/2\), then since \(r'\ge 6a\), we have \(a\ell\le a\cdot \frac{3x}{2r'}\le x/2\), so \(x-a\ell\ge x/2\). Applying \Cref{lem:ideal-detectability} to the ideal circuit on input \(\rho'\), with ground energy \(E_0=x\), gives
\[
\Pr_{C(\rho')}[|y|\le \ell]\le \frac{2^{-c\ell}}{x-a\ell}\le \frac{2^{-cx/r'}}{x/2}.
\]
Because \(x\ge r'/2\) and \(r'\ge 4\), we have \(\frac{2}{x}\le \frac{4}{r'}\le 1\le 2^{cx/(2r')}\). Therefore
\[
\Pr_{C(\rho')}[|y|\le \ell]\le 2^{-cx/r'}\cdot 2^{cx/(2r')}=2^{-cx/(2r')}=2^{-r''_2x}\le 2^{-r''x}.
\]

Thus in both cases, \(\Pr_{\mathrm{faulty}}[|y|\le t]\le 2^{-r''x}\). Since \(x=\epsilon M\) and \(t=x/r'=(\epsilon/r')M\), this shows that for every fault set of size at most \((\epsilon/r)|C|\),
\[
\Pr[|y|\le (\epsilon/r')M]\le 2^{-r''\epsilon M}.
\]
This is exactly the claimed \((\epsilon/r,\epsilon/r',2^{-r''\epsilon M})\)-robustness.
\end{proof}